\documentclass[12pt]{article} 

\usepackage[a4paper, total={7in, 8.5in}]{geometry}

\usepackage{booktabs}

\usepackage{makecell}

\usepackage[sectionbib]{natbib}
\usepackage[font=small,labelfont=bf]{caption}
\usepackage[pdfencoding=auto,
            colorlinks = true,
            urlcolor  = blue,
            ]{hyperref}

\usepackage{comment}
\usepackage{todonotes}

\usepackage{graphicx}
\usepackage{float}
\usepackage{tikz}

\usepackage{xcolor}

\definecolor{lightblue}{RGB}{173, 216, 230} 
\colorlet{lightblueDark}{lightblue!80!black}

\usepackage{enumerate}
\usepackage{enumitem}
\usepackage{algorithm}
\usepackage{algpseudocode}

\usepackage{amsmath,amssymb,amsthm,mathtools}
\mathtoolsset{showonlyrefs=true}
\usepackage{etoolbox}

\DeclareMathOperator*{\argmax}{argmax}

\newcommand{\supp}[1]{\textbf{supp}\left(#1\right)}
\newcommand{\norm}[1]{\left\lVert#1\right\rVert}
\newcommand{\abs}[1]{\left\lvert#1 \right\rvert}

\DeclareFontFamily{U}{matha}{\hyphenchar\font45}
\DeclareFontShape{U}{matha}{m}{n}{
<-6> matha5 <6-7> matha6 <7-8> matha7
<8-9> matha8 <9-10> matha9
<10-12> matha10 <12-> matha12
}{}
\DeclareSymbolFont{matha}{U}{matha}{m}{n}

\DeclareFontFamily{U}{mathx}{\hyphenchar\font45}
\DeclareFontShape{U}{mathx}{m}{n}{
<-6> mathx5 <6-7> mathx6 <7-8> mathx7
<8-9> mathx8 <9-10> mathx9
<10-12> mathx10 <12-> mathx12
}{}
\DeclareSymbolFont{mathx}{U}{mathx}{m}{n}

\DeclareMathDelimiter{\vvvert} {0}{matha}{"7E}{mathx}{"17}%

\DeclarePairedDelimiterX{\normiii}[1]
{\vvvert}
{\vvvert}
{\ifblank{#1}{\:\cdot\:}{#1}}

\newtheorem{theorem}{Theorem} 
\newtheorem{assumption}{Assumption}
\newtheorem{lemma}[theorem]{Lemma} 
\newtheorem{proposition}[theorem]{Proposition} 

\newtheorem{corollary}[theorem]{Corollary}
\newtheorem{definition}[theorem]{Definition}

\def\ba{{\mathbf a}}

\def\bom{{\mathbf m}}
\def\bof{{\mathbf f}}
\def\bl{{\mathbf z}}
\def\bw{{\mathbf w}}
\def\bb{{\mathbf b}}
\def\bc{{\mathbf c}}
\def\bW{{\mathbf W}}
\def\bx{{\mathbf x}}
\def\bX{{\mathbf X}}
\def\by{{\mathbf y}}
\def\bY{{\mathbf Y}}

\newcommand{\boeta}{\boldsymbol{\eta}}
\newcommand{\bPi}{\boldsymbol{\Pi}}
\newcommand{\bpi}{\boldsymbol{\pi}}
\newcommand{\bmu}{\boldsymbol{\mu}}
\newcommand{\bnu}{\boldsymbol{\nu}}
\newcommand{\blambda}{\boldsymbol{\lambda}}

\title{Scalable calibration of individual-based epidemic models through categorical approximations}

\date{\today}

\author{Lorenzo Rimella \\
        \small{Dipartimento di Scienze Economiche, Università degli studi di Bergamo, Bergamo, IT}\\
        Nick Whiteley \\
        \small{School of Mathematics, University of Bristol, Bristol, UK}\\
        Chris Jewell \\
        \small{School of Mathematical Sciences, Lancaster University, Lancaster, UK}\\
        Paul Fearnhead \\
        \small{School of Mathematical Sciences, Lancaster University, Lancaster, UK}\\
        Michael Whitehouse \\
        \small{School of Public Health, Imperial College London, London, UK}}
\begin{document}
 \maketitle


\begin{abstract}
    Traditional compartmental models capture population-level dynamics but fail to characterize individual-level risk. The computational cost of exact likelihood evaluation for partially observed individual-based models, however, grows exponentially with the population size, necessitating approximate inference. Existing sampling-based methods usually require multiple simulations of the individuals in the population and rely on bespoke proposal distributions or summary statistics. We propose a deterministic approach to approximating the likelihood using categorical distributions. The approximate likelihood is amenable to automatic differentiation so that parameters can be estimated by maximization or posterior sampling using standard software libraries such as Stan or TensorFlow with little user effort. We prove the consistency of the maximum approximate likelihood estimator. We empirically test our approach on several classes of individual-based models for epidemiology: different sets of disease states, individual-specific transition rates, spatial interactions, under-reporting and misreporting. We demonstrate ground truth recovery and comparable marginal log-likelihood values at substantially reduced cost compared to competitor methods. Finally, we show the scalability and effectiveness of our approach with a real-world application on the 2001 UK Foot-and-Mouth outbreak, where the simplicity of the CAL allows us to include 162775 farms.
\end{abstract}

\section{Introduction} 

The traditional approach to modeling epidemics involves dividing a population into compartments representing epidemic states through which individuals progress over time. The dynamics of such models are mathematically represented by, for instance, a Markov chain \citep{lekone2006statistical}, a system of ODEs \citep{chowell2004basic}, or a system of SDEs \citep{ionides2006inference}. These models treat the population as homogeneous in two senses: individuals are equal in their epidemiological attributes, such as susceptibility and infectivity (\emph{homogeneous attributes}); and individuals contact each other at equal rates  (\emph{homogeneous-mixing}). These homogeneity assumptions can be relaxed by stratifying individuals into meta-populations \citep{whitehouse2023consistent}, where individuals have homogeneous attributes and are homogeneous-mixing within each meta-population but they have heterogeneous attributes and are heterogeneous-mixing across different meta-populations. While homogeneous\slash heterogeneous-mixing is a well-known terminology in the epidemiology literature \citep{jewell2009bayesian}, homogeneous\slash heterogeneous attributes is somewhat uncommon but we use it to emphasize the different levels of heterogeneity we are considering.

Individual-based models (IBMs) fully relax the homogeneity assumptions by modeling each individual's disease states explicitly, instead of partitioning and aggregating them across different compartments. The interpretation of ``individual'' is context-specific and could refer, for example, to an individual person \citep{bu2024stochastic} or farm \citep{jewell2009bayesian}. Whilst stratified compartmental models may handle covariates taking values in a discrete set, IBMs allow practitioners to work with both continuous and discrete covariates defined at an individual level, such as physical and physiological characteristics (age, health records, etc. \citep{cocker2023drum}), location (geographical and\slash or community membership \citep{jewell2009bayesian,rimella2022inference}) or contact networks \citep[Chapter~4]{estrada2010network}. The increasing availability of both epidemiological testing data and accompanying covariates at an individual level has motivated the development of many fine-scale IBMs of disease transmission \citep{cocker2023drum}. However, the computational cost of performing exact inference for such IBMs with noisy and incomplete data necessarily grows astronomically in the population size $N$, with exact computation of likelihoods having complexity which is exponential in $N$ \citep{rimella2025simulation}. Fast, simple, and theoretically justified calibration of IBMs is one of the open challenges in infectious disease modeling, motivating the present work.

We now summarize our novel contributions. \textbf{(1)} We propose a new form of approximate and recursive likelihood evaluation in partially observed individual-based epidemic models via categorical distributions, which involve no simulation from the model. \textbf{(2)} We prove strong consistency of parameter estimators obtained by maximizing our approximate likelihood when data are generated by the exact model. \textbf{(3)} The computational simplicity of our methodology allows it to scale up to large population sizes, even with a simple Python implementation. Moreover, it is particularly suited to automatic differentiation, for example gradient-based optimization or Hamiltonian Monte Carlo. \textbf{(4)} We calibrate an individual-based model to the 2001 Aphtovirus (Foot-and-Mouth) outbreak in the United Kingdom (UK), scaling up to include 162775 farms in the study. 

The paper is organized as follows. We conclude this section with a motivating example and related works. In Section \ref{sec:model_dynamic_obs}, we introduce our model, show how the motivating example can be formulated in that framework, and discuss closed-form likelihood computation. Section \ref{sec:CAL} presents the main algorithm and explains the rationale behind the approximation. In Section \ref{sec:theory_main}, we state our consistency result, Theorem \ref{thm:main_consistency}, and outline its proof. Numerical results on both synthetic and real data are reported in Section \ref{sec:exper}. Section \ref{sec:discussion_future} summarizes the paper and discusses limitations and future work.

\subsection{Motivating example} \label{sec:motivating_example_first}

Consider a discrete-time Susceptible-Infective-Susceptible (SIS) model for a population of size $N$. At each discrete time step $t=0,1,\ldots$, each individual assumes one of the disease states $\{S,I\}$. The following example could easily be extended to models with more than two disease states, we focus on SIS for ease of presentation. We present alternative building blocks of the model in separate paragraphs. Similar models incorporating individual-specific covariates have been considered in previous works, for example, by \cite{ju2021sequential,bu2022likelihood,rimella2022inference,bu2024stochastic}.

\paragraph{Heterogeneous attributes.}

IBMs allow individuals' covariates to be reflected in disease state transition probabilities. Let $\beta_{n k} > 0 $ denote the rate at which the $k$-th individual infects the $n$-th when they come into contact, assuming the former is infective and the latter is susceptible. One may consider a regression model for the logarithm of the pairwise individual-specific transmission rate, e.g. $\label{eq:beta_regression model}
\log \beta_{n k} = \log \beta + \bc_n^\top \bb_{S} + \bc_{k}^\top \bb_{I},
    $ where $\beta$ is the background infection rate, $\bc_n$ is a vector of observed covariates associated with the $n$-th individual, and $\bb_S, \bb_I$ are parameter vectors respectively determining susceptibility to infection and propensity to pass the infection on. Similarly, one may consider a regression model for the recovery rate $\gamma_n > 0$ at which an infective individual recovers and returns to being susceptible: $\log \gamma_n = \log \gamma + \bc_n^\top \bb_R$, where $\bb_R$ is a parameter vector.

\paragraph{Homogeneous- and heterogeneous-mixing dynamics.}

Under the assumption that the population mixes homogeneously in discrete time, the $n$-th individual is equally likely to contact any one individual, and the probability of the transition $S\to I$ at time $t>0$ is:
\begin{equation}
1-\exp\left(-\frac{h}{N}\sum_{k\in \mathcal{I}_{t-1}}\beta_{n k}\right) = 1-\exp\left(-h\beta\exp\{\bc_n^\top \bb_S\}\frac{1}{N}\sum_{k\in \mathcal{I}_{t-1}}
\exp\{\bc_{k}^\top \bb_I\} \right),\label{eq:homog_mixing}
\end{equation}
where $h>0$ is the time period length (set equal to $1$ unless stated otherwise) and $\mathcal{I}_{t-1}$ is the set of individuals who are infective at time $t-1$. Here $\exp\{\bc_n^\top \bb_S\}$ has the interpretation of the susceptibility of the $n$-th individual, while $\exp\{\bc_k^\top \bb_I\}$ is the infectivity. At each time step, the probability of the $n$-th individual transitioning $I\to S$ is:
$
1-\exp(-h\gamma_n).
$

In some cases, information may be available about the geographic location of individuals, network structure, or other factors that influence transmission rates between pairs of individuals. For example, if $\bl_n$ denotes the Euclidean position of the $n$-th individual, then spatial weighting could be introduced in the $S\to I$ transition probability \eqref{eq:homog_mixing} in  the form:
$$
1-\exp\left(-h\beta\exp\{\bc_n^\top \bb_{S}\}\frac{1}{N}\sum_{k\in\mathcal{I}_{t-1}}\frac{1}{\sqrt{2 \pi \phi^2}} \exp\left\{-\frac{\|\bl_n-\bl_{k}\|^2}{2\phi^2}+\bc_{k}^\top \bb_{I}\right\}\right),
$$
where $\phi>0$ is a parameter and $\norm{\cdot}$ is the Euclidean distance. Here, the term $\frac{1}{N}\frac{1}{\sqrt{2 \pi \phi^2}}\exp\left\{-\frac{\|\bl_n-\bl_{k}\|^2}{2\phi^2}\right\}$ can be interpreted as the rate at which the $n$-th individual contacts the $k$-th individual, whilst the other terms are as in the homogeneous-mixing case. One could similarly model heterogeneity arising from a known network rather than spatial structure, by choosing the $S\to I$ transition probability for individual $n$ to reflect its connectivity to other members of the population. 

\paragraph{Observations.} At each time step $t$, an individual is either reported in their true state, misreported in some other state, or not reported at all: in our SIS model, the observed state of each individual is therefore one of  $\{U,S,I\}$, where $U$ represents being unreported (e.g. missing test results). For any individual, let $q_S$ (resp. $q_I$) be the probability of either correctly reporting or misreporting their state, given that they are susceptible (resp. infected). Conditional that the state of the individual is either reported or misreported, let $q_{Se}$ (resp. $q_{Sp}$) be the probability of observing $I$ if $I$ (resp. $S$ if $S$) is the true state. In the context of testing  $q_{Se}$ and $q_{Sp}$ are the sensitivity and specificity.

\paragraph{Inference challenges.} 

A typical inference task would be to calibrate the model by estimating the parameters, $\beta, \bb_S,\bb_I,\gamma, \bb_R,\phi,q_S,q_I,q_{Se},q_{Sp}$ or some subset thereof, allowing us to understand how the covariates $\bc_n$ and spatial or network interactions contribute to the dynamics of the disease in question \citep{rimella2022inference,seymour2022bayesian}. 
Due to the partial observation structure of the above model, for population size and time horizon $N,T\in \mathbb{N}$ exact likelihood evaluation would involve marginalizing over $2^{NT}$ possible latent states. 

\subsection{Related work}
The literature on IBMs is vast, and a full review would be impossible within the length constraints of the present work, here we present a brief summary. Inference for partially observed stochastic epidemic models is difficult, and, even when homogeneity assumptions are permitted, simplifications \citep{king2015avoidable} or simulation-based procedures \citep{ionides2006inference} are required. Analogously, many studies of partially observed IBMs make simplifications or approximations, e.g. by deterministic modeling \citep{sharkey2008deterministic}, mean-field approximations \citep{sherborne2018mean}, or noiseless observation mechanisms \citep{deardon2010inference}. Sophisticated simulation-based techniques have also been developed for inference in IBMs: bespoke proposals for sequential Monte Carlo \citep{rimella2022approximating}, approximate Bayesian computation procedures \citep{mckinley2018approximate}, composite likelihood methods \citep{rimella2025simulation}, data augmentation schemes \citep{bu2022likelihood, bu2024stochastic}, neural posterior estimation \citep{chatha2024neural}, and Markov chain Monte Carlo samplers \citep{touloupou2020scalable}; alongside Bayesian non-parametric approaches \citep{seymour2022bayesian} and kernel-linearization with imputation based techniques \citep{deardon2010inference}. Typically, when applied to noisy observations, it is hard to effectively scale the above methods to large population sizes.

The approximation techniques in the present work extend ideas for homogeneous population compartmental models  \citep{whiteley2021inference, whitehouse2023consistent, whitehouse2025accelerated} to the case of individual-based models. This approach of making distributional approximations is related to assumed density filtering \citep{sorenson1968non} and expectation propagation algorithms \cite[Ch.1]{minka2001family}, but the details are different and specially designed to exploit the structure of the individual-based model. Furthermore, as far as the authors are aware, this is the first work to provide results on the consistency of parameter estimates for IBMs under analysis of the large population regime.

\section{Individual-based compartmental model}\label{sec:model_dynamic_obs}

\subsection{Notation}
Given $M \in \mathbb{N}$, we use  $x_{0:M} \coloneqq x_0,\dots,x_M$ for indexing sequences, $[M]\coloneqq \left \{1,\dots, M \right \}$ for the set of the first $M$ integers, and $\bx \coloneqq \left [\bx^{(1)},\dots,\bx^{(M)} \right ]^\top$ for an $M$-dimensional vector. Given two $M$-dimensional vectors $\bx_1,\bx_2$, or $M\times M$-dimensional matrices, we denote with $\bx_1 \odot \bx_2$ the element-wise product and with $\bx_1 \oslash \bx_2$ the element-wise division, and we use $[\bx_1,\bx_2]$ for the vector stacking together $\bx_1,\bx_2$, i.e. $[\bx_1,\bx_2]\coloneqq \left [\bx_1^{(1)}, \dots, \bx_1^{(M)},\bx_2^{(1)},\dots,\bx_2^{(M)}\right ]^\top$. We write $1_M$ for the $M$-dimensional vector of all ones, $\Delta_M$ for the $M$-dimensional probability simplex, i.e. $\Delta_M \coloneqq \left \{\bx \in [0,1]^M \text{ : } \sum_{i=1}^M \bx^{(i)}=1 \right \}$, and $\mathbb{O}_M$ for the set of one-hot encoding vectors with dimension $M$, i.e. $\mathbb{O}_M \coloneqq  \{\bx \in \{0,1\}^M \text{ : } \exists j \in [M]:\bx^{(j)}=1 \text{ and } \bx^{(i)}=0 \text{ if } i\neq j  \}$, with $\mathbb{O}_M \subset \Delta_M$. Given $\bpi \in \Delta_M$ we denote with $\mbox{Cat}(\cdot|\bpi)$ the categorical distribution over $\mathbb{O}_M$ which assigns probability $\bpi^{(i)}$ to the vector $\bx \in \mathbb{O}_M$ with $\bx^{(i)}=1$ and $\bx^{(j)}=0$ for $j\neq i$. 

\subsection{Model}\label{sec:model}

We now introduce a generic form of individual-based model. We consider a population of $N \in \mathbb{N}$ individuals and assume that a vector of known covariates $\bw_n \in \mathbb{W}$ is associated with individual $n \in [N]$, where $\mathbb{W}$ is a subset of Euclidean space. We denote by $\mathbf{W}$ the collection of covariates of the entire population, $\mathbf{W} \coloneqq (\bw_1, \dots, \bw_N)$. These covariates allow us to express heterogeneity in how the disease propagates through the population. Each individual assumes any one of $M \in \mathbb{N}$ latent disease states at any one discrete time step $t\geq 0$. We use one-hot encoding vectors to represent the states of the individuals; this is a little non-standard but it will simplify mathematical expressions. The state of the $n$-th individual at time $t$ is denoted  $\bx_{n,t} \in \mathbb{O}_M$, meaning that the $i$-th component of $\bx_{n,t}$ is $1$ if and only if the $n$ individual is in state $i$ at time $t$. The state of the entire population is written $\bX_t \coloneqq (\bx_{1,t},\dots,\bx_{N,t})$.  With the covariates $\mathbf{W}$ fixed, the process $(\bX_t)_{t \geq 0}$ is a Markov chain and the individual disease states $\bx_{n,t}$ are distributed as follows.

\paragraph{Latent dynamics} At time step $t=0$, the state of each individual is drawn independently from an initial distribution which is a function of the individual-specific covariates $\bx_{n,0} | \bw_n \sim \mbox{Cat} \left (\,\cdot\,|p_0( \bw_n) \right )$,
for some probability vector $p_0(\bw_n)$.  At time steps $t\geq 1$, conditional on the population state $\bX_{t-1}$, the $n$-th individual evolves according to an $M\times M$ row-stochastic transition matrix $K_{\boeta_{n,t}}(\bw_n)$:
\begin{align}
	&\bx_{n,t}|\bX_{t-1}, \bW  \sim \mbox{Cat} \left ( \,\cdot\,\left| \left [ \bx_{n,t-1}^\top K_{\boeta_{n,t-1}}(\bw_n) \right ]^\top \right.\right ).
\end{align}
Here $\boeta_{n,t} \coloneqq \eta \left ( \bw_n, \bW, \bX_t \right )$ where $\eta:\mathbb{W} \times \mathbb{W}^{N} \times \Delta_M^N \to [0,C]$ with $C \in \mathbb{R}_+$ is a function which will allow us to express how individuals interact with the population in terms of their respective covariates and disease states.  In this formulation, the transition matrix $K_{\boeta_{n,t}}(\bw_n)$ depends on $t$ only via $\bX_t$. We consider this case for ease of presentation; our model, algorithm, and theory can be extended to transition matrices that evolve over time.

\paragraph{Observations} At time $t\geq 1$, we observe a collection of vectors $\bY_t \coloneqq (\by_{1,t},\dots,\by_{N,t})$, where each $\by_{n,t}$ is a $\mathbb{O}_{M+1}$-valued random  measurement associated with the $n$-th individual. Given $\bX_t$, $\by_{1,t},\dots,\by_{N,t}$ are conditionally independent and distributed:
\begin{align}
    \by_{n,t}|\bx_{n,t},\bw_n \sim \mbox{Cat} \left ( \,\cdot\,\left| \left [ \bx_{n,t}^\top  G(\bw_n) \right.\right ]^\top \right ),
\end{align}
where $G(\bw_n)$ is a $M \times (M+1)$-dimensional row-stochastic matrix. This matrix allows probabilities to be assigned to the $n$-th individual being: unreported, representing the extra compartment; correctly reported as in the disease state specified $\bx_{n,t}$; erroneously reported as assuming one of the other $M-1$ disease states. The matrix $G(\bw_n)$ could be allowed to depend on $t$ with only notational changes to our algorithm and theory needed. We also assume the observations are evenly spaced in time, i.e. a period length of $h=1$, but, once again, this is just for presentation purposes.

\subsection{Motivating example} \label{sec:motivating_example_second}

We now show how the motivating example from Section \ref{sec:motivating_example_first} can be cast as an instance of the generic model described in Section \ref{sec:model}. As it is an SIS model, $M=2$, $\bx_{n,t} \in \{[1,0]^\top,[0,1]^\top\}$ is a $2$-dimensional one-hot encoding vector representing disease states $\{S,I\}$. The individual-specific covariates are $\bw_n=[\bc_n,\bl_n]$, where the latter are as in Section \ref{sec:motivating_example_first}. As initial infection probabilities, we consider: $p_0(\bw_n) = [ 1-p_0,\;  p_0 ]^\top$, i.e. each individual has the same probability of being infected at the beginning of the epidemic.

\paragraph{Homogeneous- and heterogeneous-mixing dynamics.}

Following the formulation from Section \ref{sec:model}, we can express homogeneous- and heterogeneous-mixing dynamics by reformulating the interaction term $\boeta_{n,t-1}$. For the homogeneous-mixing case we can write $
\boeta_{n,t-1}
    =
    \frac{1}{N} \sum_{k \in [N]}
    \exp\{\bc_{k}^\top \bb_I\} \bx_{k,t-1}^{(2)},
$
while for the heterogeneous-mixing case we have $
    \boeta_{n,t-1}
    =
    \frac{1}{N} \sum_{k \in [N]} 
    \exp\{\bc_{k}^\top \bb_I\} \frac{\exp\left\{-\frac{\|\bl_n-\bl_{k}\|^2}{2\phi^2}\right\}}{\sqrt{2 \pi \phi^2}} \bx_{k,t-1}^{(2)}.$
Then in either case we can write:
\begin{equation}\label{eq:HetMix_K}
K_{\boeta_{n,t-1}}(\bw_n) 
    = 
    \begin{bmatrix} 
    \exp{ \left (-h\beta\exp\{\bc_n^\top \bb_{S}\} \boeta_{n,t-1} \right ) }
    &
    1 - \exp{ \left (-h\beta\exp\{\bc_n^\top \bb_{S}\} \boeta_{n,t-1} \right ) }
    \\ 
    1 - \exp{ \left ( -h\gamma_n \right )}
    &
    \exp{ \left ( -h\gamma_n  \right )}
    \end{bmatrix}.
\end{equation}

\paragraph{Observation model.} 
The observation for each individual $\by_{n,t}$ is a $3$-dimensional one-hot encoding vector representing the states: not reported (e.g. missing test results), reported as $S$, and reported as $I$. In this  scenario, a stochastic matrix for our observation model is:
$$
    G(\bw_n) 
    =
    \begin{bmatrix}
    1 - q_S 
    & 
    q_S q_{Sp}  
    & 
    q_S (1-q_{Sp})\\
    1 - q_I 
    & 
    q_I (1- q_{Se}) 
    & 
    q_I q_{Se}
    \end{bmatrix},
$$
where $q_S,q_I,q_{Se},q_{Sp} \in [0,1]$ represent the reporting probabilities when $S$ and when $I$, and the sensitivity and specificity of the test. More generally, these probabilities could be a function of covariates $\mathbf{w}_n$ and time-varying.

\subsection{Exact likelihood} \label{sec:exact_likelihood}

Under the definitions in Section \ref{sec:model}, with the covariates $\bW$ fixed, the joint process of population states $(\bX_t)_{t\geq 0}$ and observations $(\bY_t)_{t\geq 1}$ is a Hidden Markov Model (HMM) \citep{Chopin2020}. Over a time horizon $T$, the marginal likelihood of $\bY_{1:T}$ is:
\begin{equation} \label{eq:truelikelihood}
    \begin{split}
        p(\bY_{1:T}|\bW)
        &= 
        \sum_{\bX_{0:T} \in \mathbb{O}_M^{NT}} p(\bX_0|\bW) \prod_{t=1}^T p(\bX_t|\bX_{t-1},\bW) p(\bY_t|\bX_t,\bW),
    \end{split}
 \end{equation}
 where, using some HMM terminology, the initial distribution is $p(\bX_0|\bW) \coloneqq \prod_{n \in [N]} \bx_{n,0}^\top p_0( \bw_n)$, the transition kernel is $p(\bX_t|\bX_{t-1},\bW) 
    \coloneqq 
  \prod_{n \in [N]} \bx_{n,t-1}^\top K_{\boeta_{n,t-1}}(\bw_n) \bx_{n,t}$, and the emission distribution is $p(\bY_t|\bX_t,\bW) 
    \coloneqq 
 \prod_{n \in [N]} \bx_{n,t}^\top  G(\bw_n) \by_{n,t}$. The computation of the marginal likelihood requires a summation over the set $\mathbb{O}_M^{NT}$. The forward algorithm \citep{Chopin2020} computes the sum at a cost linear in $T$, by recursively computing the prediction distributions $p(\bX_t|\bY_{1:t-1},\bW)$, the filtering distributions $p(\bX_t|\bY_{1:t},\bW)$, and the marginal likelihood increments $p(\bY_t|\bY_{1:t-1},\bW)$, via the so-called  ``prediction'' and ``correction'' steps of the filtering recursion. Indeed, given $p(\bX_0|\bY_{1:0},\bW) \coloneqq p(\bX_0|\bW)$:
\begin{align}
    \textbf{Prediction: }&
    p(\bX_t|\bY_{1:t-1},\bW) \coloneqq \sum\limits_{\bX_{t-1} \in \mathbb{O}_M^N} p(\bX_t|\bX_{t-1},\bW) p(\bX_{t-1}|\bY_{1:t-1},\bW);\\
    \textbf{Correction: }&
    p(\bY_t|\bY_{1:t-1},\bW)\coloneqq \sum\limits_{\bX_{t} \in \mathbb{O}_M^N} p(\bY_t|\bX_{t},\bW) p(\bX_{t}|\bY_{1:t-1},\bW) \text{ and }\\
    &p(\bX_t|\bY_{1:t},\bW)\coloneqq \dfrac{p(\bY_t|\bX_t,\bW) p(\bX_{t}|\bY_{1:t-1},\bW)}{p(\bY_t|\bY_{1:t-1},\bW)};
\end{align}
from which get $p(\bY_{1:T}|\bW) = p(\bY_1 |\bW)\prod_{t=2}^T p(\bY_t|\bY_{1:t-1},\bW)
$. Whilst the forward algorithm simplifies the likelihood computation with respect to time, it still requires summation over the set $\mathbb{O}_M^N$, making it infeasible for even small population sizes.

\section{CAL: Categorical Approximate Likelihood} \label{sec:CAL}

Given the state at time $t-1$, the transitions and observations of each individual at time $t$ are independent of one another. The challenge arises from uncertainty in $\bX_{t-1}$, which introduces dependence. To address this, we substitute $\bX_{t-1}$ with its expectation during the prediction step and propagate the resulting approximation of $p(\bX_t \mid \bY_{1:t-1}, \bW)$ through the correction step. This leads to approximations of the predictive distributions, filtering distributions, and marginal likelihood as products of categorical distributions.

\subsection{Approximate filtering algorithm}\label{sec:approx_filtering}

In this section we propose the approximations $p(\bX_t|\bY_{1:t-1}, \bW) \approx \prod_{n \in [N]}\text{Cat}({\bx_{n,t}} \mid \bpi_{n,t|t-1})$, and $p(\bY_t|\bY_{1:t-1}, \bW) \approx \prod_{n \in [N]}\text{Cat}({\by_{n,t}} \mid \bmu_{n,t})$, and $p(\bX_t|\bY_{1:t}, \bW) \approx \prod_{n \in [N]}\text{Cat}({\bx_{n,t}} \mid \bpi_{n,t})$. We next explain how the probability vectors $\bpi_{n,t|t-1}, \bpi_{n,t} \in \Delta_M$ and $\bmu_{n,t} \in \Delta_{M+1}$ are computed for $n = 1, \dots, N$ and $t \geq 1$.

Starting from $p(\bX_0 \mid \bY_{1:0}, \bW)$, we set, for each $n \in [N]$, $\bpi_{n,0}$ to be the length-$M$ probability vector associated with the initial distribution $p_0(\bw_n)$. Hence, at the initial time, no approximation is required, as the distribution of $\bX_0$ already factorizes as a product of categorical distributions. Now consider a general time $t \geq 1$. Suppose we have computed $\bPi_{t-1} \coloneqq (\bpi_{1,t-1}, \dots, \bpi_{N,t-1})$, so that $p(\bX_{t-1}|\bY_{1:t-1} \bW) \approx \prod_{n \in [N]}\text{Cat}({\bx_{n,t-1}} \mid \bpi_{n,t-1})$. In the \textbf{prediction} step, we approximate the transition probability $p(\bX_t \mid \bX_{t-1}, \bW)$ by substituting $\bPi_{t-1}$ in place of $\bX_{t-1}$ in the quantity $\boeta_{n,t-1} = \eta(\bw_n, \bW, \bX_{t-1})$. Under the categorical approximation, $\bPi_{t-1}$ coincides with the expectation of $\bX_{t-1}$. After this substitution, the \textbf{prediction} step admits a closed-form expression. Precisely:
\begin{equation}\label{eq:CAL_prediction}
    \bpi_{n,t|t-1} \coloneqq  \left [ \bpi_{n,t-1}^{\top} K_{\widetilde{\boeta}_{n,t-1}}(\bw_n) \right ]^{\top}, \quad \text{for } n \in [N],
\end{equation}
where $\widetilde{\boeta}_{n,t-1}\coloneqq \eta(\bw_n,\bW,\bPi_{t-1})$. The \textbf{correction} step is then applied to the categorical approximation $\bigotimes_{n \in [N]} \text{Cat}(\cdot \mid \bpi_{n,t|t-1})$ of $p(\bX_t \mid \bY_{1:t-1}, \bW)$, using the exact observation model $p(\bY_t \mid \bX_t, \bW)$. This model factorizes across individuals under the conditional independence structure specified in Section~\ref{sec:model}, giving us:
\begin{equation}\label{eq:CAL_update}
    \bmu_{n,t} \coloneqq \left [ \bpi_{n,t|t-1}^\top  G(\bw_n) \right ]^\top 
    \text{ and }
    \bpi_{n,t} \coloneqq \bpi_{n,t|t-1} \odot \left \{  \left [  G(\bw_n) \oslash  \left ( 1_M \bmu_{n,t}^\top \right ) \right ] \by_{n,t} \right \} \quad \text{for } n \in [N],
\end{equation}
where we use the convention $\frac{0}{0}=0$. By sequentially combining these approximate prediction and correction steps we get Algorithm \ref{alg:CAL}, which computes all the aforementioned quantities. Algorithm \ref{alg:CAL} could output all the categorical approximations, but for the sake of presentation we make it output only the Categorical Approximate Likelihood (CAL):
\begin{equation}\label{eq:CAL}
p(\bY_{1:T}|\bW) \approx\prod_{t=1}^T \prod_{n \in [N]} \text{Cat}({\by_{n,t}} \mid \bmu_{n,t})  =\prod_{t=1}^T \prod_{n \in [N]}  \by_{n,t}^\top \bmu_{n,t}.
\end{equation}

We refer the reader to Section~\ref{sec:closed_form_CAL} of the supplementary material for complete derivations.

\begin{algorithm}[t!]
\caption{Categorical Approximate Likelihood} \label{alg:CAL}
    \begin{algorithmic}
    \Require $\bW, \bY_{1:T}, p_0(\cdot), K_{\cdot}(\cdot), G(\cdot)$
    \State Initialize $\bpi_{n,0}$ with $p_0(\bw_n)$ for all $n \in [N]$
    \For{$t \in 1,\dots,T$}
        \State $\bPi_{t-1} = (\bpi_{1,t-1}, \dots, \bpi_{N,t-1})$
        \For{$n \in [N]$}
            \State $\widetilde{\boeta}_{n,t-1} = \eta(\bw_n,\bW,\bPi_{t-1})$
            \State $\bpi_{n,t|t-1} =  \left [ \bpi_{n,t-1}^{\top} K_{\widetilde{\boeta}_{n,t-1}}(\bw_n) \right ]^{\top}$ 
            \State $\bmu_{n,t} = \left [ \bpi_{n,t|t-1}^\top  G(\bw_n) \right ]^\top$
            \State $\bpi_{n,t} = \bpi_{n,t|t-1} \odot \left \{  \left [  G(\bw_n) \oslash  \left ( 1_M \bmu_{n,t}^\top \right ) \right ] \by_{n,t} \right \}$  
        \EndFor
    \EndFor
    
    \State Return the approximate likelihood $\prod_{t=1}^T \prod_{n \in [N]}  \by_{n,t}^\top \bmu_{n,t}$
    \end{algorithmic}
\end{algorithm}

\subsection{CAL as an exact likelihood in an approximate model}\label{sec:CAL_as_exact}

Although the CAL is derived as an approximation to the marginal likelihood for the model in Section \ref{sec:model}, it can also be interpreted as the exact marginal likelihood for the approximate model where $\widetilde{\boeta}_{n,t}$ is used instead of $\boeta_{n,t}$. Here, as in Section \ref{sec:approx_filtering}, $\widetilde{\boeta}_{n,t-1}=\eta(\bw_n,\bW,\bPi_{t-1})$, and $\bPi_{t-1}$ is computed as in Algorithm \ref{alg:CAL}. Indeed, we can define the state process $\widetilde{\bX}_t = (\widetilde{\bx}_{1,t},\dots, \widetilde{\bx}_{N,t})$ and the observation process $\widetilde{\bY}_t = (\widetilde{\by}_{1,t},\dots, \widetilde{\by}_{N,t})$ of the approximate model as follows:
 \begin{align}
    &\widetilde{\bx}_{n,0}|\bw_n \sim \mbox{Cat} \left ( \cdot|p_0(\bw_n) \right ), \quad \widetilde{\bx}_{n,t}|\widetilde{\bX}_{t-1}, \widetilde{\bY}_{1:t-1}, \bW \sim \mbox{Cat} \left ( \cdot| \left [ \widetilde{\bx}_{n,t-1}^\top K_{\widetilde{\boeta}_{n,t-1}}(\bw_n) \right ]^\top \right ),\\
    &\widetilde{\by}_{n,t}|\widetilde{\bx}_{n,t},\bw_n \sim \mbox{Cat} \left ( \cdot| \left [ \widetilde{\bx}_{n,t}^\top  G(\bw_n) \right ]^\top \right ).
\end{align}
Under the approximate model above, the marginal likelihood of $\widetilde{\bY}_{1:T}$ can be computed in closed-form using Algorithm \ref{alg:CAL}. Indeed, we have $p(\widetilde{\bY}_{1:T}|\bW) = \prod_{n \in [N]} \prod_{t=1}^T \widetilde{\by}_{n,t}^\top \bmu_{n,t}$, which coincides with the CAL when the observations are $\widetilde{\bY}_1,\dots,\widetilde{\bY}_T$. 

\section{Consistency of the maximum CAL estimator}\label{sec:theory_main}

In this section, we state our results about consistency of the maximum CAL estimator of model parameters, when data are generated from the exact model from Section \ref{sec:model}. Proofs and supporting results are given in Section  \ref{suppsec:consistency_section} of the supplementary materials. This theory has links to mean-field approximations \citep{sherborne2018mean} and propagation of chaos \citep{sharrock2023online,le2007generic} as it relies on the construction of what we call a ``saturated'' system of independently evolving individuals.

\subsection{Notation, definitions and assumptions}

We denote with $\Theta$ the parameter space, with $\mathbb{W}$ the covariate space. We assume all the random variables appearing in our theory to be defined on a common probability space $(\Omega, \mathcal{F}, \mathbb{P})$. We augment our notation from Section \ref{sec:model} by writing $\bW^N, \bx_{n,t}^N, \bX^N_t,\by_{n,t}^N, \bY^N_t$ for $\bW, \bx_{n,t}, \bX_t$, $\by_{n,t}, \bY_t$; given a parameter vector $\theta \in \Theta$, the initial distribution, transition matrix, and emission matrix are denoted: $p_0(\bw_n,\theta)$, $K_{\cdot}(\bw_n,\theta)$, and $G(\bw_n,\theta)$, with $\eta(\bw_n,\bW,\bX)$ becoming $\eta^N(\bw_n,\theta, \bW^N,\bX^N)$, and $\boeta_{n,t}$ becoming $\boeta^N_{t}(\bw_n,\theta)$. Similarly, the CAL quantities $\bpi_{n,t|t-1}, \bmu_{n,t}, \bpi_{n,t}, \widetilde{\boeta}_{n,t}$ in Algorithm \ref{alg:CAL} become $\bpi^N_{n,t|t-1}(\bw_n,\theta), \bmu^N_{n,t}(\bw_n,\theta)$, $\bpi^N_{n,t}(\bw_n,\theta), \widetilde{\boeta}_{t}^N(\bw_n,\theta)$.  We denote by $\theta^\star \in \Theta$ the data-generating parameter (DGP) value, which determines the distributions of $(\bX_{t})_{t\geq 0}$ and $(\bY_{t})_{t\geq 1}$ under $\mathbb{P}$ conditional on $\bW$. For an $M$-dimensional vector $\pi$, an $M\times M$-dimensional matrix $K$, and an $\mathbb{R}$-valued random variable $\bx$ we define the following norms:
$$
    \norm{\pi}_\infty \coloneqq \max\limits_{i\in[M]} \abs{\pi^{(i)}}, \quad 
    \norm{K}_\infty \coloneqq \max_{i \in [M]} \sum_{j = 1}^M \abs{K^{(i,j)}}, \quad \text{and} \quad
    \normiii{\bx}_4 \coloneqq \left ( \mathbb{E} \left[ |\bx|^4 \right ]\right )^{\frac{1}{4}}.
$$
We refer to the non-zero elements of probability vectors and matrices as the support, which we define as $\textbf{supp}(\bpi) \coloneqq \{i \in [M]: \bpi^{(i)} \neq 0\}$ for an $M$-dimensional vector $\bpi$, and $\textbf{supp}(\bPi) \coloneqq \{(i,j) \in [M]^2: \bPi^{(i,j)} \neq 0\}$ for a $M \times M$ matrix $\bPi$.

We next state our assumptions, which we first comment on and then list. Assumption \ref{ass:main_compactness_continuity} collects standard compactness and continuity assumptions. Assumption \ref{ass:main_w_iid} is a random design assumption on the covariate vectors, which is inspired by classical theoretical analysis of regression models. Assumption \ref{ass:main_HMM_support} represents a technical assumption on the support of initial distribution, transition matrix, and emission matrix, which guarantees the invariance of the support when considering different parameters and covariates. We conclude with Assumption \ref{ass:main_eta_structure} on the structure of $\eta$, which ensures a law of large numbers for this interaction term, and Assumption \ref{ass:main_kernel_continuity} about the Lipschitz continuity of the transition matrix in $\eta$.

\begin{assumption}\label{ass:main_compactness_continuity}
    The parameter space $\Theta$ and the covariate space $\mathbb{W}$ are compact subsets of Euclidean spaces. Moreover, the initial distribution $p_0(w,\theta)$, the transition matrix $K_{\eta}(w,\theta)$, and the emission matrix $G(w,\theta)$ are all continuous functions in their arguments $w,\theta$.
\end{assumption}
 
\begin{assumption}\label{ass:main_w_iid}
   The covariates $\bw_1,\bw_2,\dots, $ are independent and identically distributed according to a distribution $\Gamma$ on $\mathbb{W}$. 
\end{assumption}

\begin{assumption}\label{ass:main_HMM_support}
    The following hold: for any $w \in \mathbb{W}$ and $\theta,\theta' \in \Theta$ we have that $\textbf{supp}(p_0(w, \theta)) = \textbf{supp}(p_0(w, \theta'))$;  for any $w \in \mathbb{W}$, $\eta,\eta' \in [0,C]$ and $\theta,\theta' \in \Theta$ we have that $\textbf{supp}(K_{\eta}(w,\theta)) = \textbf{supp}(K_{\eta'}(w,\theta'))$; for any $w \in \mathbb{W}$ and $\theta,\theta' \in \Theta$ we have that $\textbf{supp}( G(w,\theta)) = \textbf{supp}( G(w,\theta'))$.
\end{assumption}
 
\begin{assumption}\label{ass:main_eta_structure}
    For any $\theta \in \Theta, w \in \mathbb{W}, N\in \mathbb{N}$, and for any $W^N = (w_1,\dots, w_N),\Pi^N = (\pi_1,\dots,\pi_N)$ with $w_n \in \mathbb{W},\pi_n \in \Delta_M$ for all $n \in [N]$, we have:
    \begin{equation}
        \eta^N(w, \theta, W^N, \Pi^N) = \frac{1}{N} \sum_{n \in [N]} d(w, w_{n}, \theta)^\top \pi_{n},
    \end{equation}
    where $d:\mathbb{W} \times \mathbb{W} \times \Theta \to [0,C]^M$ is a bounded function, i.e. $\norm{d(w,\widetilde{w},\theta)}_\infty \leq C < \infty$ for all $w,\widetilde{w} \in \mathbb{W}$ and $\theta \in \Theta$. This assumption also implies $\eta^N(w, \theta, W^N, \Pi^N) \in [0,C]$ for any $N, w,\pi, W^N, \Pi^N$.
\end{assumption}

\begin{assumption}\label{ass:main_kernel_continuity}
    For any $\theta \in \Theta$ and $w \in \mathbb{W}$, the matrix $K_{\eta}(w,\theta)$ is Lipschitz continuous in $\eta$ with Lipschitz constant $L$, that is for any $\eta, \eta^\prime \in [0,C]$ we have:
    \begin{equation}
        \norm{K_{\eta}(w,\theta) - K_{\eta^\prime}(w,\theta)}_\infty \leq L \abs{\eta - \eta^\prime}.
    \end{equation}
\end{assumption}

\subsection{Main consistency theorem and outline of the proof}

Consider a fixed time horizon $T \geq 1$, and the log-CAL evaluated at $\theta \in \Theta$:
\begin{equation}\label{eq:main_sup_log_cal}
\ell_{1:T}^N(\theta) \coloneqq \sum_{t=1}^T \sum_{n \in [N]} \log \left [ \left ( \by_{n,t}^N \right )^\top \bmu_{n,t}^N(\bw_n,\theta) \right ].
\end{equation}
This section outlines the proof that the maximum CAL estimator $\hat{\theta}^N \coloneqq \argmax_{\theta \in \Theta} \ell_{1:T}^N(\theta)$ is consistent in the large population limit. All the details are available in Section \ref{suppsec:consistency_section} of the supplementary material. The main challenge is to prove that $N^{-1}(\ell_{1:T}^N(\theta)-\ell_{1:T}^N(\theta^\star))$ converges uniformly $\mathbb{P}$-almost surely to a contrast function which is maximised by $\theta^\star$. This then allows standard continuity arguments to be used in proving almost sure convergence of the maximizer $\hat{\theta}^N$ to some equivalence set containing $\theta^\star$. Due to the presence of covariates, the details of the analysis are substantially richer than those of \cite{whitehouse2023consistent}.

\paragraph{Saturated process and saturated CAL.}

From Section \ref{sec:model} it is clear that all the individuals are interacting via the interaction term $\boeta^N_{t-1}(\bw_n,\theta^\star)$ in the transition matrix. We can prove that under Assumptions \ref{ass:main_w_iid},\ref{ass:main_eta_structure},\ref{ass:main_kernel_continuity} for any $t\geq 0$ there exists a deterministic function $w \mapsto \boeta^\infty_{t}(w,\theta^\star)$ from $\mathbb{W}$ to $[0,C]$ such that for any $n \in [N]$:
\begin{equation}\label{eq:deterministic_effect_saturated}
    \normiii[\Bigg]{\boeta^N_{t}(\bw_n,\theta^\star) - \boeta^\infty_{t}(\bw_n,\theta^\star)}_4 = \mathcal{O}\left ( N^{-\frac{1}{2}} \right ),
\end{equation}
see Section \ref{sec:asympt_DGP} of the supplementary material for  details. From \eqref{eq:deterministic_effect_saturated} we observe that when the system becomes ``saturated'' with individuals, i.e. $N \to \infty$, the effect from the population has a deterministic behavior. Substituting $\boeta^\infty_{t}(\cdot,\theta^\star)$ in our latent dynamic defines a saturated process at the individual level. Specifically, for an individual with covariate $\bw^\infty \sim \Gamma$ the saturated process is:
\begin{equation}\label{eq:main_limiting_process}
    \begin{split}
    &\bx_0^\infty|\bw^\infty \sim \mbox{Cat}\left(\, \cdot\,|\,p_0(\bw^\infty,\theta^\star)\right),\\
    &\bx_t^\infty|\bx_{t-1}^\infty,\bw^\infty \sim \mbox{Cat}\left(\, \cdot\,\left|\left [ (\bx_{t-1}^\infty)^\top K_{ \boeta_{t-1}^{\infty} (\bw^\infty,\theta^\star ) }(\bw^\infty,\theta^\star) \right ]^\top\right.\right),\\
    &\by_t^\infty|\bx_{t}^\infty,\bw^\infty \sim \mbox{Cat}\left( \cdot\left|\left [ (\bx_{t}^\infty)^\top G(\bw^\infty,\theta^\star) \right ]^\top\right.\right).
    \end{split}
\end{equation}
We can observe that in the saturated process the individuals evolve independently, providing an asymptotic justification for the CAL prediction step.

Similarly, we can prove that under Assumptions \ref{ass:main_w_iid},\ref{ass:main_eta_structure},\ref{ass:main_kernel_continuity} for any $t\geq 0$ there exists a deterministic function $w \mapsto \bar{\boeta}^\infty_{t}(w,\theta)$ from $\mathbb{W}$ to $[0,C]$, which is such that for any $w \in \mathbb{W}$ we have $\bar{\boeta}^\infty_{t-1}(w,\theta^\star) = \boeta^\infty_{t-1}(w,\theta^\star)$ and for any $n \in [N]$:
\begin{equation}
    \normiii[\Bigg]{\widetilde{\boeta}^N_{t}(\bw_n,\theta) - \bar{\boeta}^\infty_{t}(\bw_n,\theta)}_4 = 
    \mathcal{O} \left ( N^{-\frac{1}{2}} \right ),
\end{equation}
see Section \ref{sec:asympt_CAL} of the supplementary material. We can then define the saturated CAL recursion by substituting individual saturated process observations from \eqref{eq:main_limiting_process} and $\bar{\boeta}^\infty_{t}(\bw_n,\theta)$ into Algorithm \ref{alg:CAL}. Precisely, set $\bpi_{0}^\infty(\bw^\infty,\theta) \coloneqq p_{0}(\bw^\infty,\theta)$ and then for $t\geq 1$:
\begin{equation}\label{rec:main_limiting_CAL}
    \begin{split}
    &\bpi_{t|t-1}^\infty(\bw^\infty,\theta)  \coloneqq  \left [ \bpi_{t-1}^\infty(\bw^\infty,\theta)^{\top} K_{\bar{\boeta}_{t-1}^\infty(\bw^\infty,\theta)} (\bw^\infty, \theta)\right ]^\top, \\
    &\bmu_{t}^\infty(\bw^\infty,\theta)  \coloneqq  \left [ \bpi_{t|t-1}^\infty(\bw^\infty,\theta)^{\top} G(\bw^\infty,\theta) \right ]^\top,\\
    &\bpi_{t}^\infty(\bw^\infty,\theta)  \coloneqq \bpi_{t|t-1}^\infty(\bw^\infty,\theta) \odot \left \{ \left [   G(\bw^\infty,\theta) \oslash \left ( 1_M \bmu_{t}^\infty(\bw^\infty,\theta)^\top \right ) \right ] \by_{t}^\infty  \right \}.
    \end{split}
\end{equation} 
As, conditional on $\bw^\infty$, the joint process $(\bx_t^\infty)_{t\geq 0}$, $(\by_t^\infty)_{t\geq 1}$ in \eqref{eq:main_limiting_process} is a HMM,  Recursion \eqref{rec:main_limiting_CAL} becomes the forward algorithm associated with this HMM when $\theta=\theta^\star$, i.e. at the DGP.

\paragraph{Contrast function and set of maximizers.} 

Under Assumptions \ref{ass:main_compactness_continuity},\ref{ass:main_w_iid},\ref{ass:main_HMM_support},\ref{ass:main_eta_structure},\ref{ass:main_kernel_continuity}, for any $\theta \in \Theta$ we prove that $N^{-1}({\ell_{1:T}^N(\theta)} -  {\ell_{1:T}^N(\theta^\star)})$ converges $\mathbb{P}$-almost surely to a contrast function $\mathcal{C}_T(\theta,\theta^\star)$ as $N \to \infty$, which takes the form of an expected Kullback-Leibler (KL) divergence:
\begin{equation}\label{eq:main_contrast_function}
    \mathcal{C}_T(\theta,\theta^\star)
    \coloneqq -\sum_{t=1}^T \mathbb{E} \left \{ \mathbf{KL} \left [ \mbox{Cat} \left (\cdot | \bmu_t^\infty(\bw^\infty,\theta^\star) \right )|| \mbox{Cat} \left (\cdot | \bmu_t^\infty(\bw^\infty,\theta) \right ) \right ] \right \}.
\end{equation}
Moreover, we use properties of the KL divergence to show the DGP belongs to the set of maximizers of the contrast function, i.e. $\theta^\star \in \Theta^\star \coloneqq \argmax_{\theta \in \Theta} \mathcal{C}_T(\theta,\theta^\star)$. Full proof is available in Section \ref{sec:constrast_function_pointwise_conv} and Section \ref{sec:limiting_quantities_support} of the supplementary material.

\paragraph{Convergence of the maximum CAL estimator and identifiability.} After proving some technical results, we can complete the proof of Theorem \ref{thm:main_consistency}, which states the consistency of the maximum CAL estimator, see Section \ref{sec:constrast_function_pointwise_conv} of the supplementary material.

\begin{theorem}\label{thm:main_consistency}
    Let Assumptions \ref{ass:main_compactness_continuity},\ref{ass:main_w_iid},\ref{ass:main_HMM_support},\ref{ass:main_eta_structure},\ref{ass:main_kernel_continuity} hold and let $\hat{ {\theta}}_N$ be a maximizer of $  \ell_{1:T}^N({\theta})$. Then $\hat{ {\theta}}_N$ converges to $\Theta^\star$ as $N \to \infty$,  $\mathbb{P}$-almost surely.
\end{theorem}

The theorem states that the maximum CAL estimator converges to a set of maximizers $\Theta^\star$, which is a set of parameters that define statistically indistinguishable one-individual saturated processes. More formally, denote with $\mathbb{P}_{\infty}^{\theta^\star,w}$ the law of $(\by_{t}^\infty)_{t\geq 1}$ conditional on $\bw^\infty = w$ and with DGP $\theta^\star$. We can show that for any $\theta_1^\star,\theta_2^\star \in \Theta^\star$ we have $\mathbb{P}_{\infty}^{\theta_1^\star, w} = \mathbb{P}_{\infty}^{\theta_2^\star,w}$ for $\Gamma$-almost all $w\in \mathbb{W}$, see Section \ref{sec:constrast_function_pointwise_conv} of the supplementary material.

\section{Experiments} \label{sec:exper}

We now implement the CAL over a range of IBMs with heterogeneous attributes and different heterogeneous-mixing behaviors. The results can be reproduced following the GitHub repository \href{https://github.com/LorenzoRimella/CAL}{LorenzoRimella/CAL}. All the experiments were run on a 32GB Tesla V100 GPU available on “The High-End Computing” (HEC) facility at Lancaster University.

\paragraph{Computational considerations}
The CAL is ``embarrassingly parallel'' in $N$ at each time step, making it well-suited for parallel architectures such as GPUs. Furthermore, the CAL does not rely on simulations from the model or permutations of indices, which makes it particularly amenable to just-in-time (JIT) compilation \citep{aycock2003brief}, enabling efficient execution without complex code design. Because of its simplicity, the gradient of the log-CAL with respect to the model parameters can be computed via automatic differentiation (AD). As a result, popular AD libraries \citep{tensorflow2015} can be used for efficient optimization, or the CAL can be embedded within a Hamiltonian Monte Carlo (HMC) sampler for Bayesian inference via probabilistic programming languages \citep{carpenter2017stan}. Both approaches are used and combined in our experiments, with the former implemented using the Adam optimizer \citep{kingma2014Adam}. More computational considerations can be found in Section \ref{sec:supp_computational_cost} of the supplementary materials.

\subsection{CAL-posterior inference using HMC in TensorFlow} \label{sec:exp_HMC}

We demonstrate Bayesian inference using an HMC sampler in TensorFlow \citep{tensorflow2015} to target a posterior distribution defined in terms of the CAL. The appeal of this approach is that once the model is formulated as in Section \ref{sec:model_dynamic_obs}, evaluating the CAL involves no tuning parameters and can be readily embedded within an ``off the shelf'' HMC program. 

\begin{figure}[httb!]
    \centering
    \includegraphics[width=\linewidth]{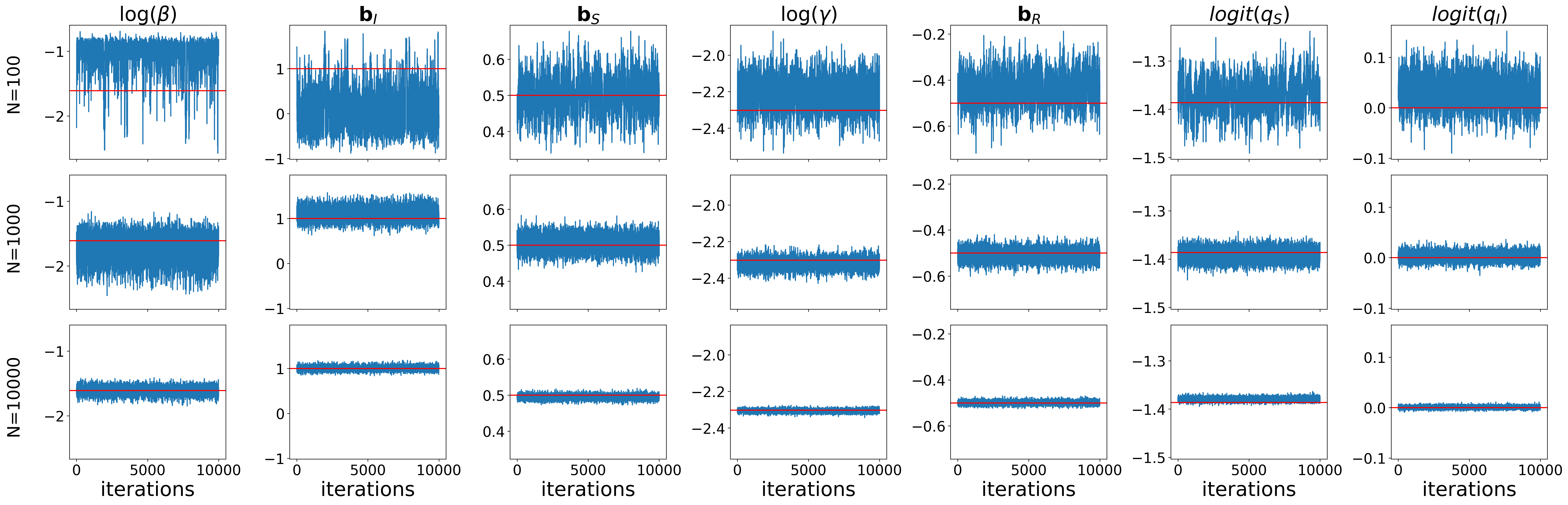}
    \caption{Trace plots for HMC under different population sizes. Solid red lines denote the DGP.}
    \label{fig:HMC}
\end{figure}

Consider the homogeneous-mixing SIS model from Section \ref{sec:motivating_example_first} and Section \ref{sec:motivating_example_second} with synthetic covariates $\bc_n \sim \mathcal{N}(\cdot|0,1)$. We simulate data from models with increasing population sizes $N=100, 1000, 10000$ and time horizon $T=200$. The full parameter settings can be found in Section \ref{sec:app_exp_HMC} of the supplementary materials. The chains show no signs of poor mixing and recovery of the DGP, with a posterior distribution that becomes increasingly concentrated as the population size grows, complementing our consistency theory. The running time for the experiment with $N=10000$ was around $0.5$s for each iteration, see Section \ref{sec:app_exp_HMC} of the supplementary material for full details of the HMC scheme.

We further assess the coverage of the credible intervals by considering $N = 1000$ and rerunning the experiment 100 times. Specifically, we simulate 100 epidemics and replicate the previous HMC procedure for each of them. This results in 95\% marginal credible interval coverages of 0.83 for $\log(\beta)$, 0.90 for $\bb_S$, 0.80 for $\bb_I$, 0.91 for $\bb_R$, 0.93 for $ logit(q_S)$, and 0.95 for $logit(q_I)$. Overall, the coverages are close to the target value of 0.95, except for $\log(\beta)$ and $\bb_I$, for which uncertainty is underestimated.

\subsection{Gradient-based calibration for heterogeneous-mixing SIS} \label{sec:exp_gradient}

We consider two heterogeneous-mixing IBMs: one with a continuous spatial interpretation, and the other with a network interpretation. We demonstrate recovery of the DGP by optimizing the CAL with Adam \citep{kingma2014Adam} and with an accuracy that increases in $N$. 

\paragraph{Model 1.} Consider the heterogeneous-mixing SIS model from Section \ref{sec:motivating_example_first} and Section \ref{sec:motivating_example_second}. Precisely, we have initial infection probabilities and $\boeta_{n,t-1}$ as in Section \ref{sec:motivating_example_second}, and a transition matrix as in \eqref{eq:HetMix_K}, where we also include the term  $\epsilon$ to represent a constant rate of infection from the environment. The covariates $\bw_n = [\bl_n, \bc_n]$ are synthetic, such that $\bl_n$ is the location in space of each individual and $\bc_n \sim \mathcal{N}(\cdot|0,1)$. The location $\bl_n$ is drawn from a mixture of 10 bivariate Gaussian distributions, each component of which can be interpreted as a geographic hub, e.g. a city. A full mathematical description of the model can also be found in Section \ref{sec:app_exp_grad} of the supplementary material.

\paragraph{Model 2.} We now group individuals into communities. All the quantities are as in Model 1 except the individuals' location $\bl_n$ which is now replaced by $\bom_n$ the mean of the mixture component the $n$-th individual was assigned to in Model 1. It is important to note that the computational cost of computing all interaction terms for Model 1 is $N^2$, while this cost can be reduced to $N$ times the number of communities for Model 2. More details on the model are available in Section \ref{sec:app_exp_grad} of the supplementary material.

\begin{figure}[httb!]
    \centering
    \includegraphics[width=0.75\linewidth]{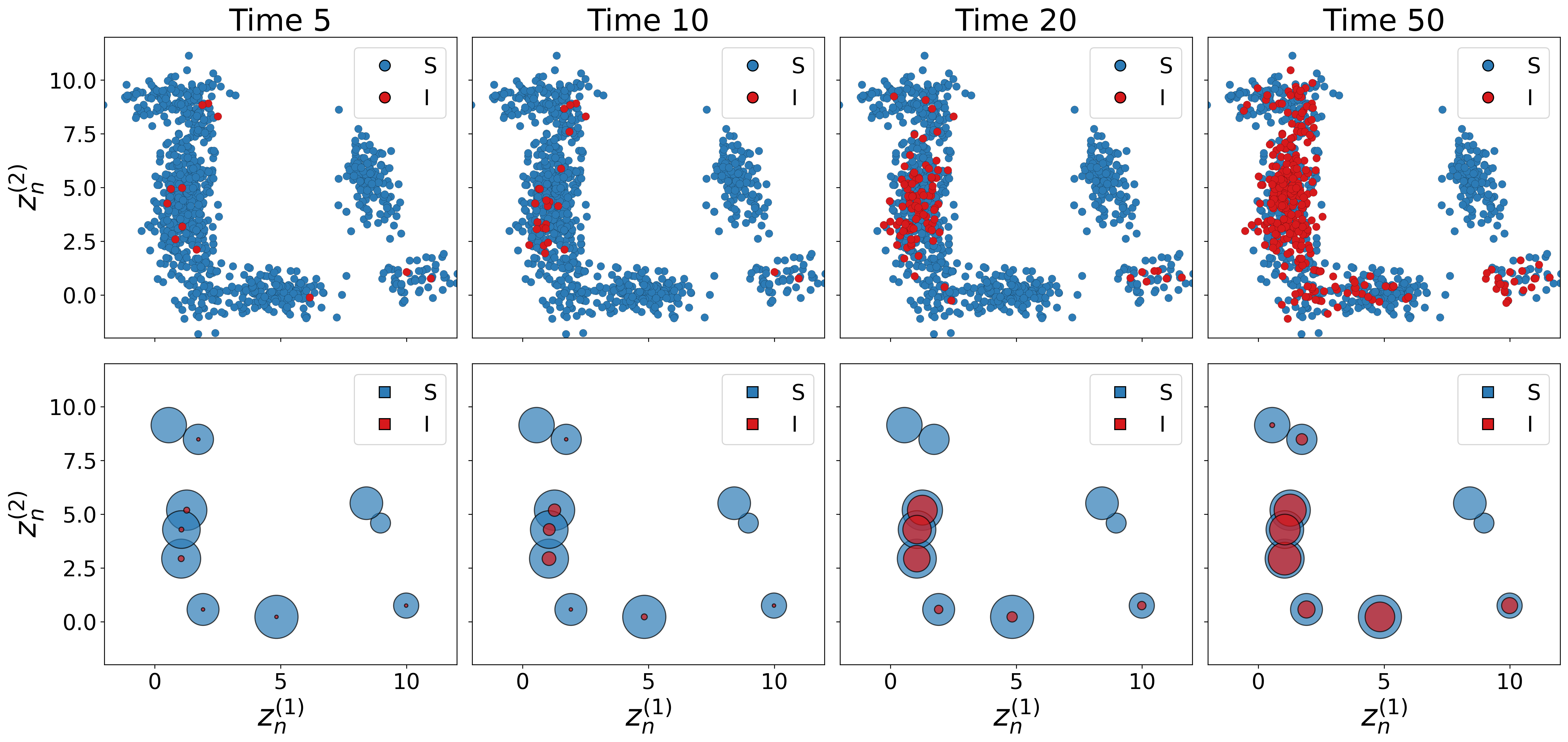}
    \caption{A realization of the latent process from Model 1 (first row) and Model 2 (second row) when $N=1000$. Different columns are associated with different time steps. For Model 1, blue and red dots refer to susceptible and infected individuals, respectively. For Model 2, the communities are blue circles with a radius that is proportional to their population, while the red circles are proportional to the number of infected inside the communities.}
    \label{fig:spatial_illustration}
\end{figure}

\paragraph{} Even though the two models have different heterogeneous-mixing properties, they share the same parameters and we set the DGP to the same values for both models. Precisely, $p_0=0.01, \beta = 2.0, \bb_I = 1.0, \bb_S = 0.5, \gamma = 0.1, \bb_R= -0.5, \phi = 1, \epsilon = 0.0001, q_S=0.2, q_I=0.5, q_{Se}=0.9, q_{Sp}=0.95$. Given the DGP we can simulate from the two models for a fixed time horizon and population, see Figure \ref{fig:spatial_illustration} for a graphical representation of the disease's spread for $t=5,10,20,50$ and when $N=1000$. Here we can observe the effect of the spatial component in both models, with the disease spreading faster in regions\slash communities with a higher number of infected, while isolated regions\slash communities are difficult to reach, and remain untouched by the disease, see Figure \ref{fig:spatial_illustration}.

\begin{table}[h!]
\centering
\caption{The effect of increasing $N$ on the maximum CAL estimator for Model 1 and Model 2. The first column shows the DGP. In brackets the standard deviation of the maximum CAL estimator computed over $100$ simulations.}\label{tab:tab_inference}
\resizebox{\textwidth}{!}{
\begin{tabular}{l|ccc|ccc}
        & \multicolumn{3}{c|}{{Model 1}} & \multicolumn{3}{c}{{Model 2}} \\ 
        \cmidrule(lr){2-4} \cmidrule(lr){5-7}
     Parameter                & $N=500$       & $N=1000$      & $N=2000$       & $N=500$      & $N=5000$       & $N=50000$      \\
      \midrule
     $\log(\beta)$=0.69       & 0.71(0.21)  & 0.71(0.17)  & 0.72(0.085)  & 0.72(0.2)   & 0.69(0.08)   & 0.7(0.03)    \\
     $\mathbf{b}_{S}$=0.5     & 0.45(0.06)  & 0.48(0.03)  & 0.49(0.022)  & 0.42(0.08)  & 0.5(0.01)    & 0.5(0.003)   \\
     $\mathbf{b}_{I}$=1.0     & 0.95(0.16)  & 0.97(0.15)  & 0.98(0.063)  & 0.91(0.2)   & 1.0(0.06)    & 1.0(0.025)   \\
     $\log(\gamma)$=-2.3      & -2.31(0.04) & -2.29(0.03) & -2.3(0.016)  & -2.32(0.05) & -2.3(0.01)   & -2.3(0.003)  \\
     $\mathbf{b}_{R}$=-0.5    & -0.51(0.05) & -0.5(0.03)  & -0.5(0.018)  & -0.55(0.07) & -0.5(0.01)   & -0.5(0.004)  \\
     $\log(\phi)$=0.0         & 0.01(0.05)  & -0.0(0.05)  & -0.0(0.026)  & -0.02(0.07) & 0.0(0.02)    & -0.0(0.007)  \\
     $logit(q_{\cdot})$=-1.39 & -1.39(0.01) & -1.39(0.01) & -1.39(0.005) & -1.39(0.01) & -1.39(0.004) & -1.39(0.001) \\
     $logit(q_{\cdot})$=0.0   & 0.0(0.02)   & 0.0(0.01)   & 0.0(0.007)   & 0.01(0.02)  & 0.0(0.004)   & 0.0(0.001)  \\
     \bottomrule
\end{tabular}
}
\end{table}

For the experiment, we consider $N= 500,1000,2000$ for Model 1 and $N= 500,5000,50000$ for Model 2, where we can use a larger population for Model 2 because of the reduced computational cost. We then simulate for each population size and for each model $100$ realizations according to the considered dynamics and observation model, with the covariates fixed.

For both models, we treat $p_0,\epsilon,q_{Se},q_{Sp}$ as known, and infer for each simulated dataset $p_0,\beta,$$\bb_S,$\\$\bb_R, \phi, q_S, q_I$ by running Adam with $10$ different initial conditions for $1000$ gradients steps. At the end of the optimization, we choose the best out of the $10$ in terms of CAL log-likelihood for each dataset. We report the results of this optimization in Table \ref{tab:tab_inference}. Here, we can observe that, for both models, the mean of our estimator is close to the true value of the DGP and that the variance of the maximum CAL estimator is shrinking with the population size.

\subsection{Calibration and filtering for heterogeneous-mixing SIR} \label{sec:exp_mispec}

In this section, we present a simple pipeline explaining how the CAL can be used to track individuals' disease states within the population when considering an individual-based susceptible-infected-removed (SIR) model. We analyze both a scenario where the model is well-specified and a scenario where the model is misspecified. 

\paragraph{Well-specified model.} Consider an individual-based SIR where the individuals have the same covariates as Model 1 from Section \ref{sec:exp_gradient}, including the same spatial locations. We consider a $p_0(\bw_n)$ such that the individuals in the top-left of the spatial region become infected with probability $p_0$, while the others are susceptible with probability $1$. The transition matrix is now $3\times 3$ with transition probabilities that are governed by the same parameters as in Model 1 from Section \ref{sec:exp_gradient}. The observation model now needs three parameters $q_S,q_I,q_R$ which represent the probability of reporting $S$ as $S$, $I$ as $I$, $R$ as $R$, respectively. We do not allow for misreporting and we force half of the population to always remain unreported. Full details are available in Section \ref{sec:app_exp_mispec} of the supplementary material. We consider this as the model that generates the data and we set the DGP to $p_0=0.5, \beta = 3.0, \bb_I = 1.0, \bb_S = 0.5, \bb_R= -0.1, \phi = 1.5, \epsilon = 0.0001, q_S=0.1, q_I=0.2, q_R = 0.5$. 

\paragraph{Misspecified model.} We now present a misspecified model, with the same initial distribution as the well-specified model, but a transition kernel with an interaction term $\boeta_{n,t}$ that groups individuals into communities as in Model 2 of Section \ref{sec:exp_gradient}. Precisely, we use the same formulation as in Model 2 but an interaction term that considers $\bar{\bl}_n$, the mean distance between all pairs of individuals within the community of individual $n$:
\begin{equation}
    \boeta_{n,t-1}
    =
    \frac{1}{N} \sum_{k \in [N]}
    \exp\{\bc_{k}^\top \bb_I\} \frac{1}{\sqrt{2 \pi \phi^2}}\exp\left\{-\frac{\|\bom_n-\bom_{k}\|^2 B_{n,k} + \bar{\bl}_n^2 (1-B_{n,k})}{2\phi^2}\right\}\bx_{k,t-1}^{(2)},
\end{equation}
where $B_{n,k}\coloneqq \mathbb{I}(\|\bom_n-\bom_{k}\| \neq 0)$. The covariates of the misspecified model are then $\bw_n = [\bom_n,\bar{\bl}_n,\bc_n]$. Similarly to Model 2, the computational cost of computing all the interaction terms is $N$ times the number of communities, making it significantly cheaper to fit compared to the well-specified model. More details are available in Section \ref{sec:app_exp_mispec} of the supplementary material.

\begin{figure}[httb!]
    \centering
    \includegraphics[width=0.75\linewidth]{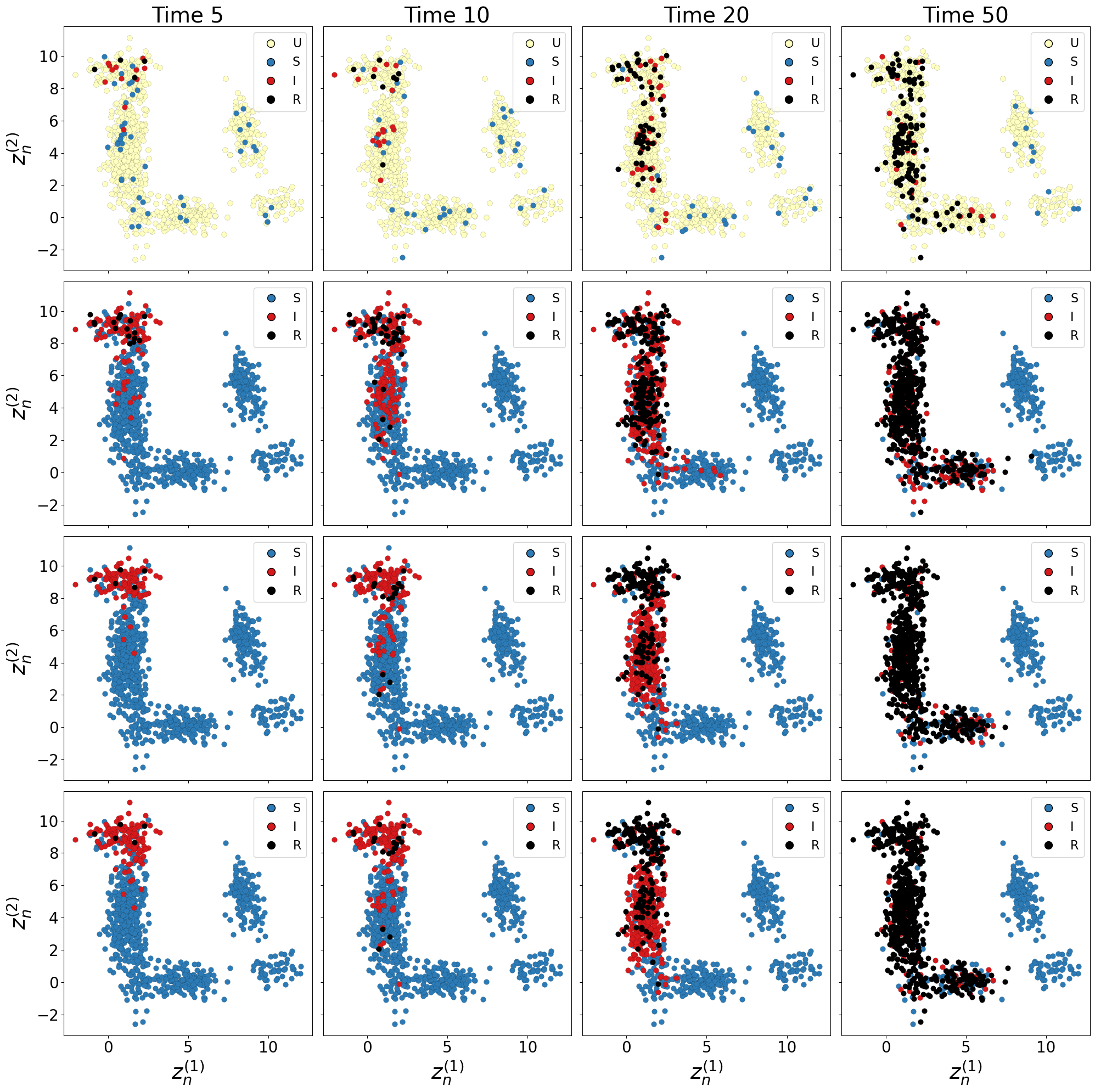}
    \caption{CAL filtering for $t=5,10,20,50$ under the well-specified and misspecified scenario. Rows from top to bottom: observed data, true latent disease states, and inferred latent disease states from CAL filtering under the well-specified and misspecified models. The yellow dots are used for unreported individuals, while blue, red, black are susceptible, infected, removed.}
    \label{fig:filter_spec_misspec}
\end{figure}

We generate an epidemic from the well-specified model, and we then optimize the parameters of both the well-specified and the misspecified model by running Adam with a learning rate of $0.1$ for $500$ iterations. After optimization, we set $\theta$ to the maximum CAL estimator and run Algorithm \ref{alg:CAL} for both models, where the CAL filter $\bpi_{n,t}$ is stored and used as an approximation of the true state. Figure \ref{fig:filter_spec_misspec} reports for $t=5,10,20,50$ the observations in the first row, the true states in the second row, the CAL estimate under the well-specified model and the misspecified model in the third and fourth rows. The CAL estimate on the state of the $n$-th individual at time $t$ is obtained as the argmax of  $\bpi_{n,t}$. It can be observed that the CAL filter is able to reproduce the spread of the epidemic from top to bottom, and even track infected individuals that are unreported for both models. 

\begin{table}[httb!]
    \centering
    \caption{Cross-entropy loss (the lower the better) and accuracy (the higher the better) for the CAL well-specified and misspecified, along with some baselines. The predicted state for accuracy is the argmax of the probability vector.}
    \label{tab:well_mis_accuracy}
    \begin{tabular}{l|ccccc}
     Metric                     &   Random &   Prev. uncertain &   Prev. certain &     CAL &   CAL missp. \\
    \midrule
     Cross-entropy & 1.10 &          1.09 &       0.67 & 0.28 &          0.29 \\
     Accuracy                   & 34.85\%  &          65.28\%  &       65.28\%   & 88.48\% &          88.08\%   \\
    \bottomrule
    \end{tabular}
\end{table}

Graphically, it seems we do not lose much with the misspecified model. To verify this, we consider the prediction performance on $\bx_{n,t}$ of $\bpi_{n,t}$ for both models. This is measured via two metrics: the cross-entropy loss \citep{de2005tutorial}, or equivalently minus the mean categorical log-likelihood; the accuracy, which is the percentage of correct estimates of $\bx_{n,t}$. We also consider three baselines: ``Random'' where we predict individual $n$ randomly unless we report their state; ``Prev. uncertain'' where we predict individual $n$ with their latest reported state with probability $0.34$ (and the other states with $(1-0.34) \slash 2$) unless we report their state; and ``Prev. certain'' where we predict individual $n$ with their latest reported state with probability $0.99$ (and the other with $(1-0.99) \slash 2$) unless we report their state. As in Figure \ref{fig:filter_spec_misspec}, the estimate is obtained as the argmax of the probabilities. Table \ref{tab:well_mis_accuracy} reports the results and shows that the CAL methods perform better than the baselines in both metrics. We observe that, as expected from the graphical interpretation, little accuracy is lost when switching to the misspecified model. All the mathematical definitions of the metrics and baselines are provided in Section \ref{sec:app_exp_mispec} of the supplementary material, which also discusses how to run the same experiment with different parameter values.

\subsection{Comparing CAL with sequential Monte Carlo} \label{sec:exp_SMC_only}

In this section, we compare the runtime and marginal likelihood values obtained by the CAL with those produced by two sequential Monte Carlo (SMC) algorithms: the Auxiliary Particle Filter (APF) \citep{johansen2008note} and the Block Auxiliary Particle Filter (Block APF) \citep{rebeschini2015can}, as both the number of particles and the population size increase. We also consider a just-in-time (JIT) compiled version of the CAL \citep{aycock2003brief} and a batched version of the Block APF. The CAL can be JIT-compiled without any additional effort because of the simplicity of its operations. It is worth noting that SMC variants exist that can be both JIT-compiled and automatically differentiated \citep{corenflos2021differentiable,tan2024accelerated}. For the batched Block APF, both individuals and particles are split into batches to reduce memory consumption; see Section \ref{sec:app_exp_SMC_only} of the supplementary materials for further details.

We consider the homogeneous-mixing SIS model from Section \ref{sec:motivating_example_first} and Section \ref{sec:motivating_example_second} with synthetic covariates $\bc_n \sim \mathcal{N}(\cdot|0,1)$. We simulate data from models with increasing population sizes $N=10, 100, 1000, 10000$ and time horizon $T=100$. The full parameter settings can again be found in Section \ref{sec:app_exp_SMC_only} of the supplementary materials.

\begin{table*}[httb!]
    \centering
    \caption{Log-marginal likelihood means and standard deviations, over $100$ runs, for the SIS model from Section \ref{sec:motivating_example_first}. Running times are reported in seconds for a single run and as averages across the $100$ runs. }
    \label{tab:SMC_comparison_table}
    \resizebox{\textwidth}{!}{
     \begin{tabular}{ll|llll|llllll}
    \hline
     N     & P    & CAL       & time  & CAL compiled & time  & APF                & time  & Block APF         & time & Block APF batched & time    \\
     \hline
     10    & 256  & -54.16    & 1.34 & -54.16       & 0.001 & -54.09(0.03)       & 0.11  & -54.14(0.06)      & 0.13 & -54.14(0.06)      & 2.43    \\
     10    & 512  &           &       &              &       & -54.09(0.03)       & 0.12  & -54.14(0.04)      & 0.13 & -54.14(0.04)      & 2.43    \\
     10    & 1024 &           &       &              &       & -54.09(0.02)       & 0.12  & -54.14(0.03)      & 0.13 & -54.14(0.03)      & 2.43    \\
     10    & 2048 &           &       &              &       & -54.09(0.01)       & 0.12  & -54.14(0.02)      & 0.13 & -54.14(0.02)      & 2.43    \\
     10    & 4096 &           &       &              &       & -54.09(0.01)       & 0.12  & -54.14(0.01)      & 0.16 & -54.14(0.01)      & 2.43    \\
     10    & 8192 &           &       &              &       & -54.09(0.01)       & 0.12  & -54.14(0.01)      & 0.3  & -54.14(0.01)      & 2.44    \\
     \hline
     100   & 256  & -5275.99  & 1.77 & -5275.99     & 0.001 & -5275.15(0.12)     & 1.12  & -5275.94(0.2)     & 1.31 & -5275.91(0.25)    & 24.0   \\
     100   & 512  &           &       &              &       & -5275.13(0.07)     & 1.13  & -5275.95(0.16)    & 1.31 & -5275.96(0.2)     & 24.1    \\
     100   & 1024 &           &       &              &       & -5275.14(0.05)     & 1.14  & -5275.97(0.14)    & 1.47 & -5275.95(0.14)    & 24.2   \\
     100   & 2048 &           &       &              &       & -5275.14(0.04)     & 1.18  & -5275.96(0.1)     & 2.39 & -5275.96(0.09)    & 24.6   \\
     100   & 4096 &           &       &              &       & -5275.14(0.03)     & 1.2   & -5276.0(0.07)     & 6.01 & -5276.0(0.07)     & 24.5   \\
     100   & 8192 &           &       &              &       & -5275.14(0.02)     & 1.28  & Out of memory     &      & -5275.97(0.05)    & 36.2   \\
     \hline
     1000  & 256  & -71208.4 & 1.87 & -71208.4    & 0.001 & -74890.3(87.9)   & 1.23  & -71247.7(10.2)  & 1.47 & -71248.4(9.0)   & 25.8   \\
     1000  & 512  &           &       &              &       & -74736.0(72.0)    & 1.26  & -71228.8(6.3)   & 2.09 & -71229.1(6.5)    & 25.8   \\
     1000  & 1024 &           &       &              &       & -74563.2(66.8)   & 1.33  & -71218.3(4.0)   & 4.55 & -71217.7(4.8)   & 26.2    \\
     1000  & 2048 &           &       &              &       & -74432.7(85.9)   & 1.48  & Out of memory     &      & -71213.1(2.9)    & 31.3    \\
     1000  & 4096 &           &       &              &       & -74322.0(74.5)   & 1.81  & Out of memory     &      & -71210.8(2.1)     & 67.6   \\
     1000  & 8192 &           &       &              &       & -74178.2(71.1)   & 2.63  & Out of memory     &      & -71209.7(1.7)   & 211  \\
     \hline
     10000 & 256  & -731878 & 1.81 & -731877    & 0.001 & -786077(294)  & 1.44  & -732295(31) & 3.8  & -732294(28)  & 223  \\
     10000 & 512  &           &       &              &       & -785395(327)  & 1.99  & Out of memory     &      & -732085(18)  & 231  \\
     10000 & 1024 &           &       &              &       & -784872(328) & 3.05  & Out of memory     &      & -731983(16)   & 231  \\
     10000 & 2048 &           &       &              &       & -784332(261)  & 5.16  & Out of memory     &      & -731928(10)  & 287  \\
     10000 & 4096 &           &       &              &       & -783742(313) & 9.78  & Out of memory     &      & -731904(7)   & 649  \\
     10000 & 8192 &           &       &              &       & -783216(297)  & 19.3 & Out of memory     &      & -731891(5)   & 2097 \\
    \hline
    \end{tabular}
    }
\end{table*}

For each population size, we run the considered algorithms on the DGP to estimate the log-marginal likelihood, which we report in Table~\ref{tab:SMC_comparison_table} as the mean and standard deviation over $100$ runs (the CAL requires only a single run, as it is a deterministic algorithm). Recall that the APF provides unbiased estimates of the marginal likelihood \citep{johansen2008note}, but it is affected by the curse of dimensionality: typically the number of particles must increase exponentially with the dimension to ensure reliable (i.e. low-variance) estimates. Moreover, although the marginal likelihood estimates are unbiased, the corresponding log-marginal likelihood estimates are negatively biased due to Jensen’s inequality. 

For $N = 10$ and $N = 100$, the APF exhibits low variance, giving us a suitable proxy for the ground truth. Both the CAL and the Block APF produce results close to those of the APF, with the Block APF being slightly closer. In terms of running times, the CAL (not JIT-compiled) is slower than both the APF and the Block APF, this is likely because sampling is more efficient than linear-algebra operations when the population is small. For $N = 1000$ and $N = 10000$, two issues arise for the SMC algorithms. First, the variance of the APF is large and the estimated log-marginal likelihood increases with the number of particles, suggesting more particles are needed. This is also confirmed by the fact that the variance does not decrease with the number of particles. Second, the Block APF requires more memory to run in parallel and must be converted into a sequential version (the batched Block APF). In this large-population regime, the CAL outperforms the SMC algorithms in terms of running time, and, being deterministic, it avoids noisy log-marginal likelihood estimates. It also appears that the (batched) Block APF has a similar asymptotic behavior, as $N, P \to \infty$, to the CAL, when $N \to \infty$, which may be of independent interest, see also Section \ref{sec:app_exp_SMC_only} of the supplementary materials.

It is important to note that if we consider more compartments, and models in which not all movements across compartments are possible (e.g. in an SIR model individuals cannot move from $S$ to $R$, and viceversa) the APF and Block APF will become degenerate as $N$ increases. Indeed, the proposal distributions of the APF and Block APF rely only on the observation at the current time step and cannot prevent future mismatches. In such situations, it may be necessary to use lookahead filters \citep{rimella2022approximating}, see Section \ref{sec:app_exp_SMC} of the supplementary materials for some experiments.

\subsection{2001 UK foot-and-mouth disease outbreak}\label{sec:exp_FM}

In 2001 a foot-and-mouth disease outbreak infected 2026 out of 188361 farms in the United Kingdom, resulting in damages and costs of about 8 billion pounds. This dataset has been studied in the context of individual-based models \citep{jewell2009bayesian,FergusonFM2009,deardon2010inference}, where farms are considered as individuals with the location and the number of animals forming the individual-specific covariates. A farm can be susceptible if no animals are infected, infected if at least one animal is infected, and removed if the farm exits the epidemic due to quarantine or culling. For this study, we consider $162775$ farms, from England, Wales, and the south of Scotland, see Figure \ref{fig:FM_spread}. A fully dense spatial kernel model would be costly to run for all farms, hence we adopt the network model strategy of Section \ref{sec:exp_mispec}, and assign each farm to a local authority, full details are in Section \ref{sec:app_exp_FM} of the supplementary materials. 

\begin{figure}[httb!]
    \centering
    \includegraphics[width=0.8\textwidth]{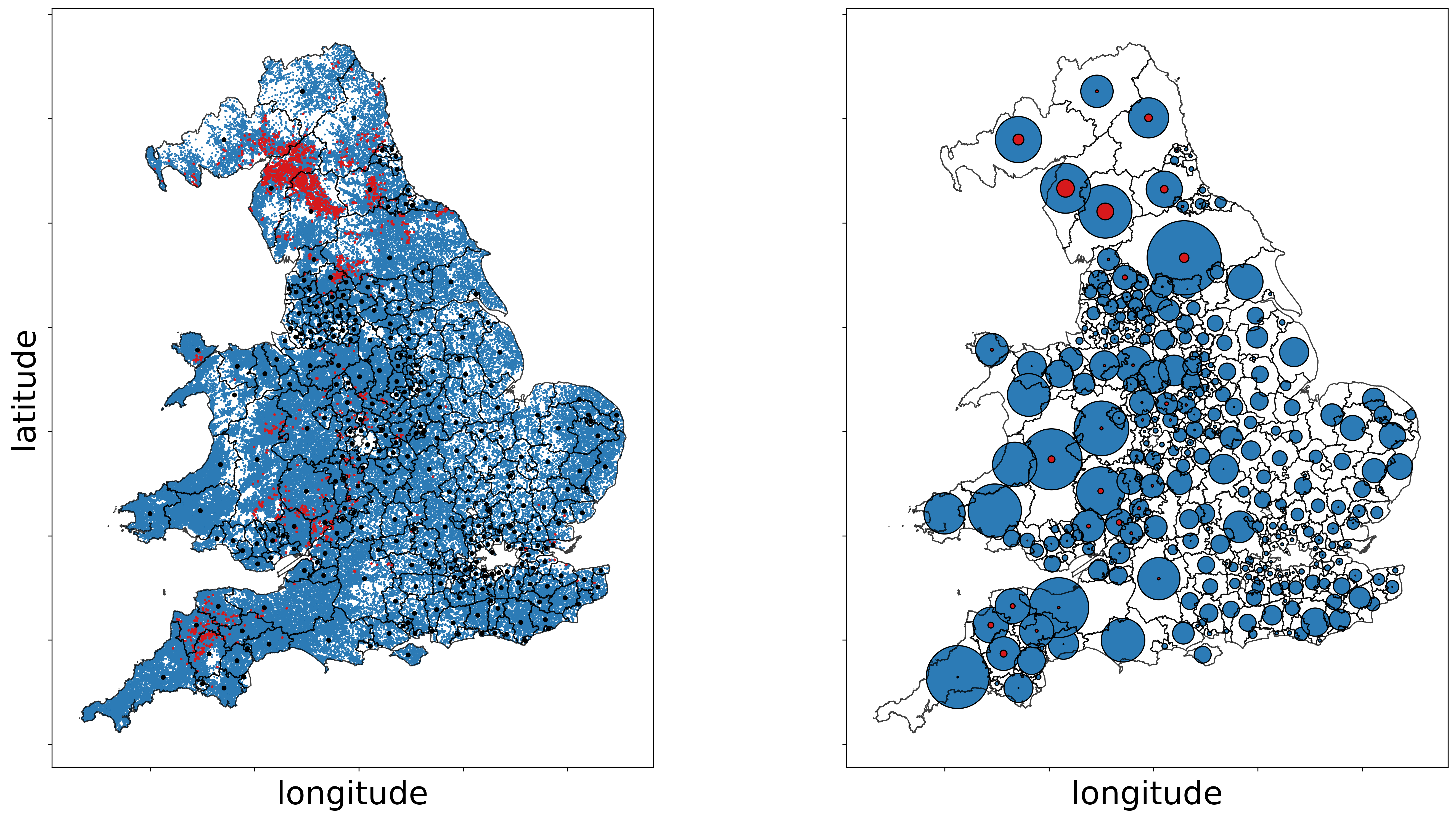}
     \caption{The farms and the local authorities included in the study. On the left, dots represent farms, with red indicating that the farm was reported infected at some point in time. On the right, the local authorities are blue circles with a radius proportional to the number of farms. Red inner circles are proportional to the number of farms within the local authority that were reported infected during the outbreak. Black contours represent the geometries of the local authorities.}
    \label{fig:FM_spread}
\end{figure}

We consider a heterogeneous-mixing individual-based SIR model as in Section \ref{sec:exp_mispec}, where transitions from $S$ to $R$ are also allowed, representing the culling\slash quarantine of healthy farms to create a containment zone around infected farms. In this model, we have two interaction terms: one controlling the spread of the disease, the other controlling the intensity of culling\slash quarantining. For the observation model, we do not allow for misreporting and we assume that susceptible and removed are always unreported. 

We optimized the parameters with Adam in two steps. In the first step, we performed a pre-optimization over different initial conditions, with each optimization taking around 70 minutes (10000 gradient steps). This resulted in a log-CAL of $–17947.27$ for the best combination of parameters. In the second step, we selected the best log-CAL parameters from the first step and used them as a warm start for a longer optimization with a varying step size, which took about 48 hours (400000 gradient steps). This produced a log-CAL of $–17946.734$. Note that the improvement is marginal, and we could probably have used fewer gradient steps. After the optimization we further used the resulting parameters as a warm start for an HMC to provide credible bands. Full details on the model, optimization and HMC are provided in Section \ref{sec:app_exp_FM} of the supplementary materials.

\begin{figure}[httb!]
    \centering
    \includegraphics[width=\textwidth]{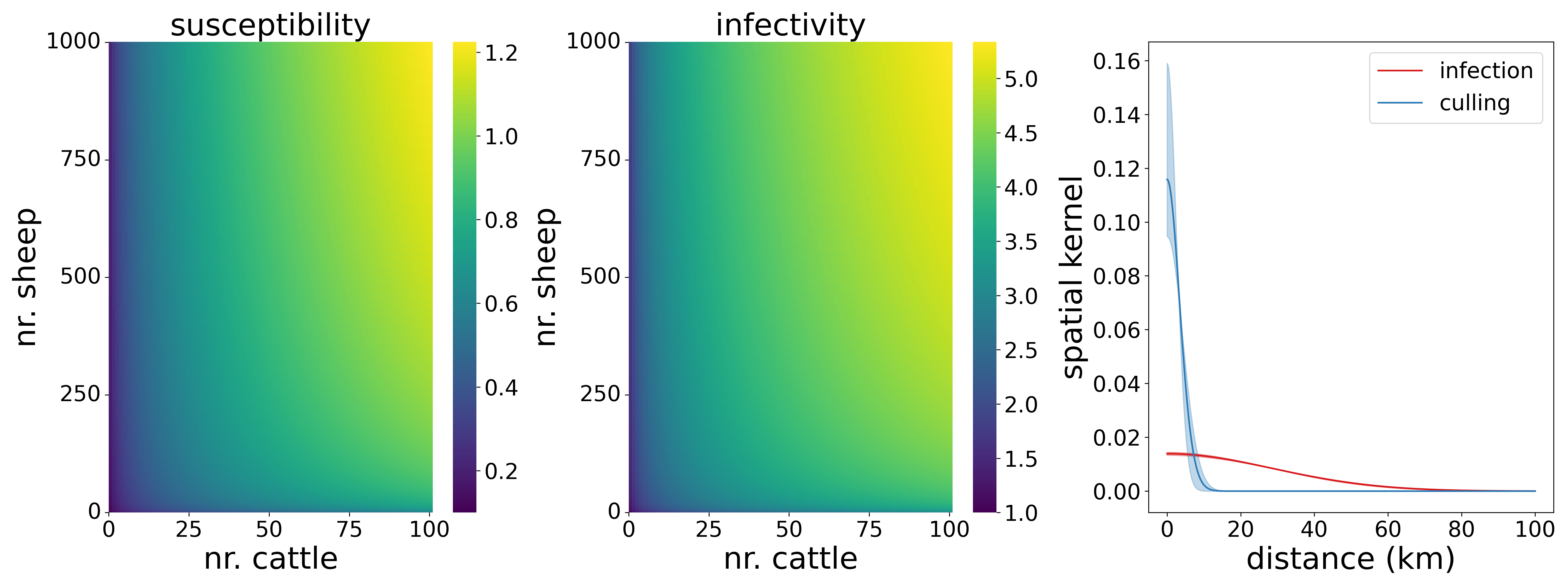}
    \caption{On the left and in the center are the heat maps of the inferred susceptibility and infectivity. On the right is the inferred spatial kernel effect on both infection and culling as a function of the distance in kilometers. The solid lines represent posterior means and the shaded bands represent $95\%$ credible intervals.}
    \label{fig:inf_susc}
\end{figure}

\begin{figure}[httb!]
    \centering
    \includegraphics[width=\linewidth]{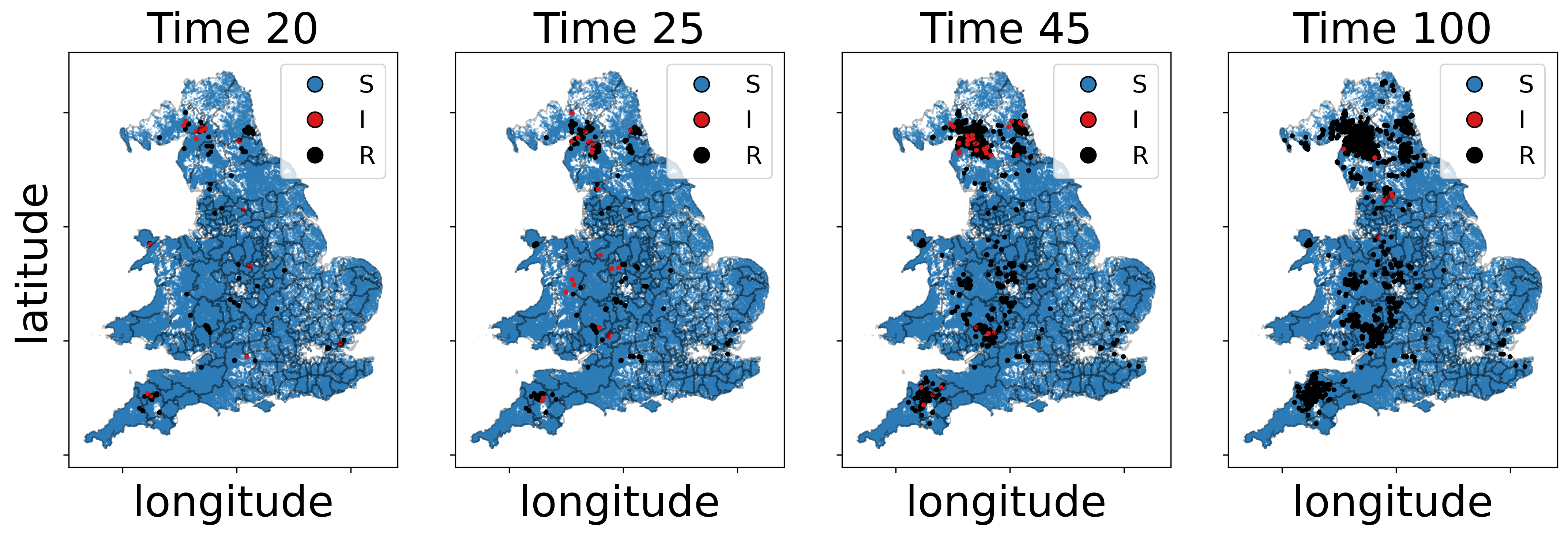}
    \caption{The CAL prediction over time of susceptible farms (blue), infected farms (red), and removed farms (black), over $t=20,25,45,100$. Parameters are set to the posterior mean of the HMC.}
    \label{fig:FM_spat_track}
\end{figure}

After the HMC run, we can study the inferred mean effects of owning cattle\slash sheep on susceptibility\slash infectivity. The plots on the left and in the center of Figure \ref{fig:inf_susc} show that owning cattle affects both susceptibility and infectivity more than owning sheep, which is a similar conclusion to \cite{jewell2009bayesian} and \cite{rimella2025simulation}. We can also analyze how distance affects both the infection and the culling\slash quarantine processes. No data about culled farms are included in the calibration, meaning that the culling\slash quarantine process is inferred from the infected farms only. The right plot of Figure \ref{fig:inf_susc} shows that the infection process travels on the order of 60km, with a culling\slash quarantine process that is applied up to a radius of 10km, which is also in line with the ``surveillance zone'' protocol considered by the government \citep{SZ_foot}. We further included $95\%$ credible interval as shaded bands. Following the reasoning of Section \ref{sec:exp_mispec}, we can plot over time the predicted state via the CAL filter, see Figure \ref{fig:FM_spat_track}, when the parameters are set to the posterior mean of the HMC. We consider $t=20, 25$ to show the spatial spread and culling\slash quarantine process in the short term, $t=45$ to show the spread when we are close to the peak of the epidemic, and $t=100$ to show when we are close to the end. We observe that the epidemic rapidly spread in the north of England, Cumbria in particular, and Cornwall, to then affect Wales on a lower scale. Even though no information about culling\slash quarantine was fed to the model, the CAL is able to spatially detect them and match the locations of confirmed and suspected premises towards the end of the outbreak (October 2001) as reported graphically on the UK government website \citep{locFM}.

\subsubsection{Benchmarking, overdispersion and extensions}\label{sec:FM_bench_misp}

A natural question that arises is whether the considered model is correctly specified. This can be assessed by comparing its log-likelihood values with those of simple benchmarks that are much easier to fit. With this in mind, we consider $\widetilde{\by}_{n,t}$ which takes values $1,0$ depending on whether the farm is reported as infected or not. We then model $p(\widetilde{\by}_{n,0}=1),p(\widetilde{\by}_{n,t}|\widetilde{\by}_{n,t-1})$ with an AR logistic regression:
$$
p(\widetilde{\by}_{n,1}=1) = \frac{\gamma}{1+e^{ -\bb^\top \bw_n}}, \quad p(\widetilde{\by}_{n,t}=1|\widetilde{\by}_{n,t-1}) = \frac{\gamma}{1+e^{-\beta \widetilde{\by}_{n,t-1} -\bb^\top \bw_n}},
$$
where $\gamma\in[0,1]$, $\beta \in \mathbb{R}$ and $\bb \in \mathbb{R}^{C+1}$, with $C$ covariates for each individual (plus one for the intercept). We optimize $\gamma, \beta, \bb$ using gradient ascent on the log-likelihood with the Adam optimizer, obtaining a log-likelihood of $-20858.99$. We further consider another AR logistic regression benchmark in which $\gamma, \beta, \bb$ are specific to each local authority, and fit the parameters in the same way. This second benchmark yields a log-likelihood of $-18866.53$. Both benchmarks provide sensible log-likelihood values that are not far from the $-17946.73$ achieved by the CAL, but they remain lower, suggesting that the individual-based model approach explains the data better. The usefulness of such benchmarks extends beyond the foot-and-mouth application and the CAL: these AR logistic regressions can be easily generalized to multiple states and can be used to check whether an individual-based model is misspecified. 

In our study, we focused on models that may underestimate uncertainty and therefore produce overconfident parameter estimates. Evaluating overdispersed models is thus crucial for assessing the quality of the inferential procedure and avoiding model misspecification \citep{stocks2020model,whitehouse2023consistent,li2024inference}. 

We next consider the same SIR model as in Section \ref{sec:exp_FM} but replace the transmission parameter $\beta$ with $\beta \xi_t$, where $\log \xi_t$ is distributed according to a Gaussian distribution with unknown mean $\mu_o$ and unknown standard deviation $\sigma_o$. We examine two cases: a model with ``shared'' overdispersion, where $\xi_t$ is common to all farms, and a model with ``local authority'' overdispersion, where a separate $\xi_{B,t}$ is defined for each local authority $B$ and shared across the farms within that authority.

Following the procedure of \cite{whitehouse2023consistent}, we nest the CAL within an SMC algorithm to estimate the marginal likelihood of the overdispersed models. Here the CAL marginalizes over the latent individuals' states, while the SMC marginalizes over the stochastic parameters. We perform a grid search over $(\mu_o, \sigma_o)$ for the ``shared'' overdispersion model, where the values on the grid are given by the CAL within SMC algorithm. This yields $\mu_o = 0$, $\sigma_o = 0.25$, and a log-marginal likelihood of $-17926.99(0.46)$ (the standard deviation is estimated over $100$ runs of the SMC). We then estimate the log-marginal likelihood for the ``local authority'' model and obtain a value of $-17907.55(0.84)$ (the standard deviation is estimated over $100$ runs of the SMC). As the ``local authority'' model considers 289 stochastic parameters (one for each local authority) we considered a block SMC approach \citep{rebeschini2015can}. More details on the models, the SMC algorithms, and the grid search are available in Section \ref{sec:app_exp_FM} of the supplementary materials.

We observe that both overdispersed models achieve better log-marginal likelihoods than the other model, but this improvement comes at a high computational cost: a single run of CAL within SMC takes approximately 7 minutes for the ``shared'' overdispersion model with 1024 particles and about 43 minutes for the ``local authority'' overdispersion model with 512 particles. Moreover, neither JIT compilation nor automatic differentiation is straightforward in this context. It is therefore important to consider alternative SMC methods that preserve these properties \citep{corenflos2021differentiable,tan2024accelerated} and/or leverage parallel-in-time implementations \citep{corenflos2022sequentialized}.

From a theoretical perspective, we know that the CAL performs well when individuals decouple in the large-population limit. In this regime, neighboring individuals being infected have negligible predictive consequences for a considered individual as long as we know the regional force of infection. In the case of the foot-and-mouth disease, the spreading process is not necessarily determined by adjacent contacts but can also involve animal or insect vectors and the livestock transport network \citep{keeling2001dynamics,diggle2006spatio,jewell2009bayesian}. These considerations align with the assumptions and modeling framework of the CAL, we leave to future work the development of modeling approaches in which farms do not decouple.

\section{Discussion, limitations and future work} \label{sec:discussion_future}

We have proposed a computationally and mathematically simple algorithm to enable approximate likelihood-based inference for a broad class of individual-based models of epidemics, supported by both theoretical foundations and practical implementations.

At first glance, the CAL and block particle filters \citep{rebeschini2015can} are similar, as they both exploit certain factorization properties of the model. However, there is a key difference in how their approximations are constructed. The CAL replaces the state of the system at time $t-1$ with its expectation under the previous filtering distribution, whereas the block particle filter applies a blocking approximation after prediction, in the form of independent resampling for each block. Hence, for finite $N$, the CAL targets the likelihood of the approximate model described in Section \ref{sec:CAL_as_exact}, while the block particle filter targets the likelihood of the approximate model obtained by marginalizing over the blocks at each time step. Their theoretical justifications also differ: the CAL is motivated by a large-population limit (i.e. the dimension goes to infinity), while the block particle filter relies on large-block asymptotics (i.e. the partition on the dimensions becomes coarser and coarser). Table \ref{tab:SMC_comparison_table} and the graphical illustration in Section \ref{sec:app_exp_SMC_only} of the supplementary material both suggest that the marginal likelihood approximation from the block particle filter behaves, as $N,P\to\infty$, similarly to the log-CAL, as $N\to\infty$. This indicates that the CAL theory may provide a justification for the use of block particle filters in high-dimensional settings \citep{li2024inference, wheeler2024informing}, which is an interesting direction for further investigation.

One limitation of the CAL is that individuals are updated independently, without accounting for correlations between them. This becomes problematic when the individuals do not decouple in the large population limit, and hence when the state of one individual provides information about others. Some clear examples of such models are household models \citep{rimella2022inference}, where individuals are assigned to households and as the number of individuals grows so does the number of households. In such a scenario, we expect the CAL to provide a bad approximation of the likelihood as it does not account for the dependencies within the households. One possible solution might be to model the state of each household instead of the individuals, hence considering the households to decouple in the large population limit. We are confident that the CAL can be extended to this modeling framework as the categorical representation can be used in any discrete state-space. It is important to note that this will trade off the quality of the approximation with the performance as the computational cost will deteriorate exponentially in the households sizes.


Another key consideration is the implicit geometric distribution assumed for compartmental waiting periods, e.g. infectious periods. Negative binomial waiting times can be approximated by introducing additional compartments, though this approach is not always satisfactory. A promising avenue for future research would be to extend our model to higher-order Markovian or semi-Markovian dynamics \citep{jewell2009bayesian,touloupou2020scalable}.


Throughout our examples, we assume $\bW$ to be static over time, though this assumption could be relaxed to incorporate dynamic covariates, and even integrate a time-evolving contact network process \citep{bu2022likelihood} within the dynamics. Additionally, there are interesting avenues for extending our methodology to continuous observation spaces, with applications in areas such as target tracking \citep{whiteley2010auxiliary} and epidemic modeling using continuous serological data \citep{hay2024serodynamics}.


\section*{Funding}

MW acknowledges funding from the MRC Centre for Global Infectious Disease Analysis (Reference No. MR/X020258/1) funded by the UK Medical Research Council. This UK-funded grant is carried out in the frame of the Global Health EDCTP3 Joint Undertaking. PF acknowledges funding from the EPSRC grant Prob\_AI (Reference No. EP/Y028783/1). CJ acknowledges funding from MRC (MR/S004793/1), BBSRC (BB/T004312/1), EPSRC (EP/V042866/1), and Research England as part of the E3: Expanding Excellence in England program.

\section*{Acknowledgements}

The authors thank ``The High-End Computing'' (HEC) facility at Lancaster University for providing the computational resources to run the experiments. The authors thank Dr. Sam Power for the useful discussion on ``exact inference for an approximate model'' and Dr. Louis Sharrock for providing helpful references on the propagation of chaos. 

\bibliographystyle{chicago}
\bibliography{references.bib}

\appendix

\section{Notation and assumptions}

\subsection{General notation}
Given an integer $M \in \mathbb{N}$, we use $x_{0:M} \coloneqq x_0,\dots,x_M$ for indexing sequences, $[M]\coloneqq \left \{1,\dots, M \right \}$ for the set of the first $M$ positive integers, and $\bx \coloneqq \left (\bx^{(1)},\dots,\bx^{(M)} \right )$ for an $M$-dimensional vector. Given two $M$-dimensional vectors $\bx_1$ and $\bx_2$ we denote with $\bx_1 \odot \bx_2$ the element-wise product and with $\bx_1 \oslash \bx_2$ the element-wise division. We write $1_M$ for the $M$-dimensional vector of all ones, $\Delta_M$ for the $M$-dimensional probability simplex, i.e. $\Delta_M \coloneqq \left \{\bx \in [0,1]^M \text{ : } \sum_{i=1}^M \bx^{(i)}=1 \right \}$, and $\mathbb{O}_M$ for the set of one-hot encoding vectors with dimension $M$, i.e. the set $\mathbb{O}_M \coloneqq \left \{\bx \in \{0,1\}^M \text{ : } \exists j \in [M]: \bx^{(j)}=1 \text{ and } \bx^{(i)}=0 \text{ if } i\neq j \right \}$. Note that $\mathbb{O}_M \subset \Delta_M$. Given $\bpi \in \Delta_M$ we denote with $\mbox{Cat}(\cdot|\bpi)$ the categorical distribution over $\mathbb{O}_M$ which assigns probability $\bpi^{(i)}$ to the vector $\bx \in \mathbb{O}_M$ with $\bx^{(i)}=1$ and $\bx^{(j)}=0$ for $j\neq i$. 

We shall work with the following norms.
    For an $M$-dimensional vector $\pi$ and an $M\times M$-matrix $K$ we define 
    $$
    \norm{\pi}_\infty \coloneqq \max\limits_{i\in[M]} \abs{\pi^{(i)}}, \quad 
    \norm{K}_\infty \coloneqq \max_{i \in [M]} \sum_{j = 1}^M \abs{K^{(i,j)}}.
    $$
    For a vector-valued function $f: \mathbb{S} \to \mathbb{R}^M$, we write 
    $$
    \norm{f}_\infty \coloneqq  \sup_{s \in \mathbb{S}}\max_{i \in [M]} \abs{f(s)^{(i)}},
    $$
    For an $\mathbb{R}$-valued random variable $\bx$ the $L^4$ norm is written:
    $$
    \normiii{\bx}_4 \coloneqq \left ( \mathbb{E} \left[ |\bx|^4 \right ]\right )^{\frac{1}{4}}.
    $$

We need a notation for the support of probability vectors and matrices. For an $M$-dimensional vector $\bpi$ we define $\textbf{supp}(\bpi) \coloneqq \{i \in [M]: \bpi^{(i)} \neq 0\}$. Similarly for a $M \times M$ matrix $\bPi$ we define $\textbf{supp}(\bPi) \coloneqq \{(i,j) \in [M]^2: \bPi^{(i,j)} \neq 0\}$.

\subsection{Model and CAL notation}

We consider the individual-based model and the CAL algorithm described in the main paper.  To formulate and prove our theoretical results, we need to make explicit the dependence of various quantities on the parameter vector $\theta \in \Theta$, covariates $\bw_n$, and/or the population size $N$. 

We write $\bW^N=(\bw_1,...,\bw_N)$ for the first $N$ population covariate vectors, $\bx_{n,t}^N$ for the state at time $t$ of the $n$-th individual, $\bX^N_t = (\bx_{1,t}^N,\ldots,\bx_{N,t}^N)$ for the state of the population at time $t$, $\by_{n,t}^N$ and $\bY^N_t$ for the observations at an individual and population level.

We use for the initial distribution, transition matrix, and emission matrix the notation: $p_0(\bw_n,\theta)$, $K_{\cdot}(\bw_n,\theta)$, and $G(\bw_n,\theta)$, with $\eta(\bw_n,\bW,\bX)$ becoming $\eta^N(\bw_n,\theta, \bW^N,\bX^N)$, and $\boeta_{n,t}$ becoming $\boeta^N_{t}(\bw_n,\theta)$. For the quantities computed in the  CAL algorithm, we highlight the functional dependence on the covariates and the considered parameter value, with $\bpi_{n,t|t-1}, \bmu_{n,t}, \bpi_{n,t}, \widetilde{\boeta}_{n,t}$ becoming $\bpi^N_{n,t|t-1}(\bw_n,\theta)$, $\bmu^N_{n,t}(\bw_n,\theta)$, $\bpi^N_{n,t}(\bw_n,\theta), \widetilde{\boeta}_{t}^N(\bw_n,\theta)$. Note that for $t\geq1$ the quantities $\bpi^N_{n,t|t-1}(\bw_n,\theta), \bmu^N_{n,t}(\bw_n,\theta)$ also depend on the population covariates $\bW^N$ and the observations $\bY_{1:t-1}^N$, and similarly the quantities $\bpi^N_{n,t}(\bw_n,\theta),\widetilde{\boeta}_{t}^N(\bw_n,\theta)$ depend on $\bW^N$ and $\bY_{1:t}^N$, but these dependencies are not shown in the notation. 

For completeness we write the CAL recursion as in Algorithm \ref{alg:CAL} with the dependence on $\bw_n$ and $\theta$ explicit in the notation, for any $n \in [N]$:
\begin{equation}\label{rec:CAL_rec}
\begin{aligned}
\bpi^N_{n,0}(\bw_n,\theta) &\coloneqq p_0(\bw_n, \theta),\\
    \bPi_{t-1}^N &\coloneqq \left ( \bpi^N_{n,t-1}(\bw_1, \theta), \dots, \bpi_{n,t-1}^N(\bw_N, \theta) \right ), \\
    \widetilde{\boeta}_{t-1}^N(\bw_n,\theta) &\coloneqq \eta^N(\bw_n, \theta, \bW^N, \bPi^N_{t-1}),\\
    \bpi^N_{n,t|t-1}(\bw_n,\theta)  &\coloneqq  \left [ \bpi^N_{n,t-1}(\bw_n,\theta)^\top K_{\widetilde{\boeta}_{t-1}^N(\bw_n,\theta)}(\bw_n,\theta) \right ]^\top,\\
    \bmu^N_{n,t}(\bw_n,\theta)   &\coloneqq \left [ \bpi^N_{n,t|t-1}(\bw_n,\theta)^\top  G(\bw_n,\theta)  \right ]^\top,\\
    \bpi^N_{n,t}(\bw_n,\theta)  &\coloneqq \bpi_{n,t|t-1}^N(\bw_n,\theta) \odot \left \{  \left [  G(\bw_n,\theta) \oslash  \left ( 1_M \bmu_{n,t}^N(\bw_n,\theta)^\top \right ) \right ] \by_{n,t}^N \right \}.
\end{aligned}
\end{equation}

\subsection{Assumptions}
The compactness and continuity in the following assumption are standard conditions in the consistency theory of maximum likelihood estimators.
\begin{assumption}\label{ass:compactness_continuity}
    The parameter space $\Theta$ and the covariate space $\mathbb{W}$ are compact subsets of Euclidean spaces. Moreover, the initial distribution $p_0(w,\theta)$, the transition matrix $K_{\eta}(w,\theta)$, and the emission matrix $G(w,\theta)$ are all continuous functions in their arguments $w,\theta$.
\end{assumption}

For the purposes of our theory, we shall treat the individual-specific covariate vectors as random, independent, and identically distributed across the population. This can be interpreted as a random design assumption of a sort, which is commonly adopted in the asymptotic studies of regression and classification methods. 

\begin{assumption}\label{ass:w_iid}
    The covariates $\bw_1,\bw_2,\dots$ are independent and identically distributed according to a distribution $\Gamma$ on $\mathbb{W}$. 
\end{assumption}

The following assumption will be used to establish that the logarithm of the CAL is well-defined across all possible values of parameters, covariates, and almost all realizations of the data.
\begin{assumption}\label{ass:HMM_support}
    The following hold:
    \begin{itemize}
        \item for any $w \in \mathbb{W}$ and $\theta,\theta' \in \Theta$ we have that $\textbf{supp}(p_0(w, \theta)) = \textbf{supp}(p_0(w, \theta'))$;
        \item for any $w \in \mathbb{W}$, $\eta,\eta' \in [0,C]$ and $\theta,\theta' \in \Theta$ we have that $\textbf{supp}(K_{\eta}(w,\theta)) = \textbf{supp}(K_{\eta'}(w,\theta'))$;
        \item for any $w \in \mathbb{W}$ and $\theta,\theta' \in \Theta$ we have that $\textbf{supp}( G(w,\theta)) = \textbf{supp}( G(w,\theta'))$.
    \end{itemize}
\end{assumption}

The following assumption constrains the form of the function $\eta$ which determines the mechanism of interaction amongst the population.
\begin{assumption}\label{ass:eta_structure}
    For any $\theta \in \Theta, w \in \mathbb{W}, N\in \mathbb{N}$, and for any $W^N = (w_1,\dots, w_N),\Pi^N = (\pi_1,\dots,\pi_N)$ with $w_n \in \mathbb{W},\pi_n \in \Delta_M$ for all $n \in [N]$, we have:
    \begin{equation}
        \eta^N(w, \theta, W^N, \Pi^N) = \frac{1}{N} \sum_{n \in [N]} d(w, w_{n}, \theta)^\top \pi_{n},
    \end{equation}
    where $d:\mathbb{W} \times \mathbb{W} \times \Theta \to [0,C]^M$ is a bounded function, i.e. $\norm{d}_\infty \leq C < \infty$, from which we also obtain $\eta^N(w, \theta, W^N, \Pi^N) \in [0,C]$ for any $N, w,\pi, W^N, \Pi^N$.
\end{assumption}

\begin{assumption}\label{ass:kernel_continuity}
    For any $\theta \in \Theta$ and $w \in \mathbb{W}$, the matrix $K_{\eta}(w,\theta)$ is Lipschitz continuous in $\eta$ with Lipschitz constant $L$, that is for any $\eta, \eta^\prime \in [0,C]$ we have:
    \begin{equation}
        \norm{K_{\eta}(w,\theta) - K_{\eta^\prime}(w,\theta)}_\infty \leq L \abs{\eta - \eta^\prime}.
    \end{equation}
\end{assumption}

\subsection{The data-generating process}\label{sec:data_generating_process}
All the random variables appearing in our theory are assumed to be defined on a common probability space $(\Omega, \mathcal{F}, \mathbb{P})$. Thus in accordance with Assumption \ref{ass:w_iid}, under $\mathbb{P}$, $\bw_1,\bw_2,\ldots$ are i.i.d. 

Let $\theta^\star$ be an arbitrarily chosen but a fixed member of $\Theta$, which will be referred to as the data-generating parameter (DGP). For each $N\geq1$, the DGP determines the  distributions of $(\bX_{t}^N)_{t\geq 0}$ and $(\bY_{t}^N)_{t\geq 1}$ under $\mathbb{P}$ conditional on $\bW^N$, in that for each $n\in[N]$,
\begin{align}
	&\bx^N_{n,0} | \bw_n \sim \mbox{Cat} \left (\, \cdot\,|\,p_0( \bw_n,\theta^\star) \right ), \\
	&\bx_{n,t}^N|\bX_{t-1}^N, \bW^N  \sim \mbox{Cat} \left(\, \cdot\,\left| \left [ (\bx_{n,t-1}^N)^\top K_{\boeta_{t-1}^N(\bw_n,\theta^\star)}(\bw_n,\theta^\star) \right ]^\top \right.\right ),\\
        &\by_{n,t}^N|\bx_{n,t}^N,\bw_n \sim \mbox{Cat} \left (\, \cdot\,\left| \left [ (\bx_{n,t}^N)^\top  G(\bw_n,\theta^\star) \right ]^\top \right.\right).
\end{align}

\section{Closed-forms under categorical approximations} \label{sec:closed_form_CAL}

Proposition \ref{prop:CAL_prediction_conj} and Proposition \ref{prop:CAL_update_conj} below show how the formulae appearing in the CAL algorithm in Section \ref{sec:approx_filtering} of the main paper are derived and can be interpreted as prediction and correction operations associated with the approximate model specified in Section \ref{sec:CAL_as_exact} of the main paper.

\begin{proposition} \label{prop:CAL_prediction_conj}
    If $\widetilde{\bX}_{t-1}^N| \bW^N \sim \bigotimes_{n \in [N]} \mbox{Cat}(\cdot|\bpi_{n,t-1}^N(\bw_n, \theta))$ and:
    $$
    \widetilde{\bx}_{n,t}^N|\widetilde{\bX}_{t-1}^N,\bW^N \sim \mbox{Cat} \left ( \cdot | \left[ (\widetilde{\bx}_{n,t-1}^N)^\top K_{\widetilde{\boeta}_{t-1}^N(\bw_n,\theta)}(\bw_n,\theta) \right ]^\top \right )  \quad \text{for } n \in [N],
    $$ 
    then $\widetilde{\bX}_{t}^N|\bW^N \sim \bigotimes_{n \in [N]} \mbox{Cat}(\cdot|\bpi_{n,t|t-1}^N(\bw_n,\theta))$ with:
    $$
    \bpi_{n,t|t-1}^N(\bw_n,\theta) =  \left (\bpi_{n,t-1}^N(\bw_n, \theta)^{\top} K_{\widetilde{\boeta}_{t-1}^N(\bw_n,\theta)}(\bw_n,\theta) \right )^{\top} \quad \text{for } n \in [N].
    $$
\end{proposition}

\begin{proof}
    We just want to compute the marginal at time $t$ after applying the approximate transition kernel:
    \begin{equation}
            \begin{split}
            &\mathbb{P}(\widetilde{\bX}_t^N = \widetilde{X}_t^N|\bW^N) 
                \\
                &= 
                \sum_{X^N = (x_{1}, \dots, x_{N})} \mathbb{P}(\widetilde{\bX}_{t-1}^N=X^N|\bW^N)\mathbb{P}(\widetilde{\bX}_t^N = \widetilde{X}_t^N |\bW^N,\widetilde{\bX}_{t-1}^N = X^N) \\
                &= 
                \sum_{X^N = (x_{1}, \dots, x_{N})} \prod_{n \in [N]} \left [ \left( x_{n}^\top \bpi_{n,t-1}^N(\bw_n, \theta) \right) \left ( x_{n}^\top K_{\widetilde{\boeta}_{t-1}^N(\bw_n,\theta)}(\bw_n,\theta) \widetilde{x}_{n,t}\right )\right ]\\
                &=
                \prod_{n \in [N]} \sum_{x_{n}} \left [ \left( x_{n}^\top \bpi_{n,t-1}^N(\bw_n, \theta) \right) \left ( x_{n}^\top K_{\widetilde{\boeta}_{t-1}^N(\bw_n,\theta)}(\bw_n,\theta) \widetilde{x}_{n,t}^N \right )\right ] 
                \\
                &=
                \prod_{n \in [N]} \sum_{i=1}^M \left [\bpi_{n,t-1}^N(\bw_n, \theta)^{(i)} K_{\widetilde{\boeta}_{t-1}^N(\bw_n,\theta)}(\bw_n,\theta)^{(i,\cdot)} \widetilde{x}_{n,t}^N \right ] \\
                & = \prod_{n \in [N]} \left [ \bpi_{n,t-1}^N(\bw_n, \theta)^{\top}  K_{\widetilde{\boeta}_{t-1}^N(\bw_n,\theta)}(\bw_n,\theta) \right ]^\top \widetilde{x}_{n,t}^N.
            \end{split}
    \end{equation}
\end{proof}

The next proposition considers the correction step, where we perform a Bayes update with a factorized categorical distribution prior and a likelihood that is given by the emission distribution of the individual-based model. As the posterior distribution is in the family of the prior distribution, we say that the factorized categorical distribution is a conjugate prior for the individual-based model observation model.

\begin{proposition}\label{prop:CAL_update_conj}
    If $\widetilde{\bX}_{t}^N|\bW^N \sim \bigotimes_{n \in [N]} \mbox{Cat}(\cdot|\bpi_{n,t|t-1}^N(\bw_n,\theta))$ and:
    $$
    \widetilde{\by}_{n,t}^N|\widetilde{\bx}_{n,t}^N,\bw_n \sim \mbox{Cat}\left (\cdot| \left [ (\widetilde{\bx}_{n,t}^N)^\top  G(\bw_n,\theta) \right ]^\top \right ) \quad \text{for } n \in [N],
    $$
    then $\widetilde{\bY}_{t}^N|\bW^N \sim \bigotimes_{n \in [N]} \mbox{Cat}(\cdot|\bmu_{n,t}^N(\bw_n,\theta))$ with:
    \begin{equation}
        \bmu_{n,t}^N(\bw_n,\theta) = \left [ \bpi_{n,t|t-1}^N(\bw_n,\theta)^\top  G(\bw_n,\theta) \right ]^\top \quad \text{for } n \in [N],
    \end{equation}
    and $\widetilde{\bX}_{t}^N|\widetilde{\bY}_{t}^N, \bW^N \sim \bigotimes_{n \in [N]} \mbox{Cat}(\cdot|\bpi_{n,t}^N(\bw_n,\theta))$ with:
    \begin{equation}
        \bpi_{n,t}^N(\bw_n,\theta) = \bpi_{n,t|t-1}^N(\bw_n,\theta) \odot \left \{  \left [  G(\bw_n,\theta) \oslash  \left ( 1_M \bmu_{n,t}^N(\bw_n,\theta)^\top \right ) \right ] \widetilde{\by}_{n,t}^N \right \} \quad \text{for } n \in [N].
    \end{equation} 
\end{proposition}

\begin{proof}
    We start by computing the marginal likelihood:
    \begin{equation}
        \begin{split}
            \mathbb{P}(\widetilde{\bY}_t^N=\widetilde{Y}_t^N|\bW^N) 
            &= 
            \sum_{X^N=(x_1,\dots,x_N)} \mathbb{P}(\widetilde{\bY}_t^N=\widetilde{Y}_t^N |\bW^N, \widetilde{\bX}_t^N = X^N) \mathbb{P}(\widetilde{\bX}_t^N = X^N|\bW^N)\\
            &= 
            \sum_{X^N=(x_1,\dots,x_N)} \prod_{n \in [N]}  \left ( x_{n}^\top  G(\bw_n,\theta) \widetilde{y}_{n,t}^N \right ) \left ( x_{n}^\top \bpi_{n,t|t-1}^N(\bw_n,\theta) \right ) \\
            &= 
            \prod_{n \in [N]} \sum_{x_{n}}   \left ( x_{n}^\top G(\bw_n,\theta) \widetilde{y}_{n,t}^N \right ) \left ( x_{n}^\top \bpi_{n,t|t-1}^N(\bw_n,\theta) \right )\\
            &= 
            \prod_{n \in [N]} \sum_{i=1}^M   \left (  G(\bw_n,\theta)^{(i,\cdot)} \widetilde{y}_{n,t}^N \right )\bpi_{n,t|t-1}^N(\bw_n,\theta)^{(i)}\\
            &= \prod_{n \in [N]} \left [ \bpi_{n,t|t-1}^N(\bw_n,\theta)^\top  G(\bw_n,\theta)  \right ] \widetilde{y}_{n,t}^N,
        \end{split}
    \end{equation}
    showing that the observations are marginally distributed as categorical distributions with the desired parameters.  We can now compute the posterior distribution using Bayes theorem:
    \begin{equation}
        \begin{split}
            \mathbb{P}&(\widetilde{\bX}_t^N=\widetilde{X}_t^N|\bW^N,\widetilde{\bY}_t^N=\widetilde{Y}_t^N) 
            = 
            \frac{\mathbb{P}(\widetilde{\bY}_t^N=\widetilde{Y}_t^N|\bW^N,\widetilde{\bX}_t^N=\widetilde{X}_t^N) \mathbb{P}(\widetilde{\bX}_t^N=\widetilde{X}_t^N|\bW^N)}{\mathbb{P}(\widetilde{\bY}_t^N=\widetilde{Y}_t^N|\bW^N)}\\
            &= 
            \prod_{n \in [N]} \frac{  \left((\widetilde{x}_{n,t}^N)^\top  G(\bw_n,\theta) \widetilde{y}_{n,t}^N \right ) \left ( (\widetilde{x}_{n,t}^N)^\top \bpi_{n,t|t-1}^N(\bw_n,\theta) \right )}{\bmu_{n,t}^N(\bw_n,\theta)^\top \widetilde{y}_{n,t}^N} \\
            &= 
            \prod_{n \in [N]} \frac{(\widetilde{x}_{n,t}^N)^\top\left( G(\bw_n,\theta) \widetilde{y}_{n,t}^N \odot\ \bpi_{n,t|t-1}^N(\bw_n,\theta) \right )}{\bmu_{n,t}^N(\bw_n,\theta)^\top \widetilde{y}_{n,t}^N}\\
            &= 
            \prod_{n \in [N]} (\widetilde{x}_{n,t}^N)^\top\left [ \left (  G(\bw_n,\theta) \widetilde{y}_{n,t}^N \odot\ \bpi_{n,t|t-1}^N(\bw_n,\theta) \right ) \oslash \left ( \bmu_{n,t}^N(\bw_n,\theta)^\top \widetilde{y}_{n,t}^N \right ) \right ]\\
            &= 
            \prod_{n \in [N]} \sum_{i =1}^M \sum_{j = 1}^M (\widetilde{x}_{n,t}^N)^{(i)}  \frac{G(\bw_n,\theta)^{(i,j)} (\widetilde{y}_{n,t}^N)^{(j)} \bpi_{n,t|t-1}^N(\bw_n,\theta)^{(i)}}{\sum_{k \in [M]} \bmu_{n,t}^N(\bw_n,\theta)^{(k)} (\widetilde{y}_{n,t}^N)^{(k)}}\\
            &=
            \prod_{n \in [N]} \sum_{i =1}^M \sum_{j = 1}^M (\widetilde{x}_{n,t}^N)^{(i)} \bpi_{n,t|t-1}^N(\bw_n,\theta)^{(i)} \frac{G(\bw_n,\theta)^{(i,j)}}{\bmu_{n,t}^N(\bw_n,\theta)^{(j)}}(\widetilde{y}_{n,t}^N)^{(j)}\\
            &=
            \prod_{n \in [N]} \left \{ \bpi_{n,t|t-1}^N(\bw_n,\theta) \odot \left \{  \left [  G(\bw_n,\theta) \oslash  \left ( 1_M \bmu_{n,t}^N(\bw_n,\theta)^\top \right ) \right ] \widetilde{y}_{n,t}^N \right \} \right \}^\top \widetilde{x}_{n,t}^N.
        \end{split}
    \end{equation}
\end{proof}

As a consequence of Proposition \ref{prop:CAL_prediction_conj} and \ref{prop:CAL_update_conj}, we can see the CAL as an exact marginal likelihood for the approximate model:
\begin{align}    
&\widetilde{\bx}^N_{n,0}|\bw_n \sim \mbox{Cat} \left ( \cdot|p_0(\bw_n,\theta) \right ), \\
    &\widetilde{\bx}^N_{n,t}|\widetilde{\bX}_{t-1}^N, \widetilde{\bY}^N_{1:t-1}, \bW^N \sim \mbox{Cat} \left ( \cdot| \left [ (\widetilde{\bx}_{n,t-1}^N)^\top K_{\widetilde{\boeta}_{t-1}^N(\bw_n,\theta)}(\bw_n,\theta) \right ]^\top \right ),\\
    &\widetilde{\by}^N_{n,t}|\widetilde{\bx}^N_{n,t},\bw_n \sim \mbox{Cat} \left ( \cdot| \left [ (\widetilde{\bx}_{n,t}^N)^\top  G(\bw_n,\theta) \right ]^\top \right ),
\end{align}
and we can compute both its joint likelihood and marginal likelihood, see the next proposition. 

\begin{proposition} \label{prop:joint_CAL_distribution}
    Over a time horizon $T$, the marginal likelihood of $\widetilde{\bY}_{1:T}^N$ is given by:
    \begin{equation}      
    p(\widetilde{\bY}_{1:T}^N|\bW^N,\theta) \coloneqq  \prod_{t=1}^T \prod_{n \in [N]} \mbox{Cat}\left (\widetilde{\by}_{n,t}^N| \bmu_{n,t}^N(\bw_n,\theta) \right ),
    \end{equation}
which is the marginal of:
    \begin{equation}
    \begin{split}        
    p(\widetilde{\bX}_{0:T}^N, \widetilde{\bY}_{1:T}^N|\bW^N,\theta) &\coloneqq 
        \prod_{n \in [N]}      (\widetilde{\bx}_{n,t}^N)^\top p_0(\bw_n,\theta) 
        \prod_{t=1}^T \left [ (\widetilde{\bx}_{n,t-1}^N)^\top K_{\widetilde{\boeta}_{t-1}^N(\bw_n, \theta)}(\bw_n, \theta) \right ]\\
        &\qquad \cdot \prod_{t=1}^T \left [ (\widetilde{\bx}_{n,t}^N)^\top G(\bw_n,\theta)\widetilde{\by}_{n,t}^N \right ],
    \end{split}
    \end{equation}
    i.e. $\sum_{\widetilde{\bX}_{0:T}^N} p(\widetilde{\bX}_{0:T}^N, \widetilde{\bY}_{1:T}^N|\bW^N,\theta) = p(\widetilde{\bY}_{1:T}^N|\bW^N,\theta)$.
\end{proposition}

\begin{proof}
For ease of presentation throughout the proof we omit conditioning on $\bW^N$ from the notation.  We prove the statement of the proposition by induction on $T$ and start by showing that $p(\widetilde{\bY}_1^N|\theta)$ satisfies the statement:
    \begin{equation}
        \begin{split}
            p(\widetilde{\bY}_1^N|\theta) 
            &= 
            \prod_{n \in [N]} \bmu_{n,1}^N(\bw_n,\theta)^\top \widetilde{\by}_{n,1}^N
            = 
            \prod_{n \in [N]} \bpi_{n,0}^N(\bw_n, \theta)^\top K_{\widetilde{\boeta}_0^N(\bw_n,\theta)}(\bw_n,\theta) G(\bw_n,\theta) \widetilde{\by}_{n,1}^N \\
            &= 
            \prod_{n \in [N]}  \sum_{i, j, k} \bpi_{n,0}^N(\bw_n, \theta)^{(i)} K_{\widetilde{\boeta}_0^N(\bw_n,\theta)}(\bw_n,\theta)^{(i,j)} G(\bw_n,\theta)^{(j,k)} (\widetilde{\by}_{n,1}^N)^{(k)}\\
            &= 
            \prod_{n \in [N]} \sum_{\widetilde{\bx}_{n,0}^N, \widetilde{\bx}_{n,1}^N} \left ( (\widetilde{\bx}_{n,0}^N)^\top \bpi_{n,0}^N(\bw_n, \theta) \right ) \left ( (\widetilde{\bx}_{n,0}^N)^\top K_{\widetilde{\boeta}_0^N(\bw_n,\theta)}(\bw_n,\theta) \widetilde{\bx}_{n,1}^N \right ) \\
            &\qquad\qquad\qquad\qquad \cdot \left ( (\widetilde{\bx}_{n,1}^N)^\top G(\bw_n,\theta) \widetilde{\by}_{n,1}^N \right )\\
            &= 
            \sum_{\widetilde{\bX}_{0}^N, \widetilde{\bX}_{1}^N} \prod_{n \in [N]} \left ( (\widetilde{\bx}_{n,0}^N)^\top \bpi_{n,0}^N(\bw_n, \theta) \right ) \left ( (\widetilde{\bx}_{n,0}^N)^\top K_{\widetilde{\boeta}_0^N(\bw_n,\theta)}(\bw_n,\theta) \widetilde{\bx}_{n,1}^N \right ) \\
            &\qquad\qquad\qquad\qquad \cdot \left ( (\widetilde{\bx}_{n,1}^N)^\top G(\bw_n,\theta) \widetilde{\by}_{n,1}^N \right )\\
            &=
            \sum_{\widetilde{\bX}_{0}^N, \widetilde{\bX}_{1}^N} p(\widetilde{\bX}_{0:1}^N, \widetilde{\bY}_{1}^N|\theta)
        \end{split}
    \end{equation}
    for $\widetilde{\bX}_{t}^N = (\widetilde{\bx}^N_{1,t},\dots,\widetilde{\bx}^N_{N,t})$ with $t=0,1$, which complete the proof for the first time step. From the above we also get:
    \begin{equation}
    \begin{split}
        \bpi_{n,1}^N(\bw_n,
        \theta) 
        &=
        \left [ \sum\limits_{\widetilde{\bx}_{n,0}^N} \left ( (\widetilde{\bx}_{n,0}^N)^\top \bpi_{n,0}^N(\bw_n, \theta) \right ) \odot \left ( (\widetilde{\bx}_{n,0}^N)^\top K_{\widetilde{\boeta}_0^N(\bw_n,\theta)}(\bw_n,\theta) \right )^\top \right ]\\
        &\qquad \oslash
        \left [ \sum\limits_{\widetilde{\bx}_{n,0}^N, \widetilde{\bx}_{n,1}^N} \left ( (\widetilde{\bx}_{n,0}^N)^\top \bpi_{n,0}^N(\bw_n, \theta) \right ) \odot \left ( (\widetilde{\bx}_{n,0}^N)^\top K_{\widetilde{\boeta}_0^N(\bw_n,\theta)}(\bw_n,\theta) \right )^\top \right ]\\
        &\qquad \odot \left ( G(\bw_n,\theta) \widetilde{\by}_{n,1}^N \right ).
    \end{split}
    \end{equation}
    
    Now assume that the statement is valid for $T-1$:
    \begin{equation}
        \begin{split}
            p(\widetilde{\bY}_{1:T-1}^N|\theta) 
            &=
            \sum_{\widetilde{\bX}_{0:T-1}^N} \prod_{n \in [N]} \left [ (\widetilde{\bx}_{n,0}^N)^\top \bpi_{n,0}^N(\bw_n, \theta) \prod_{t=1}^{T-1} \left ( (\widetilde{\bx}_{n,t-1}^N)^\top K_{\widetilde{\boeta}_{t-1}^N(\bw_n,\theta)}(\bw_n,\theta) \widetilde{\bx}_{n,t}^N \right )\right.\\
            & \left. \qquad\qquad\qquad\qquad \cdot \left ( (\widetilde{\bx}_{n,t}^N)^\top G(\bw_n,\theta) \widetilde{\by}_{n,t}^N \right ) \right ]\\
            &= 
            \prod_{n \in [N]} \sum_{\widetilde{\bx}_{n,0:T-1}^N} (\widetilde{\bx}_{n,0}^N)^\top \bpi_{n,0}^N(\bw_n, \theta) \prod_{t=1}^{T-1} \left ( (\widetilde{\bx}_{n,t-1}^N)^\top K_{\widetilde{\boeta}_{t-1}^N(\bw_n,\theta)}(\bw_n,\theta) \widetilde{\bx}_{n,t}^N \right )\\
            & \qquad\qquad\qquad\qquad \cdot  \left ( (\widetilde{\bx}_{n,t}^N)^\top G(\bw_n,\theta) \widetilde{\by}_{n,t}^N \right ),
        \end{split}
    \end{equation}
    from which we also get:
    \begin{equation}
        \begin{split}
            &\bpi_{n,T-1}^N(\bw_n, \theta) = 
            \left \{ \sum_{\widetilde{\bx}_{n,0:T-2}^N} \left [ (\widetilde{\bx}_{n,0}^N)^\top \bpi_{n,0}^N(\bw_n, \theta) \prod_{t=1}^{T-2} \left ( (\widetilde{\bx}_{n,t-1}^N)^\top K_{\widetilde{\boeta}_{t-1}^N(\bw_n,\theta)}(\bw_n,\theta) \widetilde{\bx}_{n,t}^N \right )\right.\right.\\
            & \left. \left. \qquad\qquad\qquad\qquad \cdot  \left ( (\widetilde{\bx}_{n,t}^N)^\top G(\bw_n,\theta) \widetilde{\by}_{n,t}^N \right ) \right ] \left ( (\widetilde{\bx}_{n,T-2}^N)^\top K_{\widetilde{\boeta}^N_{T-2}(\bw_n,\theta)}(\bw_n,\theta) \right )^\top  \right \}\\
            &\qquad \oslash
            \left [ \sum_{\widetilde{\bx}_{n,0:T-1}^N} (\widetilde{\bx}_{n,0}^N)^\top \bpi_{n,0}^N(\bw_n, \theta) \prod_{t=1}^{T-1} \left ( (\widetilde{\bx}_{n,t-1}^N)^\top K_{\widetilde{\boeta}_{t-1}^N(\bw_n,\theta)}(\bw_n,\theta) \widetilde{\bx}_{n,t}^N \right ) \right. \\
            & \left. \qquad\qquad\qquad\qquad \cdot  \left ( (\widetilde{\bx}_{n,t}^N)^\top G(\bw_n,\theta) \widetilde{\by}_{n,t}^N \right ) \right ] \odot G(\bw_n,\theta) \widetilde{\by}_{n,T-1}^N.
        \end{split}
    \end{equation}
    Consider now time $T$:
    \begin{equation}
        \begin{split}
            p(\widetilde{\bY}_{1:T}^N|\theta) 
            &= 
            p(\widetilde{\bY}_{1:T-1}^N|\theta) \prod_{n \in [N]} \bmu_{n,t}^N(\bw_n,\theta)^\top \widetilde{\by}_{n,T}^N\\
            &= 
            p(\widetilde{\bY}_{1:T-1}^N|\theta) \prod_{n \in [N]} \bpi_{n,t-1}^N(\bw_n, \theta)^\top K_{\widetilde{\boeta}_{t-1}^N(\bw_n,\theta)}(\bw_n,\theta) G(\bw_n,\theta) \widetilde{\by}_{n,T}^N\\
            &= 
            p(\widetilde{\bY}_{1:T-1}^N|\theta) \prod_{n \in [N]} \sum_{\widetilde{\bx}_{n,T-1}^N, \widetilde{\bx}_{n,T}^N} \left ( (\widetilde{\bx}_{n,T-1}^N)^\top \bpi_{n,t-1}^N(\bw_n, \theta) \right ) \\
            & \qquad\qquad\qquad\qquad \cdot   \left ( (\widetilde{\bx}_{n,T-1}^N)^\top K_{\widetilde{\boeta}_{t-1}^N(\bw_n,\theta)}(\bw_n,\theta) \widetilde{\bx}_{n,T}^N \right ) \left ( (\widetilde{\bx}_{n,T}^N)^\top G(\bw_n,\theta) \widetilde{\by}_{n,T}^N\right ),
        \end{split}
    \end{equation}
    from which we remark that:
    \begin{equation}
        \begin{split}
            &(\widetilde{\bx}_{n,T-1}^N)^\top \bpi_{n,T-1}^N(\bw_n, \theta)\\
            &=
            \left \{ \sum_{\widetilde{\bx}_{n,0:T-2}^N} \left [ (\widetilde{\bx}_{n,0}^N)^\top \bpi_{n,0}^N(\bw_n, \theta) \prod_{t=1}^{T-2} \left ( (\widetilde{\bx}_{n,t-1}^N)^\top K_{\widetilde{\boeta}_{t-1}^N(\bw_n,\theta)}(\bw_n,\theta) \widetilde{\bx}_{n,t}^N \right )\right.\right.\\
            & \left. \left. \qquad\qquad\qquad\qquad \cdot  \left ( (\widetilde{\bx}_{n,t}^N)^\top G(\bw_n,\theta) \widetilde{\by}_{n,t}^N \right ) \right ] \left ( (\widetilde{\bx}_{n,T-2}^N)^\top K_{\widetilde{\boeta}^N_{T-2}(\bw_n,\theta)}(\bw_n,\theta)  \widetilde{\bx}_{n,T-1}^N \right )  \right \}\\
            &\qquad \slash
            \left [ \sum_{\widetilde{\bx}_{n,0:T-1}^N} (\widetilde{\bx}_{n,0}^N)^\top \bpi_{n,0}^N(\bw_n, \theta) \prod_{t=1}^{T-1} \left ( (\widetilde{\bx}_{n,t-1}^N)^\top K_{\widetilde{\boeta}_{t-1}^N(\bw_n,\theta)}(\bw_n,\theta) \widetilde{\bx}_{n,t}^N \right ) \right. \\
            & \left. \qquad\qquad\qquad\qquad \cdot  \left ( (\widetilde{\bx}_{n,t}^N)^\top G(\bw_n,\theta) \widetilde{\by}_{n,t}^N \right ) \right ]\\
            &\qquad\cdot (\widetilde{\bx}_{n,T-1}^N)^\top G(\bw_n,\theta) \widetilde{\by}_{n,T-1}^N.
        \end{split}
    \end{equation}
    As from our inductive hypothesis we have:
    \begin{equation}
        \begin{split}
            p(\widetilde{\bY}_{1:T-1}^N|\theta) 
            &= 
            \prod_{n \in [N]} \sum_{\widetilde{\bx}_{n,0:T-1}^N} (\widetilde{\bx}_{n,0}^N)^\top \bpi_{n,0}^N(\bw_n, \theta) \prod_{t=1}^{T-1} \left ( (\widetilde{\bx}_{n,t-1}^N)^\top K_{\widetilde{\boeta}_{t-1}^N(\bw_n,\theta)}(\bw_n,\theta) \widetilde{\bx}_{n,t}^N \right )\\
            & \qquad\qquad\qquad\qquad \cdot  \left ( (\widetilde{\bx}_{n,t}^N)^\top G(\bw_n,\theta) \widetilde{\by}_{n,t}^N \right ),
        \end{split}
    \end{equation}
    which is the denominator in $(\widetilde{\bx}_{n,T-1}^N)^\top \bpi_{n,t-1}^N(\bw_n, \theta)$, we can conclude:
    \begin{equation}
        \begin{split}
            p(\widetilde{\bY}_{1:t}^N|\theta) &= 
            p(\widetilde{\bY}_{1:T-1}^N|\theta)\\
            &\quad \cdot \prod_{n \in [N]} \sum_{\widetilde{\bx}_{n,T-1}^N, \widetilde{\bx}_{n,T}^N} \left ( (\widetilde{\bx}_{n,T-1}^N)^\top \bpi_{n,t-1}^N(\bw_n, \theta) \right ) \left ( (\widetilde{\bx}_{n,T-1}^N)^\top K_{\widetilde{\boeta}_{t-1}^N(\bw_n,\theta)}(\bw_n,\theta) \widetilde{\bx}_{n,T}^N \right )\\
            &\qquad\qquad
            \cdot \left ( (\widetilde{\bx}_{n,T}^N)^\top G(\bw_n,\theta) \widetilde{\by}_{n,T}^N\right )\\
            &= 
            \prod_{n \in [N]} \sum_{\widetilde{\bx}_{n,T-1}^N, \widetilde{\bx}_{n,T}^N} \sum_{\widetilde{\bx}_{n,0:T-2}^N}  (\widetilde{\bx}_{n,0}^N)^\top \bpi_{n,0}^N(\bw_n, \theta)\\
            &\qquad\qquad \cdot
            \prod_{t=1}^{T-2} \left ( (\widetilde{\bx}_{n,t-1}^N)^\top K_{\widetilde{\boeta}_{t-1}^N(\bw_n,\theta)}(\bw_n,\theta) \widetilde{\bx}_{n,t}^N \right ) \left ( (\widetilde{\bx}_{n,t}^N)^\top G(\bw_n,\theta) \widetilde{\by}_{n,t}^N \right )\\
            &\qquad\qquad \cdot 
            \left ( (\widetilde{\bx}_{n,T-2}^N)^\top K_{\widetilde{\boeta}^N_{T-2}(\bw_n,\theta)}(\bw_n,\theta)  \widetilde{\bx}_{n,T-1}^N \right ) \left ( (\widetilde{\bx}_{n,T-1}^N)^\top G(\bw_n,\theta) \widetilde{\by}_{n,T-1}^N \right ) \\
            &\qquad\qquad\cdot \left ( (\widetilde{\bx}_{n,T-1}^N)^\top K_{\widetilde{\boeta}_{t-1}^N(\bw_n,\theta)}(\bw_n,\theta) \widetilde{\bx}_{n,T}^N \right ) \left ( (\widetilde{\bx}_{n,T}^N)^\top G(\bw_n,\theta) \widetilde{\by}_{n,T}^N\right )\\
            &= 
            \prod_{n \in [N]} \sum_{\widetilde{\bx}_{n,0:T}^N}  (\widetilde{\bx}_{n,0}^N)^\top \bpi_{n,0}^N(\bw_n, \theta)\\
            &\qquad\qquad \cdot
            \prod_{t=1}^{T} \left ( (\widetilde{\bx}_{n,t-1}^N)^\top K_{\widetilde{\boeta}_{t-1}^N(\bw_n,\theta)}(\bw_n,\theta) \widetilde{\bx}_{n,t}^N \right ) \left ( (\widetilde{\bx}_{n,t}^N)^\top G(\bw_n,\theta) \widetilde{\by}_{n,t}^N \right )\\
            &=
            \sum_{\widetilde{\bX}_{0:T}^N} p(\widetilde{\bX}_{0:T}^N, \widetilde{\bY}_{1:T}^N|\theta),
        \end{split}
    \end{equation}
    which completes the proof.
\end{proof}

\section{Consistency of the maximum CAL estimator}  \label{suppsec:consistency_section}

With the definition:
\begin{equation}\label{eq:sup_log_cal}
\ell_t^N(\theta) \coloneqq \sum_{n \in [N]} \log \left [ \left ( \by_{n,t}^N \right )^\top \bmu_{n,t}^N(\bw_n,\theta) \right ],
\end{equation}
our ultimate goal is to prove Theorem \ref{thm:consistency}, which establishes consistency of the maximum CAL estimator:
\begin{equation} \label{eq:maxCAL_estimator}
\begin{split}
    \hat{\theta}^N \coloneqq 
    \argmax_{\theta \in \Theta} \sum_{t=1}^T \frac{1}{N} \ell_t^N(\theta),\\
\end{split}
\end{equation}
in the sense that $\hat{\theta}^N$ converges to $\Theta^\star$ as $N\to \infty$, $\mathbb{P}$-almost surely, where $\Theta^\star\subset\Theta$ is a set of parameter values which are, in a sense to be made precise, equivalent to the DGP $\theta^\star$.

The main steps are:
\begin{enumerate}
    \item prove almost sure pointwise in $\theta$ convergence of $\sum_{t=1}^T \frac{1}{N} \ell_t^N(\theta)-\frac{1}{N} \ell_t^N(\theta^\star)$ to a contrast function $\mathcal{C}(\theta,\theta^\star)$, this is the subject of Theorem \ref{thm:constrast_function_conv};
    \item Use stochastic equi-continuity \cite{andrews1992generic} to prove that the almost sure convergence $\sum_{t=1}^T \frac{1}{N} \ell_t^N(\theta)-\frac{1}{N} \ell_t^N(\theta^\star) \rightarrow \mathcal{C}(\theta,\theta^\star)$  is uniform in $\theta$, this is the subject of Lemma \ref{lem:stochequnfiorm};
    \item prove Theorem \ref{thm:consistency}, and so the convergence of the maximum CAL estimator to a set of maximizers $\Theta^\star$;
    \item Characterize the set of maximizers of the contrast function $\Theta^\star$ and prove that $\Theta^\star$ contains $\theta^\star$; this is the subject of Theorem \ref{thm:constrast_function_conv} and Lemma \ref{lem:identifiability}.
\end{enumerate}

Step 1. is by far the most complicated of the three. It involves proving $L^4$ bounds for averages across the population in the data-generating process (Section \ref{sec:asympt_DGP}),  averages of various quantities computed in the CAL algorithm (Section \ref{sec:asympt_CAL}), and comparison to averages across what we call the \emph{saturated processes} and \emph{saturated CAL algorithm} (Section \ref{sec:limiting_quantities_support}), which are processes we construct for purposes of our proofs in which members of the population are statistically decoupled. All of these ingredients are then combined to establish convergence to the contrast function (Theorem \ref{thm:constrast_function_conv}).

\subsection{Preliminaries}

\subsubsection{\texorpdfstring{$L^4$}{L4} bound for conditionally independent random variables}
We state and prove a result from \cite{whitehouse2023consistent}, which is useful to find $L^4$ bounds of averages of random variables that are conditionally independent, bounded, and mean zero. This is going to be one of the main building blocks in our proof strategy.

\begin{lemma}\label{lemma:mean_0_bound}
    Consider a collection of random variable $\boldsymbol{\delta}_n$ with $n \in [N]$. Assume that given a filtration $\mathcal{F}$ the random variables $\boldsymbol{\delta}_1,\dots,\boldsymbol{\delta}_N$ are conditionally independent, bounded by a constant $B < \infty$, i.e. $|\boldsymbol{\delta}_n| \leq B$ almost surely, and satisfy $\mathbb{E}[\boldsymbol{\delta}_n|\mathcal{F}] = 0$, then:
    \begin{equation}
        \normiii[\Bigg]{\frac{1}{N} \sum_{n \in [N]} \boldsymbol{\delta}_n}_4 \leq B \sqrt[4]{6}  N^{-\frac{1}{2}}.
    \end{equation}
\end{lemma}

\begin{proof}
    From the Multinomial theorem we can see that:
    \begin{equation}
        \left ( \sum_{n \in [N]} \boldsymbol{\delta}_n \right )^4 = \sum\limits_{k_1,\dots,k_N \in \mathbb{N}: k_1+\dots+k_N=4} \binom{4}{k_1,\dots,k_N} \prod_{n \in [N]} (\boldsymbol{\delta}_n)^{k_n},
    \end{equation}
    hence if we compute expectations with respect to the filtration we have:
    \begin{equation}
    \mathbb{E}\left [\left ( \sum_{n \in [N]} \boldsymbol{\delta}_n \right )^4 \Big{|} \mathcal{F} \right ] = \sum\limits_{k_1,\dots,k_N \in \mathbb{N}: k_1+\dots+k_N=4} \binom{4}{k_1,\dots,k_N} \prod_{n \in [N]} \mathbb{E}\left [ (\boldsymbol{\delta}_n)^{k_n} \big{|} \mathcal{F} \right ] ,
    \end{equation}
    because of the conditional independence assumption. Remark that we can combine a maximum of four terms as we are considering a power of 4 and we are constrained to $\sum_{n \in [N]} k_n=4$, meaning that we have only these possible combinations:
    \begin{itemize}
        \item $k_{i_1}+k_{i_2}+k_{i_3}+k_{i_4}=4$;
        \item $k_{i_1}+k_{i_2}+k_{i_3}=4$;
        \item $k_{i_1}+k_{i_2}=4$;
        \item $k_{i_1}=4$;
    \end{itemize}
    with all the other $k$'s being zero. Given that $\mathbb{E}[\boldsymbol{\delta}_n|\mathcal{F}] = 0$ then all the combinations which involve a $k_i=1$ can be safely removed, we then end up with:
    \begin{equation}
    \begin{split}
        \mathbb{E}\left [\left ( \sum_{n \in [N]} \boldsymbol{\delta}_n \right )^4 \Big{|} \mathcal{F} \right ]   &= \sum\limits_{n \in [N]} \mathbb{E}\left [ (\boldsymbol{\delta}_n)^{4} \Big{|} \mathcal{F} \right ]  + \binom{4}{2,2}\sum_{n,n^\prime \in [N], n\neq n^\prime} \mathbb{E}\left [ (\boldsymbol{\delta}_n)^{2} | \mathcal{F} \right ] \mathbb{E} \left [ (\boldsymbol{\delta}_{n^\prime})^{2} | \mathcal{F} \right ]\\
        &\leq \sum\limits_{n \in [N]} (B)^{4} + 6 \sum_{n,n^\prime \in [N], n\neq n^\prime} (B)^{2} (B)^{2}\\
        &= N B^4 + 6 N(N-1) B^4\\
        &= B^4 (N + 6 N^2- 6N) \leq 6 B^4 N^2,
    \end{split}
    \end{equation}
    from which the statement of the Lemma follows by the tower rule, dividing by $N^4$ and exponentiating by $\frac{1}{4}$.
\end{proof}

\subsubsection{Checking the CAL is almost surely well-defined} \label{sec:CAL_well_defined} 

If $(\by_{n,t}^N)^\top \bmu_{n,t}^N(\bw_n,\theta)=0$, then \eqref{eq:sup_log_cal} would evaluate to $\log(0)$. The aim of this section is to prove that the CAL is almost surely well-defined, in the sense that $(\by_{n,t}^N)^\top \bmu_{n,t}^N(\bw_n,\theta)=0$ happens with zero probability no matter the values of $N$ and $\theta$.

The main result is Theorem \ref{thm:CAL_well_definess}, and its proof builds upon propositions \ref{prop:initial_as} - \ref{prop:sequential_welldefined} below, which exploit the recursive nature of the CAL algorithm and the data-generating process. 

\begin{proposition}\label{prop:initial_as}
    Under Assumption \ref{ass:HMM_support}, for any $n \in [N]$ if there exists $\theta \in \Theta, i \in [M]$ such that $ p_0(w_n, \theta)^{(i)}=0$ then $(\bx_{n,0}^N)^{(i)}=0$, $\mathbb{P}(\cdot|\bW^N=W^N)$-almost surely.
\end{proposition}

\begin{proof}
    Under Assumption \ref{ass:HMM_support} we have that $p_0(w_n, \theta)^{(i)}=0$ implies $p_0(w_n, \theta^\star)^{(i)}=0$, hence:
    $$
    \mathbb{P} \left ( (\bx_{n,0}^N)^{(i)}=0|\bW^N=W^N \right ) = p_0(w_n, \theta^\star)^{(i)}=0,
    $$
    which concludes the proof. 
\end{proof}

\begin{proposition}\label{prop:prediction_as}
    Under Assumption \ref{ass:HMM_support}, for any $t \geq 1$ and $n \in [N]$ if there exists $\theta \in \Theta, i \in [M]$ such that $\bpi_{n,t-1}^N(w_n, \theta)^{(i)}=0$ implies $(\bx_{n,t-1}^N)^{(i)}=0$,  $\mathbb{P}(\cdot|\bW^N=W^N,\bY_{1:t-1}^N=Y_{1:t-1}^N)$-almost surely, then, there exists $\theta \in \Theta,i \in [M]$ such that $\bpi_{n,t|t-1}^N(w_n,\theta)^{(i)}=0$ implies $(\bx_{n,t}^N)^{(i)}=0$, $\mathbb{P}(\cdot|\bW^N=W^N,\bY_{1:t-1}^N=Y_{1:t-1}^N)$-almost surely.
\end{proposition}

\begin{proof}
    Note that for $\theta \in \Theta, i \in [M]$ the following are equivalent:
    \begin{equation}
        \begin{split}
            \bpi_{n,t|t-1}^N(w_n,\theta)^{(i)}=0
            &\iff 
            \sum_{j } \bpi_{n,t-1}^N(w_n, \theta)^{(j)} K_{\widetilde{\boeta}^N_{t-1}(w_n,\theta)}(w_n,\theta)^{(j,i)}=0\\
            &\iff
            \forall j \in [M] \quad \bpi_{n,t-1}^N(w_n, \theta)^{(j)} K_{\widetilde{\boeta}^N_{t-1}(w_n,\theta)}(w_n,\theta)^{(j,i)}=0.
        \end{split}
    \end{equation}
    We then have that for any $j$ either:
    \begin{enumerate}
        \item $\mathbb\bpi_{n,t-1}^N(w_n, \theta)^{(j)}=0$ which implies $(\bx_{n,t-1}^N)^{(j)}=0$ almost surely in $\mathbb{P}(\cdot|\bW^N=W^N,\bY_{1:t-1}^N=Y_{1:t-1}^N)$ by assumption, or
        \item $K_{\widetilde{\boeta}^N_{t-1}(w_n,\theta)}(w_n,\theta)^{(j,i)}=0$, which implies that there exists $\eta = \widetilde{\boeta}^N_{t-1}(w_n,\theta)$ such that $K_{\eta}(w_n,\theta)^{(j,i)}=0$. Then by Assumption \ref{ass:HMM_support} we have that $K_{\eta}(w_n,\theta)^{(j,i)}=0$ for all $\eta \in [0,C]$ and $\theta \in \Theta$ meaning $K_{\boeta_{t-1}^N(w_n,\theta^\star)}(w_n,\theta^\star)^{(j,i)}=0$, $\mathbb{P}(\cdot|\bW^N=W^N,\bY_{1:t-1}^N=Y_{1:t-1}^N)$-almost surely as it holds for almost any realization of ${\bX}_{t-1}^N \in \mathbb{O}_M^N \subset \Delta_M^N$, where we remark that ${\bX}_{t-1}^N$ appears in $\boeta_{t-1}^N(w_n,\theta^\star) = \eta^N(w_n,\theta^\star,W^N, {\bX}_{t-1}^N)$.
    \end{enumerate}
    Now given that: 
    \begin{equation}
        \begin{split}
           \mathbb{P}\left ((\bx_{n,t}^N)^{(i)}=0|\bW^N=W^N,\bY_{1:t-1}^N=Y_{1:t-1}^N \right ) = 1 - \mathbb{P}\left ((\bx_{n,t}^N)^{(i)}=1|\bW^N=W^N,\bY_{1:t-1}^N=Y_{1:t-1}^N \right ),
        \end{split}
    \end{equation}
    and by defining ${\bX}_{t-1}^{N\setminus n}$ as ${\bX}_{t-1}^N$ with the $n$-th individual removed, we can notice that:
    \begin{equation}
        \begin{split}
            &\mathbb{P}\left ( (\bx_{n,t}^N)^{(i)}=1|\bW^N=W^N,\bY_{1:t-1}^N=Y_{1:t-1}^N \right )\\
            &= 
            \sum_{X_{t-1}^{N\setminus n}} \mathbb{P}\left ({\bX}_{t-1}^{N\setminus n} = X_{t-1}^{N\setminus n}|\bW^N=W^N,\bY_{1:t-1}^N=Y_{1:t-1}^N \right)\\
            &\quad \sum_{x_{n,t-1}^N} \mathbb{P}\left (\bx_{n,t-1}^N = x_{n,t-1}^N|\bW^N=W^N,{\bX}_{t-1}^{N\setminus n} = X_{t-1}^{N\setminus n},\bY_{1:t-1}^N=Y_{1:t-1}^N \right)\\
            & \quad \cdot \mathbb{P}\left ((\bx_{n,t}^N)^{(i)}=1|\bW^N=W^N,{\bX}_{t-1}^N = X_{t-1}^N,\bY_{1:t-1}^N=Y_{1:t-1}^N \right)\\
            &= \sum_{X_{t-1}^{N\setminus n}} \mathbb{P}\left ({\bX}_{t-1}^{N\setminus n} = X_{t-1}^{N\setminus n} | \bW^N=W^N,\bY_{1:t-1}^N=Y_{1:t-1}^N\right )\\
            &\quad \sum_{j} \mathbb{P}\left ((\bx_{n,t-1}^N)^{(j)} = 1|\bW^N=W^N,{\bX}_{t-1}^{N\setminus n} = X_{t-1}^{N\setminus n},\bY_{1:t-1}^N=Y_{1:t-1}^N \right ) K_{\boeta_{t-1}^N(w_n,\theta^\star)}(w_n,\theta^\star)^{(j,i)}\\
            &= 0,
        \end{split}
    \end{equation}
    where the last step follows from the fact that for all $j$ we either have $(\bx_{n,t-1}^N)^{(j)} = 0$ almost surely in $\mathbb{P}(\cdot|\bW^N=W^N,\bY_{1:t-1}^N=Y_{1:t-1}^N)$, or $K_{\eta}(w_n,\theta^\star)^{(j,i)} = 0$ for any $\eta \in [0,C]$ and so $K_{\boeta_{t-1}^N(w_n,\theta^\star)}(w_n,\theta^\star)^{(j,i)} = 0$ for any $X_{t-1}^N$. As $\mathbb{P}\left ((\bx_{n,t}^N)^{(i)}=1 |\bW^N=W^N,\bY_{1:t-1}^N=Y_{1:t-1}^N\right)= 0$ we can conclude $\mathbb{P}\left ((\bx_{n,t}^N)^{(i)}=0 |\bW^N=W^N,\bY_{1:t-1}^N=Y_{1:t-1}^N\right)= 1$, which concludes the proof.
\end{proof}

\begin{proposition}\label{prop:correction_as}
    Under Assumption \ref{ass:HMM_support}, for any $t \geq 1$ and $n \in [N]$ if there exists $\theta \in \Theta,i \in [M]$ such that $\bpi_{n,t|t-1}^N(w_n,\theta)^{(i)}=0$ implies $(\bx_{n,t}^N)^{(i)}=0$ almost surely in $\mathbb{P}(\cdot|\bW^N=W^N,\bY_{1:t-1}^N=Y_{1:t-1}^N)$, then: 
    \begin{itemize}
        \item there exists $\theta \in \Theta, i \in [M]$ such that $\bmu_{n,t}^N(w_n,\theta)^{(i)}=0$ implies $(\by_{n,t}^N)^{(i)}=0$ almost surely in $\mathbb{P}(\cdot|\bW^N=W^N,\bY_{1:t-1}^N=Y_{1:t-1}^N)$;
        \item there exists $\theta \in \Theta, i \in [M]$ such that $\bpi_{n,t}^N(w_n,\theta)^{(i)}=0$ implies $(\bx_{n,t}^N)^{(i)}=0$ almost surely in $\mathbb{P}(\cdot|\bW^N=W^N,\bY_{1:t}^N=Y_{1:t}^N)$.
    \end{itemize}
\end{proposition}

\begin{proof}
    Let us start with the results regarding the observations. Note that for $\theta \in \Theta, i \in [M]$ the following are equivalent:
    \begin{equation}
        \begin{split}
            \bmu_{n,t}^N(w_n,\theta)^{(i)}=0 
            &\iff 
            \sum_{j } \bpi_{n,t|t-1}^N(w_n,\theta)^{(j)}  G(w_n,\theta)^{(j,i)}=0 \\
            &\iff
            \forall j \in [M] \quad \bpi_{n,t|t-1}^N(w_n,\theta)^{(j)}  G(w_n,\theta)^{(j,i)}=0.
        \end{split}
    \end{equation}
    We then have that for all $j$ either:
    \begin{enumerate}
        \item $\mathbb\bpi_{n,t|t-1}^N(w_n,\theta)^{(j)}=0$ which implies $(\bx_{n,t}^N)^{(j)}=0$ almost surely in $\mathbb{P}(\cdot|\bW^N=W^N,\bY_{1:t-1}^N=Y_{1:t-1}^N)$ by assumption, or
        \item $G(w_n,\theta)^{(j,i)}=0$ which implies $ G(w_n,\theta^\star)^{(j,i)}=0$ because of Assumption \ref{ass:HMM_support}. 
    \end{enumerate}
    Now given that: 
    \begin{equation}
        \begin{split}
           \mathbb{P}\left ((\by_{n,t}^N)^{(i)}=0 |\bW^N=W^N,\bY_{1:t-1}^N=Y_{1:t-1}^N \right ) = 1 - \mathbb{P}\left ((\by_{n,t}^N)^{(i)}=1 |\bW^N=W^N,\bY_{1:t-1}^N=Y_{1:t-1}^N \right ),
        \end{split}
    \end{equation}
    and:
    \begin{equation}
        \begin{split}
            &\mathbb{P}\left ((\by_{n,t}^N)^{(i)}=1 |\bW^N=W^N,\bY_{1:t-1}^N=Y_{1:t-1}^N \right)\\
            &= 
            \sum_{x_{n,t}^N} \mathbb{P}\left (\bx_{n,t}^N= x_{n,t}^N |\bW^N=W^N,\bY_{1:t-1}^N=Y_{1:t-1}^N \right)\\
            &\quad \cdot \mathbb{P}\left ((\by_{n,t}^N)^{(i)}=1|\bW^N=W^N, \bx_{n,t}^N = x_{n,t}^N, \bY_{1:t-1}^N=Y_{1:t-1}^N \right)\\
            &= 
            \sum_{j} \mathbb{P}\left ((\bx_{n,t}^N)^{(j)} = 1 |\bW^N=W^N,\bY_{1:t-1}^N=Y_{1:t-1}^N \right)  G(w_n,\theta^\star)^{(j,i)} = 0,
        \end{split}
    \end{equation}
    where the last step follows from the fact that for all $j$ we either have $(\bx_{n,t-1}^N)^{(j)} = 0$ almost surely in $\mathbb{P}(\cdot|\bW^N=W^N,\bY_{1:t-1}^N=Y_{1:t-1}^N)$ or $ G(w_n,\theta^\star)^{(j,i)} = 0$. As 
    $$\mathbb{P}\left ( (\by_{n,t}^N)^{(i)}=1 |\bW^N=W^N,\bY_{1:t-1}^N=Y_{1:t-1}^N \right )= 0,$$ 
    we can conclude 
    $$\mathbb{P}\left ( (\by_{n,t}^N)^{(i)}=0 |\bW^N=W^N,\bY_{1:t-1}^N=Y_{1:t-1}^N \right )= 1,$$
    which concludes the proof of the first part. As a consequence if there exists $j\in[M]$ such that $\mathbb{P}((\by_{n,t}^N)^{(j)}=1|\bW^N=W^N,\bY_{1:t-1}^N=Y_{1:t-1}^N)>0$ then $\bmu_{n,t}^N(w_n,\theta)^{(j)} \neq 0$, meaning that 
    \begin{equation} \label{eq:as_non_zero_denom}
        \mathbb{P}\left (\sum_{j} (\by_{n,t}^N)^{(j)}\bmu_{n,t}^N(w_n,\theta)^{(j)} \neq 0|\bW^N=W^N,\bY_{1:t-1}^N=Y_{1:t-1}^N
        \right )=1.
    \end{equation}

    Consider now $\bpi_{n,t}^N(w_n,\theta)^{(i)}$, and observe that:
    \begin{equation}
        \begin{split}
            &\bpi_{n,t}^N(w_n,\theta)^{(i)}= 0  
            \iff
            \bpi_{n,t|t-1}^N(w_n,\theta)^{(i)} \frac{\sum_{j}  G(w_n,\theta)^{(i,j)} (\by_{n,t}^N)^{(j)} }{\sum_{j} (\by_{n,t}^N)^{(j)}\bmu_{n,t}^N(w_n,\theta)^{(j)}}=0\\
            &\iff
            \bpi_{n,t|t-1}^N(w_n,\theta)^{(i)} \sum_{j}  G(w_n,\theta)^{(i,j)}(\by_{n,t}^N)^{(j)}=0 \text{ and } \sum_{j} (\by_{n,t}^N)^{(j)}\bmu_{n,t}^N(w_n,\theta)^{(j)} \neq 0.
        \end{split}
    \end{equation}

    From the \eqref{eq:as_non_zero_denom}, we know that under $\mathbb{P}(\cdot|\bW^N=W^N,\bY_{1:t-1}^N=Y_{1:t-1}^N)$ we have that the denominator $\sum_{j} (\by_{n,t}^N)^{(j)}\bmu_{n,t}^N(w_n,\theta)^{(j)}$ is almost surely different from $0$ hence:
    \begin{equation}
        \begin{split}
            &\mathbb{P}\left (\bpi_{n,t|t-1}^N(w_n,\theta)^{(i)} \sum_{j}  G(w_n,\theta)^{(i,j)}(\by_{n,t}^N)^{(j)}=0\right.\\ 
            &\left. \qquad\qquad \text{ and } \sum_{j} (\by_{n,t}^N)^{(j)}\bmu_{n,t}^N(w_n,\theta)^{(j)}\neq 0 |\bW^N=W^N,\bY_{1:t-1}^N=Y_{1:t-1}^N \right )\\
            &=
            \mathbb{P}\left (\sum_{j} (\by_{n,t}^N)^{(j)}\bmu_{n,t}^N(w_n,\theta)^{(j)}=0 |\bW^N=W^N,\bY_{1:t-1}^N=Y_{1:t-1}^N \right ),
        \end{split}
    \end{equation}
    as we are considering an intersection with an almost sure event.

    We then just need to prove $\bpi_{n,t|t-1}^N(w_n,\theta)^{(i)} \sum_{j}  G(w_n,\theta)^{(i,j)}(\by_{n,t}^N)^{(j)}=0$ almost surely. We have either:
    \begin{itemize}
        \item $\bpi_{n,t|t-1}^N(w_n,\theta)^{(i)}=0$, implying $(\bx_{n,t}^N)^{(i)}=0$ almost surely in $\mathbb{P}(\cdot |\bW^N=W^N,\bY_{1:t-1}^N=Y_{1:t-1}^N)$ because we consider this statement to be true, or
        \item $\sum_{j}  G(w_n,\theta)^{(i,j)}(\by_{n,t}^N)^{(j)}=0$, which tells us that there exists $k \in [M]$ such that\\ $\mathbb{P}\left ((\by_{n,t}^N)^{(k)}=1 |\bW^N=W^N,\bY_{1:t-1}^N=Y_{1:t-1}^N \right)>0$ and $ G(w_n,\theta)^{(i,k)}=0$, note that under $\mathbb{P}(\cdot|\bW^N=W^N,\bY_{1:t}^N=Y_{1:t}^N)$ we know $k$ as we are conditioning on $\by_{n,t}^N$.
    \end{itemize}
    Now given that: 
    \begin{equation}
        \begin{split}
           \mathbb{P}((\bx_{n,t}^N)^{(i)}=0|\bW^N=W^N,\bY_{1:t}^N=Y_{1:t}^N) = 1 - \mathbb{P}((\bx_{n,t}^N)^{(i)}=1|\bW^N=W^N,\bY_{1:t}^N=Y_{1:t}^N),
        \end{split}
    \end{equation}
    we can notice that:
    \begin{equation}
        \begin{split}
            &\mathbb{P}((\bx_{n,t}^N)^{(i)}=1|\bW^N=W^N,\bY_{1:t}^N=Y_{1:t}^N)\\
            &\propto
            \mathbb{P}((\bx_{n,t}^N)^{(i)}=1, (\by_{n,t}^N)^{(k)}=1 |\bW^N=W^N,\bY_{1:t-1}^N=Y_{1:t-1}^N )\\
            &=
            G(w_n,\theta)^{(i,k)} \mathbb{P}((\bx_{n,t}^N)^{(i)}=1|\bW^N=W^N,\bY_{1:t-1}^N=Y_{1:t-1}^N)=0,
        \end{split}
    \end{equation}
    because we have either $(\bx_{n,t}^N)^{(i)}=0$ almost surely in $\mathbb{P}(\cdot|\bW^N=W^N,\bY_{1:t-1}^N=Y_{1:t-1}^N)$ or $ G(w_n,\theta)^{(i,k)}=0$, which concludes the proof.
\end{proof}

We can now combine Proposition \ref{prop:initial_as}, Proposition\ref{prop:prediction_as}, and Proposition \ref{prop:correction_as} in the following proposition.

\begin{proposition}\label{prop:sequential_welldefined}
    Under Assumption \ref{ass:HMM_support}, for any $t \geq 1$, $n \in [N]$, the following hold:
    \begin{itemize}
        \item if there exist $\theta \in \Theta, i \in [M]$ such that $\bpi_{n,t|t-1}^N(w_n,\theta)^{(i)}=0$, then $(\bx_{n,t}^N)^{(i)}=0$ almost surely in $\mathbb{P}(\cdot|\bW^N=W^N,\bY_{1:t-1}^N=Y_{1:t-1}^N)$;
        \item if there exist $\theta \in \Theta, i \in [M]$ such that $\bmu_{n,t}^N(w_n,\theta)^{(i)}=0$, then $(\by_{n,t}^N)^{(i)}=0$ almost surely in $\mathbb{P}(\cdot|\bW^N=W^N,\bY_{1:t-1}^N=Y_{1:t-1}^N)$;
        \item if there exist $\theta \in \Theta, i \in [M]$ such that $\bpi_{n,t}^N(w_n,\theta)^{(i)}=0$, then $(\bx_{n,t}^N)^{(i)}=0$ almost surely in $\mathbb{P}(\cdot|\bW^N=W^N,\bY_{1:t}^N=Y_{1:t}^N)$.
    \end{itemize}
\end{proposition}

\begin{proof}
    Suppose that the third statement is true at time $t-1$, which is valid at $t-1=0$ because of Proposition \ref{prop:initial_as}, which guarantees that if $\exists \theta \in \Theta,i \in [M]:\bpi_{n,0}^N(w_n, \theta)^{(i)}=0 \implies (\bx_{n,0}^N)^{(i)}=0$ almost surely in $\mathbb{P}(\cdot| \bW^N = W^N)$. We can in turn apply Proposition \ref{prop:prediction_as} to prove the first statement and Proposition \ref{prop:correction_as} to prove the second and the third for $t$. As the third statement is our inductive hypothesis at the next time step $t$ we can then close the induction and conclude that the three statements are valid for an arbitrary $t$. 
\end{proof}

We finally prove the main result.

\begin{theorem} \label{thm:CAL_well_definess}
    Under assumptions \ref{ass:w_iid},\ref{ass:HMM_support}, for any $N \in \mathbb{N}, t\geq 1, n \in [N]$ and $\theta \in \Theta$ we have that $\sum_i \bmu_{n,t}^N(\bw_n,\theta)^{(i)}(\by_{n,t}^N)^{(i)}\neq 0$ almost surely in $\mathbb{P}$.
\end{theorem}

\begin{proof}
    We can see that:
    \begin{equation}
	\begin{split}
	&\mathbb{P}\left (  \exists \theta \in \Theta, t\geq1,n \in [N] : \quad \sum_i \bmu_{n,t}^N(\bw_n,\theta)^{(i)}(\by_{n,t}^N)^{(i)}=0 | \bW^N = W^N \right) 
		\\
		&=
		\mathbb{P}\left ( \exists \theta \in \Theta, t\geq1,n \in [N] : \quad \forall i \in [M] \quad \bmu_{n,t}^N(\bw_n,\theta)^{(i)}(\by_{n,t}^N)^{(i)}=0 | \bW^N = W^N \right) 
		\\
		&=
		\mathbb{P}\left ( \exists \theta \in \Theta, t\geq1,n \in [N],k \in [M] : \quad \bmu_{n,t}^N(\bw_n,\theta)^{(k)}=0 \text{ and }(\by_{n,t}^N)^{(k)}=1 | \bW^N = W^N \right),
	\end{split}
    \end{equation}
where the second equality uses the fact that $\by_{n,t}^N$ is a one-hot encoding vector, i.e. the $k$th component is $1$ while the others are $0$. Moreover:
    \begin{equation}
	\begin{split}
        &\mathbb{P}\left ( \left. \exists \theta \in \Theta, t\geq1,n \in [N] : \quad \sum_i \bmu_{n,t}^N(\bw_n,\theta)^{(i)}(\by_{n,t}^N)^{(i)}=0 \right| \bW^N = W^N \right)\\
		&=
		\sum_{Y_{1:t-1}^N} \mathbb{P}\left ( \exists \theta \in \Theta, t\geq1,n \in [N],k \in [M] : \right.\\
        &\left. \qquad \qquad\qquad \bmu_{n,t}^N(\bw_n,\theta)^{(k)}=0 \text{ and } (\by_{n,t}^N)^{(k)}=1 |\bW^N=W^N,\bY_{1:t-1}^N=Y_{1:t-1}^N \right)\\
		&\qquad\quad \cdot \mathbb{P}\left ( \bY_{1:t-1}^N= Y_{1:t-1}^N | \bW^N = W^N \right )=0,
	\end{split}
    \end{equation}
    as from the second statement of Proposition \ref{prop:sequential_welldefined} we know that event $\left \{\bmu_{n,t}^N(\bw_n,\theta)^{(k)}=0 \right \} \cap \left \{ (\by_{n,t}^N)^{(k)}=1 \right \}$ is probability zero conditionally on $\bW^N=W^N,\bY_{1:t-1}^N=Y_{1:t-1}^N$, indeed:
    \begin{equation}
    \begin{split}
        &\mathbb{P}\left ( \exists \theta \in \Theta, t\geq1,n \in [N],k \in [M] : \right. \\ 
        &\left. \qquad\qquad \bmu_{n,t}^N(\bw_n,\theta)^{(k)}=0 \text{ and } (\by_{n,t}^N)^{(k)}=1 |\bW^N=W^N,\bY_{1:t-1}^N=Y_{1:t-1}^N \right)\\
        &=
        1 - \mathbb{P}\left ( \forall \theta \in \Theta, t\geq1,n \in [N],k \in [M] : \right. \\ 
        &\left. \qquad\qquad\qquad \text{if } \bmu_{n,t}^N(\bw_n,\theta)^{(k)}=0 \text{ then } (\by_{n,t}^N)^{(k)}=0 |\bW^N=W^N,\bY_{1:t-1}^N=Y_{1:t-1}^N \right)\\
        &=0.
    \end{split}
    \end{equation}
    
    We can then conclude the proof of the first statement as: 
    \begin{equation}
    \begin{split}
        &\mathbb{P}\left (\left. \forall \theta \in \Theta, t\geq1, n \in [N] \quad \sum_i \bmu_{n,t}^N(\bw_n,\theta)^{(i)}(\by_{n,t}^N)^{(i)}\neq0 \right| \bW^N = W^N \right)\\
        &= 1 - \mathbb{P}\left (\left. \exists \theta \in \Theta, t\geq1,n \in [N] : \quad \sum_i \bmu_{n,t}^N(\bw_n,\theta)^{(i)}(\by_{n,t}^N)^{(i)}=0 \right| \bW^N = W^N \right)=1.
    \end{split}
    \end{equation}
    To conclude the proof we need:
    \begin{equation}
        \mathbb{P}\left ( \forall \theta \in \Theta, t\geq1, n \in [N] \quad \sum_i \bmu_{n,t}^N(\bw_n,\theta)^{(i)}(\by_{n,t}^N)^{(i)}\neq0 \right)=1,
    \end{equation}
    which can be proven by observing that:
    \begin{equation}
        \begin{split}
            &\mathbb{P}\left ( \forall \theta \in \Theta, t\geq1, n \in [N] \quad \sum_i \bmu_{n,t}^N(\bw_n,\theta)^{(i)}(\by_{n,t}^N)^{(i)}\neq0 \right)\\
            &=
            \int \mathbb{P}\left (\left.  \forall \theta \in \Theta, t\geq1, n \in [N] \quad \sum_i \bmu_{n,t}^N(w_n,\theta)^{(i)}(\by_{n,t}^N)^{(i)}\neq0 \right| \bW^N = W^N \right)\\
            &\qquad \qquad \Gamma(dw_1) \dots \Gamma(dw_N)\\
            &= 1,
        \end{split}
    \end{equation}
    where we applied Assumption \ref{ass:w_iid} and where the last step follows from what we have just proven.
\end{proof}

\subsubsection{Checking the CAL is almost surely bounded}

In this section, we want to prove that all the non-zero elements of $\bmu_{n,t}^N(\bw_n,\theta)$ are almost surely bounded below by a quantity ${m}_t>0$ that does not depend on $N$, this will be put to use in establishing \texorpdfstring{$L^4$}{L4} bounds in Section \ref{sec:asympt_CAL}.

\begin{proposition} \label{prop:CAL_as_bounded}
    Under assumptions \ref{ass:compactness_continuity},\ref{ass:HMM_support},\ref{ass:kernel_continuity}, for $t\geq 1$ there exists ${m}_{t}>0$ such that for any $N \in \mathbb{N}$ and $n \in [N]$ we have:
    $$
    \mathbb{P} \left ( \bmu_{n,t}^N(\bw_n,\theta)^{(i)}
        \geq 
        {m}_{t}\quad \forall i \in \supp{\bmu_{n,t}^N(\bw_n,\theta)} \right )=1 \quad \forall \theta \in \Theta.
    $$
\end{proposition}

\begin{proof}

    For a fixed $N\in \mathbb{N}$, consider the following inductive hypothesis. There exists $\bar{m}_{t-1} > 0$ such that:
    $$
    \mathbb{P} \left ( \bpi_{n,t-1}^N(\bw_n,\theta)^{(i)}
        \geq 
        \bar{m}_{t-1} \quad \forall i \in \supp{\bpi_{n,t-1}^N(\bw_n,\theta)} \right )=1 \quad \forall n \in [N], \theta \in \Theta.
    $$
    
    We start by proving that the inductive hypothesis is true when $t-1=0$. From Assumption \ref{ass:compactness_continuity} we have that $p_0(w,\theta)$ is continuous in $w,\theta$ and both $\mathbb{W}$ and $\Theta$ are compact, we then get from Weierstrass theorem that there exists a minimum $m_0$ such that for any realization of $\bW^N$ and for any $i \in \supp{p_0(\bw_n,\theta)^{(i)}}$:
    $$
    \bpi_{n,0}^N(\bw_n,\theta)^{(i)}= p_0(\bw_n,\theta)^{(i)} \geq \min_{w \in \mathbb{W},\theta \in \Theta} \min_{j \in \supp{p_0(w,\theta)^{(j)}}} p_0(w,\theta)^{(j)} \eqqcolon m_0,
    $$
    with $m_0>0$ as we are considering a minimum over $j \in \supp{p_0(w,\theta)^{(j)}}$ which excludes all the zeros. As $m_0$ does not depend on $\bW^N$ we conclude that there exists $m_{0} > 0$ such that:
    $$
    \mathbb{P} \left ( \bpi_{n,0}^N(\bw_n,\theta)^{(i)}
        \geq 
        m_{0} \quad \forall i \in \supp{\bpi_{n,0}^N(\bw_n,\theta)} \right )=1 \quad \forall n \in [N], \theta \in \Theta.
    $$
    
    Let us now work on a general time step $t$. For $i \in \supp{\bpi_{n,t|t-1}^N(\bw_n,\theta)}$,
    \begin{equation}
    \begin{split}
        \bpi_{n,t|t-1}^N(\bw_n,\theta)^{(i)}
        &=
        \sum_{j} \bpi_{n,t-1}^N(\bw_n,\theta)^{(j)} K_{\widetilde{\boeta}^N_{t-1}(\bw_n,\theta)}(\bw_n,\theta)^{(j,i)}\\
        &=
        \sum_{j \in \supp{\bpi_{n,t-1}^N(\bw_n,\theta)}} \bpi_{n,t-1}^N(\bw_n,\theta)^{(j)} K_{\widetilde{\boeta}^N_{t-1}(\bw_n,\theta)}(\bw_n,\theta)^{(j,i)}\\
        &\geq 
        \bar{m}_{t-1} \sum_{j \in \supp{\bpi_{n,t-1}^N(\bw_n,\theta)}} K_{\widetilde{\boeta}^N_{t-1}(\bw_n,\theta)}(\bw_n,\theta)^{(j,i)},
    \end{split}
    \end{equation}
    where the inequality holds $\mathbb{P}$-almost surely by the inductive hypothesis. Several other inequalities in the remainder of the proof hold $\mathbb{P}$-almost surely, but to avoid repetition we do not state this explicitly. As:
    \begin{align}
    i \in \supp{\bpi_{n,t|t-1}^N(\bw_n,\theta)}
    &\Longleftrightarrow
    \sum_{j \in \supp{\bpi_{n,t-1}^N(\bw_n,\theta)}} K_{\widetilde{\boeta}^N_{t-1}(\bw_n,\theta)}(\bw_n,\theta)^{(j,i)} \neq 0,
    \end{align}
    we can conclude that there exists at least one component of $K_{\widetilde{\boeta}^N_{t-1}(\bw_n,\theta)}(\bw_n,\theta)^{(\cdot,i)}$ which is different from zero, hence: 
    \begin{equation}\label{eq:strict_pos_support}
    \sum_{j \in \supp{\bpi_{n,t-1}^N(\bw_n,\theta)}} K_{\widetilde{\boeta}^N_{t-1}(\bw_n,\theta)}(\bw_n,\theta)^{(j,i)} \geq \min_{j \in \supp{K_{\widetilde{\boeta}^N_{t-1}(\bw_n,\theta)}(\bw_n,\theta)^{(\cdot,i)}} } K_{\widetilde{\boeta}^N_{t-1}(\bw_n,\theta)}(\bw_n,\theta)^{(j,i)}
    >0.
    \end{equation}
    Because of Assumption \ref{ass:HMM_support} we have:
    $$
    K_{\widetilde{\boeta}^N_{t-1}(\bw_n,\theta)}(\bw_n,\theta)^{(j,i)}=0
    \Longleftrightarrow
    K_{\eta}(\bw_n,\theta)^{(j,i)}=0 \quad \forall \eta \in [0,C],
    $$ 
    meaning that
    \begin{equation}
    \begin{split}
    \min_{j \in \supp{K_{\widetilde{\boeta}^N_{t-1}(\bw_n,\theta)}(\bw_n,\theta)^{(\cdot,i)}} } & K_{\widetilde{\boeta}^N_{t-1}(\bw_n,\theta)}(\bw_n,\theta)^{(j,i)} 
    \geq \min_{\eta \in [0,C]} \min_{j \in \supp{K_{\eta}(\bw_n,\theta)^{(\cdot,i)}} } K_{\eta}(\bw_n,\theta)^{(j,i)}.
    \end{split}
    \end{equation}
    We can then conclude:
    \begin{equation}
    \begin{split}
        \bpi_{n,t|t-1}^N(\bw_n,\theta)^{(i)}
        &\geq 
        \bar{m}_{t-1} \min_{\theta \in \Theta, w \in \mathbb{W},\eta \in [0,C]} \min_{j \in \supp{K_{\eta}(w,\theta)^{(\cdot,i)}} } K_{\eta}(w,\theta)^{(j,i)}\\
        &\geq 
        \bar{m}_{t-1} \min_{\theta \in \Theta, w \in \mathbb{W},\eta \in [0,C]} \min_{(i,j) \in \supp{K_{\eta}(w,\theta)} } K_{\eta}(w,\theta)^{(j,i)}.
    \end{split}
    \end{equation}
    Note that $K_{\eta}(w,\theta)$ is continuous in $\eta$ because of Assumption \ref{ass:kernel_continuity} and also in $w,\theta$ because of Assumption \ref{ass:compactness_continuity}. Furthermore, the domain of $\eta$ is $[0,C]$, which does not depend on $N$ and it is compact. Additionally, $\mathbb{W},\Theta$ are compact by Assumption \ref{ass:compactness_continuity}. Hence we can conclude by the Weirstrass theorem that there exist $\eta_{min}\in [0,C]$, $w_{min}\in \mathbb{W}, (i,j) \in \supp{K_{\eta}(w,\theta)}$, and $\theta_{min} \in \Theta$ such that:
\begin{equation}
 m_{K} \coloneqq K_{\eta_{min}}(w_{min},\theta_{min})^{(j,i)} =  \min_{\theta \in \Theta, w \in \mathbb{W},\eta \in [0,C]} \min_{(i,j) \in \supp{K_{\eta}(w,\theta)} } K_{\eta}(w,\theta)^{(j,i)}>0,
\end{equation}
  where strict positivity follows from the observation made in Equation \eqref{eq:strict_pos_support}. Hence: 
    \begin{equation}
    \begin{split}
        \bpi_{n,t|t-1}^N(\bw_n,\theta)^{(i)}
        &\geq 
        \bar{m}_{t-1} m_{K}>0,
    \end{split}
    \end{equation}
    where the lower bounding constants do not depend on $\bW^N,\bY_{1:t-1}^N$, $\theta$ or $N$. Therefore we can conclude that there exist $\bar{m}_{t-1}, m_{K}>0$ such that:
    $$
    \mathbb{P} \left ( \bpi_{n,t|t-1}^N(\bw_n,\theta)^{(i)}
        \geq 
        \bar{m}_{t-1} m_{K} \quad \forall i \in \supp{\bpi_{n,t|t-1}^N(\bw_n,\theta)} \right )=1 \quad \forall \theta \in \Theta.
    $$
    
    Similarly, for $i \in \supp{\bmu_{n,t}^N(\bw_n,\theta)}$,
    \begin{equation}
    \begin{split}
        \bmu_{n,t}^N(\bw_n,\theta)^{(i)}
        &=
        \sum_{j} \bpi_{n,t|t-1}^N(\bw_n,\theta)^{(j)} G(\bw_n,\theta)^{(j,i)}\\
        &\geq 
        \bar{m}_{t-1} m_{K} \sum_{j \in \supp{\bpi_{n,t|t-1}^N(\bw_n,\theta)}} G(\bw_n,\theta)^{(j,i)},
    \end{split}
    \end{equation}
    where the inequality follows from what we have proven above. Moreover:
    \begin{align}
    i \in \supp{\bmu_{n,t}^N(\bw_n,\theta)}
    &\Longleftrightarrow
    \sum_{j \in \supp{\bpi_{n,t|t-1}^N(\bw_n,\theta)}} G(\bw_n,\theta)^{(j,i)}\neq 0,
    \end{align}
    meaning that there exists at least one component of $G(\bw_n,\theta)^{(\cdot,i)}$ which is different from zero. Hence following the same reasoning as above:
    \begin{equation}
    \begin{split}
        \bmu_{n,t}^N(\bw_n,\theta)^{(i)}
        &\geq 
        \bar{m}_{t-1} m_{K} \min_{w \in \mathbb{W},\theta \in \Theta } \min_{j \in \supp{G(w,\theta)^{(\cdot,i)}}} G(w,\theta)^{(j,i)}\\
        &\geq 
        \bar{m}_{t-1} m_{K} \min_{w \in \mathbb{W},\theta \in \Theta } \min_{(i,j) \in \supp{G(w,\theta)}} G(w,\theta)^{(j,i)},
    \end{split}
    \end{equation}
    and as $G(w,\theta)$ is continuous in $w,\theta$ because of Assumption \ref{ass:compactness_continuity} and $\mathbb{W}, \Theta$ are compact because of Assumption \ref{ass:compactness_continuity} we can conclude by Weirstrass theorem that there exists a minimum $m_{G}$ such that:
    \begin{equation}
    \begin{split}
        \bmu_{n,t}^N(\bw_n,\theta)^{(i)}
        &\geq 
        \bar{m}_{t-1} m_{K} m_{G}>0,
    \end{split}
    \end{equation}
    where the strict inequality follows from considering a minimum on the support of matrix $G$. 
    
    We conclude that there exist $\bar{m}_{t-1}, m_{K}, m_{G}>0$ such that:
    $$
    \mathbb{P} \left ( \bmu_{n,t}^N(\bw_n,\theta)^{(i)}
        \geq 
        \bar{m}_{t-1} m_{K} m_{G} \quad \forall i \in \supp{\bmu_{n,t}^N(\bw_n,\theta)} \right )=1 \quad \forall \theta \in \Theta.
    $$

    Consider $i \in \supp{\bpi_{n,t}^N(\bw_n,\theta)}$ then:
    \begin{equation}
        \begin{split}
            \bpi_{n,t}^N(\bw_n,\theta)^{(i)}
            &= 
            \bpi_{n,t|t-1}^N(\bw_n,\theta)^{(i)} \frac{\sum_{j}  G(\bw_n,\theta)^{(i,j)} (\by_{n,t}^N)^{(j)} }{\sum_{j} (\by_{n,t}^N)^{(j)}\bmu_{n,t}^N(\bw_n,\theta)^{(j)}}\\
            &\geq 
            \bar{m}_{t-1} m_{K} \sum_{j}  G(\bw_n,\theta)^{(i,j)} (\by_{n,t}^N)^{(j)},
        \end{split}
    \end{equation}
    where the inequality follows from what we have proven above, from $\sum_{j} (\by_{n,t}^N)^{(j)}\bmu_{n,t}^N(\bw_n,\theta)^{(j)} \leq 1$ by definition and $\sum_{j} (\by_{n,t}^N)^{(j)}\bmu_{n,t}^N(\bw_n,\theta)^{(j)} \neq 0$ $\mathbb{P}$-almost surely because of Theorem \ref{thm:CAL_well_definess}. We now observe that:
    \begin{align}
    i \in \supp{\bpi_{n,t}^N(\bw_n,\theta)}
    &\Longleftrightarrow
    \sum_{j}  G(\bw_n,\theta)^{(i,j)} (\by_{n,t}^N)^{(j)}\neq 0,
    \end{align}
    meaning that there is at least one element of the $G(\bw_n,\theta)^{(i,\cdot)}$ which is different from zero. Following the same reasoning of $\bmu_{n,t}^N(\bw_n,\theta)^{(i)}$ we can conclude:
    \begin{equation}
        \begin{split}
            \bpi_{n,t}^N(\bw_n,\theta)^{(i)}
            &\geq 
            \bar{m}_{t-1} m_{K} \min_{j \in \supp{G(\bw_n,\theta)^{(i,\cdot)}}}  G(\bw_n,\theta)^{(i,j)}\\
            &\geq 
            \bar{m}_{t-1} m_{K} \min_{w \in \mathbb{W}, \theta \in \Theta} \min_{(i,j) \in \supp{G(w,\theta)}}  G(w,\theta)^{(i,j)}\\
            &\geq
            \bar{m}_{t-1} m_{K} m_{G}>0.
        \end{split}
    \end{equation}
    
    We can conclude that there exist constants $\bar{m}_{t-1}, m_{K}, m_{G}>0$ such that:
    $$
    \mathbb{P} \left ( \bpi_{n,t}^N(\bw_n,\theta)^{(i)}
        \geq 
        \bar{m}_{t-1} m_{K} m_{G} \quad \forall i \in \supp{\bpi_{n,t}^N(\bw_n,\theta)} \right )=1, \quad \forall \theta \in \Theta.
    $$
    
    We can then set $\bar{m}_t \coloneqq \bar{m}_{t-1} m_{K} m_{G}$ and conclude that there exists $\bar{m}_t >0$ such that:
    $$
    \mathbb{P} \left ( \bpi_{n,t}^N(\bw_n,\theta)^{(i)}
        \geq 
        \bar{m}_{t} \quad \forall i \in \supp{\bpi_{n,t}^N(\bw_n,\theta)} \right )=1, \quad \forall \theta \in \Theta,
    $$
    which closes the induction, meaning that the above equality holds for an arbitrary $t$. As a consequence, we also have that for any $t\geq 1$ there exists ${m}_{t}>0$ such that:
    $$
    \mathbb{P} \left ( \bmu_{n,t}^N(\bw_n,\theta)^{(i)}
        \geq 
        {m}_{t}\quad \forall i \in \supp{\bmu_{n,t}^N(\bw_n,\theta)} \right )=1, \quad \forall \theta \in \Theta,
    $$
    with ${m}_{t}\coloneqq \bar{m}_{t-1} m_{K} m_{G}$, which concludes the proof.
    
\end{proof}

\subsection{\texorpdfstring{$L^4$}{L4} bounds for averages of the data-generating process} \label{sec:asympt_DGP}

In Section \ref{sec:asympt_DGP} we establish $L^4$ bounds for averages across the population of disease states and observations.

\paragraph{Initial condition.} We start by considering the population at $t=0$.

\begin{proposition} \label{prop:bound_initial_cond}
    Under Assumption \ref{ass:w_iid}, there exists $\alpha_0\geq 0$ such that for any bounded function $f: \mathbb{W} \to [0,B]^M$, i.e. $\norm{f}_\infty \leq B < \infty$,
    \begin{equation}
        \normiii[\Bigg]{\frac{1}{N} \sum_{n \in [N]} 
        f(\bw_n)^\top \bx_{n,0}^N - \int f(w)^\top p_0\left( w, \theta^\star \right) \Gamma(d w)}_4 \leq 2 B \sqrt[4]{6} N^{-\frac{1}{2}} \alpha_0.
    \end{equation}
    Moreover, there exists $\widetilde{\alpha}_0\geq 0$ such that for any bounded function $f: \mathbb{W} \to [0,B]^M$,
    \begin{equation}
        \normiii[\Bigg]{\frac{1}{N} \sum_{n \in [N]} 
        f(\bw_n)^\top \bx_{n,0}^N - f(\bw_n)^\top p_0\left( \bw_n, \theta^\star \right)}_4 \leq 2 B \sqrt[4]{6} N^{-\frac{1}{2}} \widetilde{\alpha}_0.
    \end{equation}
\end{proposition}

\begin{proof}
    For the first statement note that by Minkowski inequality:
    \begin{align}
        &\normiii[\Bigg]{\frac{1}{N} \sum_{n \in [N]} 
        f(\bw_n)^\top \bx_{n,0}^N - \int f(w)^\top p_0\left( w, \theta^\star \right) \Gamma(d w)}_4\\
        &\leq
        \normiii[\Bigg]{\frac{1}{N} \sum_{n \in [N]} 
        f(\bw_n)^\top \bx_{n,0}^N - f(\bw_n)^\top p_0\left( \bw_n, \theta^\star \right) }_4 \label{eq:initial_cond_bound_A}\\
        &\quad + 
        \normiii[\Bigg]{\frac{1}{N} \sum_{n \in [N]} 
        f(\bw_n)^\top p_0\left( \bw_n, \theta^\star \right) - \int f(w)^\top p_0\left( w, \theta^\star \right) \Gamma(d w)}_4. \label{eq:initial_cond_bound_B}
    \end{align}
    Starting from \eqref{eq:initial_cond_bound_A}, which is also the second statement, we observe that:
    \begin{equation}
        \begin{split}
        \mathbb{E} \left [ f(\bw_n)^\top \bx_{n,0}^N | \bW^N\right] 
        &= 
        \sum_{i =1}^M \mathbb{E} \left [ f(\bw_n)^{(i)} (\bx_{n,0}^N)^{(i)} | \bW^N \right ] \\
        &=
        \sum_{i =1}^M f(\bw_n)^{(i)} p_0\left( \bw_n, \theta^\star \right)^{(i)} 
        =
        f(\bw_n)^\top p_0\left( \bw_n, \theta^\star \right),
        \end{split}
    \end{equation}
    moreover:
    \begin{equation}
        \begin{split}
        &\abs{ f(\bw_n)^\top \bx_{n,0}^N- f(\bw_n)^\top p_0\left( \bw_n, \theta^\star \right)} \\
        &\leq 
        \sum_{i =1}^M \left [ \abs{ f(\bw_n)^{(i)}  (\bx_{n,0}^N)^{(i)}} + \abs{f(\bw_n)^{(i)} p_0\left( \bw_n, \theta^\star \right)^{(i)}} \right ]\\
        &\leq 
        B \sum_{i =1}^M \left [ (\bx_{n,0}^N)^{(i)} + p_0\left( \bw_n, \theta^\star \right)^{(i)} \right ] 
        = 
        2 B, \quad \mathbb{P}\text{-almost surely},
        \end{split}
    \end{equation}
    as $\bx_{n,0}^N$ is a one-hot encoding vector and $p_0\left( \bw_n, \theta^\star \right)$ is a probability distribution for any $\bw_n$. As $f(\bw_n)^\top \bx_{n,0}^N- f(\bw_n)^\top p_0\left( \bw_n, \theta^\star \right)$ are also conditionally independent across $n$ given the population covariates $\bW^N$ because of the factorization of the initial distribution, we can apply Lemma \ref{lemma:mean_0_bound} and conclude:
    \begin{equation}
    \normiii[\Bigg]{\frac{1}{N} \sum_{n \in [N]} f(\bw_n)^\top \bx_{n,0}^N -  f(\bw_n)^\top p_0\left( \bw_n, \theta^\star \right)}_4 \leq 2B \sqrt[4]{6}  N^{-\frac{1}{2}},
    \end{equation}
    which also proves the second statement for $\widetilde{\alpha}_{0} = 1$.
    
    Similarly for \eqref{eq:initial_cond_bound_B} we have:
    \begin{equation}
        \begin{split}
        \mathbb{E} \left [ f(\bw_n)^\top p_0\left( \bw_n, \theta^\star \right) \right] 
        &= 
        \int f(w)^\top p_0\left( w, \theta^\star \right) \Gamma(d w),
        \end{split}
    \end{equation}
    and also:
    $$
    \abs{f(\bw_n)^\top p_0\left( \bw_n, \theta^\star \right) - \int f(w)^\top p_0\left( w, \theta^\star \right) \Gamma(d w)} \leq 2 B, \quad \mathbb{P}\text{-almost surely}.
    $$
    As $f(\bw_n)^\top p_0\left( \bw_n, \theta^\star \right) - \int f(w)^\top p_0\left( w, \theta^\star \right) \Gamma(d w)$ are independent because functions of independent random variables, indeed $\bw_n$ are i.i.d. samples from $\Gamma$, we can apply Lemma \ref{lemma:mean_0_bound} and conclude:
    \begin{equation}
        \normiii[\Bigg]{\frac{1}{N} \sum_{n \in [N]} 
        f(\bw_n)^\top p_0\left( \bw_n, \theta^\star \right) - \int f(w)^\top p_0\left( w, \theta^\star \right) \Gamma(d w)}_4 \leq 2B \sqrt[4]{6}  N^{-\frac{1}{2}}.
    \end{equation}
    
    By putting everything together we conclude the proof by setting $\alpha_0=2$, indeed:
    \begin{align}
            &\normiii[\Bigg]{\frac{1}{N} \sum_{n \in [N]} 
        f(\bw_n)^\top \bx_{n,0}^N - \int f(w)^\top p_0\left( w, \theta^\star \right) \Gamma(d w)}_4\\
        &\leq
        \normiii[\Bigg]{\frac{1}{N} \sum_{n \in [N]} 
        f(\bw_n)^\top \bx_{n,0}^N - f(\bw_n)^\top p_0\left( \bw_n, \theta^\star \right) }_4\\
        &\quad + 
        \normiii[\Bigg]{\frac{1}{N} \sum_{n \in [N]} 
        f(\bw_n)^\top p_0\left( \bw_n, \theta^\star \right) - \int f(w)^\top p_0\left( w, \theta^\star \right) \Gamma(d w)}_4\\
        &\quad \leq
        2B \sqrt[4]{6}  N^{-\frac{1}{2}} + 2B \sqrt[4]{6}  N^{-\frac{1}{2}} = 2B \sqrt[4]{6}  N^{-\frac{1}{2}} 2.
    \end{align}
\end{proof}

\paragraph{Dynamics.}
We now turn to the behavior of $\frac{1}{N}\sum_{n \in [N]}f(\bw_n)^\top \bx_{n,t}^N$, for $t\geq1$. We need the following definitions, for $w \in \mathbb{W}$ and $t\geq 1$:
\begin{equation}\label{eq:recursion_lambda}
    \begin{aligned}    &\blambda_{0}^\infty(w,\theta^\star) \coloneqq p_0\left( w, \theta^\star \right) \\
    &\boeta_{t-1}^{\infty} (w,\theta^\star) \coloneqq \int d(w,\widetilde{w},\theta^\star)^\top \blambda_{t-1}^\infty (\widetilde{w},\theta^\star) \Gamma(d\widetilde{w})\\
    &\blambda_{t}^\infty (w, \theta^\star ) \coloneqq \left [ \blambda_{t-1}^\infty (w,\theta^\star)^\top K_{ \boeta_{t-1}^{\infty} (w,\theta^\star )  }(w,\theta^\star) \right ]^\top.
    \end{aligned}
\end{equation}
The quantity $\boeta_{t-1}^\infty$ can be loosely interpreted as being $\eta^N(\cdot,\cdot,\bW^N,\bX^N_{t-1})$ but with $\bW^N$ and $\bX^N_{t-1}$ integrated out in the $N\to\infty$ limit. 

\begin{proposition}\label{prop:bound_X}
    Under assumptions \ref{ass:w_iid},\ref{ass:eta_structure},\ref{ass:kernel_continuity}, for any $t \geq 1$ there exists $\alpha_t > 0$ such that for any bounded function $f: \mathbb{W} \to [0,B]^M$
    \begin{equation}
        \normiii[\Bigg]{\frac{1}{N} \sum_{n \in [N]} f(\bw_n)^\top  \bx_{n,t}^N - \int f(w)^\top \blambda_{t}^\infty(w,\theta^\star) \Gamma(d w)}_4 \leq 2 B \sqrt[4]{6} N^{-\frac{1}{2}}  \alpha_t.
    \end{equation}
    Moreover, under assumptions \ref{ass:eta_structure},\ref{ass:kernel_continuity}, for any $t\geq 1$ there exists $\widetilde{\alpha}_t > 0$ such that for any bounded function $f: \mathbb{W} \to [0,B]^M$,
    \begin{equation}
        \normiii[\Bigg]{\frac{1}{N} \sum_{n \in [N]} f(\bw_n)^\top  \bx_{n,t}^N - f(\bw_n)^\top \blambda_{t}^\infty (\bw_n, \theta^\star)}_4 \leq 2 B \sqrt[4]{6} N^{-\frac{1}{2}} \widetilde{\alpha}_t.
    \end{equation}
\end{proposition}

\begin{proof}
    We prove the first statement by induction on $t$ and start by assuming that there exists $\alpha_{t-1} \in \mathbb{R}_+$ such that for any bounded function $f: \mathbb{W} \to [0,B]^M$ we have:
    \begin{equation}\label{ass:recursive_bound} 
        \normiii[\Bigg]{\frac{1}{N} \sum_{n \in [N]} f(\bw_n)^\top \bx_{n,t-1}^N - \int f(w)^\top \blambda_{t-1}^\infty(w,\theta^\star) \Gamma(d w)}_4 \leq 2 B \sqrt[4]{6} N^{-\frac{1}{2}} \alpha_{t-1},
    \end{equation}
    which is true at the initial time step under $\alpha_{0}=2$ because of the first statement of Proposition \ref{prop:bound_initial_cond}.
    
    We then look at time step $t$ and decompose the problem into three sub-problems:
    \begin{equation}
        \begin{split}
            &f(\bw_n)^\top \bx_{n,t}^N  - \int f(w)^\top \blambda_{t}^\infty(w,\theta^\star) \Gamma(d w)
            \\
            &= 
            f(\bw_n)^\top \bx_{n,t}^N - f(\bw_n)^\top \left [ (\bx_{n,t-1}^N)^\top K_{\boeta^N_{t-1}(\bw_n,\theta^\star)}( \bw_n,\theta^\star) \right ]^\top  \\
            &+
            f(\bw_n)^\top \left [ (\bx_{n,t-1}^N)^\top K_{\boeta^N_{t-1}(\bw_n,\theta^\star)}( \bw_n,\theta^\star) \right ]^\top - f(\bw_n)^\top \left [ (\bx_{n,t-1}^N)^\top K_{\boeta^\infty_{t-1}(\bw_n,\theta^\star)}( \bw_n, \theta^\star)  \right ]^\top \\
            &+
            f(\bw_n)^\top \left [ (\bx_{n,t-1}^N)^\top K_{\boeta^\infty_{t-1}(\bw_n,\theta^\star)}( \bw_n, \theta^\star)  \right ]^\top -  \int f(w)^\top \left [ \blambda_{t-1}^\infty(w,\theta^\star)^\top K_{\boeta^\infty_{t-1}(w,\theta^\star)}(w,\theta^\star)\right ]^\top \Gamma(d w).
        \end{split}
    \end{equation}
   By the Minkowski inequality, we can conclude that
    \begin{align}
            &\normiii[\Bigg]{\frac{1}{N}\sum_{n \in [N]} f(\bw_n)^\top \bx_{n,t}^N  - \int f(w)^\top \blambda_{t}^\infty(w,\theta^\star) \Gamma(d w)}_4
            \\
            &\leq 
            \normiii[\Bigg]{\frac{1}{N}\sum_{n \in [N]} f(\bw_n)^\top \bx_{n,t}^N - f(\bw_n)^\top \left [ (\bx_{n,t-1}^N)^\top K_{\boeta^N_{t-1}(\bw_n,\theta^\star)}( \bw_n,\theta^\star) \right ]^\top  }_4 \label{subeq:A}\\
            &+
            \normiii[\Bigg]{\frac{1}{N}\sum_{n \in [N]} f(\bw_n)^\top \left [ K_{\boeta^N_{t-1}(\bw_n,\theta^\star)}( \bw_n,\theta^\star) - K_{\boeta^\infty_{t-1}(\bw_n,\theta^\star)}( \bw_n, \theta^\star)  \right ]^\top \bx_{n,t-1}^N }_4
            \label{subeq:B}\\
            &+
            \Bigg{\vvvert}\frac{1}{N}\sum_{n \in [N]} \left [ K_{\boeta^\infty_{t-1}(\bw_n,\theta^\star)}( \bw_n, \theta^\star)  f(\bw_n) \right ]^\top \bx_{n,t-1}^N \\
            & \qquad \qquad - \int f(w)^\top \left [ \blambda_{t-1}^\infty(w,\theta^\star)^\top K_{\boeta^\infty_{t-1}(w,\theta^\star)}(w,\theta^\star) \right ]^\top \Gamma(d w) \Bigg{\vvvert}_4. \label{subeq:C}
    \end{align}
    
    Consider \eqref{subeq:A}, we can notice that:
    \begin{equation}
        \begin{split}
            \mathbb{E} \left [ f(\bw_n)^\top \bx_{n,t}^N | \bX^N_{t-1} , \bW^N \right ] 
            &=
            \sum_{j = 1}^M \mathbb{E} \left [ f(\bw_n)^{(j)} (\bx_{n,t}^N)^{(j)} | \bX^N_{t-1} , \bW^N \right ]\\
            &= 
            \sum_{j = 1}^M \left [ (\bx_{n,t-1}^N)^\top K_{\boeta^N_{t-1}(\bw_n,\theta^\star)}( \bw_n,\theta^\star) \right ]^{(j)}  f(\bw_n)^{(j)} \\
            &= f(\bw_n)^\top \left [ (\bx_{n,t-1}^N)^\top K_{\boeta^N_{t-1}(\bw_n,\theta^\star)}( \bw_n,\theta^\star) \right ]^\top.
        \end{split}
    \end{equation}
    Given that we are considering differences of scalar products of $f(\bw_n)$ with probability\slash one-hot encoding vectors, we have:
    \begin{equation}
        \abs{f(\bw_n)^\top \bx_{n,t}^N - f(\bw_n)^\top \left [ (\bx_{n,t-1}^N)^\top K_{\boeta^N_{t-1}(\bw_n,\theta^\star)}( \bw_n,\theta^\star) \right ]^\top} \leq 2 B \quad \mathbb{P}-\text{almost surely}.
    \end{equation}
    As $\bx_{n,t}^N$ are conditionally independent across $n$ given the population covariates and the state of the population at the previous time step we can apply Lemma \ref{lemma:mean_0_bound} and get:
    \begin{equation}
        \normiii[\Bigg]{\frac{1}{N}\sum_{n \in [N]} f(\bw_n)^\top \bx_{n,t}^N - f(\bw_n)^\top \left [ (\bx_{n,t-1}^N)^\top K_{\boeta^N_{t-1}(\bw_n,\theta^\star)}( \bw_n,\theta^\star) \right ]^\top  }_4 \leq 2 B \sqrt[4]{6} N^{-\frac{1}{2}}.
    \end{equation}
    
    We now consider \eqref{subeq:B} and note that:
    \begin{equation}
        \begin{split}
            &\abs{f(\bw_n)^\top \left [ (\bx_{n,t-1}^N)^\top K_{\boeta^N_{t-1}(\bw_n,\theta^\star)}( \bw_n,\theta^\star) \right ]^\top - f(\bw_n)^\top \left [ (\bx_{n,t-1}^N)^\top K_{\boeta^\infty_{t-1}(\bw_n,\theta^\star)}( \bw_n, \theta^\star)  \right ]^\top }\\ &\leq \sum_{i =1}^M \sum_{j = 1}^M \abs{(\bx_{n,t-1}^N)^{(i)} \left [ K_{\boeta^N_{t-1}(\bw_n,\theta^\star)}( \bw_n,\theta^\star) - K_{\boeta^\infty_{t-1}(\bw_n,\theta^\star)}( \bw_n, \theta^\star)  \right ]^{(i,j)} f(\bw_n)^{(j)}}\\
            &\leq B \norm{K_{\boeta^N_{t-1}(\bw_n,\theta^\star)}( \bw_n,\theta^\star) - K_{\boeta^\infty_{t-1}(\bw_n,\theta^\star)}( \bw_n, \theta^\star) }_\infty\\
	    &\leq B L \abs{\boeta^N_{t-1}(\bw_n,\theta^\star) - \boeta^\infty_{t-1}(\bw_n,\theta^\star)},
        \end{split}
    \end{equation} 
    where the final steps follow from the boundedness of $f$, the definition of one-hot encoding vectors, and the Lipschitz-continuity assumption on the transition matrix stated in Assumption \ref{ass:kernel_continuity}. We now use the structure of $\boldsymbol{\eta}^N_{n,t-1}$ given in Assumption \ref{ass:eta_structure}:
    \begin{equation} \label{eq:etabound}
    \begin{split}
        &\abs{\boeta^N_{t-1}(\bw_n,\theta^\star) - \boeta^\infty_{t-1}(\bw_n,\theta^\star)}\\
        &= \abs{ \frac{1}{N} \sum_{k \in [N]} d(\bw_n,\bw_{k},\theta^\star)^\top \bx^N_{k,t-1} - \int d(\bw_n,\bw,\theta^\star)^\top \blambda_{t-1}^\infty(\bw,\theta^\star) \Gamma(d \bw) },     
    \end{split}
    \end{equation} 
    hence we get the following bound for \eqref{subeq:B}:
    \begin{equation}
        \begin{split}
            &\normiii[\Bigg]{\frac{1}{N}\sum_{n \in [N]} f(\bw_n)^\top \left [ K_{\boeta^N_{t-1}(\bw_n,\theta^\star)}( \bw_n,\theta^\star) - K_{\boeta^\infty_{t-1}(\bw_n,\theta^\star)}( \bw_n, \theta^\star)  \right ]^\top \bx_{n,t-1}^N }_4\\
            &\leq
            \frac{1}{N} \sum_{n \in [N]} \normiii[\Bigg]{f(\bw_n)^\top \left [ K_{\boeta^N_{t-1}(\bw_n,\theta^\star)}( \bw_n,\theta^\star) - K_{\boeta^\infty_{t-1}(\bw_n,\theta^\star)}( \bw_n, \theta^\star)  \right ]^\top \bx_{n,t-1}^N }_4\\
            &\leq
            \frac{B L}{N} \sum_{n \in [N]} \normiii[\Bigg]{\frac{1}{N} \sum_{k \in [N]} d(\bw_n,\bw_{k},\theta^\star)^\top \bx^N_{k,t-1} - \int d(\bw_n,w, \theta^\star) \blambda_{t-1}^\infty(w,\theta^\star) \Gamma(d w)}_4.
        \end{split}
    \end{equation} 
    As for almost any realization $w_n \in \mathbb{W}$ of $\bw_n$ the function $d(w_n,\cdot)$ is a bounded function because of Assumption \ref{ass:eta_structure} we can apply our inductive hypothesis and get:
    \begin{equation}
        \begin{split}
            &\normiii[\Bigg]{\frac{1}{N} \sum_{k \in [N]} d(w_n,\bw_{k},\theta^\star)^\top \bx^N_{k,t-1} - \int d(w_n,w, \theta^\star) \blambda_{t-1}^\infty(w,\theta^\star) \Gamma(d w)}_4 \leq 2 C \sqrt[4]{6} N^{-\frac{1}{2}} \alpha_{t-1}.
        \end{split}
    \end{equation} 
    Since the above bound holds for any $w_n\in \mathbb{W}$, the same inequality holds $\mathbb{P}$-almost surely when the random covariate $\bw_n$ is substituted in place of $w_n$. We then conclude:
    \begin{equation}
        \begin{split}
            &\normiii[\Bigg]{\frac{1}{N}\sum_{n \in [N]} f(\bw_n)^\top \left [ K_{\boeta^N_{t-1}(\bw_n,\theta^\star)}( \bw_n,\theta^\star) - K_{\boeta^\infty_{t-1}(\bw_n,\theta^\star)}( \bw_n, \theta^\star)  \right ]^\top \bx_{n,t-1}^N }_4\\
            &\leq
            \frac{B L}{N} \sum_{n \in [N]} \normiii[\Bigg]{\frac{1}{N} \sum_{k \in [N]} d(\bw_n,\bw_{k},\theta^\star)^\top \bx^N_{k,t-1} - \int d(\bw_n,w, \theta^\star) \blambda_{t-1}^\infty(w,\theta^\star) \Gamma(d w)}_4\\
            &\leq
            \frac{B L}{N} \sum_{n \in [N]} 2 C \sqrt[4]{6} N^{-\frac{1}{2}} \alpha_{t-1}
            =
            2 B L C \sqrt[4]{6} N^{-\frac{1}{2}} \alpha_{t-1}.
        \end{split}
    \end{equation} 
    where the first step follows from Minkowski inequality, the second one from our previous calculations, and the final one from the inductive assumption and the fact that $\norm{d(\bw_n,\cdot)}_\infty \leq C$,  $\mathbb{P}$-almost surely from Assumption \ref{ass:eta_structure}. 
    
    The last term that is left to bound is \eqref{subeq:C}. Given that it can be easily proven that for any row-stochastic matrix $K$ the function $Kf(\cdot)$ given by $\bw \mapsto K f(\bw)$ is such that $\norm{Kf(\cdot)}_\infty \leq \norm{f}_\infty$, we can conclude that $\norm{K_{\boeta^\infty_{t-1}(\cdot,\theta^\star)}(\cdot,\theta^\star)f(\cdot)}_\infty \leq B$. We can then apply our inductive assumption and conclude:    
    \begin{equation}
        \begin{split}
            &\Bigg{\vvvert}
            \frac{1}{N}\sum_{n \in [N]} \left [ K_{\boeta^\infty_{t-1}(\bw_n,\theta^\star)}( \bw_n, \theta^\star)  f(\bw_n)  \right ]^\top \bx_{n,t-1}^N \\
            & \qquad \qquad - \int \left [ K_{\boeta^\infty_{t-1}(w,\theta^\star)}(w,\theta^\star) f(w) \right ]^\top \blambda_{t-1}^\infty(w,\theta^\star) \Gamma(d w) \Bigg{\vvvert}_4
            \leq
            2 B \sqrt[4]{6} N^{-\frac{1}{2}} \alpha_{t-1}.
        \end{split}
    \end{equation}
    
    By putting everything together we have:
    \begin{equation}
        \begin{split}
            &\normiii[\Bigg]{\frac{1}{N} \sum_{n \in [N]} f(\bw_n) \bx_{n,t}^N - \int f(w)^\top \blambda_{t}^\infty(w,\theta^\star) \Gamma(d w)}_4\\
            &\leq 
            2 B \sqrt[4]{6} N^{-\frac{1}{2}} + 2 B L C \sqrt[4]{6} N^{-\frac{1}{2}} \alpha_{t-1} + 2 B \sqrt[4]{6} N^{-\frac{1}{2}} \alpha_{t-1}\\
            &\leq 
            2 B \sqrt[4]{6} N^{-\frac{1}{2}} \left [ 1+ (1 + L C) \alpha_{t-1} \right ],
        \end{split}
    \end{equation}
    from which we can conclude the proof of the first statement by setting $\alpha_t = \left [ 1+ (1 + L C) \alpha_{t-1} \right ]$.

    The proof of the second statement follows similarly by induction on $t$. Start by assuming that there exists $\widetilde{\alpha}_{t-1} \in \mathbb{R}_+$ for any $f: \mathbb{W} \to [0,B]^M$ bounded function such that for time $t-1$ we have:
    \begin{equation}
        \normiii[\Bigg]{\frac{1}{N} \sum_{n \in [N]} f(\bw_n)^\top \bx_{n,t-1}^N - f(\bw_n)^\top \blambda_{t-1}^\infty(\bw_n,\theta^\star)}_4 \leq 2 B \sqrt[4]{6} N^{-\frac{1}{2}} \widetilde{\alpha}_{t-1},
    \end{equation}
    which is true at the initial time step under $\widetilde{\alpha}_{0}=1$ because of the second statement of Proposition \ref{prop:bound_initial_cond}. We then follow the same steps as the previous proof, but we substitute \eqref{subeq:C} with:
    \begin{equation}
        \begin{split}
             \normiii[\Bigg]{\frac{1}{N}\sum_{n \in [N]} \left [ K_{\boeta^\infty_{t-1}(\bw_n,\theta^\star)}( \bw_n, \theta^\star)  f(\bw_n) \right ]^\top \left [ \bx_{n,t-1}^N - \blambda_{t-1}^\infty(\bw_n, \theta^\star) \right ]}_4 
             \leq 
             2 B \sqrt[4]{6} N^{-\frac{1}{2}} \widetilde{\alpha}_{t-1},
        \end{split}
    \end{equation}
    which is bounded because of the induction hypothesis. We can then conclude:
    \begin{equation}
        \begin{split}
            &\normiii[\Bigg]{\frac{1}{N} \sum_{n \in [N]} f(\bw_n)^\top \bx_{n,t}^N - f(\bw_n)^\top \blambda_{t}^\infty (\bw_n, \theta^\star)^\top }_4\\
            &\leq 
            2 B \sqrt[4]{6} N^{-\frac{1}{2}} + 2 B C \sqrt[4]{6} N^{-\frac{1}{2}} \alpha_{t-1} + 2 B \sqrt[4]{6} N^{-\frac{1}{2}} \widetilde{\alpha}_{t-1}\\
            &\leq 
            2 B \sqrt[4]{6} N^{-\frac{1}{2}} \left [ 1+ C \alpha_{t-1} + \widetilde{\alpha}_{t-1} \right ],
        \end{split}
    \end{equation}
    from which we can conclude the proof of the second statement by setting $\widetilde{\alpha}_t = \left [ 1+ C \alpha_{t-1} + \widetilde{\alpha}_{t-1} \right ]$.
\end{proof}

The following corollary tells us that in the large population limit,  the interaction term $\boeta^N_{t-1}(\bw_n,\theta^\star)$ converges to the quantity $\boeta^\infty_{t-1}(\bw_n,\theta^\star)$ defined in \eqref{eq:recursion_lambda} which depends on individual-specific covariate $\bw_n$, but not on the covariates or diseases states of the rest of the population. This can be interpreted as meaning that individuals become statistically decoupled as the population size grows. 

\begin{corollary}\label{corol:eta_bound}
    Under assumptions \ref{ass:w_iid},\ref{ass:eta_structure},\ref{ass:kernel_continuity}, for any $t\geq 1$ there exists $\alpha_{t-1} > 0$ such that:
    \begin{equation}
        \normiii[\Bigg]{\boeta^N_{t-1}(\bw_n,\theta^\star) - \boeta^\infty_{t-1}(\bw_n,\theta^\star)}_4 \leq 2 C \sqrt[4]{6} N^{-\frac{1}{2}} \alpha_{t-1}.
    \end{equation}
\end{corollary}

\begin{proof}
    The result is a byproduct of the proof of Proposition \ref{prop:bound_X}. In particular, the passage proving the bound on Equation \eqref{eq:etabound} proves exactly the statement of the present corollary.
\end{proof}

\paragraph{Observations.}
Using Proposition \ref{prop:bound_X}, we next establish an $L^4$ bound for averages across the population of observations. We need the following definition, for $w \in \mathbb{W}$ and $t\geq 1$:
\begin{equation} \label{eq:recursion_nu}  \bnu_t^\infty(w,\theta^\star) \coloneqq \left [\blambda_t^\infty(w,\theta^\star)^\top G(w,\theta^\star) \right ]^\top.
\end{equation}

\begin{proposition}\label{prop:bound_Y}
    Under assumptions \ref{ass:w_iid},\ref{ass:eta_structure},\ref{ass:kernel_continuity}, for any $t\geq 1$ there exists $\alpha_t > 0$ such that for any bounded  $f: \mathbb{W} \to [0,B]^{M+1}$
    \begin{equation}
        \normiii[\Bigg]{\frac{1}{N} \sum_{n \in [N]} f(\bw_n)^\top \by_{n,t}^N - \int f(w)^\top \bnu_t^\infty(w,\theta^\star) \Gamma(d w)}_4 \leq 2 B \sqrt[4]{6} N^{-\frac{1}{2}} \left ( 1 + \alpha_t \right ).
    \end{equation}
    Moreover, under assumptions \ref{ass:eta_structure},\ref{ass:kernel_continuity}, for any $t\geq 1$ there exists $\widetilde{\alpha}_t > 0$ such that for bounded function $f: \mathbb{W} \to [0,B]^{M+1}$:
    \begin{equation}
        \normiii[\Bigg]{\frac{1}{N} \sum_{n \in [N]} f(\bw_n)^\top \by_{n,t}^N - f(\bw_n)^\top \bnu_t^\infty (\bw_n,\theta^\star) }_4 \leq 2 B \sqrt[4]{6} N^{-\frac{1}{2}} \left ( 1 + \widetilde{\alpha}_t \right ).
    \end{equation}
\end{proposition}

\begin{proof}
    To prove the first statement, the first step is to note that by Minkowski inequality:
    \begin{equation}
        \begin{split}
        &\normiii[\Bigg]{\frac{1}{N} \sum_{n \in [N]} f(\bw_n)^\top \by_{n,t}^N -  \int f(w)^\top \bnu_t^\infty(w,\theta^\star) \Gamma(d w)}_4 \\
        &\leq 
        \normiii[\Bigg]{\frac{1}{N} \sum_{n \in [N]} f(\bw_n)^\top \by_{n,t}^N - f(\bw_n)^\top \left [ (\bx_{n,t}^N)^\top G(\bw_n,\theta^\star) \right ]^\top}_4 \\
        &\quad + 
        \normiii[\Bigg]{\frac{1}{N} \sum_{n \in [N]} f(\bw_n)^\top \left [ (\bx_{n,t}^N)^\top G(\bw_n,\theta^\star) \right ]^\top - \int f(w)^\top \left [ \blambda_t^\infty(w,\theta^\star)^\top G(w,\theta^\star) \right ]^\top \Gamma(d w)}_4,
        \end{split}
    \end{equation}
    as by definition $ \int f(w)^\top \bnu_t^\infty(w,\theta^\star) \Gamma(d w) = \int f(w)^\top \left [ \blambda_t^\infty(w,\theta^\star)^\top G(w,\theta^\star) \right ]^\top \Gamma(d w)$. For the first term, we notice that: 
    \begin{equation}
        \mathbb{E} \left [f(\bw_n)^\top \by_{n,t}^N \Big{|} \bX_{t}^N, \bW^N \right ] =  f(\bw_n)^\top \left [ (\bx_{n,t}^N)^\top G(\bw_n,\theta^\star)  \right ] = \left [ G(\bw_n,\theta^\star) f(\bw_n) \right ]^\top  \bx_{n,t}^N,
    \end{equation}
    meaning that the arguments of the sum are all conditionally independent and with mean zero. Moreover, given that $[(\bx_{n,t}^N)^\top G(\bw_n,\theta^\star)]^\top$ is a probability vector and $\by_{n,t}^N$ is a one-hot encoding vector, we can also conclude that the arguments of the sums are all bounded by $2B$, hence we can apply Lemma \ref{lemma:mean_0_bound} and conclude:
    \begin{equation}
        \begin{split}
        \normiii[\Bigg]{\frac{1}{N} \sum_{n \in [N]} f(\bw_n)^\top \by_{n,t}^N - f(\bw_n)^\top \left [ (\bx_{n,t}^N)^\top G(\bw_n,\theta^\star) \right ]^\top}_4 
        \leq 2 B \sqrt[4]{6} N^{-\frac{1}{2}}.
        \end{split}
    \end{equation}
    For the second term we just need some refactoring:
    \begin{equation}
        \begin{split}
        &\normiii[\Bigg]{\frac{1}{N} \sum_{n \in [N]} f(\bw_n)^\top \left [ (\bx_{n,t}^N)^\top G(\bw_n,\theta^\star) \right ]^\top - \int f(w)^\top \left [ \blambda_t^\infty(w,\theta^\star)^\top G(w,\theta^\star) \right ]^\top \Gamma(d w)}_4\\
        &=
        \normiii[\Bigg]{\frac{1}{N} \sum_{n \in [N]} \left [ G(\bw_n,\theta^\star) f(\bw_n) \right ]^\top \bx_{n,t}^N - \int \left [ G(w,\theta^\star) f(w) \right ]^\top\blambda_t^\infty(w,\theta^\star)  \Gamma(d w)}_4.
        \end{split}
    \end{equation}
    As $\norm{G(\cdot,\theta^\star)f(\cdot)}_\infty \leq \norm{f}_\infty$ because $G(\cdot,\theta^\star)$ is a row-stochastic matrix, we can just apply the first statement of Proposition \ref{prop:bound_X} and conclude:
    \begin{equation}
        \begin{split}
        &\normiii[\Bigg]{\frac{1}{N} \sum_{n \in [N]} f(\bw_n)^\top \left [ (\bx_{n,t}^N)^\top G(\bw_n,\theta^\star) \right ]^\top - \int f(w)^\top \left [ \blambda_t^\infty(w,\theta^\star)^\top G(w,\theta^\star) \right ]^\top \Gamma(d w)}_4\\
        &\leq 2 B \sqrt[4]{6} N^{-\frac{1}{2}} \alpha_t.
        \end{split}
    \end{equation}
    We then conclude the proof by putting everything together:
    \begin{equation}
        \begin{split}
        &\normiii[\Bigg]{\frac{1}{N} \sum_{n \in [N]} f(\bw_n)^\top \by_{n,t}^N -  \int f(w)^\top \bnu_t^\infty(w,\theta^\star) \Gamma(d w)}_4 
        \leq 2 B \sqrt[4]{6} N^{-\frac{1}{2}}\left ( 1 + \alpha_t \right ).
        \end{split}
    \end{equation}

    Given the above proof strategy, it is trivial to establish the second statement by simply substituting the last term of the Minkowski inequality with:
    \begin{equation}
        \begin{split}
        &\normiii[\Bigg]{\frac{1}{N} \sum_{n \in [N]} f(\bw_n)^\top \left [ (\bx_{n,t}^N)^\top G(\bw_n,\theta^\star) \right ]^\top - f(\bw_n)^\top \blambda_t^\infty(\bw_n,\theta^\star)}_4,
        \end{split}
    \end{equation}
    and by applying the second statement of Proposition \ref{prop:bound_X}.
\end{proof}

\subsection{\texorpdfstring{$L^4$}{L4} bounds for the CAL filtering algorithm} \label{sec:asympt_CAL}

The aim of this section is to establish $L^4$ bounds for averages across the population of the various probability vectors computed in the CAL filtering recursion, \eqref{rec:CAL_rec}. In Proposition \ref{prop:bound_CAL_initial} and Proposition \ref{prop:bound_CAL_correction} below, we shall see that the large-population behavior of these averages is determined by the following quantities.
Specifically, for $w \in \mathbb{W}$:
\begin{equation}\label{rec:asympt_CAL_rec_full}
    \begin{aligned}
    &\bar{\bpi}_{0}^\infty(w,\theta) \coloneqq p_{0}(w,\theta),\\
    &\bar{\boeta}_{t-1}^\infty(w,\theta) \coloneqq \int d(w,\widetilde{w},\theta)^\top \bar{\bpi}_{t-1}^\infty(\widetilde{w},\theta) \Gamma(d \widetilde{w}),\\
    &\bar{\bpi}_{t|t-1}^\infty(w,\theta)  \coloneqq  \left [ \bar{\bpi}_{t-1}^\infty(w,\theta)^{\top} K_{\bar{\boeta}_{t-1}^\infty(w,\theta)} (w, \theta)\right ]^\top,\\
  &\bar{\bmu}_{t}^\infty(w,\theta)  \coloneqq  \left [ \bar{\bpi}_{t|t-1}^\infty(w,\theta)^{\top} G(w,\theta) \right ]^\top,\\
    &\bar{\bpi}_{t}^\infty(w,\theta)  \coloneqq \bar{\bpi}_{t|t-1}^\infty(w,\theta) \odot \left \{ \left [   G(w,\theta) \oslash \left ( 1_M \bar{\bmu}_{t}^\infty(w,\theta)^\top \right ) \right ] \bnu_{t}^\infty(w,\theta^\star)  \right \},
    \end{aligned}
\end{equation}
where $\bnu_{t}^\infty(w,\theta^\star)$ is defined in \eqref{eq:recursion_nu}. 

\begin{proposition}\label{prop:bound_CAL_initial}
    Under Assumption \ref{ass:w_iid}, for any bounded function $f: \mathbb{W} \to [0,B]^M$ it holds that:
    \begin{equation}
        \normiii[\Bigg]{\frac{1}{N} \sum_{n \in [N]} f(\bw_n)^\top \bpi_{n,0}^N(\bw_n,\theta) - \int f(\bw)^\top \bar{\bpi}_{0}^\infty(\bw,\theta) \Gamma(d \bw)}_4 \leq 2 B \sqrt[4]{6} N^{-\frac{1}{2}}.
    \end{equation}
    Moreover, for any bounded function $f: \mathbb{W} \to [0,B]^M$:
    \begin{equation}
        \normiii[\Bigg]{\frac{1}{N} \sum_{n \in [N]} f(\bw_n)^\top \bpi_{n,0}^N(\bw_n,\theta) - f(\bw_n)^\top \bar{\bpi}_{0}^\infty(\bw_n,\theta)}_4 = 0.
    \end{equation}
\end{proposition}

\begin{proof}
    The first statement follows the same reasoning of Proposition \ref{prop:bound_initial_cond}, but we include it here for completeness. As: 
    $$
    \mathbb{E}\left [ f(\bw_n)^\top \bpi_{n,0}^N(\bw_n,\theta) \right ] = \int f(w)^\top \bar{\bpi}_{0}^\infty(w,\theta) \Gamma(d w),
    $$
    and we have a sequence of independent and bounded random variables:
    $$\abs{f(\bw_n)^\top \bpi_{n,0}^N(\bw_n,\theta) - \int f(w)^\top \bar{\bpi}_{0}^\infty(w,\theta) \Gamma(d w)}\leq 2B,$$
    we can conclude the proof by applying Lemma \ref{lemma:mean_0_bound}. The second statement follows trivially from the definition of the limiting process, indeed $\bpi_{n,0}^N(\bw_n,\theta) = \bar{\bpi}_{0}^\infty(\bw_n,\theta)$.
\end{proof}

\begin{proposition}\label{prop:bound_CAL_prediction}
    If there exists a constant $\gamma_{t-1} > 0$ such that for any $f: \mathbb{W} \to [0,B]^M$
    \begin{equation}
        \normiii[\Bigg]{\frac{1}{N} \sum_{n \in [N]} f(\bw_n)^\top \left [ \bpi_{n,t-1}^N(\bw_n,\theta) - \bar{\bpi}_{t-1}^\infty(\bw_n, \theta)  \right ]}_4 \leq 2 B \sqrt[4]{6} N^{-\frac{1}{2}} \gamma_{t-1},
    \end{equation}
    then under assumptions \ref{ass:eta_structure},\ref{ass:kernel_continuity} there exists $\gamma_{t|t-1} > 0$ such that for any bounded function  $f: \mathbb{W} \to [0,B]^M$,
    \begin{equation}
        \normiii[\Bigg]{\frac{1}{N} \sum_{n \in [N]} f(\bw_n)^\top \left [ \bpi_{n,t|t-1}^N(\bw_n,\theta) - \bar{\bpi}_{t|t-1}^\infty(\bw_n,\theta) \right ]}_4 \leq 2 B \sqrt[4]{6} N^{-\frac{1}{2}} \gamma_{t|t-1}.
    \end{equation}
    Moreover, under assumptions \ref{ass:w_iid},\ref{ass:eta_structure},\ref{ass:kernel_continuity}, for any  bounded function $f: \mathbb{W} \to [0,B]^M$ it also holds:
    \begin{equation}
        \normiii[\Bigg]{\frac{1}{N} \sum_{n \in [N]} f(\bw_n)^\top \bpi_{n,t|t-1}^N(\bw_n,\theta) - \int f(w)^\top \bar{\bpi}_{t|t-1}^\infty(w,\theta) \Gamma(d w) ) }_4 \leq 2 B \sqrt[4]{6} N^{-\frac{1}{2}} \left ( \gamma_{t|t-1} + 1 \right ).
    \end{equation}
\end{proposition}

\begin{proof}
   As a preliminary notice that: 
    \begin{equation} \label{eq:minkowski_application}
    \begin{split}
        &\normiii[\Bigg]{\frac{1}{N} \sum_{n \in [N]} f(\bw_n)^\top \bpi_{n,t-1}^N(\bw_n,\theta) - \int f(w)^\top \bar{\bpi}_{t-1}^\infty(w, \theta) \Gamma( d w) }_4\\
        &\leq 
        \normiii[\Bigg]{\frac{1}{N} \sum_{n \in [N]} f(\bw_n)^\top \left [ \bpi_{n,t-1}^N(\bw_n,\theta) - \bar{\bpi}_{t-1}^\infty(\bw_n, \theta)  \right ]}_4\\
        &+
        \normiii[\Bigg]{\frac{1}{N} \sum_{n \in [N]} f(\bw_n)^\top \bar{\bpi}_{t-1}^\infty(\bw_n, \theta) - \int f(w)^\top \bar{\bpi}_{t-1}^\infty(w, \theta) \Gamma( d w)}_4\\
	&\leq
        2 B \sqrt[4]{6} N^{-\frac{1}{2}} (\gamma_{t-1} + 1),
    \end{split}
    \end{equation}
    as we can apply Minkowski inequality and then bound the first quantity with our assumption at time $t-1$ and the second quantity with Lemma \ref{lemma:mean_0_bound}, as we are dealing with an average of independent random variables that are mean zero and bounded (this follows the same steps as the proof of Proposition \ref{prop:bound_CAL_initial}). 

    As a consequence, we have the following implication of our assumption:
    \begin{equation} 
        \normiii[\Bigg]{\frac{1}{N} \sum_{n \in [N]} f(\bw_n)^\top \bpi_{n,t-1}^N(\bw_n,\theta) - \int f(w)^\top \bar{\bpi}_{t-1}^\infty(w, \theta) \Gamma( d w)}_4
        \leq
        2 B \sqrt[4]{6} N^{-\frac{1}{2}} (\gamma_{t-1} + 1),
    \end{equation}
    and we know that if we prove the first statement it is enough to follow the above reasoning to prove the second one.
    
    We now start the proof of the first statement by noting:
    \begin{equation}
    \begin{split}
        &f(\bw_n)^\top \bpi_{n,t|t-1}^N(\bw_n,\theta) - f(\bw_n)^\top \bar{\bpi}_{t|t-1}^\infty(\bw_n,\theta)  \\
        &= 
        f(\bw_n)^\top \left [ \bpi_{n,t-1}^N(\bw_n,\theta)^\top K_{\widetilde{\boeta}^N_{t-1}(\bw_n,\theta)}(\bw_n,\theta) \right ]^\top - f(\bw_n)^\top \left [ \bpi_{n,t-1}^N(\bw_n,\theta)^\top K_{\bar{\boeta}^\infty_{t-1}(\bw_n,\theta)}(\bw_n,\theta) \right ]^\top\\
        &+
        f(\bw_n)^\top \left [ \bpi_{n,t-1}^N(\bw_n,\theta)^\top K_{\bar{\boeta}^\infty_{t-1}(\bw_n,\theta)}(\bw_n,\theta) \right ]^\top - f(\bw_n)^\top \left [ \bar{\bpi}_{t-1}^\infty(\bw_n,\theta)^\top K_{\bar{\boeta}^\infty_{t-1}(\bw_n,\theta)}(\bw_n,\theta) \right ]^\top.
    \end{split}
    \end{equation}
    Hence we can apply Minkowski inequality:
    \begin{align}
        &\normiii[\Bigg]{\frac{1}{N} \sum_{n \in [N]} f(\bw_n)^\top \bpi_{n,t|t-1}^N(\bw_n,\theta) - f(\bw_n)^\top \bar{\bpi}_{t|t-1}^\infty(\bw_n,\theta)}_4\\
        &\leq
        \normiii[\Bigg]{\frac{1}{N} \sum_{n \in [N]} f(\bw_n)^\top \left [  K_{\widetilde{\boeta}^N_{t-1}(\bw_n,\theta)}(\bw_n,\theta) - K_{\bar{\boeta}^\infty_{t-1}(\bw_n,\theta)}(\bw_n,\theta) \right ]^\top \bpi_{n,t-1}^N(\bw_n,\theta)}_4 \label{eq:asympt_step_1_prediction}\\
        &\quad + 
        \normiii[\Bigg]{\frac{1}{N} \sum_{n \in [N]} \left [ K_{\bar{\boeta}^\infty_{t-1}(\bw_n,\theta)}(\bw_n,\theta) f(\bw_n) \right ]^\top \left [ \bpi_{n,t-1}^N(\bw_n,\theta) - \bar{\bpi}_{t-1}^\infty(\bw_n,\theta) \right ]}_4.\label{eq:asympt_step_2_prediction}
    \end{align}
    
    On \eqref{eq:asympt_step_1_prediction} we can apply Assumption \ref{ass:kernel_continuity} and get:
    \begin{equation}
    \begin{split}
        &\normiii[\Bigg]{\frac{1}{N} \sum_{n \in [N]} f(\bw_n)^\top \left [ K_{\widetilde{\boeta}^N_{t-1}(\bw_n,\theta)}(\bw_n,\theta) -  K_{\bar{\boeta}^\infty_{t-1}(\bw_n,\theta)}(\bw_n,\theta) \right ]^\top \bpi_{n,t-1}^N(\bw_n,\theta)}_4 \\
        &
        \leq 
        \frac{B}{N} \sum_{n \in [N]} \normiii[\Bigg]{\widetilde{\boeta}^N_{t-1}(\bw_n,\theta) -\bar{\boeta}^\infty_{t-1}(\bw_n,\theta)}_4,
    \end{split}
    \end{equation}
    which follows the same steps as the proof of Proposition \ref{prop:bound_X}. Again, similarly to Proposition \ref{prop:bound_X}, we can apply Assumption \ref{ass:kernel_continuity} on the Lipschitz continuity of the transition matrix, Assumption \ref{ass:eta_structure} on the structure of $\eta^N$ and, given the definition of $\widetilde{\boeta}_t^N$, $\bar{\boeta}^\infty_t$, we get:
    \begin{equation} \label{eq:etacal_bound}
        \begin{split}
            &\normiii[\Bigg]{\widetilde{\boeta}^N_{t-1}(\bw_n,\theta) -\bar{\boeta}^\infty_{t-1}(\bw_n,\theta)}_4\\
            &=
            \normiii[\Bigg]{\frac{1}{N} \sum_{k \in [N]} d(\bw_n, \bw_{k},\theta)^\top \bpi_{k,t-1}^N(\bw_{k},\theta) - \int d(\bw_n,w,\theta)^\top \bar{\bpi}_{t-1}^\infty(w,\theta) \Gamma(d w)}_4\\
            &\leq 
            2 L C \sqrt[4]{6} N^{-\frac{1}{2}} (\gamma_{t-1} +1),
        \end{split}
    \end{equation}
    because of the inductive assumption, and which holds $\mathbb{P}$-almost surely following the same reasoning of Proposition \ref{prop:bound_X}. Hence: 
    \begin{equation}\label{eq:bound_tilte_eta_bar_eta}
    \begin{split}
        &\normiii[\Bigg]{\frac{1}{N} \sum_{n \in [N]} f(\bw_n)^\top \left [ \bpi_{n,t-1}^N(\bw_n,\theta)^\top K_{\widetilde{\boeta}^N_{t-1}(\bw_n,\theta)}(\bw_n,\theta) - \bpi_{n,t-1}^N(\bw_n,\theta)^\top K_{\bar{\boeta}^\infty_{t-1}(\bw_n,\theta)}(\bw_n,\theta) \right ]^\top}_4 \\
        &
        \leq \frac{B}{N} \sum_{n \in [N]} \normiii[\Bigg]{\frac{1}{N} \sum_{k \in [N]} d(\bw_n, \bw_{k},\theta)^\top \bpi_{k,t-1}^N(\bw_{k},\theta) - \int d(\bw_n,w,\theta)^\top \bar{\bpi}_{t-1}^\infty(w,\theta) \Gamma(d w)}_4\\
        &\leq \frac{B}{N} \sum_{n \in [N]} 2 L C \sqrt[4]{6} N^{-\frac{1}{2}} (\gamma_{t-1} +1) = 2 B L C \sqrt[4]{6} N^{-\frac{1}{2}} (\gamma_{t-1} +1),
    \end{split}
    \end{equation}
    where the last step follows from \eqref{eq:minkowski_application} and from the same reasoning of Proposition \ref{prop:bound_X}, i.e. we prove the inequality $\mathbb{P}$-almost surely.

    On \eqref{eq:asympt_step_2_prediction}, we can observe that the vectorial function $\bw \mapsto K_{\bar{\boeta}^\infty_{t-1}(\bw,\theta)}(\bw,\theta) f(\bw)$ is bounded, i.e. $\norm{K_{\bar{\boeta}^\infty_{t-1}(\cdot,\theta)}(\cdot,\theta)f}_\infty\leq B$, as $K_{\bar{\boeta}^\infty_{t-1}(\bw,\theta)}$ is a row-stochastic matrix, hence we can apply our assumption on the time step $t-1$ and conclude:
    \begin{equation}
        \begin{split}
            &\normiii[\Bigg]{\frac{1}{N} \sum_{n \in [N]} \left [ K_{\bar{\boeta}^\infty_{t-1}(\bw_n,\theta)}(\bw_n,\theta) f(\bw_n) \right ]^\top \left [ \bpi_{n,t-1}^N(\bw_n,\theta) - \bar{\bpi}_{t-1}^\infty(\bw_n, \theta)  \right ] }_4\\
            & \leq 2 B \sqrt[4]{6} N^{-\frac{1}{2}} \gamma_{t-1}.
        \end{split}
    \end{equation}

    By putting everything together we get:
    \begin{equation}
        \begin{split}
        &\normiii[\Bigg]{\frac{1}{N} \sum_{n \in [N]} f(\bw_n)^\top \bpi_{n,t|t-1}^N(\bw_n,\theta) - f(\bw_n)^\top \bar{\bpi}_{t|t-1}^\infty(\bw_n,\theta)}_4\\
        &\leq
        2 B L C \sqrt[4]{6} N^{-\frac{1}{2}} (\gamma_{t-1}+1) + 2 B \sqrt[4]{6} N^{-\frac{1}{2}} \gamma_{t-1} = 2 B \sqrt[4]{6} N^{-\frac{1}{2}} \left[ L C (\gamma_{t-1}+1) + \gamma_{t-1} \right],
        \end{split}
    \end{equation}
    by setting $\gamma_{t|t-1}=\left[ C (\gamma_{t-1}+1) + \gamma_{t-1} \right]$ we conclude the proof of the first statement.

    Remark that as we already mentioned at the beginning of the proof we have:
    \begin{equation}
        \mathbb{E} \left [ f(\bw_n)^\top \bar{\bpi}_{t|t-1}^\infty(\bw_n,\theta) \right ] = \int f(w)^\top \bar{\bpi}_{t|t-1}^\infty(w,\theta) \Gamma(d w),
    \end{equation}
    and the random variables $f(\bw_n)^\top \bar{\bpi}_{t|t-1}^\infty(\bw_n,\theta) - \int f(w)^\top \bar{\bpi}_{t|t-1}^\infty(w,\theta) \Gamma(d w)$ are all bounded by $2B$ and independent, hence by Lemma \ref{lemma:mean_0_bound}:
    \begin{equation}
        \normiii[\Bigg]{\frac{1}{N} \sum_{n \in [N]} f(\bw_n)^\top \bar{\bpi}_{t|t-1}^\infty(\bw_n,\theta) - \int f(w)^\top \bar{\bpi}_{t|t-1}^\infty(w,\theta) \Gamma(d w)}_4 \leq 2 B \sqrt[4]{6} N^{-\frac{1}{2}},
    \end{equation}
    which proves the second statement.
\end{proof}

\begin{proposition}\label{prop:bound_CAL_correction}
    There exists $\gamma_{t|t-1} > 0$ for any $f: \mathbb{W} \to [0,B]^M$ bounded function  such that
    \begin{equation}
        \normiii[\Bigg]{\frac{1}{N} \sum_{n \in [N]} f(\bw_n)^\top \bpi_{n,t|t-1}^N(\bw_n,\theta) - f(\bw_n)^\top \bar{\bpi}_{t|t-1}^\infty(\bw_n,\theta)}_4 \leq 2 B \sqrt[4]{6} N^{-\frac{1}{2}} \gamma_{t|t-1},
    \end{equation}
    then under assumptions \ref{ass:compactness_continuity},\ref{ass:HMM_support},\ref{ass:eta_structure},\ref{ass:kernel_continuity} for any $f: \mathbb{W} \to [0,B]^M$ bounded function we have:
    \begin{equation}
        \normiii[\Bigg]{\frac{1}{N} \sum_{n \in [N]} f(\bw_n)^\top \bmu_{n,t}^N(\bw_n,\theta) - f(\bw_n)^\top \bar{\bmu}_{t}^\infty(\bw_n,\theta)}_4 \leq 2 B \sqrt[4]{6} N^{-\frac{1}{2}} \gamma_{t|t-1},
    \end{equation}
    and there exists $\gamma_{t} > 0$ such that for any bounded function $f: \mathbb{W} \to [0,B]^M$,
    \begin{equation}
        \normiii[\Bigg]{\frac{1}{N} \sum_{n \in [N]} f(\bw_n)^\top \bpi_{n,t}^N(\bw_n,\theta) - f(\bw_n)^\top \bar{\bpi}_{t}^\infty(\bw_n,\theta)}_4 \leq 2 B \sqrt[4]{6} N^{-\frac{1}{2}} \gamma_{t}.\label{stat:pi_N_infty}
    \end{equation}
    Moreover, under assumptions \ref{ass:compactness_continuity},\ref{ass:w_iid},\ref{ass:HMM_support},\ref{ass:eta_structure},\ref{ass:kernel_continuity}, for any $f: \mathbb{W} \to [0,B]^M$ bounded function
    \begin{equation}
        \normiii[\Bigg]{\frac{1}{N} \sum_{n \in [N]} f(\bw_n)^\top \bmu_{n,t}^N(\bw_n,\theta) - \int f(w)^\top \bar{\bmu}_{t}^\infty(w,\theta) \Gamma(d w) }_4 \leq 2 B \sqrt[4]{6} N^{-\frac{1}{2}} (1+\gamma_{t|t-1}),
    \end{equation}
    and:
    \begin{equation}
        \normiii[\Bigg]{\frac{1}{N} \sum_{n \in [N]} f(\bw_n)^\top \bpi_{n,t}^N(\bw_n,\theta) - \int f(w)^\top \bar{\bpi}_{t}^\infty(w,\theta) \Gamma( d w) }_4 \leq 2 B \sqrt[4]{6} N^{-\frac{1}{2}} (1+\gamma_{t}).
    \end{equation}
\end{proposition}

\begin{proof}
    The first statement is straightforward, indeed:
    \begin{equation}
        \begin{split}
            &\normiii[\Bigg]{\frac{1}{N} \sum_{n \in [N]} f(\bw_n)^\top \bmu_{n,t}^N(\bw_n,\theta) - f(\bw_n)^\top \bar{\bmu}_{t}^\infty(\bw_n,\theta)}_4\\
            &= 
            \normiii[\Bigg]{\frac{1}{N} \sum_{n \in [N]} f(\bw_n)^\top \left [ \bpi_{n,t|t-1}^N(\bw_n,\theta)^\top G(\bw_n,\theta) \right ]^\top - f(\bw_n)^\top \left [ \bar{\bpi}_{t|t-1}^\infty(\bw_n,\theta)^\top G(\bw_n,\theta) \right ]^\top }_4\\
            &=
            \normiii[\Bigg]{\frac{1}{N} \sum_{n \in [N]} \left [ G(\bw_n,\theta) f(\bw_n) \right ]^\top \bpi_{n,t|t-1}^N(\bw_n,\theta) - \left [ G(\bw_n,\theta) f(\bw_n) \right ]^\top \bar{\bpi}_{t|t-1}^\infty(\bw_n,\theta)}_4.
        \end{split}
    \end{equation}
    As $G(\bw,\theta)$ is a row-stochastic matrix the function $\bw \mapsto G(\bw,\theta) f(\bw)$ is bounded, and precisely $\norm{G(\cdot,\theta) f(\cdot)}_\infty \leq B$, meaning that the proof of the statement follows simply by the assumption on step $t-1$:
    \begin{equation}\label{eq:first_statement_mu}
        \begin{split}
            &\normiii[\Bigg]{\frac{1}{N} \sum_{n \in [N]} f(\bw_n)^\top \bmu_{n,t}^N(\bw_n,\theta) - f(\bw_n)^\top \bar{\bmu}_{t}^\infty(\bw_n,\theta)}_4\leq 
            2 B \sqrt[4]{6} N^{-\frac{1}{2}} \gamma_{t|t-1}.
        \end{split}
    \end{equation}

    The proof of the second statement is more involved and we need to start by reformulating $f(\bw_n)^\top \bpi_{n,t}^N(\bw_n,\theta)$:
    \begin{equation}
        \begin{split}
            f(\bw_n)^\top \bpi_{n,t}^N(\bw_n,\theta) 
            &= 
            \sum_{i =1}^M f(\bw_n)^{(i)} \bpi_{n,t}^N(\bw_n,\theta)^{(i)} \\
            &= 
            \sum_{i =1}^M f(\bw_n)^{(i)} \bpi_{n,t|t-1}^N(\bw_n,\theta)^{(i)} \sum_{j = 1}^M \frac{ G(\bw_n,\theta)^{(i,j)}}{\bmu_{n,t}^N(\bw_n,\theta)^{(j)}}(\by_{n,t}^N)^{(j)}\\
            &= 
            \left [ f(\bw_n) \odot \bpi_{n,t|t-1}^N(\bw_n,\theta) \right ]^\top \left [ G_{\bmu_{n,t}^N(\bw_n,\theta)}(\bw_n,\theta) \by_{n,t}^N \right ], \label{eq:pi_t_reformulation}
        \end{split}
    \end{equation}
    where we define $G_{\bmu}(\bw,\theta)$ as the matrix with elements $G_{\bmu}(\bw,\theta)^{(i,j)} = \frac{ G(\bw,\theta)^{(i,j)}}{\bmu^{(j)}}$ where $\frac{0}{0}=0$ by convention. Remark that, because of Theorem \ref{thm:CAL_well_definess}, the CAL is well-defined  $\mathbb{P}$-almost surely. Also a similar reformulation to the one in \eqref{eq:pi_t_reformulation} can be done on $\bar{\bpi}_t^\infty(\bw,\theta)$ meaning:
    $
            f(\bw)^\top \bar{\bpi}_t^\infty(\bw,\theta) 
            = 
            \left [ f(\bw) \odot \bar{\bpi}_{t|t-1}^\infty(\bw,\theta) \right ]^\top \left [ G_{\bar{\bmu}_t^\infty(\bw,\theta)}(\bw,\theta) \bnu_t^\infty(\bw,\theta^\star) \right ]
    $.
    Then we note that:
    \begin{equation}
    \begin{split}
            &f(\bw_n)^\top \bpi_{n,t}^N(\bw_n,\theta) - f(\bw_n)^\top \bar{\bpi}_{t}^\infty(\bw_n,\theta) \\
            &= 
            \left [ f(\bw_n) \odot \bpi_{n,t|t-1}^N(\bw_n,\theta) \right ]^\top \left [ G_{\bmu_{n,t}^N(\bw_n,\theta)}(\bw_n,\theta) \by_{n,t}^N \right ]\\
            &\qquad-\left [ f(\bw_n) \odot \bar{\bpi}_{t|t-1}^\infty(\bw_n,\theta) \right ]^\top \left [ G_{\bmu_{n,t}^N(\bw_n,\theta)}(\bw_n,\theta) \by_{n,t}^N \right ] \\
            &\qquad+\left [ f(\bw_n) \odot \bar{\bpi}_{t|t-1}^\infty(\bw_n,\theta) \right ]^\top \left [ G_{\bmu_{n,t}^N(\bw_n,\theta)}(\bw_n,\theta) \by_{n,t}^N \right ]\\
            &\qquad- \left [ f(\bw_n) \odot \bar{\bpi}_{t|t-1}^\infty(\bw_n,\theta) \right ]^\top \left [ G_{\bar{\bmu}_t^\infty(\bw_n,\theta)}(\bw_n,\theta) \by_{n,t}^N \right ] \\
            &\qquad+ \left [ f(\bw_n) \odot \bar{\bpi}_{t|t-1}^\infty(\bw_n,\theta) \right ]^\top \left [ G_{\bar{\bmu}_t^\infty(\bw_n,\theta)}(\bw_n,\theta) \by_{n,t}^N \right ]\\
            &\qquad- \left [ f(\bw_n) \odot \bar{\bpi}_{t|t-1}^\infty(\bw_n,\theta) \right ]^\top \left [ G_{\bar{\bmu}_t^\infty(\bw_n,\theta)}(\bw_n,\theta) \bnu_t^\infty(\bw_n,\theta^\star) \right ].
        \end{split}
    \end{equation}
    With the above decomposition we can apply Minkowski inequality and conclude:
        \begin{align}
            &\normiii[\Bigg]{ \frac{1}{N} \sum_{n \in [N]} f(\bw_n)^\top \bpi_{n,t}^N(\bw_n,\theta) - f(\bw_n)^\top \bar{\bpi}_{t}^\infty(\bw_n,\theta)}_4 \\
            &\leq 
            \Bigg{\vvvert}\frac{1}{N} \sum_{n \in [N]} \left [ f(\bw_n) \odot \bpi_{n,t|t-1}^N(\bw_n,\theta) \right ]^\top \left [ G_{\bmu_{n,t}^N(\bw_n,\theta)}(\bw_n,\theta) \by_{n,t}^N \right ] \\
            &\qquad-
             \left [ f(\bw_n) \odot \bar{\bpi}_{t|t-1}^\infty(\bw_n,\theta) \right ]^\top \left [ G_{\bmu_{n,t}^N(\bw_n,\theta)}(\bw_n,\theta) \by_{n,t}^N \right ] \Bigg{\vvvert}_4 \label{eq:correction_A}\\
            &\qquad+
            \Bigg{\vvvert} \frac{1}{N} \sum_{n \in [N]} \left [ f(\bw_n) \odot \bar{\bpi}_{t|t-1}^\infty(\bw_n,\theta) \right ]^\top \left [ G_{\bmu_{n,t}^N(\bw_n,\theta)}(\bw_n,\theta) \by_{n,t}^N \right ] \\
            &\qquad- 
             \left [ f(\bw_n) \odot \bar{\bpi}_{t|t-1}^\infty(\bw_n,\theta) \right ]^\top \left [ G_{\bar{\bmu}_t^\infty(\bw_n,\theta)}(\bw_n,\theta) \by_{n,t}^N \right ] \Bigg{\vvvert}_4 \label{eq:correction_B}\\
            &\qquad+ 
            \Bigg{\vvvert} \frac{1}{N} \sum_{n \in [N]} \left [ f(\bw_n) \odot \bar{\bpi}_{t|t-1}^\infty(\bw_n,\theta) \right ]^\top \left [ G_{\bar{\bmu}_t^\infty(\bw_n,\theta)}(\bw_n,\theta) \by_{n,t}^N \right ] \\
            &\qquad- 
             \left [ f(\bw_n) \odot \bar{\bpi}_{t|t-1}^\infty(\bw_n,\theta) \right ]^\top \left [ G_{\bar{\bmu}_t^\infty(\bw_n,\theta)}(\bw_n,\theta) \bnu_t^\infty(\bw_n,\theta^\star) \right ]\Bigg{\vvvert}_4.\label{eq:correction_C}
    \end{align}
    
    Starting from \eqref{eq:correction_A}, we notice that:
    \begin{equation}
    \begin{split}
        &\Bigg{\vvvert}\frac{1}{N} \sum_{n \in [N]} \left [ f(\bw_n) \odot \bpi_{n,t|t-1}^N(\bw_n,\theta) \right ]^\top \left [ G_{\bmu_{n,t}^N(\bw_n,\theta)}(\bw_n,\theta) \by_{n,t}^N \right ] \\
        &\qquad-
       \left [ f(\bw_n) \odot \bar{\bpi}_{t|t-1}^\infty(\bw_n,\theta) \right ]^\top \left [ G_{\bmu_{n,t}^N(\bw_n,\theta)}(\bw_n,\theta) \by_{n,t}^N \right ]\Bigg{\vvvert}_4 \\
        &=
        \normiii[\Bigg]{\frac{1}{N} \sum_{n \in [N]} \left [ \bpi_{n,t|t-1}^N(\bw_n,\theta) - \bar{\bpi}_{t|t-1}^\infty(\bw_n,\theta) \right ]^\top \left [ f(\bw_n) \odot G_{\bmu_{n,t}^N(\bw_n,\theta)}(\bw_n,\theta) \by_{n,t}^N \right ]}_4\\
        &=
        \normiii[\Bigg]{\frac{1}{N} \sum_{n \in [N]} \left [ f(\bw_n) \odot G_{\bmu_{n,t}^N(\bw_n,\theta)}(\bw_n,\theta) \by_{n,t}^N \right ]^\top \left [ \bpi_{n,t|t-1}^N(\bw_n,\theta) - \bar{\bpi}_{t|t-1}^\infty(\bw_n,\theta) \right ]}_4,
    \end{split}
    \end{equation}
    as $G_{\bmu_{n,t}^N(\bw_n,\theta)}(\bw_n,\theta) \by_{n,t}^N$ is a vector of probabilities, i.e. elements that are less or equal than 1, we can conclude that:
    $$
    \norm{f(\cdot) \odot G_{\bmu_{n,t}^N(\cdot,\theta)}(\cdot,\theta) \by_{n,t}^N}_\infty\leq B, \quad \mathbb{P} \text{-almost surely},
    $$
    hence, similarly to what we do with $d(\bw_n,\cdot)$ in the proof of Proposition \ref{prop:bound_X}, we can apply our assumption on time step $t-1$ for almost any realization of $\by_{n,t}^N$, and conclude:
    \begin{equation}
    \begin{split}
        &\Bigg{\vvvert}\frac{1}{N} \sum_{n \in [N]} \left [ f(\bw_n) \odot \bpi_{n,t|t-1}^N(\bw_n,\theta) \right ]^\top \left [ G_{\bmu_{n,t}^N(\bw_n,\theta)}(\bw_n,\theta) \by_{n,t}^N \right ] \\
        &\qquad-
         \left [ f(\bw_n) \odot \bar{\bpi}_{t|t-1}^\infty(\bw_n,\theta) \right ]^\top \left [ G_{\bmu_{n,t}^N(\bw_n,\theta)}(\bw_n,\theta) \by_{n,t}^N \right ] \Bigg{\vvvert}_4 
        \leq 2 B \sqrt[4]{6} N^{-\frac{1}{2}} \gamma_{t|t-1}.
    \end{split}
    \end{equation}

    As we define $G_\mu$ as the matrix with elements $G^{(i,j)} \slash \mu^{(j)}$, we can notice that given two vectors $a,b$ with the same dimensions of $G_\mu$ we have:
    \begin{equation}
        x^\top G_\mu b - x^\top G_{\widetilde{\mu}} b = \sum_{i,j} x^{(i)} y^{(j)} \frac{G^{(i,j)} \widetilde{\mu}^{(j)} - G^{(i,j)} \mu^{(j)}}{\widetilde{\mu}^{(j)}{\mu}^{(j)}} = \sum_{i,j} x^{(i)} \frac{y^{(j)}}{\widetilde{\mu}^{(j)}{\mu}^{(j)}}G^{(i,j)} \left ( \widetilde{\mu}^{(j)} - \mu^{(j)} \right )
    \end{equation}
    
    Hence we can reformulate \eqref{eq:correction_B}:
    \begin{equation}
        \begin{split}
            &\Bigg{\vvvert}  \frac{1}{N} \sum_{n \in [N]} \left [ f(\bw_n) \odot \bar{\bpi}_{t|t-1}^\infty(\bw_n,\theta) \right ]^\top \left [ G_{\bmu_{n,t}^N(\bw_n,\theta)}(\bw_n,\theta) \by_{n,t}^N \right ]\\
            &\qquad- 
            \left [ f(\bw_n) \odot \bar{\bpi}_{t|t-1}^\infty(\bw_n,\theta) \right ]^\top \left [ G_{\bar{\bmu}_t^\infty(\bw_n,\theta)}(\bw_n,\theta) \by_{n,t}^N \right ]\Bigg{\vvvert}_4\\
            &=
            \Bigg{\vvvert} \frac{1}{N} \sum_{n \in [N]} \left \{ \left [ f(\bw_n) \odot \bar{\bpi}_{t|t-1}^\infty(\bw_n,\theta) \right ]^\top G(\bw_n,\theta) \right \}^\top \odot \left [ \by_{n,t}^N \oslash \bmu_{n,t}^N(\bw_n,\theta) \oslash \bar{\bmu}_t^\infty(\bw_n,\theta)\right ]^\top\\
            &\qquad \qquad 
            \left [\bar{\bmu}_t^\infty(\bw_n,\theta) - \bmu_{n,t}^N(\bw_n,\theta) \right ] \Bigg{\vvvert}_4,
        \end{split}
    \end{equation}
    from which we can notice that for any $\bw_n$:
    \begin{equation}
        \begin{split}
            &\norm{\left \{ \left [ f(\bw_n) \odot \bar{\bpi}_{t|t-1}^\infty(\bw_n,\theta) \right ]^\top G(\bw_n,\theta) \right \} \odot \left [ \by_{n,t}^N \oslash \bmu_{n,t}^N(\bw_n,\theta) \oslash \bar{\bmu}_t^\infty(\bw_n,\theta)\right ]}_\infty\\
            &\leq
            \norm{f}_\infty \norm{ \by_{n,t}^N \oslash \bmu_{n,t}^N(\bw_n,\theta)}_\infty, \quad \mathbb{P} \text{-almost surely},
        \end{split}
    \end{equation}
    where the first step follows from $\bar{\bmu}_t^\infty(\bw,\theta) = \left [ \bar{\bpi}_{t|t-1}^\infty(\bw,\theta)^\top G(\bw,\theta) \right ]^\top$ and the elementwise ratio $\by_{n,t}^N \oslash \bmu_{n,t}^N(\bw_n,\theta)$ is well-defined because of Theorem \ref{thm:CAL_well_definess}. 

    Because of Proposition \ref{prop:CAL_as_bounded} we know that there exists ${m}_{t}>0$ such that $\bmu_{n,t}^N(\bw_n,\theta) \geq {m}_{t}$ almost surely. We can then conclude that for any $n \in [N]$:
    \begin{equation}
        \begin{split}
        &\norm{\left \{ \left [ f(\bw_n) \odot \bar{\bpi}_{t|t-1}^\infty(\bw_n,\theta) \right ]^\top G(\bw_n,\theta) \right \} \odot \left [ \by_{n,t}^N \oslash \bmu_{n,t}^N(\bw_n,\theta) \oslash \bar{\bmu}_t^\infty(\bw_n,\theta)\right ]}_\infty\\ &\leq \frac{\norm{f}_\infty}{{m}_{t}}
        \leq \frac{B}{{m}_{t}}, \quad \mathbb{P} \text{-almost surely}. 
        \end{split}
    \end{equation}
    Hence we can apply to \eqref{eq:correction_B} the same reasoning that we do with $d(\bw_n,\cdot)$ in the proof of Proposition \ref{prop:bound_X} and apply \eqref{eq:first_statement_mu} for almost any realization of $\bY^N_{1:t-1},\by^N_{n,t}$, i.e. $\mathbb{P}$-almost surely, and conclude:
    \begin{equation}
        \begin{split}
            &\Bigg{\vvvert}  \frac{1}{N} \sum_{n \in [N]} \left [ f(\bw_n) \odot \bar{\bpi}_{t|t-1}^\infty(\bw_n,\theta) \right ]^\top \left [ G_{\bmu_{n,t}^N(\bw_n,\theta)}(\bw_n,\theta) \by_{n,t}^N \right ]\\
            &\qquad- 
            \left [ f(\bw_n) \odot \bar{\bpi}_{t|t-1}^\infty(\bw_n,\theta) \right ]^\top \left [ G_{\bar{\bmu}_t^\infty(\bw_n,\theta)}(\bw_n,\theta) \by_{n,t}^N \right ] \Bigg{\vvvert}_4 \leq 2 B \sqrt[4]{6} N^{-\frac{1}{2}} \frac{\gamma_{t|t-1}}{{m}_{t}}.
        \end{split}
    \end{equation}

    Finally, we notice that for \eqref{eq:correction_C}:
    \begin{equation}
        \begin{split}
            &\Bigg{\vvvert}  \frac{1}{N} \sum_{n \in [N]} \left [ f(\bw_n) \odot \bar{\bpi}_{t|t-1}^\infty(\bw_n,\theta) \right ]^\top \left [ G_{\bar{\bmu}_t^\infty(\bw_n,\theta)}(\bw_n,\theta) \by_{n,t}^N \right ] \\
            &\qquad- 
           \left [ f(\bw_n) \odot \bar{\bpi}_{t|t-1}^\infty(\bw_n,\theta) \right ]^\top \left [ G_{\bar{\bmu}_t^\infty(\bw_n,\theta)}(\bw_n,\theta) \bnu_t^\infty(\bw_n,\theta^\star) \right ]\Bigg{\vvvert}_4\\
            &\qquad = 
           \Bigg{\vvvert}\frac{1}{N} \sum_{n \in [N]} \left \{ \left [ f(\bw_n) \odot \bar{\bpi}_{t|t-1}^\infty(\bw_n,\theta) \right ]^\top  G_{\bar{\bmu}_t^\infty(\bw_n,\theta)}(\bw_n,\theta) \right \} \left [ \by_{n,t}^N - \bnu_t^\infty(\bw_n,\theta^\star) \right ]\Bigg{\vvvert}_4,
        \end{split}
    \end{equation}
    from which we can see that:
    \begin{equation}
        \begin{split}
            \norm{\left [ f(\cdot) \odot \bar{\bpi}_{t|t-1}^\infty(\cdot,\theta) \right ]^\top  G_{\bar{\bmu}_t^\infty(\cdot,\theta)}(\cdot,\theta)}_\infty 
            &\leq 
            \norm{f}_\infty
            \norm{\left [ \bar{\bpi}_{t|t-1}^\infty(\cdot,\theta) \right ]^\top  G_{\bar{\bmu}_t^\infty(\cdot,\theta)}(\cdot,\theta)}_\infty\\
            &\leq 
            \norm{f}_\infty
            \norm{ \sum_i \bar{\bpi}_{t|t-1}^\infty(\cdot,\theta)^{(i)} \frac{G(\cdot,\theta)^{(i,j)}}{\bar{\bmu}_t^\infty(\cdot,\theta)^{(j)}}  }_\infty\\
            &= 
            \norm{f}_\infty,
        \end{split}
    \end{equation}
    where the last step follows from $\bar{\bpi}_{t|t-1}^\infty(\cdot,\theta), G(\cdot,\theta), \bar{\bmu}_t^\infty(\cdot,\theta)$ being almost surely positive and from $[\bar{\bmu}_t^\infty(\cdot,\theta) = \bar{\bpi}_{t|t-1}^\infty(\cdot,\theta)^\top G(\cdot,\theta)]^\top$. We can then apply Proposition \ref{prop:bound_Y} and conclude
    \begin{equation}
        \begin{split}
             &\Bigg{\vvvert}  \frac{1}{N} \sum_{n \in [N]} \left [ f(\bw_n) \odot \bar{\bpi}_{t|t-1}^\infty(\bw_n,\theta) \right ]^\top \left [ G_{\bar{\bmu}_t^\infty(\bw_n,\theta)}(\bw_n,\theta) \by_{n,t}^N \right ] \\
            &\qquad- 
             \left [ f(\bw_n) \odot \bar{\bpi}_{t|t-1}^\infty(\bw_n,\theta) \right ]^\top \left [ G_{\bar{\bmu}_t^\infty(\bw_n,\theta)}(\bw_n,\theta) \bnu_t^\infty(\bw_n,\theta^\star) \right ]\Bigg{\vvvert}_4
            \leq 2 B \sqrt[4]{6} N^{-\frac{1}{2}} \left ( 1 + \alpha_t \right ).
        \end{split}
    \end{equation}
    By putting everything together we then have:
    \begin{equation}
        \begin{split}
            &\normiii[\Bigg]{ \frac{1}{N} \sum_{n \in [N]} f(\bw_n)^\top \bpi_{n,t}^N(\bw_n,\theta) - f(\bw_n)^\top \bar{\bpi}_{t}^\infty(\bw_n,\theta)}_4 \\
            &\leq 2 B \sqrt[4]{6} N^{-\frac{1}{2}} \gamma_{t|t-1} + 2 B \sqrt[4]{6} N^{-\frac{1}{2}} \frac{\gamma_{t|t-1}}{{m}_{t}} + 2 B \sqrt[4]{6} N^{-\frac{1}{2}} \left ( 1 + \alpha_t \right )\\
            &=
            2 B \sqrt[4]{6} N^{-\frac{1}{2}} \left [ \gamma_{t|t-1} + \frac{\gamma_{t|t-1}}{{m}_{t}} + \left ( 1 + \alpha_t \right ) \right ],
        \end{split}
    \end{equation}
    which concludes the proof for $\gamma_t = \left [ \gamma_{t|t-1} + \frac{\gamma_{t|t-1}}{{m}_{t}} + \left ( 1 + \alpha_t \right ) \right ]$.

    The third and the fourth statements are simple consequences of the first and the second statement, Minkowski inequality and Lemma \ref{lemma:mean_0_bound}, see the reasoning to prove the second statement of Proposition \ref{prop:bound_CAL_prediction}.
\end{proof}

We now combine propositions \ref{prop:bound_CAL_initial} - \ref{prop:bound_CAL_correction}.
\begin{proposition}\label{prop:CAL_bound}
    Under assumptions \ref{ass:compactness_continuity},\ref{ass:w_iid},\ref{ass:HMM_support},\ref{ass:eta_structure},\ref{ass:kernel_continuity}, for any $t\geq 1$ there exist $\gamma_{t|t-1}>0$ and $\gamma_t>0$ such that for any bounded function $f: \mathbb{W} \to [0,B]^M$
    \begin{equation}
    \begin{aligned}
       &\normiii[\Bigg]{\frac{1}{N} \sum_{n \in [N]} f(\bw_n)^\top \bpi_{n,t|t-1}^N(\bw_n,\theta) -  \int f(w)^\top \bar{\bpi}_{t|t-1}^\infty(w,\theta) \Gamma(d w) }_4 
        \leq 2 B \sqrt[4]{6} N^{-\frac{1}{2}} \gamma_{t|t-1},\\
        &\normiii[\Bigg]{\frac{1}{N} \sum_{n \in [N]} f(\bw_n)^\top \bmu_{n,t}^N(\bw_n,\theta) - \int f(w)^\top \bar{\bmu}_{t}^\infty(w,\theta) \Gamma(d w) }_4 \leq 2 B \sqrt[4]{6} N^{-\frac{1}{2}} \gamma_{t|t-1},\\
        &\normiii[\Bigg]{\frac{1}{N} \sum_{n \in [N]} f(\bw_n)^\top \bpi_{n,t}^N(\bw_n,\theta) - \int f(w)^\top \bar{\bpi}_{t}^\infty(w,\theta) \Gamma(d w) }_4 \leq 2 B \sqrt[4]{6} N^{-\frac{1}{2}} \gamma_{t}.
    \end{aligned}
    \end{equation}
\end{proposition}

\begin{proof}
    Assume that the third statement is true at time $t-1$, which is verified for $t-1=0$ because of Proposition \ref{prop:bound_CAL_initial}. Then we can in turn apply  Proposition \ref{prop:bound_CAL_prediction} to prove the first statement, and Proposition \ref{prop:bound_CAL_correction} to prove the second and the third. As the third statement is also our inductive hypothesis and we have proven it to hold for the next time step $t$ we can state that the three statements hold for an arbitrary time step $t$.
\end{proof}

The following corollary is analogous to Corollary \ref{corol:eta_bound}, but addresses the quantity $\widetilde{\boeta}_{t-1}^N(\cdot,\cdot)$ computed in the CAL algorithm (recall \ref{rec:CAL_rec}). Corollary \ref{corol:eta_bound_CAL} can be interpreted as meaning that the probability vectors computed in the CAL algorithm become decoupled across the population as $N\to\infty$. Recall from \eqref{eq:recursion_lambda} and \eqref{rec:asympt_CAL_rec_full} the definitions of $\boeta_{t-1}^\infty(\cdot,\cdot)$ and $\bar{\boeta}_{t-1}^\infty(\cdot,\cdot)$.

\begin{corollary}\label{corol:eta_bound_CAL}
    Under assumptions \ref{ass:w_iid},\ref{ass:eta_structure},\ref{ass:kernel_continuity}, for any $t\geq 1$ there exists $\gamma_{t-1} > 0$ such that for any $w\in\mathbb{W}$ and $\theta\in\Theta$,
    \begin{equation}
        \normiii[\Bigg]{\widetilde{\boeta}^N_{t-1}(w,\theta) - \bar{\boeta}^\infty_{t-1}(w,\theta)}_4 \leq 
        2 L C \sqrt[4]{6} N^{-\frac{1}{2}} (\gamma_{t-1} +1).
    \end{equation}
    Moreover, $\bar{\boeta}^\infty_{t-1}(w,\theta^\star) = \boeta^\infty_{t-1}(w,\theta^\star)$.
\end{corollary}

\begin{proof}
    The first statement is a byproduct of Proposition \ref{prop:CAL_bound} and Equation \eqref{eq:etacal_bound} in the proof of Proposition \ref{prop:bound_CAL_prediction}. Indeed, Proposition \ref{prop:CAL_bound} ensures that the assumptions of Proposition \ref{prop:bound_CAL_prediction} hold.

     The second statement follows by induction, indeed $ \bar{\bpi}_{0}^\infty(w,\theta^\star) =  \blambda_{0}^\infty(w,\theta^\star)$ by definition hence $\bar{\bpi}_{t-1}^\infty(w,\theta^\star) =  \blambda_{t-1}^\infty(w,\theta^\star)$ holds for $t=1$. Suppose now $\bar{\bpi}_{t-1}^\infty(w,\theta^\star) =  \blambda_{t-1}^\infty(w,\theta^\star)$ is true, then we get that for any $i \in [M]$:
    \begin{equation}
        \begin{split}
            \bar{\bpi}_{t}^\infty(w,\theta^\star)^{(i)}
            &= 
            \bar{\bpi}_{t|t-1}^\infty(w,\theta^\star)^{(i)} \left \{ \left [   G(w,\theta^\star) \oslash \left ( 1_M \bar{\bmu}_{t}^\infty(w,\theta^\star)^\top \right ) \right ] \bnu_{t}^\infty(w,\theta^\star)  \right \}^{(i)}\\
            &= 
            \bar{\bpi}_{t|t-1}^\infty(w,\theta^\star)^{(i)} \left \{ \left [   G(w,\theta^\star) \oslash \left ( 1_M {\nu}_{t}^\infty(w,\theta^\star)^\top \right ) \right ] \bnu_{t}^\infty(w,\theta^\star)  \right \}^{(i)}\\
            &= 
            \bar{\bpi}_{t|t-1}^\infty(w,\theta^\star)^{(i)} \sum_{j = 1}^M  \frac{G(w,\theta^\star)^{(i,j)}}{\bnu_{t}^\infty(w,\theta^\star)^{(j)}} \bnu_{t}^\infty(w,\theta^\star)^{(j)}\\
            &= 
            \sum_{j = 1}^M \bar{\bpi}_{t-1}^\infty(w,\theta^\star)^{(j)} K_{\bar{\boeta}^\infty_{t-1}(w,\theta^\star)}(w,\theta^\star)^{(j,i)}\\
            &= 
            \sum_{j = 1}^M \blambda_{t-1}^\infty(w,\theta^\star)^{(j)} K_{\boeta^\infty_{t-1}(w,\theta^\star)}(w,\theta^\star)^{(j,i)} = \blambda_{t}^\infty(w,\theta^\star)^{(i)},
        \end{split}
    \end{equation}
    which concludes the proof.
\end{proof}

\subsection{The saturated processes and saturated CAL algorithm} \label{sec:limiting_quantities_support}

In Section \ref{sec:limiting_quantities_support} we introduce key objects which will allow us to understand the large-population behavior of the data-generating process and the CAL algorithm.  The first is what we call the \emph{one-individual saturated process}. This process can be interpreted as describing the evolution of a single individual in an infinite population, where the covariate vector associated with this one individual is drawn from the distribution $\Gamma$ appearing in Assumption \ref{ass:w_iid}. The term ``saturated'' refers to the fact that the law of the process is defined in terms of the limiting interaction terms $\boeta_t^\infty(\cdot,\cdot)$, $t\geq 0$, cf.  Corollary \ref{corol:eta_bound}.

We then introduce the \emph{population saturated process}, which consists of independent individuals, each of which is distributed in the same way as the one-individual saturated process, but with covariate vectors $(\bw_n)_{n\in[N]}$ which are shared with those in the data-generating process specified in Section \ref{sec:data_generating_process}. In Proposition \ref{thm:LLN_f_y_yinf} we shall show that averages across the population in the data-generating process approximate averages across the population saturated process.

Lastly, we introduce the \emph{saturated CAL algorithm}, which is very similar to the CAL algorithm, but uses a model that is decoupled across individuals. Proposition \ref{prop:mu_muhat_bound} below will tell us that the logarithm of the CAL approximates the logarithm of the corresponding quantity from the saturated CAL algorithm.
 
\paragraph{One-individual saturated process.}
We define the one-individual saturated process as:
\begin{equation}\label{eq:limiting_process}
    \begin{split}
    &\bw^\infty \sim \Gamma,\\
    &\bx_0^\infty|\bw^\infty \sim \mbox{Cat}\left(\, \cdot\,|\,p_0(\bw^\infty,\theta^\star)\right),\\
    &\bx_t^\infty|\bx_{t-1}^\infty,\bw^\infty \sim \mbox{Cat}\left(\, \cdot\,\left|\left [ (\bx_{t-1}^\infty)^\top K_{ \boeta_{t-1}^{\infty} (\bw^\infty,\theta^\star ) }(\bw^\infty,\theta^\star) \right ]^\top\right.\right),\\
    &\by_t^\infty|\bx_{t}^\infty,\bw^\infty \sim \mbox{Cat}\left( \cdot\left|\left [ (\bx_{t}^\infty)^\top G(\bw^\infty,\theta^\star) \right ]^\top\right.\right).
    \end{split}
\end{equation}

Conditional on $\bw^\infty$, the joint process $(\bx_t^\infty)_{t\geq 0}$, $(\by_t^\infty)_{t\geq 1}$ described in \eqref{eq:limiting_process} is a HMM. If in Recursion \eqref{rec:rand_asympt_CAL_rec_full} below $\bw^\infty$ and $\theta^\star$ are substituted in place of $w$ and $\theta$, then \eqref{rec:rand_asympt_CAL_rec_full} becomes the forward algorithm associated with this HMM.

Recall from \eqref{rec:asympt_CAL_rec_full} the definition of $\bar{\boeta}_{t-1}^\infty(\cdot,\cdot)$, and define for $w \in \mathbb{W}$ and $t\geq 1$:
\begin{equation}\label{rec:rand_asympt_CAL_rec_full}
    \begin{aligned}
    &\bpi_{0}^\infty(w,\theta) \coloneqq p_{0}(w,\theta),\\
    &\bpi_{t|t-1}^\infty(w,\theta)  \coloneqq  \left [ \bpi_{t-1}^\infty(w,\theta)^{\top} K_{\bar{\boeta}_{t-1}^\infty(w,\theta)} (w, \theta)\right ]^\top, \\
    &\bmu_{t}^\infty(w,\theta)  \coloneqq  \left [ \bpi_{t|t-1}^\infty(w,\theta)^{\top} G(w,\theta) \right ]^\top,\\
    &\bpi_{t}^\infty(w,\theta)  \coloneqq \bpi_{t|t-1}^\infty(w,\theta) \odot \left \{ \left [   G(w,\theta) \oslash \left ( 1_M \bmu_{t}^\infty(w,\theta)^\top \right ) \right ] \by_{t}^\infty  \right \}.
    \end{aligned}
\end{equation} 

Recall from \eqref{rec:rand_asympt_CAL_rec_full} the definition of $\bmu_t^\infty(\cdot,\cdot)$.

\begin{proposition} \label{prop:joint_rand_aymptoticCAL_distribution}
    Over a time horizon $T$, let $w \in \mathbb{W}$ and $y_{1:T} \in \mathbb{O}_{M+1}$:
    \begin{equation}
        p(\by_{1:T}^\infty|\bw^\infty,\theta^\star) \coloneqq 
        \prod_{t=1}^T \mbox{Cat}\left ( \by_{t}^\infty | \bmu_{t}^\infty(\bw^\infty, \theta^\star) \right ) = \prod_{t=1}^T (\by_{t}^\infty)^\top \bmu_t^\infty(\bw^\infty,\theta^\star),
    \end{equation}
    and it is obtained as a marginal distribution over $x_{0:T} \in \mathbb{O}_M$ of:
    \begin{equation}
    \begin{split}
        &p(\bx_{0:T}^\infty, \by_{1:T}^\infty|\bw^\infty,\theta^\star) \coloneqq \\ &\mbox{Cat}\left ( \bx^\infty_{0}| p_0(\bw^\infty, \theta^\star) \right ) \prod_{t=1}^T  \mbox{Cat}\left ( \bx^\infty_{t} | \left[ (\bx^\infty_{t-1})^\top K_{ \boeta_{t-1}^{\infty} (\bw^\infty, \theta^\star)}(\bw^\infty, \theta^\star) \right ]^\top \right )\\ 
        &\qquad\qquad\qquad\qquad \qquad \cdot \prod_{t=1}^T  \mbox{Cat}\left ( \by^\infty_t | \left [ ( \bx^\infty_{t})^\top G(\bw^\infty, \theta^\star) \right ]^\top \right )\\
        &= \left ( (\bx_{0}^\infty)^\top p_0(\bw^\infty, \theta^\star) \right ) \prod_{t=1}^T \left [ \left ( (\bx^\infty_{t-1})^\top K_{ \boeta_{t-1}^{\infty} (\bw^\infty, \theta^\star)}(\bw^\infty, \theta^\star) \bx^\infty_{t} \right ) \left (( \bx^\infty_{t})^\top G(\bw^\infty, \theta^\star) \by^\infty_{t} \right ) \right ].
    \end{split}
    \end{equation}
    Moreover, we have
    $$
    \begin{aligned}       &\bx_t^\infty|\by_{1:t-1}^\infty,\bw^\infty  \sim \mbox{Cat}\left ( \cdot | \bpi_{t|t-1}^\infty(\bw^\infty,\theta^\star) \right ),\\      &\bx_t^\infty|\by_{1:t}^\infty,\bw^\infty \sim \mbox{Cat}\left ( \cdot | \bpi_{t}^\infty(\bw^\infty,\theta^\star) \right ),\\
    &  \by_t^\infty|\by_{1:t-1}^\infty,\bw^\infty \sim \mbox{Cat}\left ( \cdot | \bmu_{t}^\infty(\bw^\infty,\theta^\star) \right ).
    \end{aligned}
    $$
\end{proposition}

\begin{proof}
    Note that because of Corollary \ref{corol:eta_bound_CAL} we have $\boeta_{t}^{\infty} (\bw^\infty, \theta^\star)=\bar{\boeta}_{t}^{\infty} (\bw^\infty, \theta^\star)$. The proof of the first statement follows then the same calculation of the proof of Proposition \ref{prop:joint_CAL_distribution}, with the infinite transition matrix $K_{\boeta_t^\infty(\bw^\infty,\theta^\star)}(\bw^\infty,\theta^\star)$. The proof of the remaining statements is trivial by noting that $(\bx_t^\infty)_{t \geq 0},(\by_t^\infty)_{t \geq 1}$ given $\bw^\infty$ is a HMM. Indeed, Recursion \ref{rec:rand_asympt_CAL_rec_full} can be interpreted as the forward algorithm for the HMM $(\bx_t^\infty)_{t \geq 0},(\by_t^\infty)_{t \geq 1}$ given $\bw^\infty$.
\end{proof}

\begin{proposition} \label{prop:rand_aympt_CAL_well_definess}
    Under assumptions \ref{ass:w_iid},\ref{ass:HMM_support}, for any $t\geq 1$ and $\theta \in \Theta$ we have that\\ $\sum_i {\bmu}_{t}^\infty(\bw^\infty,\theta)^{(i)}(\by_{t}^\infty)^{(i)}\neq 0$ \quad $\mathbb{P}-\text{almost surely}$. 
\end{proposition}

\begin{proof}
    The proof follows the same steps as the proof of Theorem \ref{thm:CAL_well_definess}, but we work with the one-individual saturated process.
    
    We can note that because of Assumption \ref{ass:HMM_support} if there exists $\theta \in \Theta$, $i \in [M]$ such that $\bpi_0^\infty(w,\theta)^{(i)} = p_0(w,\theta)^{(i)}=0$, then $p_0(w,\theta^\star)^{(i)}=0$ and so:
    $$
    \mathbb{P}((\bx_0^\infty)^{(i)}=0|\bw^\infty=w)= p_0(w,\theta)= \bpi_0^\infty(w,\theta)^{(i)}=0,
    $$
    meaning that $(\bx_0^\infty)^{(i)}=0$ almost surely in $\mathbb{P}(\cdot|\bw^\infty=w)$.

    Assume now by induction that if there exists $\theta \in \Theta$, $i \in [M]$ such that $\bpi_{t-1}^\infty(w,\theta)^{(i)} =0$, then $(\bx_{t-1}^\infty)^{(i)}=0$ almost surely in $\mathbb{P}(\cdot|\bw^\infty=w, \by_{1:t-1}^\infty = y_{1:t-1})$, which is satisfied at time $t=1$ because of our previous reasoning. 

    We can note that:
    \begin{equation}
        \begin{split}
            \bpi_{t|t-1}^\infty(w,\theta)^{(i)} =0
            &\iff 
            \sum_{j } \bpi_{t-1}^\infty(w, \theta)^{(j)} K_{\bar{\boeta}^\infty_{t-1}(w,\theta)}(w_n,\theta)^{(j,i)}=0\\
            &\iff
            \forall j \in [M] \quad \bpi_{t-1}^\infty(w, \theta)^{(j)} K_{\bar{\boeta}^\infty_{t-1}(w,\theta)}(w_n,\theta)^{(j,i)}=0,
        \end{split}
    \end{equation}
    hence for all $j\in [M]$ we either have:
    \begin{enumerate}
        \item $\mathbb\bpi_{t-1}^\infty(w, \theta)^{(j)}=0$ which implies $(\bx_{t-1}^\infty)^{(j)}=0$ almost surely in $\mathbb{P}(\cdot|\bw^\infty=w, \by_{1:t-1}^\infty = y_{1:t-1})$ by our inductive assumption, or
        \item $K_{\bar{\boeta}^\infty_{t-1}(w,\theta)}(w,\theta)^{(j,i)}=0$, which implies that there exists $\eta$ such that $K_{\eta}(w,\theta)^{(j,i)}=0$ and so, by Assumption \ref{ass:HMM_support}, $K_{\eta}(w,\theta)^{(j,i)}=0$ for all $\eta \in [0,C]$ and $\theta \in \Theta$ meaning $K_{\bar{\boeta}^\infty_{t-1}(w,\theta^\star)}(w,\theta^\star)^{(j,i)}=0$.
    \end{enumerate}
    from which we conclude:
    \begin{equation}
        \begin{split}
            &\mathbb{P}((\bx_t^\infty)^{(i)}=0|\bw^\infty=w, \by_{1:t-1}^\infty = y_{1:t-1}) \\
            &=
            1- \sum_{x_{t-1}} \mathbb{P}((\bx_t^\infty)^{(i)}=1|\bw^\infty=w, \bx_{t-1}^\infty = x_{t-1},\by_{1:t}^\infty = y_{1:t})\\
            &\qquad \qquad \qquad \qquad \cdot \mathbb{P}(\bx_{t-1}^\infty = x_{t-1}|\bw^\infty=w, \by_{1:t-1}^\infty = y_{1:t-1})\\
            &=
            1- \sum_{j} K_{\bar{\boeta}^\infty_{t-1}(w,\theta)}(w,\theta)^{(j,i)} \mathbb{P}((\bx_{t-1}^\infty)^{(j)} = 1|\bw^\infty=w, \by_{1:t-1}^\infty = y_{1:t-1})=1,
        \end{split}
    \end{equation}
    hence if $\bpi_{t|t-1}^\infty(w,\theta)^{(i)} =0$ then $(\by_t^\infty)^{(i)}=0$ almost surely in $\mathbb{P}(\cdot|\bw^\infty=w, \by_{1:t-1}^\infty = y_{1:t-1})$.

    Moving to:
    \begin{equation}
        \begin{split}
            \bmu_{t}^\infty(w,\theta)^{(i)}=0
            &\iff 
            \sum_{j } \bpi_{t|t-1}^\infty(w,\theta)^{(j)}  G(w,\theta)^{(j,i)}=0 \\
            &\iff
            \forall j \in [M] \quad \bpi_{t|t-1}^\infty(w,\theta)^{(j)}  G(w,\theta)^{(j,i)}=0,
        \end{split}
    \end{equation}
    we then have that for all $j$ either:
    \begin{enumerate}
        \item $\mathbb\bpi_{t|t-1}^\infty(w,\theta)^{(j)}=0$ which implies $(\bx_{t}^\infty)^{(j)}=0$ almost surely in $\mathbb{P}(\cdot|\bw^\infty=w, \by_{1:t-1}^\infty = y_{1:t-1})$ because of what we have proven above, or
        \item $ G(w,\theta)^{(j,i)}=0$ which implies $ G(w,\theta^\star)^{(j,i)}=0$ because of Assumption \ref{ass:HMM_support}, 
    \end{enumerate}
    meaning that: 
    \begin{equation}
        \begin{split}
           &\mathbb{P}\left ((\by_{t}^\infty)^{(i)}=0 |\bw^\infty=w, \by_{1:t-1}^\infty = y_{1:t-1}\right ) 
           = 
           1 - \mathbb{P}\left ((\by_{t}^\infty)^{(i)}=1 |\bw^\infty=w, \by_{1:t-1}^\infty = y_{1:t-1}\right )\\
           &=
           1- \sum_{x_{t}} \mathbb{P}\left (\bx_{t}^\infty= x_{t} |\bw^\infty=w, \by_{1:t-1}^\infty = y_{1:t-1} \right)\\
           &\qquad \qquad \qquad \qquad \cdot \mathbb{P}\left ((\by_{t}^\infty)^{(i)}=1|\bw^\infty=w, \bx_{t}^\infty = x_{t}, \by_{1:t-1}^\infty = y_{1:t-1} \right)\\
           &= 
           1 - \sum_{j} \mathbb{P}\left ((\bx_{t}^\infty)^{(j)} = 1 |\bw^\infty=w, \by_{1:t-1}^\infty = y_{1:t-1} \right)  G(w,\theta^\star)^{(j,i)} = 1,
        \end{split}
    \end{equation}
    hence if $\bmu_{t}^\infty(w,\theta)^{(i)} =0$ then $(\by_t^\infty)^{(i)}=0$ almost surely in $\mathbb{P}(\cdot|\bw^\infty=w, \by_{1:t-1}^\infty = y_{1:t-1})$.

    Consider now:
    \begin{equation}
        \begin{split}
            &\bpi_{t}^\infty(w,\theta)^{(i)}= 0
            \iff
            \bpi_{t|t-1}^\infty(w,\theta)^{(i)} \frac{\sum_{j}  G(w,\theta)^{(i,j)} (\by_{t}^\infty)^{(j)} }{\sum_{j} (\by_{t}^\infty)^{(j)}\bmu_{t}^\infty(w,\theta)^{(j)}}=0\\
            &\iff
            \bpi_{t|t-1}^\infty(w,\theta)^{(i)} \sum_{j}  G(w,\theta)^{(i,j)}(\by_{t}^\infty)^{(j)}=0 \text{ and } \sum_{j} (\by_{t}^\infty)^{(j)}\bmu_{t}^\infty(w,\theta)^{(j)}\neq 0.
        \end{split}
    \end{equation}
    Similarly to Theorem \ref{thm:CAL_well_definess} we have $\sum_{j} (\by_{t}^\infty)^{(j)}\bmu_{t}^\infty(w,\theta)^{(j)}\neq 0$ being an almost sure event under $\mathbb{P}(\cdot|\bw^\infty=w, \by_{1:t-1}^\infty = y_{1:t-1})$, hence:
    \begin{equation}
    \begin{split}
        &\mathbb{P}\left(\bpi_{t|t-1}^\infty(w,\theta)^{(i)} \sum_{j}  G(w,\theta)^{(i,j)}(\by_{t}^\infty)^{(j)}=0\right.\\
        &\left. \qquad \qquad \text{ and } \sum_{j} (\by_{t}^\infty)^{(j)}\bmu_{t}^\infty(w,\theta)^{(j)}\neq 0|\bw^\infty=w, \by_{1:t-1}^\infty = y_{1:t-1}\right)\\
        &=
        \mathbb{P}\left (\bpi_{t|t-1}^\infty(w,\theta)^{(i)} \sum_{j}  G(w,\theta)^{(i,j)}(\by_{t}^\infty)^{(j)}=0|\bw^\infty=w, \by_{1:t-1}^\infty = y_{1:t-1} \right ).
    \end{split}
    \end{equation}
    We then just need to work on the event $\bpi_{t|t-1}^\infty(w,\theta)^{(i)} \sum_{j}  G(w,\theta)^{(i,j)}(\by_{t}^\infty)^{(j)}=0$.

    We then have either:
    \begin{itemize}
        \item $\bpi_{t|t-1}^\infty(w,\theta)^{(i)}=0$, implying $(\bx_{t}^\infty)^{(i)}=0$ almost surely in $\mathbb{P}(\cdot|\bw^\infty=w, \by_{1:t-1}^\infty = y_{1:t-1})$ because of what we have proven above, or
        \item $\sum_{j}  G(w,\theta)^{(i,j)}(\by_{t}^\infty)^{(j)}=0$, which tells us that there exists $k \in [M]$ such that\\ $\mathbb{P}\left ((\by_{t}^\infty)^{(k)}=1 |\bw^\infty=w, \by_{1:t-1}^\infty = y_{1:t-1} \right)>0$ and $ G(w,\theta)^{(i,k)}=0$, note that under $\mathbb{P}(\cdot |\bw^\infty=w, \by_{1:t}^\infty = y_{1:t})$ we know $k$ as we are conditioning on $\by_{t}^\infty$;
    \end{itemize}
    as:
    \begin{equation}
        \begin{split}
            &\mathbb{P}((\bx_{t}^\infty)^{(i)}=1|\bw^\infty=w, \by_{1:t}^\infty = y_{1:t})\\
            &\propto
            \mathbb{P}((\bx_{t}^\infty)^{(i)}=1, (\by_{t}^\infty)^{(k)}=1|\bw^\infty=w, \by_{1:t-1}^\infty = y_{1:t-1})\\
            &=
            G(w,\theta)^{(i,k)} \mathbb{P}((\bx_{t}^\infty)^{(i)}=1|\bw^\infty=w, \by_{1:t-1}^\infty = y_{1:t-1})=0,
        \end{split}
    \end{equation}
    we can then conclude $(\bx_{t}^\infty)^{(i)}=0$ almost surely in $\mathbb{P}(\cdot|\bw^\infty=w, \by_{1:t}^\infty = y_{1:t})$.

    With this final result we have shown that our inductive assumption is true at time $t$, hence by induction we can conclude that for any $t \geq 1$:
    \begin{itemize}
        \item if there exist $\theta \in \Theta, i \in [M]$ such that $\bpi_{t|t-1}^\infty(w,\theta)^{(i)}=0$, then $(\bx_{t}^\infty)^{(i)}=0$ almost surely in $\mathbb{P}(\cdot|\bw^\infty=w, \by_{1:t-1}^\infty = y_{1:t-1})$;
        \item if there exist $\theta \in \Theta, i \in [M]$ such that $\bmu_{t}^\infty(w,\theta)^{(i)}=0$, then $(\by_{t}^\infty)^{(i)}=0$ almost surely in $\mathbb{P}(\cdot|\bw^\infty=w, \by_{1:t-1}^\infty = y_{1:t-1})$;
        \item if there exist $\theta \in \Theta, i \in [M]$ such that $\bpi_{t}^\infty(w,\theta)^{(i)}=0$, then $(\bx_{t}^\infty)^{(i)}=0$ almost surely in $\mathbb{P}(\cdot|\bw^\infty=w, \by_{1:t}^\infty = y_{1:t})$.
    \end{itemize}

    To prove that the random asymptotic CAL is a well-defined algorithm in $\mathbb{P}$ we need to prove that:
    \begin{equation}
        \mathbb{P}\left ( \forall \theta \in \Theta, t\geq1 \quad \sum_i \bmu_{t}^\infty(\bw^\infty,\theta)^{(i)}(\by_{t}^\infty)^{(i)}\neq0 \right)=1,
    \end{equation}
    which can be proven by observing that:
    \begin{equation}
        \begin{split}
            &\mathbb{P}\left ( \forall \theta \in \Theta, t\geq1,  \quad \sum_i \bmu_{t}^\infty(\bw^\infty,\theta)^{(i)}(\by_{t}^\infty)^{(i)} \neq 0 \right)\\
            &=
            \mathbb{P}\left (  \forall \theta \in \Theta, t\geq1, \quad \sum_i \bmu_{t}^\infty(\bw^\infty,\theta)^{(i)}(\by_{t}^\infty)^{(i)}\neq0 \right)\\
            &=
            \int \mathbb{P}\left (  \forall \theta \in \Theta, t\geq1 \quad \sum_i \bmu_{t}^\infty(w,\theta)^{(i)}(\by_{t}^\infty)^{(i)}\neq0 |\bw^\infty=w\right) \Gamma(dw)\\
            &=
            \int \sum_{y_{1:t-1}} \mathbb{P}\left (  \forall \theta \in \Theta, t\geq1 \quad \sum_i \bmu_{t}^\infty(w,\theta)^{(i)}(\by_{t}^\infty)^{(i)}\neq0 |\bw^\infty=w, \by_{1:t-1}^\infty = y_{1:t-1} \right) \\
            & \quad \quad \cdot \mathbb{P} (\by^\infty_{1:t-1} = y_{1:t-1} |\bw^\infty=w) \Gamma(dw)\\
            &= 1.
        \end{split}
    \end{equation}
\end{proof}

\begin{proposition} \label{prop:infCAL_as_bounded}
    Under assumptions \ref{ass:compactness_continuity},\ref{ass:HMM_support},\ref{ass:kernel_continuity}, there exists ${m}_{t}>0$ as in Proposition \ref{prop:CAL_as_bounded} for any $t\geq 1$such that:
    $$
    \mathbb{P} \left ( \bmu_{t}^\infty(\bw^\infty,\theta)^{(i)}
        \geq 
        {m}_{t}\quad \forall i \in \supp{\bmu_{t}^\infty(\bw^\infty,\theta)} \right )=1 \quad \forall \theta \in \Theta.
    $$
\end{proposition}

\begin{proof}
    Consider the following inductive hypothesis. There exists $\bar{m}_{t-1} > 0$ such that:
    $$
    \mathbb{P} \left ( \bpi_{t-1}^\infty(\bw^\infty,\theta)^{(i)}
        \geq 
        \bar{m}_{t-1} \quad \forall i \in \supp{\bpi_{t-1}^\infty(\bw^\infty,\theta)} \right )=1 \quad \forall n \in [N], \theta \in \Theta.
    $$

    As for Proposition \ref{prop:CAL_as_bounded}, from Assumption \ref{ass:compactness_continuity} we have that $p_0(w,\theta)$ is continuous in $w,\theta$ and both $\mathbb{W}$ and $\Theta$ are compact, we then get from Weierstrass theorem that there exists a minimum $m_0$ such that for any realization of $\bw^\infty$ and for any $i \in \supp{p_0(\bw^\infty,\theta)^{(i)}}$:
    $$
    \bpi_{0}^\infty(\bw^\infty,\theta)^{(i)}= p_0(\bw^\infty,\theta)^{(i)} \geq \min_{w \in \mathbb{W},\theta \in \Theta} \min_{j \in \supp{p_0(w,\theta)^{(j)}}} p_0(w,\theta)^{(j)} \eqqcolon m_0,
    $$
    with $m_0>0$ as we are considering a minimum over $j \in \supp{p_0(w,\theta)^{(j)}}$ which excludes all the zeros. As $m_0$ does not depend on $\bw^\infty$ we conclude that the inductive hypothesis holds for $t-1=0$.
    
    We now move to $\bpi_{t|t-1}^\infty(\bw^\infty,\theta)$, 
     and let $i \in \supp{\bpi_{t|t-1}^\infty(\bw^\infty,\theta)}$ then:
    \begin{equation}
    \begin{split}
        \bpi_{t|t-1}^\infty(\bw^\infty,\theta)^{(i)}
        &\geq 
        \bar{m}_{t-1} \sum_{j \in \supp{\bpi_{t-1}^\infty(\bw^\infty,\theta)}} K_{\bar{\boeta}^\infty_{t-1}(\bw^\infty,\theta)}(\bw^\infty,\theta)^{(j,i)},
    \end{split}
    \end{equation}
    where the inequality holds with probability $1$ under the inductive hypothesis. Several of the remaining inequalities in the proof also hold with probability $1$, although to avoid repetition we do not state this explicitly. Similarly to Proposition \ref{prop:CAL_as_bounded}: 
    $$
    \sum_{j \in \supp{\bpi_{t-1}^\infty(\bw^\infty,\theta)}} K_{\bar{\boeta}^\infty_{t-1}(\bw^\infty,\theta)}(\bw^\infty,\theta)^{(j,i)} \geq \min_{j \in \supp{K_{\bar{\boeta}^\infty_{t-1}(\bw^\infty,\theta)}(\bw^\infty,\theta)^{(\cdot,i)}} } K_{\bar{\boeta}^\infty_{t-1}(\bw^\infty,\theta)}(\bw^\infty,\theta)^{(j,i)}.
    $$
    Because of Assumption \ref{ass:HMM_support} we have:
    $$
    K_{\bar{\boeta}^\infty_{t-1}(\bw^\infty,\theta)}(\bw^\infty,\theta)^{(j,i)}=0
    \Longleftrightarrow
    K_{\eta}(\bw^\infty,\theta)^{(j,i)}=0 \quad \forall \eta \in [0,C],
    $$ 
    meaning that
    \begin{equation}
    \begin{split}
    \min_{j \in \supp{K_{\bar{\boeta}^\infty_{t-1}(\bw^\infty,\theta)}(\bw^\infty,\theta)^{(\cdot,i)}} } & K_{\bar{\boeta}^\infty_{t-1}(\bw^\infty,\theta)}(\bw^\infty,\theta)^{(j,i)} 
    \geq 
    \min_{\eta \in [0,C]} \min_{j \in \supp{K_{\eta}(\bw^\infty,\theta)^{(\cdot,i)}} } K_{\eta}(\bw^\infty,\theta)^{(j,i)}.
    \end{split}
    \end{equation}
    We can then conclude:
    \begin{equation}
    \begin{split}
        \bpi_{t|t-1}^\infty(\bw^\infty,\theta)^{(i)}
        &\geq 
        \bar{m}_{t-1} \min_{\theta \in \Theta, w \in \mathbb{W},\eta \in [0,C]} \min_{j \in \supp{K_{\eta}(w,\theta)^{(\cdot,i)}} } K_{\eta}(w,\theta)^{(j,i)}\\
        &\geq 
        \bar{m}_{t-1} \min_{\theta \in \Theta, w \in \mathbb{W},\eta \in [0,C]} \min_{(i,j) \in \supp{K_{\eta}(w,\theta)} } K_{\eta}(w,\theta)^{(j,i)}.
    \end{split}
    \end{equation}
    As $K_{\eta}(w,\theta)$ is continuous in $\eta$ because of Assumption \ref{ass:kernel_continuity} and also in $w,\theta$ because of Assumption \ref{ass:compactness_continuity}, and both $[0,C]$ is compact by definition and $\mathbb{W},\Theta$ are compact because of Assumption \ref{ass:compactness_continuity}, we can conclude by Weirstrass theorem that there exists a minimum $m_{K}$, hence: 
    \begin{equation}
    \begin{split}
        \bpi_{t|t-1}^\infty(\bw^\infty,\theta)^{(i)}
        &\geq 
        \bar{m}_{t-1} m_{K}>0,
    \end{split}
    \end{equation}
    where the strictly greater than zero follows from considering the minimum on the support of the transition matrix. As there is no dependence on $\bw^\infty,\by_{1:t-1}^\infty$ we conclude that there exist $\bar{m}_{t-1}, m_{K}>0$ such that:
    $$
    \mathbb{P} \left ( \bpi_{t|t-1}^\infty(\bw^\infty,\theta)^{(i)}
        \geq 
        \bar{m}_{t-1} m_{K} \quad \forall i \in \supp{\bpi_{t|t-1}^\infty(\bw^\infty,\theta)} \right )=1 \quad \forall \theta \in \Theta.
    $$
    
    Following again the same steps as in Proposition \ref{prop:CAL_as_bounded}, we have that for an arbitrary realization of $\bw^\infty,\by_{1:t-1}^\infty$, and $i \in \supp{\bmu_{t}^\infty(\bw^\infty,\theta)}$:
    \begin{equation}
    \begin{split}
        \bmu_{t}^\infty(\bw^\infty,\theta)^{(i)}
        &\geq 
        \bar{m}_{t-1} m_{K} \sum_{j \in \supp{\bpi_{t|t-1}^\infty(\bw^\infty,\theta)}} G(\bw^\infty,\theta)^{(j,i)},
    \end{split}
    \end{equation}
    where the inequality follows from what we have proven above. Moreover:
    \begin{equation}
    \begin{split}
        \bmu_{t}^\infty(\bw^\infty,\theta)^{(i)}
        &\geq 
        \bar{m}_{t-1} m_{K} \min_{w \in \mathbb{W},\theta \in \Theta } \min_{(i,j) \in \supp{G(w,\theta)}} G(w,\theta)^{(j,i)},
    \end{split}
    \end{equation}
    and as $G(w,\theta)$ is continuous in $w,\theta$ because of Assumption \ref{ass:compactness_continuity} and $\mathbb{W}, \Theta$ are compact because of Assumption \ref{ass:compactness_continuity} we can conclude by Weirstrass theorem that there exists a minimum $m_{G}$:
    \begin{equation}
    \begin{split}
        \bmu_{t}^\infty(\bw^\infty,\theta)^{(i)}
        &\geq 
        \bar{m}_{t-1} m_{K} m_{G}>0,
    \end{split}
    \end{equation}
    where the strictly greater than zero follows from considering a minimum on the support of the emission matrix. As there is no dependence on $\bw^\infty,\by_{1:t-1}^\infty$ we can then conclude that there exist $\bar{m}_{t-1}, m_{K}, m_{G}>0$ such that:
    $$
    \mathbb{P} \left ( \bmu_{t}^\infty(\bw^\infty,\theta)^{(i)}
        \geq 
        \bar{m}_{t-1} m_{K} m_{G} \quad \forall i \in \supp{\bmu_{t}^\infty(\bw^\infty,\theta)} \right )=1 \quad \forall \theta \in \Theta.
    $$

    Finally, consider a realization of $\bw^\infty,\by_{1:t}^\infty$ and $i \in \supp{\bpi_{t}^\infty(\bw^\infty,\theta)}$ then:
    \begin{equation}
        \begin{split}
            \bpi_{t}^\infty(\bw^\infty,\theta)^{(i)}
            &\geq 
            \bar{m}_{t-1} m_{K} \sum_{j}  G(\bw^\infty,\theta)^{(i,j)} (\by_{t}^\infty)^{(j)},
        \end{split}
    \end{equation}
    where the inequality follows from what we have proven above and we know that by definition and by Proposition \ref{prop:rand_aympt_CAL_well_definess}:
    $$
    \sum_{j}  G(\bw^\infty,\theta)^{(i,j)} (\by_{t}^\infty)^{(j)} \leq 1 \text{ and } \sum_{j}  G(\bw^\infty,\theta)^{(i,j)} (\by_{t}^\infty)^{(j)}\neq 0,
    $$
    $\mathbb{P}$-almost surely. Following the same steps as in Proposition \ref{prop:CAL_as_bounded} we can conclude that there exist $\bar{m}_{t-1}, m_{K}, m_{G}>0$ such that:
    $$
    \mathbb{P} \left ( \bpi_{t}^\infty(\bw^\infty,\theta)^{(i)}
        \geq 
        \bar{m}_{t-1} m_{K} m_{G} \quad \forall i \in \supp{\bpi_{t}^\infty(\bw^\infty,\theta)} \right )=1 \quad \forall \theta \in \Theta.
    $$
    
    We can then set $\bar{m}_t \coloneqq \bar{m}_{t-1} m_{K} m_{G}$ and conclude that there exists $\bar{m}_t >0$ such that:
    $$
    \mathbb{P} \left ( \bpi_{t}^\infty(\bw^\infty,\theta)^{(i)}
        \geq 
        \bar{m}_t \quad \forall i \in \supp{\bpi_{t}^\infty(\bw^\infty,\theta)} \right )=1 \quad \forall \theta \in \Theta,
    $$
    which closes the induction, meaning that the above statement holds for an arbitrary $t$. As a consequence, we also have that for any $t\geq 1$ there exists ${m}_{t}>0$ such that:
    $$
    \mathbb{P} \left ( \bmu_{t}^\infty(\bw^\infty,\theta)^{(i)}
        \geq 
        {m}_{t}\quad \forall i \in \supp{\mu_{t}^\infty(\bw^\infty,\theta)} \right )=1 \quad \forall \theta \in \Theta,
    $$
    with ${m}_{t}\coloneqq \bar{m}_{t-1} m_{K} m_{G}$, which concludes the proof.
\end{proof}

\paragraph{Population saturated process.} The population saturated process consists of $\mathbb{O}_M$-valued disease states $(\bx^\infty_{n,t})_{t\geq 0}$ and $\mathbb{O}_{M+1}$-valued observations $(\by^\infty_{n,t})_{t\geq 1}$, for each  $n\in\mathbb{N}$.  

Given $\bw_1,\bw_2,\ldots,$ (which are the same covariate vectors as in the data-generating process), the individuals and observations $(\bx^\infty_{n,t})_{t\geq 0}$, $(\by^\infty_{n,t})_{t\geq 1}$, are defined to be conditionally independent across $n
$, and distributed as follows:
\begin{equation}\label{eq:limiting_process_n}
    \begin{split}
    &\bx_{n,0}^\infty|\bw_n \sim \mbox{Cat}\left( \cdot|p_0(\bw_n,\theta^\star)\right),\\
    &\bx_{n,t}^\infty|\bx_{n,t-1}^\infty,\bw_n \sim \mbox{Cat}\left( \cdot|\left [ (\bx_{t-1}^\infty)^\top K_{ \boeta_{t-1}^{\infty} (\bw_n,\theta^\star ) }(\bw_n,\theta^\star) \right ]^\top\right),\\
    &\by_{n,t}^\infty|\bx_{n,t}^\infty,\bw_n \sim \mbox{Cat}\left( \cdot|\left [ (\bx_{t}^\infty)^\top G(\bw_n,\theta^\star) \right ]^\top\right).
    \end{split}
\end{equation}
where $\boeta_{t-1}^{\infty}$ is as in Corollary \ref{corol:eta_bound}.

\begin{proposition} \label{thm:LLN_f_y_yinf}
    Under Assumption \ref{ass:w_iid},\ref{ass:eta_structure},\ref{ass:kernel_continuity}, for any $t\geq 1$ there exists constants $e_t >0$ and $B_t>0$ such that for any function $f_t$ with $f_t(\bw_n,\by_{n,1:t}^N) \in [-B_t,B_t]$ and $f_t(\bw_n,\by_{n,1:t}^\infty) \in [-B_t,B_t]$ almost surely we have:
    \begin{equation}
        \begin{split}
            \normiii[\Bigg]{\frac{1}{N} \sum_{n \in [N]} f_t(\bw_n,\by_{n,1:t}^N) - f_t(\bw_n,\by_{n,1:t}^\infty)}_4 \leq 
            2 B_t \sqrt[4]{6} N^{-\frac{1}{2}} e_t.
        \end{split}
    \end{equation}
\end{proposition} 

\begin{proof}
    Remark that both $\bx_{n,t}^N$ and $\bx_{n,t}^\infty$ are random variables that take values on $\mathbb{O}_M$, and similarly $\by_{n,t}^N$ and $\by_{n,t}^\infty$ take values on $\mathbb{O}_{M+1}$. We often write expectations over $\bx_{n,t}^N$ or $\bx_{n,t}^\infty$ (resp. $\by_{n,t}^N$ or $\by_{n,t}^\infty$) as summations over ``$x$'' (resp. ``$y$''), implicitly it should be understood that we are marginalizing over $x\in \mathbb{O}_M$ (resp. $y \in \mathbb{O}_{M+1}$). 
    
    Consider the following inductive hypothesis. There exists $e_{t-1} >0$ such that for any function $f_{t-1}$ with $f_{t-1}(\bw_n,\bx_{n,1:{t-1}}^N) \in [-B_{t-1},B_{t-1}]$ and $f_{t-1}(\bw_n,\bx_{n,1:{t-1}}^\infty) \in [-B_{t-1},B_{t-1}]$ we have:
    \begin{equation}
        \begin{split}
            \normiii[\Bigg]{\frac{1}{N} \sum_{n \in [N]} f_{t-1}(\bw_n,\bx_{n,{0:t-1}}^N) - f_{t-1}(\bw_n,\bx_{n,{0:t-1}}^\infty)}_4 \leq 2 B_{t-1} \sqrt[4]{6} N^{-\frac{1}{2}} \bar{e}_{t-1}.
        \end{split}
    \end{equation}
    We start by proving it at $t-1=0$, and specifically we want:
    \begin{equation}
        \begin{split}
            \normiii[\Bigg]{\frac{1}{N} \sum_{n \in [N]} f_0(\bw_n,\bx_{n,0}^N) - f_0(\bw_n,\bx_{n,0}^\infty)}_4 \leq 2 B_0 \sqrt[4]{6} N^{-\frac{1}{2}} \bar{e}_{0}.
        \end{split}
    \end{equation}
    Observe that $f_0(\bw_n,\bx_{n,0}^N)-f_0(\bw_n,\bx_{n,0}^\infty)$ are conditionally independent given $\bW^N$ and mean zero because:
    \begin{equation}
        \begin{split}
            \mathbb{E} \left [ f_0(\bw_n,\bx_{n,0}^N) | \bW^N \right ] = \sum_{x_0} f_0(\bw_n,x_0) p_0(\bw_n,\theta^\star) =\mathbb{E} \left [ f_0(\bw_n,\bx_{n,0}^\infty) | \bW^N \right ].
        \end{split}
    \end{equation}
    Moreover, they are almost surely bounded because we are assuming $f_0(\bw_n,\bx_{n,0}^N) \in [-B_0,B_0]$ and $f_0(\bw_n,\bx_{n,0}^\infty) \in [-B_0,B_0]$ almost surely, hence:
    \begin{equation}
        \begin{split}
            \norm{f_0(\bw_n,\bx_{n,0}^N) - f_0(\bw_n,\bx_{n,0}^\infty)}_\infty \leq 2 B_0.
        \end{split}
    \end{equation}
    meaning that by Lemma \ref{lemma:mean_0_bound}:
    \begin{equation}
        \begin{split}
            \normiii[\Bigg]{\frac{1}{N} \sum_{n \in [N]} f_0(\bw_n,\bx_{n,0}^N) - f_0(\bw_n,\bx_{n,0}^\infty)}_4 \leq 2 B_0 \sqrt[4]{6} N^{-\frac{1}{2}},
        \end{split}
    \end{equation}
    so our inductive hypothesis is true at $t-1=0$ for $\bar{e}_0=1$.

    Consider now a general time $t$, we can rewrite:
    \begin{align}
            &f_{t}(\bw_n,\bx_{n,{0:t}}^N) - f_{t}(\bw_n,\bx_{n,{0:t}}^\infty)\\
            &=
            f_{t}(\bw_n,\bx_{n,{0:t}}^N) - \sum_{x} f_{t}(\bw_n,(\bx_{n,{0:t-1}}^N,x)) (\bx_{n,t-1}^N)^\top K_{\boeta_{t-1}^N(\bw_n,\theta^\star)}(\bw_n,\theta^\star) x \\
            &\quad
            + \sum_{x} f_{t}(\bw_n,(\bx_{n,{0:t-1}}^N,x)) (\bx_{n,t-1}^N)^\top \left [ K_{\boeta_{t-1}^N(\bw_n,\theta^\star)}(\bw_n,\theta^\star) - K_{\boeta_{t-1}^\infty(\bw_n,\theta^\star)}(\bw_n,\theta^\star) \right ] x\\
            &\quad
            + \sum_{x} f_{t}(\bw_n,(\bx_{n,{0:t-1}}^N,x)) (\bx_{n,t-1}^N)^\top K_{\boeta_{t-1}^\infty(\bw_n,\theta^\star)}(\bw_n,\theta^\star) x\\
            &\quad\quad
            - \sum_{x} f_{t}(\bw_n,(\bx_{n,{0:t-1}}^\infty,x)) (\bx_{n,t-1}^\infty)^\top K_{\boeta_{t-1}^\infty(\bw_n,\theta^\star)}(\bw_n,\theta^\star) \\
            &\quad
            +\sum_{x} f_{t}(\bw_n,(\bx_{n,{0:t-1}}^\infty,x)) (\bx_{n,t-1}^\infty)^\top K_{\boeta_{t-1}^\infty(\bw_n,\theta^\star)}(\bw_n,\theta^\star) x - f_{t}(\bw_n,\bx_{n,{t}}^\infty), 
    \end{align}
    meaning that by Minkowski inequality:
    \begin{align}
            &\normiii[\Bigg]{\frac{1}{N} \sum_{n \in [N]} f_{t}(\bw_n,\bx_{n,{0:t}}^N) - f_{t}(\bw_n,\bx_{n,{0:t}}^\infty)}_4\\
            &=
            \normiii[\Bigg]{\frac{1}{N} \sum_{n \in [N]} f_{t}(\bw_n,\bx_{n,{0:t}}^N) - \sum_{x} f_{t}(\bw_n,(\bx_{n,{0:t-1}}^N,x)) (\bx_{n,t-1}^N)^\top K_{\boeta_{t-1}^N(\bw_n,\theta^\star)}(\bw_n,\theta^\star) x}_4 \label{eq:asympt_A_generalized}\\
            &\quad
            + 
            \left \vvvert\frac{1}{N} \sum_{n \in [N]} \sum_{x} f_{t}(\bw_n,(\bx_{n,{0:t-1}}^N,x)) (\bx_{n,t-1}^N)^\top \left [ K_{\boeta_{t-1}^N(\bw_n,\theta^\star)}(\bw_n,\theta^\star) \right.\right.\\
            &\left.\left. \qquad\qquad \qquad\qquad \qquad\qquad \qquad\qquad \qquad\qquad
            - K_{\boeta_{t-1}^\infty(\bw_n,\theta^\star)}(\bw_n,\theta^\star) \right ] x\right \vvvert_4 \label{eq:asympt_B_generalized}\\
            &\quad
            + 
            \left \vvvert \frac{1}{N} \sum_{n \in [N]}  \sum_{x} f_{t}(\bw_n,(\bx_{n,{0:t-1}}^N,x)) (\bx_{n,t-1}^N)^\top K_{\boeta_{t-1}^\infty(\bw_n,\theta^\star)}(\bw_n,\theta^\star) x \right.\\
            &\left. \qquad\qquad\qquad\qquad
            - \sum_{x} f_{t}(\bw_n,(\bx_{n,{0:t-1}}^\infty,x)) (\bx_{n,t-1}^\infty)^\top K_{\boeta_{t-1}^\infty(\bw_n,\theta^\star)}(\bw_n,\theta^\star) x \right \vvvert_4 \label{eq:asympt_C_generalized}\\
            &\quad
            +
            \normiii[\Bigg]{\frac{1}{N} \sum_{n \in [N]} \sum_{x} f_{t}(\bw_n,(\bx_{n,{0:t-1}}^\infty,x)) (\bx_{n,t-1}^\infty)^\top K_{\boeta_{t-1}^\infty(\bw_n,\theta^\star)}(\bw_n,\theta^\star) x - f_{t}(\bw_n,\bx_{n,{t}}^\infty)}_4. \label{eq:asympt_D_generalized}
    \end{align}

    Starting from \eqref{eq:asympt_A_generalized} we can notice that:
    $$
    \mathbb{E}  \left [ f_{t}(\bw_n,\bx_{n,{0:t}}^N) | \bW^N, \bx_{0:t-1}^N \right ] = \sum_{x} f_{t}(\bw_n,(\bx_{n,{0:t-1}}^N,x)) (\bx_{n,t-1}^N)^\top K_{\boeta_{t-1}^N(\bw_n,\theta^\star)}(\bw_n,\theta^\star) x,
    $$
    moreover the random variables $f_{t}(\bw_n,\bx_{n,{0:t}}^N) - \sum_{x} f_{t}(\bw_n,(\bx_{n,{0:t-1}}^N,x)) (\bx_{n,t-1}^N)^\top K_{\boeta_{t-1}^N(\bw_n,\theta^\star)}(\bw_n,\theta^\star) x$ are conditionally independent across $n$ given $\bW^N, \bX_{t-1}^N$ and almost surely bounded by $2B_t$, we can then conclude by Lemma \ref{lemma:mean_0_bound}:
    \begin{equation}
        \normiii[\Bigg]{\frac{1}{N} \sum_{n \in [N]} f_{t}(\bw_n,\bx_{n,{0:t}}^N) - \sum_{x} f_{t}(\bw_n,(\bx_{n,{0:t-1}}^N,x)) (\bx_{n,t-1}^N)^\top K_{\boeta_{t-1}^N(\bw_n,\theta^\star)}(\bw_n,\theta^\star) x}_4
        \leq 2 B_t \sqrt[4]{6} N^{-\frac{1}{2}}.
    \end{equation}
    Moving to \eqref{eq:asympt_B_generalized} we note that by Minkowski inequality:
    \begin{equation}
        \begin{split}
            &\left \vvvert\frac{1}{N} \sum_{n \in [N]} \sum_{x} f_{t}(\bw_n,(\bx_{n,{0:t-1}}^N,x)) (\bx_{n,t-1}^N)^\top \left [ K_{\boeta_{t-1}^N(\bw_n,\theta^\star)}(\bw_n,\theta^\star) \right.\right.\\
            &\left.\left. \qquad\qquad \qquad\qquad \qquad\qquad \qquad\qquad \qquad\qquad
            - K_{\boeta_{t-1}^\infty(\bw_n,\theta^\star)}(\bw_n,\theta^\star) \right ] x\right \vvvert_4 \\
            &\leq
            \frac{1}{N} \sum_{n \in [N]} \left \vvvert \sum_{x} f_{t}(\bw_n,(\bx_{n,{0:t-1}}^N,x)) (\bx_{n,t-1}^N)^\top \left [ K_{\boeta_{t-1}^N(\bw_n,\theta^\star)}(\bw_n,\theta^\star) \right.\right.\\
            &\left.\left. \qquad\qquad \qquad\qquad \qquad\qquad \qquad\qquad \qquad\qquad
            - K_{\boeta_{t-1}^\infty(\bw_n,\theta^\star)}(\bw_n,\theta^\star) \right ] x\right \vvvert_4.
        \end{split}
    \end{equation}
    Remark that:
    \begin{equation}
        \begin{split}
            &\abs{\sum_{x} f_{t}(\bw_n,(\bx_{n,{0:t-1}}^N,x))(\bx_{n,t-1}^N)^\top \left [ K_{\boeta_{t-1}^N(\bw_n,\theta^\star)}(\bw_n,\theta^\star) - K_{\boeta_{t-1}^\infty(\bw_n,\theta^\star)}(\bw_n,\theta^\star) \right ] x}\\
            &\leq 
            B_t \sum_{x} \abs{(\bx_{n,t-1}^N)^\top \left [ K_{\boeta_{t-1}^N(\bw_n,\theta^\star)}(\bw_n,\theta^\star) - K_{\boeta_{t-1}^\infty(\bw_n,\theta^\star)}(\bw_n,\theta^\star) \right ] x}\\
            &\leq 
            B_t \norm{K_{\boeta_{t-1}^N(\bw_n,\theta^\star)}(\bw_n,\theta^\star) - K_{\boeta_{t-1}^\infty(\bw_n,\theta^\star)}(\bw_n,\theta^\star)}_\infty,
        \end{split}
    \end{equation}
    hence by Assumption \ref{ass:kernel_continuity} we have:
    \begin{equation}
        \begin{split}
            &\abs{\sum_{x} f_{t}(\bw_n,(\bx_{n,{0:t-1}}^N,x))(\bx_{n,t-1}^N)^\top \left [ K_{\boeta_{t-1}^N(\bw_n,\theta^\star)}(\bw_n,\theta^\star) - K_{\boeta_{t-1}^\infty(\bw_n,\theta^\star)}(\bw_n,\theta^\star) \right ] x}\\
            &\leq 
            B_t L\abs{\boeta_{t-1}^N(\bw_n,\theta^\star) - \boeta_{t-1}^\infty(\bw_n,\theta^\star)}.
        \end{split}
    \end{equation}
    We can then rewrite:
    \begin{equation}
        \begin{split}
            &\left \vvvert\frac{1}{N} \sum_{n \in [N]} \sum_{x} f_{t}(\bw_n,(\bx_{n,{0:t-1}}^N,x)) (\bx_{n,t-1}^N)^\top \left [ K_{\boeta_{t-1}^N(\bw_n,\theta^\star)}(\bw_n,\theta^\star) \right.\right.\\
            &\left.\left. \qquad\qquad \qquad\qquad \qquad\qquad \qquad\qquad \qquad\qquad
            - K_{\boeta_{t-1}^\infty(\bw_n,\theta^\star)}(\bw_n,\theta^\star) \right ] x\right \vvvert_4 \\
            &\leq
            \frac{B_t L}{N} \sum_{n \in [N]} \normiii[\Bigg]{\boeta_{t-1}^N(\bw_n,\theta^\star) - \boeta_{t-1}^\infty(\bw_n,\theta^\star)}_4,
        \end{split}
    \end{equation}
    and because of Corollary \ref{corol:eta_bound} and Assumption \ref{ass:kernel_continuity}, we can conclude:
    \begin{equation}
        \begin{split}
            &\left \vvvert\frac{1}{N} \sum_{n \in [N]} \sum_{x} f_{t}(\bw_n,(\bx_{n,{0:t-1}}^N,x)) (\bx_{n,t-1}^N)^\top \left [ K_{\boeta_{t-1}^N(\bw_n,\theta^\star)}(\bw_n,\theta^\star) \right.\right.\\
            &\left.\left. \qquad\qquad \qquad\qquad \qquad\qquad \qquad\qquad \qquad\qquad
            - K_{\boeta_{t-1}^\infty(\bw_n,\theta^\star)}(\bw_n,\theta^\star) \right ] x\right \vvvert_4 \\
            &\leq
            2 B_t L C \sqrt[4]{6} N^{-\frac{1}{2}} \alpha_{t-1}.
        \end{split}
    \end{equation}
    
    Consider \eqref{eq:asympt_C_generalized}, it is just enough to apply our inductive hypothesis with test function $f_{t-1}$ given by:
    $$
    f_{t-1}(\bw_n,x_{0:t-1}) = \sum_{x} f_{t}(\bw_n,(x_{0:t-1},x)) x_{t-1}^\top K_{\boeta_{t-1}^\infty(\bw_n,\theta^\star)}(\bw_n,\theta^\star) x,
    $$
    indeed whether we have $x_{0:t-1}=\bx_{n,0:t-1}^N$ or $x_{0:t-1}=\bx_{n,0:t-1}^\infty$ in both cases:
    \begin{equation}
        \begin{split}
        &\abs{\sum_{x} f_{t}(\bw_n,(x_{0:t-1},x)) x_{t-1}^\top K_{\boeta_{t-1}^\infty(\bw_n,\theta^\star)}(\bw_n,\theta^\star) x} \\
        &\leq 
        \sum_{x} \abs{f_{t}(\bw_n,(x_{0:t-1},x)) x_{t-1}^\top K_{\boeta_{t-1}^\infty(\bw_n,\theta^\star)}(\bw_n,\theta^\star) x}\\
        &\leq 
        B_t \sum_{x} \abs{x_{t-1}^\top K_{\boeta_{t-1}^\infty(\bw_n,\theta^\star)}(\bw_n,\theta^\star) x} =B_t.
        \end{split}
    \end{equation}
    Hence:
    \begin{equation}
        \begin{split}
            &\left \vvvert \frac{1}{N} \sum_{n \in [N]}  \sum_{x} f_{t}(\bw_n,(\bx_{n,{0:t-1}}^N,x)) (\bx_{n,t-1}^N)^\top K_{\boeta_{t-1}^\infty(\bw_n,\theta^\star)}(\bw_n,\theta^\star) x \right.\\
            &\left. \qquad\qquad\qquad\qquad
            - \sum_{x} f_{t}(\bw_n,(\bx_{n,{0:t-1}}^\infty,x)) (\bx_{n,t-1}^\infty)^\top K_{\boeta_{t-1}^\infty(\bw_n,\theta^\star)}(\bw_n,\theta^\star) x \right \vvvert_4\\
            &\leq
            2 B_t \sqrt[4]{6} N^{-\frac{1}{2}} \bar{e}_{t-1}.
        \end{split}
    \end{equation}
    The final term to work on is \eqref{eq:asympt_D_generalized}, but:
    \begin{equation}
        \begin{split}
        \mathbb{E} \left [ f_{t}(\bw_n,\bx_{n,{0:t}}^\infty) | \bW^N \right ] &=
        \mathbb{E} \left \{ \mathbb{E} \left [ f_{t}(\bw_n,\bx_{n,{0:t}}^\infty)| \bx_{n,0:t-1}^\infty, \bW^N \right ] | \bW^N \right \}\\
        &=
        \mathbb{E} \left [ \sum_{x} f_{t}(\bw_n,(\bx_{n,{0:t-1}}^\infty,x)) (\bx_{n,t-1}^\infty)^\top K_{\boeta_{t-1}^\infty(\bw_n,\theta^\star)}(\bw_n,\theta^\star) x | \bW^N \right ],
        \end{split}
    \end{equation}
    and the random variables:
    $$
    \sum_{x} f_{t}(\bw_n,(\bx_{n,{0:t-1}}^\infty,x)) (\bx_{n,t-1}^\infty)^\top K_{\boeta_{t-1}^\infty(\bw_n,\theta^\star)}(\bw_n,\theta^\star) x - f_{t}(\bw_n,\bx_{n,{0:t}}^\infty),
    $$
    are defined to be conditionally independent given $\bW^N$ and bounded by $2B_t$. We can then apply Lemma \ref{lemma:mean_0_bound} and conclude:
    \begin{equation}
        \begin{split}
            &\normiii[\Bigg]{\frac{1}{N} \sum_{n \in [N]} \sum_{x} f_{t}(\bw_n,(\bx_{n,{0:t-1}}^\infty,x)) (\bx_{n,t-1}^\infty)^\top K_{\boeta_{t-1}^\infty(\bw_n,\theta^\star)}(\bw_n,\theta^\star) x - f_{t}(\bw_n,\bx_{n,{0:t}}^\infty)}_4\\
            &\leq
            2 B_t \sqrt[4]{6} N^{-\frac{1}{2}}.
        \end{split}
    \end{equation}

    By putting everything together we can conclude:
    \begin{equation}
        \begin{split}
            &\normiii[\Bigg]{\frac{1}{N} \sum_{n \in [N]} f_{t}(\bw_n,\bx_{n,{0:t}}^N) - f_{t}(\bw_n,\bx_{n,{0:t}}^\infty)}_4\\
            &\leq 
            2 B_t \sqrt[4]{6} N^{-\frac{1}{2}}
            +
            2 B_t L C \sqrt[4]{6} N^{-\frac{1}{2}} \bar{e}_{t-1}
            +
            2 B_t \sqrt[4]{6} N^{-\frac{1}{2}} \bar{e}_{t-1}
            +
            2 B_t \sqrt[4]{6} N^{-\frac{1}{2}}\\
            &=
            2 B_t \sqrt[4]{6} N^{-\frac{1}{2}} \left [ 2 + ( L C + 1) \bar{e}_{t-1} \right ].
        \end{split}
    \end{equation}
    Hence we have proven that our inductive hypothesis is valid at time $t$ with the constant $\bar{e}_{t} \coloneqq \left [ 2 + ( L M C + 1) \bar{e}_{t-1} \right ]$, which tells us that for any $t\geq 1$ we have:
    \begin{equation}
        \begin{split}
            \normiii[\Bigg]{\frac{1}{N} \sum_{n \in [N]} f_{t}(\bw_n,\bx_{n,{0:t}}^N) - f_{t}(\bw_n,\bx_{n,{0:t}}^\infty)}_4 \leq 2 B_t \sqrt[4]{6} N^{-\frac{1}{2}} \bar{e}_{t},
        \end{split}
    \end{equation}
    which also implies:
    \begin{align}
            &\normiii[\Bigg]{\frac{1}{N} \sum_{n \in [N]} f_{t}(\bw_n,\bx_{n,{1:t}}^N) - f_{t}(\bw_n,\bx_{n,{1:t}}^\infty)}_4 \leq 2 B_t \sqrt[4]{6} N^{-\frac{1}{2}} \bar{e}_{t}, \label{eq:test_function_x_ind_path}
    \end{align}
    as we can rewrite $f_{t}(\bw_n,\bx_{n,{1:t}}^N)$ as $f_{t}(\bw_n,\bx_{n,{1:t}}^N)\mathbb{I}(\bx_{n,0} \in \mathbb{O}_M)$, where the indicator condition is always satisfied.

    We can now move to prove the statement of the proposition, with $f_t$ as therein:
    \begin{equation}
        \begin{split}
            \normiii[\Bigg]{\frac{1}{N} \sum_{n \in [N]} f_{t}(\bw_n,\by_{n,{1:t}}^N) - f_{t}(\bw_n,\by_{n,{1:t}}^\infty)}_4 \leq 2 B_t \sqrt[4]{6} N^{-\frac{1}{2}} e_{t}.
        \end{split}
    \end{equation}
    Observe that:
    \begin{align}
            &f_{t}(\bw_n,\by_{n,{1:t}}^N) - f_{t}(\bw_n,\by_{n,{1:t}}^\infty)\\
            &=
            f_{t}(\bw_n,\by_{n,{1:t}}^N) - \sum_{y_{1:t}} f_{t}(\bw_n,y_{1:t}) \prod_{s=1}^t (\bx_{n,s}^N)^\top G(\bw_n,\theta^\star) y_s\\
            &\quad
            + \sum_{y_{1:t}} f_{t}(\bw_n,y_{1:t}) \left [ \prod_{s=1}^t (\bx_{n,s}^N)^\top G(\bw_n,\theta^\star) y_s
            - \prod_{s=1}^t (\bx_{n,s}^\infty)^\top G(\bw_n,\theta^\star) y_s \right ]\\
            &\quad
            +\sum_{y_{1:t}} f_{t}(\bw_n,y_{1:t}) \prod_{s=1}^t (\bx_{n,s}^\infty)^\top G(\bw_n,\theta^\star) y_s - f_{t}(\bw_n,\by_{n,{t}}^\infty), 
    \end{align}
    then by Minkowski inequality:
    \begin{align}
            &\normiii[\Bigg]{\frac{1}{N} \sum_{n \in [N]} f_{t}(\bw_n,\by_{n,{1:t}}^N) - f_{t}(\bw_n,\by_{n,{1:t}}^\infty)}_4\\
            &=
            \normiii[\Bigg]{\frac{1}{N} \sum_{n \in [N]} f_{t}(\bw_n,\by_{n,{1:t}}^N) - \sum_{y_{1:t}} f_{t}(\bw_n,y_{1:t}) \prod_{s=1}^t (\bx_{n,s}^N)^\top G(\bw_n,\theta^\star) y_s }_4 \label{eq:asympt_yA_generalized}\\
            &
            + 
            \normiii[\Bigg]{\frac{1}{N} \sum_{n \in [N]} \sum_{y_{1:t}} f_{t}(\bw_n,y_{1:t}) \left [ \prod_{s=1}^t (\bx_{n,s}^N)^\top G(\bw_n,\theta^\star) y_s
            - \prod_{s=1}^t (\bx_{n,s}^\infty)^\top G(\bw_n,\theta^\star) y_s \right ] }_4 \label{eq:asympt_yB_generalized}\\
            &
            +
            \normiii[\Bigg]{\frac{1}{N} \sum_{n \in [N]} \sum_{y_{1:t}} f_{t}(\bw_n,y_{1:t}) \prod_{s=1}^t (\bx_{n,s}^\infty)^\top G(\bw_n,\theta^\star) y_s - f_{t}(\bw_n,\by_{n,{1:t}}^\infty)}_4. \label{eq:asympt_yC_generalized}
    \end{align}
    Starting from \eqref{eq:asympt_yA_generalized} we can notice that:
    $$
    \mathbb{E} \left [ f_{t}(\bw_n,\by_{n,{1:t}}^N) | \bW^N, \bX_{1:t}^N \right ] 
    = 
    \sum_{y_{1:t}} f_{t}(\bw_n,y_{1:t}) \prod_{s=1}^t (\bx_{n,s}^N)^\top G(\bw_n,\theta^\star) y_s,
    $$
    moreover the random variables:
    $$
    f_{t}(\bw_n,\by_{n,{1:t}}^N) - \sum_{y_{1:t}} f_{t}(\bw_n,y_{1:t}) \prod_{s=1}^t (\bx_{n,s}^N)^\top G(\bw_n,\theta^\star) y_s,
    $$
    are conditionally independent given $\bW^N, \bX_{1:t}^N$ and bounded by $2B_t$, hence we can apply Lemma \ref{lemma:mean_0_bound} and conclude:
    $$
    \normiii[\Bigg]{\frac{1}{N} \sum_{n \in [N]} f_{t}(\bw_n,\by_{n,{1:t}}^N) - \sum_{y_{1:t}} f_{t}(\bw_n,y_{1:t}) \prod_{s=1}^t (\bx_{n,s}^N)^\top G(\bw_n,\theta^\star) y_s }_4 
    \leq 
    2 B_t \sqrt[4]{6} N^{-\frac{1}{2}}.
    $$
    Consider now \eqref{eq:asympt_yB_generalized} and notice that if we consider the test function:
    $$
    h_t(\bw_n, x_{0:t}) = \sum_{y_{1:t}} f_{t}(\bw_n,y_{1:t}) \prod_{s=1}^t (x_{s})^\top G(\bw_n,\theta^\star) y_s,
    $$
    whether we have $x_{0:t}=\bx_{n,0:t}^N$ or $x_{0:t}=\bx_{n,0:t}^\infty$ in both cases:
    $$
    \abs{\sum_{y_{1:t}} f_{t}(\bw_n,y_{1:t}) \prod_{s=1}^t (x_{s})^\top G(\bw_n,\theta^\star) y_s} \leq B_t \sum_{y_{1:t}} \abs{ \prod_{s=1}^t (x_{s})^\top G(\bw_n,\theta^\star) y_s} = B_t.
    $$
    The term \eqref{eq:asympt_yB_generalized} is then an application of \eqref{eq:test_function_x_ind_path}:
    \begin{equation}
        \begin{split}
            &\normiii[\Bigg]{\frac{1}{N} \sum_{n \in [N]} \sum_{y_{1:t}} f_{t}(\bw_n,y_{1:t}) \left [ \prod_{s=1}^t (\bx_{n,s}^N)^\top G(\bw_n,\theta^\star) y_s
            - \prod_{s=1}^t (\bx_{n,s}^\infty)^\top G(\bw_n,\theta^\star) y_s \right ] }_4 \\
            &\leq
            2 B_t \sqrt[4]{6} N^{-\frac{1}{2}} \bar{e}_{t}.
        \end{split}
    \end{equation}
    The last term we need to work on is \eqref{eq:asympt_yC_generalized}, but:
    \begin{equation}
        \begin{split}
            \mathbb{E} \left [  f_{t}(\bw_n,\by_{n,{1:t}}^\infty) | \bW^N \right ] 
            &= 
            \mathbb{E} \left \{  \mathbb{E} \left [  f_{t}(\bw_n,\by_{n,{1:t}}^\infty) | \bW^N, \bx_{n,{0:t}}^\infty \right ]| \bW^N \right \} \\
            &=
            \sum_{y_{1:t}} \mathbb{E} \left [  f_{t}(\bw_n,y_{1:t}) \prod_{s=1}^t (\bx_{n,s}^\infty)^\top G(\bw_n,\theta^\star) y_s | \bW^N \right ], 
        \end{split}
    \end{equation}
    and the random variables:
    $$
    \sum_{y_{1:t}} f_{t}(\bw_n,y_{1:t}) \prod_{s=1}^t (\bx_{n,s}^\infty)^\top G(\bw_n,\theta^\star) y_s - f_{t}(\bw_n,\by_{n,{1:t}}^\infty),
    $$
    are defined to be conditionally independent given $\bW^N$ and bounded by $2B_t$, hence we can apply Lemma \ref{lemma:mean_0_bound} and conclude:
    $$
    \normiii[\Bigg]{\frac{1}{N} \sum_{n \in [N]} \sum_{y_{1:t}} f_{t}(\bw_n,y_{1:t}) \prod_{s=1}^t (\bx_{n,s}^\infty)^\top G(\bw_n,\theta^\star) y_s - f_{t}(\bw_n,\by_{n,{1:t}}^\infty)}_4 \leq 2 B_t \sqrt[4]{6} N^{-\frac{1}{2}}.
    $$
    By putting everything together we get:
    \begin{equation}
        \begin{split}
            &\normiii[\Bigg]{\frac{1}{N} \sum_{n \in [N]} f_{t}(\bw_n,\by_{n,{1:t}}^N) - f_{t}(\bw_n,\by_{n,{1:t}}^\infty)}_4\\
            &\leq 
            2 B_t \sqrt[4]{6} N^{-\frac{1}{2}} 
            +
            2 B_t \sqrt[4]{6} N^{-\frac{1}{2}} \bar{e}_{t}
            +
            2 B_t \sqrt[4]{6} N^{-\frac{1}{2}}\\
            &= 2 B_t \sqrt[4]{6} N^{-\frac{1}{2}} ( 2+\bar{e}_{t} ),
        \end{split}
    \end{equation}
    which conclude the proof under $e_{t} \coloneqq ( 2+\bar{e}_{t} )$.
\end{proof}

\paragraph{Saturated CAL algorithm.}
We refer to the following recursion as the ``saturated CAL algorithm''. 
\begin{equation}\label{rec:rand_halfasympt_CAL_rec_full}
    \begin{aligned} &\hat\bpi_{n,0}^\infty(\bw_n,\theta) \coloneqq p_{0}(\bw_n,\theta),\\
    &\hat\bpi_{n,t|t-1}^\infty(\bw_n,\theta)  \coloneqq  \left [ \hat\bpi_{n,t-1}^\infty(\bw_n,\theta)^{\top} K_{\bar{\boeta}_{t-1}^\infty(\bw_n,\theta)} (\bw_n, \theta)\right ]^\top,\\
    &\hat\bmu_{n,t}^\infty(\bw_n,\theta)  \coloneqq  \left [ \hat\bpi_{n,t|t-1}^\infty(\bw_n,\theta)^{\top} G(\bw_n,\theta) \right ]^\top,\\
    &\hat\bpi_{n,t}^\infty(\bw_n,\theta)  \coloneqq \hat\bpi_{n,t|t-1}^\infty(\bw_n,\theta) \odot \left \{ \left [   G(\bw_n,\theta) \oslash \left ( 1_M \hat\bmu_{n,t}^\infty(\bw_n,\theta)^\top \right ) \right ] \by_{n,t}^N  \right \},
    \end{aligned}
\end{equation}
where $\bar{\boeta}_{t-1}^\infty(\cdot,\cdot)$ is defined in \eqref{rec:asympt_CAL_rec_full}.

\begin{proposition} \label{prop:satCAL_well_definess}
    Under assumptions \ref{ass:w_iid},\ref{ass:HMM_support}, for any $N \in \mathbb{N}, t\geq 1, n \in [N]$ and $\theta \in \Theta$ we have that $\sum_i \hat\bmu_{n,t}^\infty(\bw_n,\theta)^{(i)}(\by_{n,t}^N)^{(i)}\neq 0$ \quad $\mathbb{P}$-almost surely.
\end{proposition}

\begin{proof}
    The proof follows the same steps as Theorem \ref{thm:CAL_well_definess}, where we can replicate the same of Proposition \ref{prop:initial_as}, Proposition \ref{prop:prediction_as}, Proposition \ref{prop:correction_as}, and Proposition \ref{prop:sequential_welldefined} for the saturated CAL. The only difference can be found in the prediction step where instead of
    $K_{\widetilde{\boeta}^N_{t-1}(\bw_n,\theta)}(\bw_n,\theta)^{(j,i)}=0$ we have $K_{\bar{\boeta}^\infty_{t-1}(\bw_n,\theta)}(\bw_n,\theta)^{(j,i)}=0$ which similarly implies that there exists $\eta = \bar{\boeta}^\infty_{t-1}(\bw_n,\theta)$ such that $K_{\eta}(\bw_n,\theta)^{(j,i)}=0$, which allows us to follow the same argument as in Proposition \ref{prop:prediction_as}.
\end{proof}

\begin{proposition} \label{prop:satCAL_as_bounded}
    Under assumptions \ref{ass:compactness_continuity},\ref{ass:HMM_support},\ref{ass:kernel_continuity}, for any $t\geq 1$ there exists ${m}_{t}>0$ as in Proposition \ref{prop:CAL_as_bounded} for any $N \in \mathbb{N}$ and $n \in [N]$ such that:
    $$
    \mathbb{P} \left ( \hat{\bmu}_{n,t}^N(\bw_n,\theta)^{(i)}
        \geq 
        {m}_{t}\quad \forall i \in \supp{\bmu_{n,t}^N(\bw_n,\theta)} \right )=1 \quad \forall \theta \in \Theta.
    $$
\end{proposition}

\begin{proof}
    The proof follows the same steps as the proof of Proposition \ref{prop:CAL_as_bounded}, indeed the only difference between the CAL and the saturated CAL is the use of the saturated dynamic from which we still get:
    \begin{equation}
    \begin{split}
    \min_{j \in \supp{K_{\bar{\boeta}^\infty_{t-1}(\bw_n,\theta)}(\bw_n,\theta)^{(\cdot,i)}} } & K_{\bar{\boeta}^\infty_{t-1}(\bw_n,\theta)}(\bw_n,\theta)^{(j,i)} 
    \geq \min_{\eta \in [0,C]} \min_{j \in \supp{K_{\eta}(\bw_n,\theta)^{(\cdot,i)}} } K_{\eta}(\bw_n,\theta)^{(j,i)},
    \end{split}
    \end{equation}
    as in the proof of Proposition \ref{prop:CAL_as_bounded}.
\end{proof}

To conclude the section we establish that the incremental terms in the logarithm of the CAL approximated the corresponding quantities from the saturated CAL algorithm.

\begin{proposition} \label{prop:mu_muhat_bound}
    Under assumptions \ref{ass:compactness_continuity},\ref{ass:w_iid},\ref{ass:HMM_support},\ref{ass:eta_structure},\ref{ass:kernel_continuity}, there exists $\chi_t >0$ such that for any $t\geq 1$, $n\in [N]$:
    \begin{equation}
        \begin{split}
            \normiii[\Bigg]{
		 \log \left ( (\by_{n,t}^N)^\top \bmu_{n,t}^N(\bw_n,\theta) \right ) - \log \left ( (\by_{n,t}^N)^\top \hat{\bmu}_{n,t}^\infty(\bw_n,\theta) \right )
            }_4 \leq 2 \sqrt[4]{6} N^{-\frac{1}{2}} \chi_{t}.
        \end{split}
    \end{equation}
\end{proposition}

\begin{proof}
    Using Proposition \ref{prop:CAL_as_bounded} and Proposition \ref{prop:satCAL_as_bounded} for the lower bound, together with the fact that $\by_{n,t}^N$ is a one hot encoding vector and the $\bmu_\cdot^\cdot$ are probability vectors for the upper bound,  we can conclude that both $(\by_{n,t}^N)^\top {\bmu}_{n,t}^N(\bw_n,\theta)$ and $(\by_{n,t}^N)^\top \hat{\bmu}_{n,t}^\infty(\bw_n,\theta)$ are such that:
    \begin{align}
        0<{m}_{t} \leq (\by_{n,t}^N)^\top \bmu_{n,t}^N(\bw_n,\theta) \leq 1, \quad 0<{m}_{t}  \leq (\by_{n,t}^N)^\top \hat{\bmu}_{n,t}^\infty(\bw_n,\theta) \leq 1,
    \end{align}
    $\mathbb{P}$-almost surely.
    
    As the the function $u\mapsto\log(u)$ is Lipschitz on the compact interval $[{m}_{t},1]$ we have:
    \begin{equation}
        \begin{split}
             &\abs{\log \left ( (\by_{n,t}^N)^\top \bmu_{n,t}^N(\bw_n,\theta) \right ) - \log \left ( (\by_{n,t}^N)^\top \hat{\bmu}_{n,t}^\infty(\bw_n,\theta) \right )}\\
             &\leq 
             \frac{1}{{m}_{t}}
             \abs{(\by_{n,t}^N)^\top \bmu_{n,t}^N(\bw_n,\theta) - (\by_{n,t}^N)^\top \hat{\bmu}_{n,t}^\infty(\bw_n,\theta)},
        \end{split}
    \end{equation}
    $\mathbb{P}$-almost surely. From the above we can conclude:
    \begin{equation}
        \begin{split}
            &\normiii[\Bigg]{
	\log \left ( (\by_{n,t}^N)^\top \bmu_{n,t}^N(\bw_n,\theta) \right ) - \log \left ( (\by_{n,t}^N)^\top \hat{\bmu}_{n,t}^\infty(\bw_n,\theta) \right )
            }_4\\
            &\leq 
            \frac{1}{{m}_{t}} \normiii[\Bigg]{(\by_{n,t}^N)^\top \bmu_{n,t}^N(\bw_n,\theta) - (\by_{n,t}^N)^\top \hat{\bmu}_{n,t}^\infty(\bw_n,\theta) }_4,
        \end{split}
    \end{equation}
    meaning that it remains to bound:
    $$
    \normiii[\Bigg]{(\by_{n,t}^N)^\top \bmu_{n,t}^N(\bw_n,\theta) - (\by_{n,t}^N)^\top \hat{\bmu}_{n,t}^\infty(\bw_n,\theta) }_4.
    $$

    Consider as an inductive hypothesis that there exists a constant $\bar{\chi}_{t-1}>0$ such that for any random vector $\bof_n$ which satisfies $\norm{\bof_n}_\infty \leq B$, $\mathbb{P}$-almost surely:
    \begin{equation}
        \begin{split}
            \normiii[\Bigg]{
            \bof_n^\top \bpi_{n,t-1}^N(\bw_n,\theta) - \bof_n^\top \hat{\bpi}_{n,t-1}^\infty(\bw_n,\theta)
            }_4
            \leq 2 B \sqrt[4]{6} N^{-\frac{1}{2}} \bar{\chi}_{t-1}.
        \end{split}
    \end{equation}

    Observe that this is valid at $t-1=0$ since ${\bpi}_{n,0}^\infty(\bw_n,\theta)=\hat{\bpi}_{n,0}^\infty(\bw_n,\theta)$, therefore we have:
    \begin{equation}
        \begin{split}
            \normiii[\Bigg]{
            \bof_n^\top \bpi_{n,0}^N(\bw_n,\theta) - \bof_n^\top \hat{\bpi}_{n,0}^\infty(\bw_n,\theta)
            }_4
            =0.
        \end{split}
    \end{equation}
    In order to show that the inductive hypothesis holds at time $t$, let us first consider the prediction step, and observe that:
    \begin{align}
            &\normiii[\Bigg]{
            \bof_n^\top \bpi_{n,t|t-1}^N(\bw_n,\theta) - \bof_n^\top \hat{\bpi}_{n,t|t-1}^\infty(\bw_n,\theta)
            }_4\\
            &\leq
            \normiii[\Bigg]{
            \bpi_{n,t-1}^N(\bw_n,\theta)^\top K_{\widetilde{\boeta}_{t-1}^N(\bw_n,\theta)} (\bw_n, \theta) \bof_n
            - {\bpi}_{n,t-1}^N(\bw_n,\theta)^\top K_{\bar{\boeta}_{t-1}^\infty(\bw_n,\theta)} (\bw_n, \theta) \bof_n }_4 \label{eq:mu_muhat_pred_A}\\
            & 
            + 
            \normiii[\Bigg]{
            {\bpi}_{n,t-1}^N(\bw_n,\theta)^\top K_{\bar{\boeta}_{t-1}^\infty(\bw_n,\theta)} (\bw_n, \theta) \bof_n - \hat{\bpi}_{n,t-1}^\infty(\bw_n,\theta)^\top K_{\bar{\boeta}_{t-1}^\infty(\bw_n,\theta)} (\bw_n, \theta) \bof_n}_4. \label{eq:mu_muhat_pred_B}
    \end{align}
    Starting from \eqref{eq:mu_muhat_pred_A}, we have:
    \begin{equation}
        \begin{split}
            &\normiii[\Bigg]{
            \bpi_{n,t-1}^N(\bw_n,\theta)^\top K_{\widetilde{\boeta}_{t-1}^N(\bw_n,\theta)} (\bw_n, \theta) \bof_n
            - {\bpi}_{n,t-1}^N(\bw_n,\theta)^\top K_{\bar{\boeta}_{t-1}^\infty(\bw_n,\theta)} (\bw_n, \theta) \bof_n }_4\\
            &\leq 
            \norm{\bof_n}_\infty
            \normiii[\Bigg]{
            \widetilde{\boeta}_{t-1}^N(\bw_n,\theta) - \bar{\boeta}_{t-1}^\infty(\bw_n,\theta)
            }_4 \leq 
            2 B C \sqrt[4]{6} N^{-\frac{1}{2}} (\gamma_{t-1} +1),
        \end{split}
    \end{equation}
    which follows in the same way as the proof of Proposition \ref{prop:bound_CAL_prediction}, see \eqref{eq:bound_tilte_eta_bar_eta}.

    Moving to \eqref{eq:mu_muhat_pred_B}, we can apply our inductive hypothesis on the random vector: $K_{\bar{\boeta}_{t-1}^\infty(\bw_n,\theta)} (\bw_n, \theta) \bof_n$ since  $\norm{K_{\bar{\boeta}_{t-1}^\infty(\bw_n,\theta)} (\bw_n, \theta) \bof_n}_\infty \leq B$, $\mathbb{P}$-almost surely, hence:
    \begin{equation}
        \begin{split}
            &\normiii[\Bigg]{
            {\bpi}_{n,t-1}^N(\bw_n,\theta)^\top K_{\bar{\boeta}_{t-1}^\infty(\bw_n,\theta)} (\bw_n, \theta) \bof_n - \hat{\bpi}_{n,t-1}^\infty(\bw_n,\theta)^\top K_{\bar{\boeta}_{t-1}^\infty(\bw_n,\theta)} (\bw_n, \theta) \bof_n}_4\\
            & \leq 
            2 B \sqrt[4]{6} N^{-\frac{1}{2}} \bar{\chi}_{t-1}.
        \end{split}
    \end{equation}
    By putting everything together we can conclude:
    \begin{align}
            \normiii[\Bigg]{
            \bof_n^\top \bpi_{n,t|t-1}^N(\bw_n,\theta) - \bof_n^\top \hat{\bpi}_{n,t|t-1}^\infty(\bw_n,\theta)
            }_4
            &\leq
            2 B [C(\gamma_{t-1} +1) + \bar{\chi}_{t-1}] \sqrt[4]{6} N^{-\frac{1}{2}}\\
            &\leq
            2 B \sqrt[4]{6} N^{-\frac{1}{2}} \bar{\chi}_{t|t-1},
    \end{align}
    where $\bar{\chi}_{t|t-1} \coloneqq C(\gamma_{t-1} +1) + \bar{\chi}_{t-1}$. Given the above we also have:
    \begin{equation} \label{eq:mu_muhat_bound}
        \begin{split}
            &\normiii[\Bigg]{
            \bof_n^\top \bmu_{n,t}^N(\bw_n,\theta) - \bof_n^\top \hat{\bmu}_{n,t}^\infty(\bw_n,\theta)
            }_4\\
            &\leq
            \normiii[\Bigg]{
            \bpi_{n,t|t-1}^N(\bw_n,\theta)^\top G(\bw_n,\theta)  \bof_n- \hat{\bpi}_{n,t|t-1}^\infty(\bw_n,\theta)^\top  G(\bw_n,\theta)  \bof_n
            }_4\\
            &\leq
            2 B \sqrt[4]{6} N^{-\frac{1}{2}} \bar{\chi}_{t|t-1},
        \end{split}
    \end{equation}
    which is just an application of the previous result on the bounded random variable $G(\bw_n,\theta)  \bof_n$, as we know  $\norm{G(\bw_n,\theta)  f}_\infty \leq B$ is bounded $\mathbb{P}$-almost surely.

    Now we need to work on:
    \begin{align}
            \normiii[\Bigg]{
            \bof_n^\top \bpi_{n,t}^N(\bw_n,\theta) - \bof_n^\top \hat{\bpi}_{n,t}^\infty(\bw_n,\theta)
            }_4.
    \end{align}
    Note that by using $G_{\bmu}(\bw,\theta)$ for the matrix with elements $G_{\bmu}(\bw,\theta)^{(i,j)} = \frac{ G(\bw,\theta)^{(i,j)}}{\bmu^{(j)}}$ where $\frac{0}{0}=0$ by convention, we can rewrite everything in a more compact way:
    \begin{equation}
        \begin{split}
            &\normiii[\Bigg]{
            \bof_n^\top \bpi_{n,t}^N(\bw_n,\theta) - \bof_n^\top \hat{\bpi}_{n,t}^\infty(\bw_n,\theta)
            }_4\\
            &= 
            \Bigg{\vvvert} \left [\bof_n \odot \bpi_{n,t|t-1}^N(\bw_n,\theta) \right ]^\top \left [ G_{\bmu_{n,t}^N(\bw_n,\theta)}(\bw_n,\theta) \by_{n,t}^N \right ] \\
            &\qquad-
           \left [ \bof_n \odot \hat{\bpi}_{t|t-1}^\infty(\bw_n,\theta) \right ]^\top \left [ G_{\hat{\bmu}_{n,t}^\infty(\bw_n,\theta)}(\bw_n,\theta) \by_{n,t}^N \right ]\Bigg{\vvvert}_4.
        \end{split}
    \end{equation}
    By Minkowski inequality we can then conclude:
    \begin{align}
            &\normiii[\Bigg]{
            \bof_n^\top \bpi_{n,t}^N(\bw_n,\theta) - \bof_n^\top \hat{\bpi}_{n,t}^\infty(\bw_n,\theta)
            }_4\\
            &\leq
            \Bigg{\vvvert} \left [ \bof_n \odot \bpi_{n,t|t-1}^N(\bw_n,\theta) \right ]^\top \left [ G_{\bmu_{n,t}^N(\bw_n,\theta)}(\bw_n,\theta) \by_{n,t}^N \right ] \\
            &\qquad-
           \left [ \bof_n \odot {\bpi}_{t|t-1}^N(\bw_n,\theta) \right ]^\top \left [ G_{\hat{\bmu}_{n,t}^\infty(\bw_n,\theta)}(\bw_n,\theta) \by_{n,t}^N \right ]\Bigg{\vvvert}_4 \label{mu_muhat_corr_A}\\
           &+
           \Bigg{\vvvert} \left [ \bof_n \odot {\bpi}_{n,t|t-1}^N(\bw_n,\theta) \right ]^\top \left [ G_{\hat{\bmu}_{n,t}^\infty(\bw_n,\theta)}(\bw_n,\theta) \by_{n,t}^N \right ] \\
            &\qquad-
           \left [ \bof_n \odot \hat{\bpi}_{t|t-1}^\infty(\bw_n,\theta) \right ]^\top \left [ G_{\hat{\bmu}_{n,t}^\infty(\bw_n,\theta)}(\bw_n,\theta) \by_{n,t}^N \right ]\Bigg{\vvvert}_4. \label{mu_muhat_corr_B}
    \end{align}

    Starting from \eqref{mu_muhat_corr_A}, we remark that:
    \begin{equation}
        x^\top G_\mu b - x^\top G_{\widetilde{\mu}} b = \sum_{i,j} x^{(i)} y^{(j)} \frac{G^{(i,j)} \widetilde{\mu}^{(j)} - G^{(i,j)} \mu^{(j)}}{\widetilde{\mu}^{(j)}{\mu}^{(j)}} = \sum_{i,j} x^{(i)} \frac{y^{(j)}}{\widetilde{\mu}^{(j)}{\mu}^{(j)}}G^{(i,j)} \left ( \widetilde{\mu}^{(j)} - \mu^{(j)} \right ).
    \end{equation}
    
    Hence we can reformulate \eqref{mu_muhat_corr_A}:
    \begin{equation}
        \begin{split}
            &\Bigg{\vvvert} \left [ \bof_n \odot \bpi_{n,t|t-1}^N(\bw_n,\theta) \right ]^\top \left [ G_{\bmu_{n,t}^N(\bw_n,\theta)}(\bw_n,\theta) \by_{n,t}^N \right ] \\
            &\qquad-
           \left [ \bof_n \odot {\bpi}_{t|t-1}^N(\bw_n,\theta) \right ]^\top \left [ G_{\hat{\bmu}_{n,t}^\infty(\bw_n,\theta)}(\bw_n,\theta) \by_{n,t}^N \right ]\Bigg{\vvvert}_4\\
            &=
            \Bigg{\vvvert} \left \{ \left [ \bof_n \odot  \bpi_{n,t|t-1}^N(\bw_n,\theta) \right ]^\top G(\bw_n,\theta) \right \}^\top \odot \left [ \by_{n,t}^N \oslash \bmu_{n,t}^N(\bw_n,\theta) \oslash \hat{\bmu}_{n,t}^\infty(\bw_n,\theta)\right ]^\top\\
            &\qquad \qquad 
            \left [\hat{\bmu}_{n,t}^\infty(\bw_n,\theta) - \bmu_{n,t}^N(\bw_n,\theta) \right ] \Bigg{\vvvert}_4,
        \end{split}
    \end{equation}
    from which we can notice that for any $\bw_n$:
    \begin{equation}
        \begin{split}
            &\norm{\left \{ \left [ \bof_n \odot {\bpi}_{t|t-1}^N(\bw_n,\theta) \right ]^\top G(\bw_n,\theta) \right \} \odot \left [ \by_{n,t}^N \oslash \bmu_{n,t}^N(\bw_n,\theta) \oslash \hat{\bmu}_{n,t}^\infty(\bw_n,\theta)\right ]}_\infty\\
            &\leq
            \norm{\bof_n}_\infty \norm{ \by_{n,t}^N \oslash \hat{\bmu}_{n,t}^\infty(\bw_n,\theta)}_\infty,
        \end{split}
    \end{equation}
    where the first step follows from ${\bmu}_{n,t}^N(\bw,\theta) = \left [ {\bpi}_{t|t-1}^N(\bw,\theta)^\top G(\bw,\theta) \right ]^\top$ and the elementwise ratio $\by_{n,t}^N \oslash \hat{\bmu}_{n,t}^\infty(\bw_n,\theta)$ is well-defined because of Proposition \ref{prop:satCAL_well_definess}. 
    
    As from Proposition \ref{prop:satCAL_as_bounded} we know that the saturated CAL is almost surely bounded we have:
    \begin{equation}
        \begin{split}
        &\norm{\left \{ \left [ \bof_n \odot {\bpi}_{t|t-1}^N(\bw_n,\theta) \right ]^\top G(\bw_n,\theta) \right \} \odot \left [ \by_{n,t}^N \oslash \bmu_{n,t}^N(\bw_n,\theta) \oslash \hat{\bmu}_{n,t}^\infty(\bw_n,\theta)\right ]}_\infty\\ &\leq \frac{\norm{\bof_n}_\infty}{{m}_{t}}
        \leq \frac{B}{{m}_{t}},        
        \end{split}
    \end{equation}
    $\mathbb{P}$-almost surely. Hence we can apply \eqref{eq:mu_muhat_bound} as we are considering an almost surely bounded random vector:
    $$
    \left \{ \left [ \bof_n \odot {\bpi}_{t|t-1}^N(\bw_n,\theta) \right ]^\top G(\bw_n,\theta) \right \} \odot \left [ \by_{n,t}^N \oslash \bmu_{n,t}^N(\bw_n,\theta) \oslash \hat{\bmu}_{n,t}^\infty(\bw_n,\theta)\right ]
    $$
    and conclude:
    \begin{equation}
        \begin{split}
            &\Bigg{\vvvert} \left [ \bof_n \odot \bpi_{n,t|t-1}^N(\bw_n,\theta) \right ]^\top \left [ G_{\bmu_{n,t}^N(\bw_n,\theta)}(\bw_n,\theta) \by_{n,t}^N \right ] \\
            &\qquad-
           \left [ \bof_n \odot {\bpi}_{t|t-1}^N(\bw_n,\theta) \right ]^\top \left [ G_{\hat{\bmu}_{n,t}^\infty(\bw_n,\theta)}(\bw_n,\theta) \by_{n,t}^N \right ]\Bigg{\vvvert}_4\\
            &=
            \Bigg{\vvvert} \left \{ \left [ \bof_n \odot  \bpi_{n,t|t-1}^N(\bw_n,\theta) \right ]^\top G(\bw_n,\theta) \right \}^\top \odot \left [ \by_{n,t}^N \oslash \bmu_{n,t}^N(\bw_n,\theta) \oslash \hat{\bmu}_{n,t}^\infty(\bw_n,\theta)\right ]^\top\\
            &\qquad \qquad 
            \left [\hat{\bmu}_{n,t}^\infty(\bw_n,\theta) - \bmu_{n,t}^N(\bw_n,\theta) \right ] \Bigg{\vvvert}_4 \leq 2 B \sqrt[4]{6} N^{-\frac{1}{2}} \frac{\chi_t}{{m}_{t}}.
        \end{split}
    \end{equation}
    Moving to \eqref{mu_muhat_corr_B}, we can observe that:
    \begin{equation}
        \begin{split}
            &\Bigg{\vvvert} \left [ \bof_n \odot {\bpi}_{n,t|t-1}^N(\bw_n,\theta) \right ]^\top \left [ G_{\hat{\bmu}_{n,t}^\infty(\bw_n,\theta)}(\bw_n,\theta) \by_{n,t}^N \right ] \\
            &\qquad-
           \left [ \bof_n \odot \hat{\bpi}_{t|t-1}^\infty(\bw_n,\theta) \right ]^\top \left [ G_{\hat{\bmu}_{n,t}^\infty(\bw_n,\theta)}(\bw_n,\theta) \by_{n,t}^N \right ]\Bigg{\vvvert}_4\\
           &=
           \Bigg{\vvvert} \left [ \bof_n \odot G_{\hat{\bmu}_{n,t}^\infty(\bw_n,\theta)}(\bw_n,\theta) \by_{n,t}^N \right ]^\top \left [ {\bpi}_{t|t-1}^N(\bw_n,\theta) - \hat{\bpi}_{t|t-1}^\infty(\bw_n,\theta) \right ] \Bigg{\vvvert}_4\\
           &\leq 
           2 B \sqrt[4]{6} N^{-\frac{1}{2}} {\bar{\chi}_{t|t-1}},
        \end{split}
    \end{equation}
    where we can apply \eqref{eq:mu_muhat_bound} to the test vector $\bof_n \odot G_{\hat{\bmu}_{n,t}^\infty(\bw_n,\theta)}(\bw_n,\theta) \by_{n,t}^N$, as it is almost surely bounded.

    By putting everything together we can conclude:
    \begin{align}
        &\normiii[\Bigg]{
        \bof_n^\top \bpi_{n,t}^N(\bw_n,\theta) - \bof_n^\top \hat{\bpi}_{n,t}^\infty(\bw_n,\theta)
        }_4
        \leq
        2 B \sqrt[4]{6} N^{-\frac{1}{2}} \frac{\chi_t}{{m}_{t}}+
        2 B \sqrt[4]{6} N^{-\frac{1}{2}} {\bar{\chi}_{t|t-1}},
    \end{align}
    which closes our inductive hypothesis by setting ${\bar{\chi}_{t}} = \frac{\chi_t}{{m}_{t}}+{\bar{\chi}_{t|t-1}}$, and so we can conclude that for any $t\geq 1$ there exists $\bar{\chi}_{t|t-1},\bar{\chi}_t$ such that for any test random vector $\bof_n$, with $\norm{\bof_n}_\infty \leq B$, $\mathbb{P}$-almost surely:
    \begin{itemize}
        \item $\normiii[\Bigg]{
            \bof_n^\top \bpi_{n,t|t-1}^N(\bw_n,\theta) - \bof_n^\top \hat{\bpi}_{n,t|t-1}^\infty(\bw_n,\theta)
            }_4
            \leq
            2 B \sqrt[4]{6} N^{-\frac{1}{2}} \bar{\chi}_{t|t-1}$;
        \item $\normiii[\Bigg]{
            \bof_n^\top \bmu_{n,t}^N(\bw_n,\theta) - \bof_n^\top \hat{\bmu}_{n,t}^\infty(\bw_n,\theta)
            }_4
            \leq
            2 B \sqrt[4]{6} N^{-\frac{1}{2}} \bar{\chi}_{t|t-1}$;
        \item $\normiii[\Bigg]{
            \bof_n^\top \bpi_{n,t}^N(\bw_n,\theta) - \bof_n^\top \hat{\bpi}_{n,t}^\infty(\bw_n,\theta)
            }_4
            \leq
            2 B \sqrt[4]{6} N^{-\frac{1}{2}} \bar{\chi}_{t}$.
    \end{itemize}
    As $\by_{n,t}^N$ is almost surely bounded $\norm{\by_{n,t}^N}_\infty \leq 1$ we can conclude:
    $$
    \normiii[\Bigg]{(\by_{n,t}^N)^\top \bmu_{n,t}^N(\bw_n,\theta) - (\by_{n,t}^N)^\top \hat{\bmu}_{n,t}^\infty(\bw_n,\theta) }_4 \leq 2 \sqrt[4]{6} N^{-\frac{1}{2}} \bar{\chi}_{t|t-1},
    $$
    so we can conclude our proof by setting $\chi_t \coloneqq \frac{\bar{\chi}_{t|t-1}}{{m}_{t}}$ as:
    \begin{equation}
        \begin{split}
            &\normiii[\Bigg]{
	\log \left ( (\by_{n,t}^N)^\top \bmu_{n,t}^N(\bw_n,\theta) \right ) - \log \left ( (\by_{n,t}^N)^\top \hat{\bmu}_{n,t}^\infty(\bw_n,\theta) \right )
            }_4\\
            &\leq 
            \frac{1}{{m}_{t}} \normiii[\Bigg]{(\by_{n,t}^N)^\top \bmu_{n,t}^N(\bw_n,\theta) - (\by_{n,t}^N)^\top \hat{\bmu}_{n,t}^\infty(\bw_n,\theta) }_4\leq 
            2 \sqrt[4]{6} N^{-\frac{1}{2}} \frac{\bar{\chi}_{t|t-1}}{{m}_{t}}.
        \end{split}
    \end{equation}
\end{proof}

\subsection{Strong consistency} \label{sec:constrast_function_pointwise_conv}

\paragraph{Contrast function.}
In Proposition \ref{prop:convergence_log_likelihood} below we establish the convergence of the rescaled logarithm of the CAL, $\frac{\ell_{1:T}^N(\theta)}{N}$ in the large population limit, and in Theorem \ref{thm:constrast_function_conv} show how $\frac{\ell_{1:T}^N(\theta)}{N} -  \frac{\ell_{1:T}^N(\theta^\star)}{N}$ converges to a contrast function.

\begin{proposition} \label{prop:convergence_log_likelihood}
    Under assumptions \ref{ass:compactness_continuity},\ref{ass:w_iid},\ref{ass:HMM_support},\ref{ass:eta_structure},\ref{ass:kernel_continuity} and for any $\theta\in\Theta$ let:
    $$
    \ell_t^N(\theta) \coloneqq \sum_{n \in [N]} \log \left ( (\by_{n,t}^N)^\top  \bmu_{n,t}^N(\bw_n,\theta) \right ),
    $$ 
    then $N^{-1} \ell_t^N(\theta)$ converges to $\mathbb{E} \left [ \log \left ( (\by_t^\infty)^\top  \bmu_t^\infty(\bw^\infty,\theta) \right ) \right ]$ as $N\to\infty$, $\mathbb{P}$-almost surely. Moreover:
    \begin{equation}
        \mathbb{E} \left [ \log \left ( (\by_t^\infty)^\top  \bmu_t^\infty(\bw^\infty,\theta) \right ) \right ] 
        =
        \mathbb{E} \left [ \bmu_t^\infty(\bw^\infty,\theta^\star)^\top  \log \left ( \bmu_t^\infty(\bw^\infty,\theta) \right ) \right ].
    \end{equation}
\end{proposition}

\begin{proof}

    Consider the definition of $\bmu_{n,t}^N(\bw_n,\theta)$ from \eqref{rec:CAL_rec}, the definition of $\hat{\bmu}_{n,t}^\infty(\bw_n,\theta)$ from \eqref{rec:rand_halfasympt_CAL_rec_full}, the definition of $\bmu_t^\infty(w,\theta)$ from \eqref{rec:rand_asympt_CAL_rec_full}, and define:
    \begin{equation}\label{rec:single_rand_asympt_CAL_rec_full}
        \begin{aligned}
        &\bpi_{n,0}^\infty(\bw_n,\theta) \coloneqq p_{0}(\bw_n,\theta),\\
        &\bpi_{n,t|t-1}^\infty(\bw_n,\theta)  \coloneqq  \left [ \bpi_{n,t-1}^\infty(\bw_n,\theta)^{\top} K_{\bar{\boeta}_{t-1}^\infty(\bw_n,\theta)} (\bw_n, \theta)\right ]^\top, \\
        &\bmu_{n,t}^\infty(\bw_n,\theta)  \coloneqq  \left [ \bpi_{n,t|t-1}^\infty(\bw_n,\theta)^{\top} G(\bw_n,\theta) \right ]^\top,\\
        &\bpi_{n,t}^\infty(\bw_n,\theta)  \coloneqq \bpi_{n,t|t-1}^\infty(\bw_n,\theta) \odot \left \{ \left [   G(\bw_n,\theta) \oslash \left ( 1_M \bmu_{n,t}^\infty(\bw_n,\theta)^\top \right ) \right ] \by_{n,t}^\infty  \right \},
        \end{aligned}
    \end{equation} 
    where $\bar{\boeta}_{t-1}^\infty(\cdot,\cdot)$ is defined in \eqref{rec:asympt_CAL_rec_full} and $\by^\infty_{n,t}$ is defined in \eqref{eq:limiting_process_n}.
    
    We can then consider the following decomposition:
    \begin{align}
        &\frac{1}{N} \sum_{n \in [N]}  \log \left ( (\by_{n,t}^N)^\top \bmu_{n,t}^N(\bw_n,\theta) \right )- \mathbb{E} \left [  \log \left ( (\by_t^\infty)^\top \bmu_t^\infty(\bw^\infty,\theta) \right )\right ]\\
        &=
        \frac{1}{N} \sum_{n \in [N]} ( \log \left ( \by_{n,t}^N)^\top \bmu_{n,t}^N(\bw_n,\theta) \right )
        -
         \log \left ( (\by_{n,t}^N)^\top \hat{\bmu}_{n,t}^\infty(\bw_n,\theta) \right )\\
        & 
        +
        \frac{1}{N} \sum_{n \in [N]}  \log \left ( (\by_{n,t}^N)^\top \hat{\bmu}_{n,t}^\infty(\bw_n,\theta) \right ) 
        -
         \log \left ( (\by_{n,t}^\infty)^\top \bmu_{n,t}^\infty(\bw_n,\theta) \right )\\
        & 
        +\frac{1}{N} \sum_{n \in [N]}  \log \left ( (\by_{n,t}^\infty)^\top \bmu_{n,t}^\infty(\bw_n,\theta) \right ) - \mathbb{E} \left [  \log \left ( (\by_t^\infty)^\top \bmu_t^\infty(\bw^\infty,\theta) \right ) \right ],
    \end{align}
    by Minkowski inequality we conclude:
    \begin{align}
        &\normiii[\Bigg]{\frac{1}{N} \sum_{n \in [N]}  \log \left ( (\by_{n,t}^N)^\top \bmu_{n,t}^N(\bw_n,\theta) \right )- \mathbb{E} \left [  \log \left ( (\by_t^\infty)^\top \bmu_t^\infty(\bw^\infty,\theta) \right )\right ]}_4\\
        &\leq 
        \normiii[\Bigg]{\frac{1}{N} \sum_{n \in [N]}  \log \left ( (\by_{n,t}^N)^\top \bmu_{n,t}^N(\bw_n,\theta) \right )
        -
         \log \left ( (\by_{n,t}^N)^\top \hat{\bmu}_{n,t}^\infty(\bw_n,\theta) \right )}_4 \label{eq:contrast_A}\\
        & 
        +
        \normiii[\Bigg]{\frac{1}{N} \sum_{n \in [N]}  \log \left ( (\by_{n,t}^N)^\top \hat{\bmu}_{n,t}^\infty(\bw_n,\theta) \right ) 
        -
         \log \left ( (\by_{n,t}^\infty)^\top \bmu_{n,t}^\infty(\bw_n,\theta) \right )}_4 
        \label{eq:contrast_B}\\
        & 
        +\normiii[\Bigg]{\frac{1}{N} \sum_{n \in [N]}  \log \left ( (\by_{n,t}^\infty)^\top \bmu_{n,t}^\infty(\bw_n,\theta) \right ) - \mathbb{E} \left [  \log \left ( (\by_t^\infty)^\top \bmu_{t}^\infty(\bw^\infty,\theta) \right ) \right ]}_4.
        \label{eq:contrast_C}
    \end{align}
    Starting from \eqref{eq:contrast_A}, we can apply Proposition \ref{prop:mu_muhat_bound} to obtain:
    \begin{equation}
        \begin{split}
            \normiii[\Bigg]{
		 \log \left ( (\by_{n,t}^N)^\top \bmu_{n,t}^N(\bw_n,\theta) \right ) - \log \left ( (\by_{n,t}^N)^\top \hat{\bmu}_{n,t}^\infty(\bw_n,\theta) \right )
            }_4 \leq 2 \sqrt[4]{6} N^{-\frac{1}{2}} \chi_{t}.
        \end{split}
    \end{equation}

    Consider now \eqref{eq:contrast_B}, it is important to observe that when we compare \eqref{rec:rand_halfasympt_CAL_rec_full} with \eqref{rec:single_rand_asympt_CAL_rec_full}, the only difference between the two is that \eqref{rec:rand_halfasympt_CAL_rec_full} uses $\by_{n,t}^N$ while \eqref{rec:single_rand_asympt_CAL_rec_full} uses the population saturated process $\by_{n,t}^\infty$. Hence, if we look at $(\by_{n,t}^N)^\top \hat{\bmu}_{n,t}^\infty(\bw_n,\theta)$ as a function of $\by_{n,1:t}^N$ and $(\by_{n,t}^\infty)^\top {\bmu}_{n,t}^\infty(\bw_n,\theta)$ as a function of $\by_{n,1:t}^\infty$ we are considering the same function evaluated in different arguments;  we can define for a fixed $\theta$ the function $h_t^\theta(w,y_{1:t})$ which is such that $h_t^\theta(\bw_n,\by_{n,t}^N)= (\by_{n,t}^N)^\top \hat{\bmu}_{n,t}^\infty(\bw_n,\theta)$ and $h_t^\theta(\bw_n,\by_{n,t}^\infty)=(\by_{n,t}^\infty)^\top \bmu_{n,t}^\infty(\bw_n,\theta)$.
    Because of Proposition \ref{prop:satCAL_as_bounded}, for all $i \in [M]$ we have that $\hat{\bmu}_{n,t}^\infty(\bw_n,\theta)^{(i)}>m_t$, meaning that $h_t^\theta(\bw_n,\by_{n,t}^N) \in [m_t,1]$ almost surely. Similarly, because of Proposition \ref{prop:infCAL_as_bounded}, for all $i \in [M]$ we have that ${\bmu}_{n,t}^\infty(\bw_n,\theta)^{(i)}>m_t$, meaning that $h_t^\theta(\bw_n,\by_{n,t}^\infty) \in [m_t,1]$ almost surely. This follows as \eqref{eq:limiting_process_n} consists of repeating \eqref{eq:limiting_process} $N$ times. We can then consider $\log(h_t^\theta(\bw_n,\by_{n,t}^N)) \in [\log(m_t),0]$ and $\log(h_t^\theta(\bw_n,\by_{n,t}^\infty)) \in [\log(m_t),0]$ almost surely and apply Proposition \ref{thm:LLN_f_y_yinf} as both $\log(h_t^\theta(\bw_n,\by_{n,t}^N))$ and $\log(h_t^\theta(\bw_n,\by_{n,t}^\infty))$ are almost surely in $[\log(m_t),\abs{\log(m_t)}]$:
    \begin{equation}
        \begin{split}
            \normiii[\Bigg]{
            \frac{1}{N} \sum_{n \in [N]}  \log \left ( (\by_{n,t}^N)^\top \hat{\bmu}_{n,t}^\infty(\bw_n,\theta) \right ) 
            -
             \log \left ( (\by_{n,t}^\infty)^\top \bmu_t^\infty(\bw_n,\theta) \right )
             }_4 \leq 
            2 \abs{\log(m_t)} \sqrt[4]{6} N^{-\frac{1}{2}} e_t.
        \end{split}
    \end{equation}

    For \eqref{eq:contrast_C}, note that:
    \begin{equation}
        \mathbb{E} \left [ \log \left ( (\by_{n,t}^\infty)^\top \bmu_{n,t}^\infty(\bw_n,\theta) \right )  \right ]
        = 
        \mathbb{E} \left [ \log \left ( (\by_{t}^\infty)^\top  \bmu_{t}^\infty(\bw^\infty,\theta) \right ) \right ],
    \end{equation}
    see the definitions of $\by_{n,t}^\infty,\by_{t}^\infty,\bmu_{t}^\infty(\bw^\infty,\theta)$ in Section \ref{sec:limiting_quantities_support} and the definition of $\bmu_{n,t}^\infty(\bw_n,\theta)$ in \eqref{rec:single_rand_asympt_CAL_rec_full} at the beginning of the proof. Moreover, 
    \begin{equation}
        \norm{\log \left ( (\by_{n,t}^\infty)^\top \bmu_{n,t}^\infty(\bw_n,\theta) \right ) - \mathbb{E} \left [ \log \left ( (\by_t^\infty)^\top  \bmu_t^\infty(\bw^\infty,\theta) \right ) \right ]} \leq 2 \abs{\log(m_t)},
    \end{equation}
    $\mathbb{P}$-almost surely, because of the previous reasoning. As we are considering averages of random variables that are mean zero, bounded, and independent, we can apply Lemma \ref{lemma:mean_0_bound} and conclude:
    \begin{equation}
    \begin{split}
        &\normiii[\Bigg]{\frac{1}{N} \sum_{n \in [N]}  \log \left ( (\by_{n,t}^\infty)^\top \bmu_{n,t}^\infty(\bw_n,\theta) \right ) - \mathbb{E} \left [  \log \left ( (\by_t^\infty)^\top \bmu_t^\infty(\bw^\infty,\theta) \right ) \right ]}_4\leq  2 \sqrt[4]{6} N^{-\frac{1}{2}} \abs{\log(m_t)}. 
    \end{split}
    \end{equation}
    By putting everything together we obtain:
    \begin{equation} \label{eq:final_bound_borel_cantelli}
        \begin{split}
            &\normiii[\Bigg]{\frac{1}{N} \sum_{n \in [N]}  \log \left ( (\by_{n,t}^N)^\top \bmu_{n,t}^N(\bw_n,\theta) \right )- \mathbb{E} \left [  \log \left ( (\by_t^\infty)^\top \bmu_t^\infty(\bw^\infty,\theta) \right )\right ]}_4\\
            &\leq
            2 \sqrt[4]{6} N^{-\frac{1}{2}} \chi_{t}
            +
            2 \abs{\log(m_t)} \sqrt[4]{6} N^{-\frac{1}{2}} e_t
            +
            2 \sqrt[4]{6} N^{-\frac{1}{2}} \abs{\log(m_t)}\\
            &=
            2 \sqrt[4]{6} N^{-\frac{1}{2}} \left [ \chi_{t} + \abs{\log(m_t)} (e_t + 1)\right ].
        \end{split}
    \end{equation}
    By applying Markov's inequality, for any $\iota>0$,
    \begin{equation}
        \begin{split}
            &\mathbb{P} \left ( \abs{\frac{1}{N} \sum_{n \in [N]}  \log \left ( (\by_{n,t}^N)^\top \bmu_{n,t}^N(\bw_n,\theta) \right )- \mathbb{E} \left [  \log \left ( (\by_t^\infty)^\top \bmu_t^\infty(\bw^\infty,\theta) \right )\right ]} > \iota \right ) \\
        &=
        \mathbb{P} \left ( \abs{\frac{1}{N} \sum_{n \in [N]}  \log \left ( (\by_{n,t}^N)^\top \bmu_{n,t}^N(\bw_n,\theta) \right )- \mathbb{E} \left [  \log \left ( (\by_t^\infty)^\top \bmu_t^\infty(\bw^\infty,\theta) \right )\right ]}^4 > \iota^4 \right ) \\
            &\leq 
        \frac{1}{\iota^4} \mathbb{E} \left [ \abs{\frac{1}{N} \sum_{n \in [N]}  \log \left ( (\by_{n,t}^N)^\top \bmu_{n,t}^N(\bw_n,\theta) \right )- \mathbb{E} \left [  \log \left ( (\by_t^\infty)^\top \bmu_t^\infty(\bw^\infty,\theta) \right )\right ]}^4 \right ]\\
            &= 
        \iota^{-4} \normiii[\Bigg]{\frac{1}{N} \sum_{n \in [N]}  \log \left ( (\by_{n,t}^N)^\top \bmu_{n,t}^N(\bw_n,\theta) \right )- \mathbb{E} \left [  \log \left ( (\by_t^\infty)^\top \bmu_t^\infty(\bw^\infty,\theta) \right )\right ]}_4^4\\
        &\leq 
        6 \frac{2^4}{\iota^{4}} N^{-2} \left [ \chi_{t} + \abs{\log(m_t)} (e_t + 1)\right ]^4,
        \end{split}
    \end{equation}
    where the last bound follows from \eqref{eq:final_bound_borel_cantelli}. We then conclude that:
    \begin{equation}
        \sum_{N=1}^\infty \mathbb{P} \left ( \abs{\frac{1}{N} \sum_{n \in [N]}  \log \left ( (\by_{n,t}^N)^\top \bmu_{n,t}^N(\bw_n,\theta) \right )- \mathbb{E} \left [  \log \left ( (\by_t^\infty)^\top \bmu_t^\infty(\bw^\infty,\theta) \right )\right ]} > \iota \right ) 
        < \infty,
    \end{equation}
    hence by the Borel-Cantelli lemma:
    \begin{equation}
        \frac{1}{N} \sum_{n \in [N]}  \log \left ( (\by_{n,t}^N)^\top \bmu_{n,t}^N(\bw_n,\theta) \right )- \mathbb{E} \left [  \log \left ( (\by_t^\infty)^\top \bmu_t^\infty(\bw^\infty,\theta) \right )\right ] \to 0,
    \end{equation}
    as $N\to \infty$, $\mathbb{P}$-almost surely, or equivalently $\frac{1}{N}\ell^N_t(\theta)$ converges to $\mathbb{E} \left [  \log \left ( (\by_t^\infty)^\top \bmu_t^\infty(\bw^\infty,\theta) \right )\right ]$ as $N \to \infty$,  $\mathbb{P}$-almost surely.

    For the final statement of the proposition it is enough to observe that as $\by_t^\infty$ is a one-hot encoding vector:
    $$
    \log \left ( (\by_t^\infty)^\top \bmu_t^\infty(\bw^\infty,\theta) \right )
    =
    (\by_t^\infty)^\top \log \left ( \bmu_t^\infty(\bw^\infty,\theta) \right )
    $$
    under the convention $0\log0=0$, and from Proposition \ref{prop:joint_rand_aymptoticCAL_distribution} we have:
    $$
    \by_t^\infty|\by_{1:t-1}^\infty,\bw^\infty \sim \mbox{Cat}\left ( \cdot | \bmu_{t}^\infty(\bw^\infty,\theta^\star) \right ).
    $$ 
    Hence, by the tower rule:
    \begin{equation}
    \begin{split}
        \mathbb{E} \left [ (\by_t^\infty)^\top  \log \left ( \bmu_t^\infty(\bw^\infty,\theta) \right )\right ] 
        &=
        \mathbb{E} \left \{ \mathbb{E} \left [ (\by_t^\infty)^\top  \log \left ( \bmu_t^\infty(\bw^\infty,\theta) \right )| \by_{1:t-1}^\infty,\bw^\infty \right ] \right \} \\
        &=
        \mathbb{E} \left \{ \mathbb{E} \left [ (\by_t^\infty)^\top | \by_{1:t-1}^\infty,\bw^\infty \right ] \log \left ( \bmu_t^\infty(\bw^\infty,\theta) \right ) \right \} \\
        &=
        \mathbb{E} \left [ \bmu_t^\infty(\bw^\infty,\theta^\star)^\top  \log \left ( \bmu_t^\infty(\bw^\infty,\theta) \right )\right ].
    \end{split}
    \end{equation}
\end{proof}

Because of Proposition \ref{prop:convergence_log_likelihood}, we can conclude that the CAL has a contrast function which is an expected Kullback-Leibler divergence, as in the following theorem. Recall from \eqref{rec:rand_asympt_CAL_rec_full} the definition $\bmu_t^\infty(\bw^\infty,\theta)$.
\begin{theorem} \label{thm:constrast_function_conv}
    Under assumptions \ref{ass:compactness_continuity},\ref{ass:w_iid},\ref{ass:HMM_support},\ref{ass:eta_structure},\ref{ass:kernel_continuity}, for any $T\geq 1$ and $\theta \in \Theta$ let:
    $$
    \ell_{1:T}^N(\theta) \coloneqq \sum_{t=1}^T \ell_{t}^N(\theta)
    $$
    then:
    $$
        \frac{\ell_{1:T}^N(\theta)}{N} -  \frac{\ell_{1:T}^N(\theta^\star)}{N}
        \to
        \mathcal{C}_T(\theta,\theta^\star),
    $$
    as $N\to\infty$, $\mathbb{P}$-almost surely, where:
    \begin{equation}\label{eq:contrast_function}
        \mathcal{C}_T(\theta,\theta^\star)
        \coloneqq -\sum_{t=1}^T \mathbb{E} \left \{ \mathbf{KL} \left [ \mbox{Cat} \left (\cdot | \bmu_t^\infty(\bw^\infty,\theta^\star) \right )|| \mbox{Cat} \left (\cdot | \bmu_t^\infty(\bw^\infty,\theta) \right ) \right ] \right \}.
    \end{equation}
    Moreover:
    \begin{equation}\label{eq:datagenpar_maximise}
        \theta^\star \in \Theta^\star \coloneqq \argmax_{\theta \in \Theta} \mathcal{C}_T(\theta,\theta^\star).
    \end{equation}
\end{theorem}

\begin{proof}
    Because of Proposition \ref{prop:convergence_log_likelihood}:
    \begin{equation}
    \begin{split}
         &\frac{\ell_{1:T}^N(\theta)}{N} -  \frac{\ell_{1:T}^N(\theta^\star)}{N}\to \sum_{t=1}^T \int \mathbb{E} \left [\left. \bmu_t^\infty(\bw^\infty,\theta^\star)^\top \log \left ( \frac{\bmu_t^\infty(\bw^\infty,\theta)}{\bmu_t^\infty(\bw^\infty,\theta^\star)} \right )\right| \bw^\infty =w \right ]\Gamma(d w),
    \end{split}
    \end{equation}
    as $N\to \infty$, $\mathbb{P}$-almost surely, where we  notice that:
    \begin{equation}
    \begin{split}
        &\mathbb{E} \left [\left. \bmu_t^\infty(\bw^\infty,\theta^\star)^\top \log \left ( \frac{\bmu_t^\infty(\bw^\infty,\theta)}{\bmu_t^\infty(\bw^\infty,\theta^\star)} \right )\right| \bw^\infty =w \right ]\\
        &= 
        \mathbb{E} \left \{\left.
        \mathbb{E} \left [\left. \bmu_t^\infty(\bw^\infty,\theta^\star)^\top \log \left ( \frac{\bmu_t^\infty(\bw^\infty,\theta)}{\bmu_t^\infty(\bw^\infty,\theta^\star)} \right )\right| \by_{1:t-1}^\infty, \bw^\infty =w \right ]
        \right| \bw^\infty =w \right \}
        \\
        &= 
        - \mathbb{E} \left \{ \mathbf{KL} \left [ \mbox{Cat} \left (\cdot | \bmu_t^\infty(\bw^\infty,\theta^\star) \right )|| \mbox{Cat} \left (\cdot | \bmu_t^\infty(\bw^\infty,\theta) \right ) \right ] | \bw^\infty =w  \right \},
    \end{split}
    \end{equation}
    hence:
    \begin{equation}
        \begin{split}
         &\frac{\ell_{1:T}^N(\theta)}{N} -  \frac{\ell_{1:T}^N(\theta^\star)}{N} 
         \\
         &\to
         -\sum_{t=1}^T \int \mathbb{E} \left \{ \mathbf{KL} \left [ \mbox{Cat} \left (\cdot | \bmu_t^\infty(\bw^\infty,\theta^\star) \right )|| \mbox{Cat} \left (\cdot | \bmu_t^\infty(\bw^\infty,\theta) \right ) \right ] | \bw^\infty =w  \right \} \Gamma(d w),
        \end{split}
    \end{equation}
    as $N\to \infty$, $\mathbb{P}$-almost surely, from which we conclude the first part of the proof as:
    \begin{equation}
        \mathcal{C}_T(\theta,\theta^\star)
        = -\sum_{t=1}^T \int \mathbb{E} \left \{ \mathbf{KL} \left [ \mbox{Cat} \left (\cdot | \bmu_t^\infty(\bw^\infty,\theta^\star) \right )|| \mbox{Cat} \left (\cdot | \bmu_t^\infty(\bw^\infty,\theta) \right ) \right ] | \bw^\infty =w  \right \} \Gamma(d w).
    \end{equation}
    The KL-divergence is always greater than or equal to zero, and equal to zero if and only if the two distributions are equal. Hence the maximal value of the negative KL-divergence is zero, and we have:
    \begin{equation}
        \theta^\star \in \Theta^\star = \argmax_{\theta \in \Theta} \mathcal{C}_T(\theta,\theta^\star),
    \end{equation}
    which concludes the proof.
\end{proof}

\paragraph{Uniform almost sure convergence.}

Theorem \ref{thm:constrast_function_conv} proves the convergence $ \frac{\ell_{1:T}^N(\theta)}{N} -  \frac{\ell_{1:T}^N(\theta^\star)}{N} \to \mathcal{C}_T(\theta,\theta^\star)$ \textit{pointwise} in $\theta$. In order to make statements about the convergence of the maximizer of $\ell_{1:T}^N(\theta)$ in the same vein as \cite{whitehouse2023consistent}, we must show this convergence is uniform. To proceed we will use the following results.

\begin{definition}
    Let $(\mathcal{H}_N)_{N \geq 1}$ be a sequence of random functions $\mathcal{H}_N:\theta\in\Theta \mapsto \mathcal{H}_N(\theta) \in \mathbb{R}$ where $\Theta$ is a metric space. We say that $(\mathcal{H}_N)_{N \geq 1}$ are stochastically equicontinuous if there exists an event $E$ of probability $1$, such that for all $\iota>0$ and $\omega \in E$, there exists $N(\omega)$ and $\delta>0$ such that $N>N(\omega)$ implies:
    \begin{equation*}
        \sup_{|\theta_1-\theta_2|<\delta}| \mathcal{H}_N(\omega, \theta_1) - \mathcal{H}_N(\omega, \theta_2)| < \iota.
    \end{equation*}
\end{definition}

\begin{lemma}\label{stochascoli}
    Assume $\Theta$ is a compact metric space and let $(\mathcal{H}_N)_{N \geq 1}$ be a sequence of random functions $\mathcal{H}_N:\theta\in\Theta \rightarrow \mathcal{H}_N(\theta)\in\mathbb{R}$. If there exists a continuous function $\mathcal{H}$ such that for all $\theta \in \Theta$ we have  $| \mathcal{H}_N(\theta) - \mathcal{H}(\theta)| \overset{a.s.}{\rightarrow} 0$, and $(\mathcal{H}_N)_{N \geq 1}$ are stochastically equicontinuous, then:
    \begin{equation*}
        \sup_{\theta \in \Theta}| \mathcal{H}_N(\theta) - \mathcal{H}(\theta)| \overset{a.s.}{\rightarrow} 0.
    \end{equation*}
    That is $\mathcal{H}_N(\theta)$ converges to $\mathcal{H}(\theta)$ almost surely as $N \rightarrow \infty$, uniformly in $\theta$.
\end{lemma}

\begin{proof}
    See \cite{andrews1992generic}.
\end{proof}

\begin{lemma} \label{lem:stochequnfiorm}
    Under Assumptions \ref{ass:compactness_continuity},\ref{ass:w_iid},\ref{ass:HMM_support},\ref{ass:eta_structure},\ref{ass:kernel_continuity}, 
    
       \begin{equation}
        \sup_{\theta\in\Theta}\left|\frac{\ell_{1:T}^N(\theta)}{N} -  \frac{\ell_{1:T}^N(\theta^\star)}{N}  
       - \mathcal{C}_T(\theta,\theta^\star)\right|\to 0, 
    \end{equation}
    $\mathbb{P}$-almost surely.
\end{lemma}
\begin{proof}
    By Theorem \ref{thm:constrast_function_conv} we have pointwise convergence. Hence, by Lemma \ref{stochascoli} it is enough to show that
    \begin{equation}\label{eq:finite_cf}
        \mathcal{C}_{t}^N(\theta) \coloneqq  \frac{1}{N} \sum_{n \in [N]} (\by_{n,t}^N)^\top  \log \left ( \bmu_{n,t}^N(\bw_n,\theta) \right ),
    \end{equation}
    is stochastically equicontinuous. Define 
    \begin{equation}
        \mathcal{C}_{t}^{\infty}(\theta) \coloneqq \int \mathbb{E} \left [ \bmu_t^\infty(\bw^\infty,\theta^\star) ^\top  \log \left ( \bmu_t^\infty(\bw^\infty,\theta) \right )| \bw^\infty = \bw \right ] \Gamma(d \bw).
    \end{equation}
    Let $\theta_1,\theta_2 \in \Theta$ and consider the decomposition
    \begin{align} \label{eq:stocheqdecomp1}
        | \mathcal{C}_{t}^N(\theta_1) -  \mathcal{C}_{t}^N(\theta_2)| \leq &| \mathcal{C}_{t}^N(\theta_1) -  \mathcal{C}_{t}^{\infty}(\theta_1)| \\\label{eq:stocheqdecomp2}
        + &| \mathcal{C}_{t}^N(\theta_2) -  \mathcal{C}_{t}^{\infty}(\theta_2)| \\\label{eq:stocheqdecomp3}
        + &| \mathcal{C}_{t}^{\infty}(\theta_1) -  \mathcal{C}_{t}^{\infty}(\theta_2)|.
    \end{align}
    Note that all of these quantities are well defined by Theorem \ref{thm:CAL_well_definess} and by the fact that $\bmu^\infty_t(\bw^\infty,\theta)\neq 0$ because of Proposition \ref{prop:rand_aympt_CAL_well_definess}.
    
    Let $E\subset \Omega$ such that $\mathbb{P}(E) = 1$. Let $\theta,\theta^\star \in \Theta$, $\omega \in E$, and $\iota>0$. By Theorem \ref{thm:constrast_function_conv} we have almost sure convergence of \eqref{eq:stocheqdecomp1} and \eqref{eq:stocheqdecomp2} to $0$, hence there exists an ${N}(\omega)$ such that for all $N>{N}(\omega)$ these terms are bounded by $\iota/3$. 

    It remains to show that there exists a $\delta$ such that for $|\theta_1-\theta_2|<\delta$ implies that \eqref{eq:stocheqdecomp3} is bounded by $\iota/3$. This follows directly from noticing that $\mathcal{C}_{t}^{\infty}(\theta)$ comprises a composition of continuous functions of our model quantities $p_0, K,$ and $G$, which are themselves continuous functions of $\theta$ because of Assumption \ref{ass:compactness_continuity}.
\end{proof}

\begin{theorem}\label{thm:consistency}
    Let \ref{ass:compactness_continuity},\ref{ass:w_iid},\ref{ass:HMM_support},\ref{ass:eta_structure},\ref{ass:kernel_continuity} hold and let $\hat{ {\theta}}_N$ be a maximizer of $  \ell_{1:T}^N({\theta})$. Then $\hat{ {\theta}}_N$ converges to $\Theta^\star$ as $N \to \infty$,  $\mathbb{P}$-almost surely.
\end{theorem}

\begin{proof}
    The proof follows in the same manner as that of Theorem 1 in \cite{whitehouse2023consistent}, we include it for completeness.

    Let $\mathcal{C}_T( \theta^\star,  \theta)$ be as defined by \eqref{eq:contrast_function} and let $\mathcal{C}_T^N(\theta) = \sum_{t=1}^T \mathcal{C}_t^{N}(\theta)$ be as in Equation \eqref{eq:finite_cf}.  We have that $\mathcal{C}_T^{N}(\hat{ \theta}_N) \geq \mathcal{C}_T^{N}({ \theta})$ for all $ \theta \in \Theta^\star$. Furthermore  ${\mathcal{C}_T( \theta^\star,  \theta^\star) - \mathcal{C}_T( \theta^\star,  \theta) \geq 0}$ for all $ \theta \in \Theta$. We can combine these inequalities to obtain:
    \begin{equation} \label{supconv}
        \begin{aligned}
        0 &\leq \mathcal{C}_T( \theta^\star,  \theta^\star) - \mathcal{C}_T( \theta^\star, \hat{ \theta}_n) \\
        & =  \mathcal{C}_T( \theta^\star,  \theta^\star)  - \mathcal{C}_T^N(  \theta^\star)  +\mathcal{C}_T^N(  \theta^\star)
        -\mathcal{C}_T^N(  \hat{ \theta}_n)
        + \mathcal{C}_T^N(  \hat{ \theta}_n)
        - \mathcal{C}_T( \theta^\star, \hat{ \theta}_n) \\
        &\leq 2\sup_{ \theta \in \Theta}\left|\mathcal{C}_T( \theta^\star,  \theta) - \mathcal{C}_T^N( \theta) \right| \rightarrow 0 \quad \mathbb{P} \text{-almost surely,}
        \end{aligned}
    \end{equation}
    by Lemma \ref{lem:stochequnfiorm}. Hence $\mathcal{C}_T( \theta^\star, \hat{ \theta}_n) \rightarrow \mathcal{C}_T( \theta^\star, { \theta^\star})$ $\mathbb{P}$-almost surely.

    Now assume for purposes of contradiction that there is some positive probability that $\hat{ \theta}_n$ does not converge to the set $\Theta^\star$, i.e. assume that there is an event $E\subset \Omega$ with $\mathbb{P}(E)>0$ such that for all $\omega \in E$ there exists a $\delta>0$ such that for infinitely many $n \in \mathbb{N}$ we have ${\hat{ \theta}_n(\omega)}$  is not in the open neighborhood $B_\delta(\Theta^\star) = \{ \theta \in \Theta: \exists  \theta' \in \Theta^\star: \| \theta -  \theta' \|<\delta\}$.
    Since $\Theta$ is compact, the set $B_\delta(\Theta^\star)^c = \Theta \setminus B_\delta(\Theta^\star)$ is closed, bounded, and therefore compact. Furthermore, $\mathcal{C}_T( \theta^\star, \theta)$ is continuous in $ \theta$. By the extreme value theorem this means that there exists a $ \theta'\in B_\delta(\Theta^\star)^c$ such that for all  $ \theta \in B_\delta(\Theta^\star)^c$:
    $$ 
    \mathcal{C}_T( \theta^\star, \theta) \leq \mathcal{C}_T( \theta^\star, \theta')
    $$
    Furthermore, since $ \theta' \notin \Theta^\star$ there exists  $\iota>0$ such that:
    $$ 
    \mathcal{C}_T( \theta^\star, \theta') < \mathcal{C}_T( \theta^\star, \theta^\star) - \iota.
    $$
    By our assumption we have for each $\omega \in E$ there are infinitely many $n \in \mathbb{N}$ such that $\hat{ \theta}_n(\omega) \in B_\delta(\Theta^\star)^c$. But this implies that for each $\omega \in E$ there are infinitely many $n \in \mathbb{N}$ such that:
    $$ 
    \mathcal{C}_T( \theta^\star, \hat{ \theta}_n(\omega)) \leq \mathcal{C}_T( \theta^\star, \theta') < \mathcal{C}_T( \theta^\star, \theta^\star) - \iota,
    $$
    $$
    \implies |\mathcal{C}_T( \theta^\star,  \theta^\star) - \mathcal{C}_T( \theta^\star, \hat{ \theta}_n(\omega))|> \iota,
    $$
    which contradicts \eqref{supconv}. Hence we must have that $\hat{ \theta}_n$ converges to the set $\Theta^\star$ $\mathbb{P}$-almost surely.
\end{proof}

\subsubsection{Identifiability} 

\begin{definition}
    Let $\{\by_{t}^\infty\}_{t\geq 1}$ be generated according to the process defined by equations \eqref{eq:limiting_process} 
    with data-generating parameter $\theta^\star \in \Theta$. Denote the law of $\{\by_{t}^\infty\}_{t\geq 1}$ conditional on $\bw^\infty = w$ with $\mathbb{P}_{\infty}^{\theta^\star,w}$.
\end{definition}

\begin{lemma}\label{lem:identifiability}
    Let $\theta^\star \in \Theta$. 
    For any $\theta_1,\theta_2 \in \Theta^\star \coloneqq \argmax_{\theta \in \Theta} \mathcal{C}_T(\theta,\theta^\star)$ we have that $\mathbb{P}_{\infty}^{\theta_1, w} = \mathbb{P}_{\infty}^{\theta_2,w}$ for $\Gamma$-almost all $w\in \mathbb{W}$.
\end{lemma}

\begin{proof}
Recall from \eqref{eq:contrast_function} the definition of the contrast function:
    \begin{equation}
        \mathcal{C}_T(\theta,\theta^\star)
        \coloneqq -\sum_{t=1}^T \mathbb{E} \left \{ \mathbf{KL} \left [ \mbox{Cat} \left (\cdot | \bmu_t^\infty(\bw^\infty,\theta^\star) \right )|| \mbox{Cat} \left (\cdot | \bmu_t^\infty(\bw^\infty,\theta) \right ) \right ] \right \}.
    \end{equation}
    where $\theta^\star$ is the DGP and $\theta$ is a candidate parameter. With $\theta^\star$ fixed, we want to characterize the set of maximizers of $\mathcal{C}_T(\theta,\theta^\star)$ in the first argument, i.e. the set $\Theta^\star \coloneqq \argmax_{\theta \in \Theta} \mathcal{C}_T(\theta,\theta^\star)$. 
  Let $w \in \mathbb{W}$ and consider the conditional expectation
    \begin{equation}
         \mathbb{E} \left \{ \mathbf{KL} \left [ \mbox{Cat} \left (\cdot | \bmu_t^\infty(\bw^\infty,\theta^\star) \right )|| \mbox{Cat} \left (\cdot | \bmu_t^\infty(\bw^\infty,\theta) \right ) \right ] \right | \bw^\infty = w\},
    \end{equation}
    this is the expectation of a KL-divergence between categorical distributions parameterized by the random vectors $\bmu_t^\infty(\bw^\infty,\theta^\star)$ and $\bmu_t^\infty(\bw^\infty,\theta)$. Recall that according to the recursive definition of these vectors given in \eqref{rec:rand_asympt_CAL_rec_full}, conditional on $\bw^\infty = w$  these vectors are random solely as functions of $\by_{1:t-1}^\infty \sim \mathbb{P}_{\infty}^{\theta^\star,w}$, with no other sources of stochasticity. Hence this conditional expectation is a summation over $y_{1:t-1} \in \mathbb{O}_{M+1}^{t-1}$:
    \begin{align} \label{eq:contrast_condexp}
        &\mathbb{E} \left \{ \mathbf{KL} \left [ \mbox{Cat} \left (\cdot | \bmu_t^\infty(\bw^\infty,\theta^\star) \right )|| \mbox{Cat} \left (\cdot | \bmu_t^\infty(\bw^\infty,\theta) \right ) |\bw^\infty = w \right ] \right \} \\ 
        &= 
        \sum_{y_{1:t-1}} \mathbb{P}_{\infty}^{\theta^\star,w}(\by^\infty_{1:t-1}=y_{1:t-1})\mathbf{KL} \left [ \mbox{Cat} \left (\cdot | \mu_t^\infty(w,\theta^\star) \right )|| \mbox{Cat} \left (\cdot | \mu_t^\infty(w,\theta) \right ) \right ]
    \end{align}
    where the unbolded $\mu_t^\infty(w,\theta^\star)$ and $\mu_t^\infty(w,\theta)$ (similarly to the process defined by Equation \eqref{eq:limiting_process}) are calculated with the recursions: 
\begin{equation}\label{rec:pathwise_CAL}
    \begin{aligned}
    &\pi_{0}^\infty(w,\theta) \coloneqq p_{0}(w,\theta),\\
    &\bar{\eta}_{t-1}^\infty(w,\theta) =  \int d(w,\widetilde{w},\theta)^\top \bar{\pi}_{t-1}^\infty(\widetilde{w},\theta) \Gamma(d \widetilde{w}),\\
    &\pi_{t|t-1}^\infty(w,\theta)  \coloneqq  \left [ \pi_{t-1}^\infty(w,\theta)^{\top} K_{\bar{\eta}_{t-1}^\infty(w,\theta)} (w, \theta)\right ]^\top, \\
    &\mu_{t}^\infty(w,\theta)  \coloneqq  \left [ \pi_{t|t-1}^\infty(w,\theta)^{\top} G(w,\theta) \right ]^\top,\\
    &\pi_{t}^\infty(w,\theta)  \coloneqq \pi_{t|t-1}^\infty(w,\theta) \odot \left \{ \left [   G(w,\theta) \oslash \left ( 1_M \mu_{t}^\infty(w,\theta)^\top \right ) \right ] y_{t}  \right \},
    \end{aligned}
\end{equation} 
where the dependence of various quantities on $y_{1:t}$ is not shown in the notation.

By properties of the KL-divergence we have that
    \begin{equation}
        \mathbf{KL} \left [ \mbox{Cat} \left (\cdot | \mu_t^\infty(w,\theta^\star) \right )|| \mbox{Cat} \left (\cdot | \mu_t^\infty(w,\theta) \right )\right]=0 \iff \mu_t^\infty(w,\theta)=\mu_t^\infty(w,\theta^\star),
    \end{equation}
    and
    \begin{equation}
        \mathbf{KL} \left [ \mbox{Cat} \left (\cdot | \mu_t^\infty(w,\theta^\star) \right )|| \mbox{Cat} \left (\cdot | \mu_t^\infty(w,\theta) \right )\right] >0  \iff \mu_t^\infty(w,\theta) \neq \mu_t^\infty(w,\theta^\star),
    \end{equation}
    which makes it clear that $\theta^\star \in \Theta^\star$, as already mentioned in \eqref{eq:datagenpar_maximise}. It then follows from Equation \eqref{eq:contrast_condexp} that if $\theta \in \Theta^\star$ then for 
    $\mathbb{P}_\infty^{\theta^\star,w}$-almost all paths ${y}_{1:t-1} \in 
    \mathbb{O}_{M+1}^{t-1}$ we must have $\mu_t^\infty(w,\theta)=\mu_t^\infty(w,\theta^\star)$. Now, considering the full contrast function we have by the tower rule:
    \begin{align}
        \mathcal{C}_T(\theta,\theta^\star)
        &\coloneqq -\sum_{t=1}^T \mathbb{E} \left \{ \mathbf{KL} \left [ \mbox{Cat} \left (\cdot | \bmu_t^\infty(\bw^\infty,\theta^\star) \right )|| \mbox{Cat} \left (\cdot | \bmu_t^\infty(\bw^\infty,\theta) \right ) \right ] \right \} \\
        &= -\sum_{t=1}^T \mathbb{E} \left \{ \mathbb{E} \left \{ \mathbf{KL} \left [ \mbox{Cat} \left (\cdot | \bmu_t^\infty(\bw^\infty,\theta^\star) \right )|| \mbox{Cat} \left (\cdot | \bmu_t^\infty(\bw^\infty,\theta) \right ) \right ]| \bw^{\infty}
        \right \} \right\}.
    \end{align}
    Here we are summing over $t\in [T]$ and taking the expectation of \eqref{eq:contrast_condexp} over $\bw^\infty \sim \Gamma$. It therefore follows that if $\theta \in \Theta^\star$ then for all $t\in [T]$, $\Gamma-$ almost all $w\in \mathbb{W}$,  and $\mathbb{P}_\infty^{\theta^\star,w}$-almost all paths ${y}_{1:t-1} \in 
    \mathbb{O}_{M+1}^{t-1}$,  we must have that $\mu_t^\infty(w,\theta)=\mu_t^\infty(w,\theta^\star)$.
    
    
    To complete the proof, recall that by Proposition \ref{prop:joint_rand_aymptoticCAL_distribution} if $\by_{1:T}\sim \mathbb{P}_\infty^{\theta^\star,w}$  we have the conditional distributions  $\by_t^\infty|\by_{1:t-1}^\infty,\bw^\infty \sim \mbox{Cat}\left ( \cdot | \bmu_{t}^\infty(\bw^\infty,\theta^\star) \right )$. Hence for $\Gamma-$almost all $w \in \mathbb{W}$,  $\mathbb{P}_\infty^{\theta^\star,w}$-almost all paths ${y}_{1:T} \in 
    \mathbb{O}_{M+1}^{T}$ and any $\theta_1,\theta_2 \in \Theta^\star$ we have 
    \begin{equation}
        \begin{split}
            \mathbb{P}_{\infty}^{\theta_1,w}({ \by}_{1:T}^\infty = y_{1:T} ) 
            &= 
            \prod_{t=1}^T \mathbb{P}_{\infty}^{\theta_1,w}({ \by}_{t}^\infty = y_t|{\by}_{1:t-1}^\infty=y_{1:t-1} )\\ 
            &=
            \prod_{t=1}^T \mbox{Cat}\left ( y_t | \mu_{t}^\infty(w,\theta_1) \right ) \\
            &=
            \prod_{t=1}^T \mbox{Cat}\left ( y_{t} | \mu_{t}^\infty(w,\theta^\star) \right ) \\
            &=
            \prod_{t=1}^T \mbox{Cat}\left ( y_{t} | \mu_{t}^\infty(w,\theta_2) \right ) \\
            &= 
            \prod_{t=1}^T \mathbb{P}_{\infty}^{\theta_2,w}({ \by}_{t}^\infty = y_t|{\by}_{1:t-1}^\infty=y_{1:t-1} ) \\
            &= \mathbb{P}_{\infty}^{\theta_2,w}({ \by}_{1:T}^\infty = y_{1:T} ),
        \end{split}
    \end{equation}
    where we used that $\theta \in \Theta^\star$ implies $\mu_{t}^\infty(w,\theta)=\mu_{t}^\infty(w,\theta^\star)$.  Assumption \ref{ass:HMM_support} ensures that the support of $\mathbb{P}_{\infty}^{\theta,w}$ does not depend on $\theta$, and hence $\mathbb{P}_{\infty}^{\theta_1,w} = \mathbb{P}_{\infty}^{\theta_2,w}$ for $\Gamma-$almost all $w \in \mathbb{W}$.
\end{proof}

\section{Experiments}

\subsection{Computational considerations} \label{sec:supp_computational_cost}

Algorithm \ref{alg:CAL} requires only $N$ repetitions of simple linear algebra operations on $M$-dimensional vectors and $M\times M$ matrices at each time step. The resulting computational cost is $\mathcal{O}(T N M^2 + T N \mathcal{C}_\eta(N))$, where $T$ and $N$ arise from recursive operations over time and individuals, $M^2$ accounts for vector-matrix operations, and $\mathcal{C}_\eta(N)$ is the cost of evaluating the $\eta$ function for a fixed population size $N$.

\begin{table}[httb!]
    \centering
    \caption{Memory and running time summary. $\mathcal{C}_\eta(N)$ is the model specific cost of evaluating $\eta$ for a fixed population size $N$. For the SMC, $P$ is the number of particles, while $\mathcal{C}_q(M,1)$ is the cost of computing the parameters of the categorical distribution $q$ which is used as a proposal (see \cite{rimella2022approximating} for an example). The memory requirement assumes that everything is stored over time steps.}
    \label{table:comp_cost}
    \begin{tabular}{llll}
    \hline
      & Memory & Running time\\
    \hline
    Data simulation & $\mathcal{O}(T N)$ & $\mathcal{O}(T N \mathcal{C}_{Cat}(M,1) + T N \mathcal{C}_\eta(N))$\\
    \hline
    SMC & $\mathcal{O}(T N P)$ & $\mathcal{O}(T N \mathcal{C}_{Cat}(M,P) + T N P \mathcal{C}_\eta(N) + T N \mathcal{C}_q(M,P))$\\
    CAL & $\mathcal{O}(T N M )$ & $\mathcal{O}(T N M^2 + T N \mathcal{C}_\eta(N))$\\
    \hline
    \end{tabular}
\end{table}

The cost $\mathcal{C}_\eta(N)$ is model-specific. For instance, for the motivating example in sections \ref{sec:motivating_example_first} and \ref{sec:motivating_example_second}, we have a cost of $\mathcal{C}_\eta(N)=\mathcal{O}(N)$ for the homogeneous-mixing case and a cost of $\mathcal{C}_\eta(N)=\mathcal{O}(N^2)$ for the heterogeneous-mixing scenario or other dense spatial interactions \citep{jewell2009bayesian,rimella2022inference}. Remark that the cost $\mathcal{C}_\eta(N)$ is inherent to the model, rather than a feature of the CAL, as even simulating data from the model would incur in a cost of the same order: $\mathcal{O}(T N \mathcal{C}_{Cat}(M,1) + T N \mathcal{C}_{\eta}(N))$, where $\mathcal{C}_{Cat}(M,1)$ is the cost of simulating a single draw from a categorical distribution with $M$ categories. Algorithm \ref{alg:CAL} requires computing probability vectors $\bpi_{n,t|t-1}$, $\bmu_{n,t}$, and $\bpi_{n,t}$ for each individual $n$ at every time step. If all these vectors are stored over $T$ time steps, the memory requirement is $\mathcal{O}( T N M)$. The memory requirement of storing the data is $\mathcal{O}(T N)$. The resulting computational cost is significantly cheaper than SMC, which would require many simulations of each individual. Table \ref{table:comp_cost} provides an SMC-CAL complexity comparison summary. Remark that if the user is only interested in the likelihood computation, it is possible to reduce the memory consumption of the SMC to $\mathcal{O}(NP)$, and similarly to $\mathcal{O}(NM)$ for the CAL.


\subsection{Hamiltonian Monte Carlo for homogeneous-mixing SIS} \label{sec:app_exp_HMC}

We use the parameter settings $p_0=0.01, \beta = 0.2, \bb_I = 0.5, \bb_S = 1.0, \gamma=0.1, \bb_R=-0.5, q_S=0.2, q_I=0.5, q_{Se}=0.9, q_{Sp}=0.95$. We treat $p_0,q_{Se},q_{Sp}$ as known. To create a warm start for our HMC sampler we consider $100$ random initializations, where each parameter is drawn from a standard Gaussian, and perform $1000$ steps of Adam optimization with a learning rate of $0.1$. Across the $100$ initializations, we choose the resulting parameter values with the highest CAL log-likelihood as a warm-start for our HMC sampler which we then run for $200000$ iterations. We use uninformative independent Gaussian priors each with mean $0$ and standard deviation $10$.

The learning rate is adapted every $1000$ iteration to fall within the acceptance range $[0.55,0.75]$. Specifically, if the acceptance range is lower than $0.55$ then for the next $1000$ iterations it is decreased by $35\%$, while if it is higher than $0.75$ it is increased by $35\%$, otherwise if it falls within the range is kept stable. We run the 201000 HMC iterations and then we choose only the last iterations in which the acceptance rate was stable. We get acceptance rates within the range $0.60-0.70$, which are close to the optimal acceptance rate of $0.65$ \citep{neal2012mcmc}. We further apply a thinning procedure that select $10000$ samples equally spaced in time (number of iterations).



\subsection{Gradient-based calibration for heterogeneously-mixing SIS models} \label{sec:app_exp_grad}

\paragraph{Model 1}

Here we provide all the details on Model 1. As initial infection probabilities, we consider:
\begin{equation}
    p_0(\bw_n) 
    = 
    \begin{bmatrix} 
    1-p_{0}\\ p_{0} 
    \end{bmatrix},
\end{equation}
where $p_0 \in [0,1]$. We consider an interaction term:
\begin{equation}
    \boeta_{n,t-1}
    =
    \frac{1}{N} \sum_{k \in [N]} \begin{bmatrix}
    0\\
    \exp\{\bc_{k}^\top \bb_I\} \frac{1}{\sqrt{2 \pi \phi^2} {\exp\left\{-\frac{\|\bl_n-\bl_{k}\|^2}{2\phi^2}\right\}}}
    \end{bmatrix}^\top \bx_{k,t-1},
\end{equation}
where $\phi>0$ and $\bb_I \in \mathbb{R}$ as we are considering a single covariate. The $\boeta_{n,t-1}$ is then used in the transition matrix:
\begin{equation}
K_{\boeta_{n,t-1}}(\bw_n) 
    = 
    \begin{bmatrix} 
    \exp{ \left (-h\beta\exp\{\bc_n^\top \bb_{S}\} \boeta_{n,t-1} - h\epsilon  \right ) }
    &
    1 - \exp{ \left (-h\beta\exp\{\bc_n^\top \bb_{S}\} \boeta_{n,t-1} - h \epsilon  \right ) }
    \\ 
    1 - \exp{ \left ( -h\gamma_n \right )}
    &
    \exp{ \left ( -h\gamma_n  \right )}
    \end{bmatrix},
\end{equation}
with $\log \gamma_n = \log \gamma + \bc_n^\top \bb_R$ and where $\beta,\gamma >0$ and $\bb_S, \bb_R \in \mathbb{R}$.

The observation model is given by the matrix:
$$
    G(\bw_n) 
    =
    \begin{bmatrix}
    1 - q_S 
    & 
    q_S q_{Sp}  
    & 
    q_S (1-q_{Sp})\\
    1 - q_I 
    & 
    q_I (1- q_{Se}) 
    & 
    q_I q_{Se}
    \end{bmatrix},
$$
where $q_S, q_{Se}, q_I, q_{Sp} \in [0,1]$.

Figure \ref{fig:spatial_inference} provides a graphical representation of the optimization of Model 1.

\begin{figure}[httb!]
    \centering
    \includegraphics[width=\linewidth]{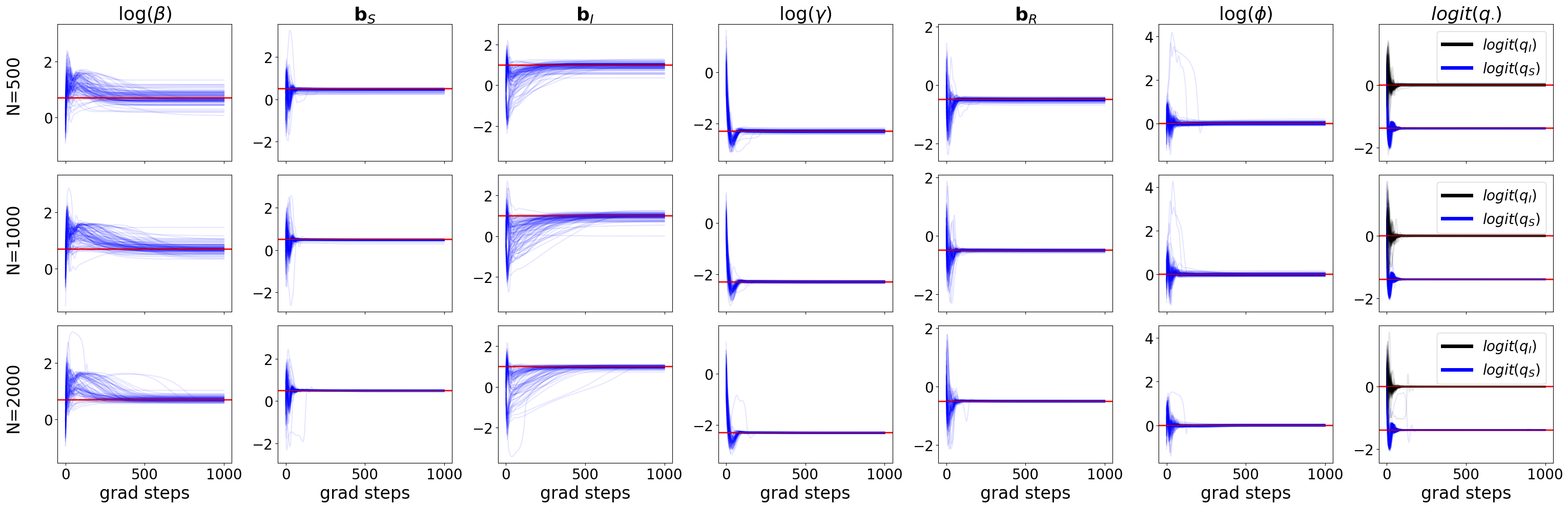}
    \caption{CAL parameters values during optimization for Model 1 from Section \ref{sec:exp_gradient} over the gradient steps of Adam and across different population sizes (from top to bottom). Lines refer to the best out of $10$ optimizations on different datasets. Blue and black colors are used in the same plot to distinguish across parameter components.}
    \label{fig:spatial_inference}
\end{figure}

\paragraph{Model 2}

Here we provide all the details on Model 2. Everything is as in Model 1 but the spatial location is now set to the centroid of the community the individual is in. As a consequence, there is also a computationally cheaper representation of $\boeta_{n,t-1}$, which exploits the fact that some individuals belong to the same community. Denote with ${\bom}^c_i$ the centroid of community $i$ and with $\ba_i$ the set of individuals within the community, i.e. $n \in \ba_i$ if $n$ is in the community $i$. We can then observe that for $n \in {\ba}_j$:
\begin{equation}
\begin{split}
    \boeta_{n,t-1}
    &=
    \frac{1}{N} \sum_{k \in [N]} \begin{bmatrix}
    0\\
    \exp\{\bc_{k}^\top \bb_I\} \frac{1}{\sqrt{2 \pi \phi^2}} {\exp\left\{-\frac{\|\bom_n-\bom_{k}\|^2}{2\phi^2}\right\}}
    \end{bmatrix}^\top \bx_{k,t-1}\\
    &=
    \frac{1}{N} \sum_i \sum_{k \in {\ba}_i} \begin{bmatrix}
    0\\
    \exp\{\bc_{k}^\top \bb_I\} \frac{1}{\sqrt{2 \pi \phi^2}} {\exp\left\{-\frac{\|\bom_n-\bom_{k}\|^2}{2\phi^2}\right\}}
    \end{bmatrix}^\top \bx_{k,t-1}\\
    &=
    \frac{1}{N} \sum_i \sum_{k \in {\ba}_i} \begin{bmatrix}
    0\\
    \exp\{\bc_{k}^\top \bb_I\} \frac{1}{\sqrt{2 \pi \phi^2}} {\exp\left\{-\frac{\|{\bom}^c_{j}-{\bom}^c_{i}\|^2}{2\phi^2}\right\}}
    \end{bmatrix}^\top \bx_{k,t-1},
\end{split}
\end{equation}
from which we observe that for any $n,k \in {\ba}_j$ we have $\boeta_{n,t-1}=\boeta_{k,t-1}$, hence we just need to compute the interaction term for each community. This reduces the computational cost of computing all $\boeta_{n,t-1}$ from $N^2$ to $N$ times the number of communities.

Figure \ref{fig:graph_inference} provides a graphical representation of the optimization of Model 2.

\begin{figure}[httb!]
    \centering
    \includegraphics[width=\linewidth]{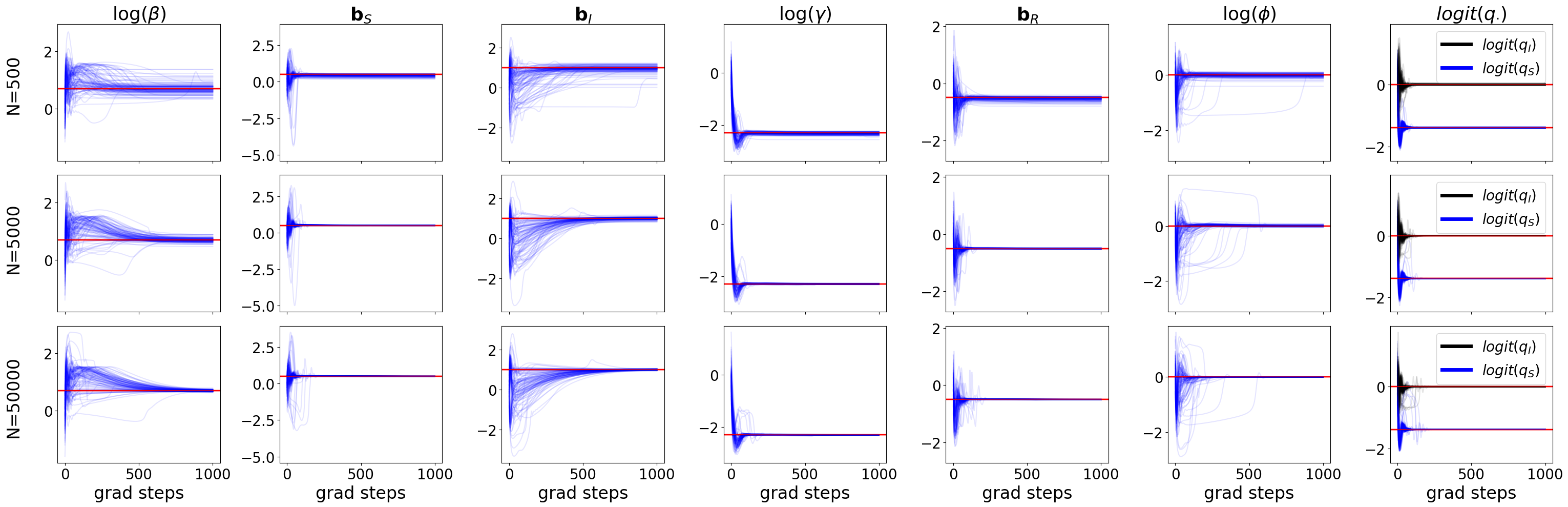}
    \caption{CAL parameters values during optimization for Model 2 from Section \ref{sec:exp_gradient} over the gradient steps of Adam and across different population sizes (from top to bottom). Lines refer to the best out of $10$ optimizations on different datasets. Blue and black colors are used in the same plot to distinguish across parameter components.}
    \label{fig:graph_inference}
\end{figure}

\subsection{Calibration and filtering for heterogeneously-mixing SIR} \label{sec:app_exp_mispec}

\paragraph{Well-specified model}

Here we provide all the details on the correctly specified model from Section \ref{sec:exp_mispec}. As initial infection probabilities, we consider:
\begin{equation}
    p_0(\bw_n) 
    = 
    \begin{bmatrix} 
    1 - p_0 \mathbb{I}\left( \bl_n \in \left[-\infty, 5 \right] \times \left[8, +\infty\right ] \right )\\ p_0 \mathbb{I}\left( \bl_n \in \left[-\infty,5\right] \times \left[8, +\infty\right ] \right )\\
    0
    \end{bmatrix},
\end{equation}
where $p_0 \in [0,1]$. We consider an interaction term:
\begin{equation}
    \boeta_{n,t-1}
    =
    \frac{1}{N} \sum_{k \in [N]} \begin{bmatrix}
    0\\
    \exp\{\bc_{k}^\top \bb_I\} \frac{\exp\left\{-\frac{\|\bl_n-\bl_{k}\|^2}{2\phi}\right\}}{\sqrt{2 \pi \phi}}
    \end{bmatrix}^\top \bx_{k,t-1},
\end{equation}
where $\phi>0$ and $\bb_I \in \mathbb{R}$ as we are considering a single covariate. The $\boeta_{n,t-1}$ is then used in the transition matrix:
\begin{equation}
\begin{split}
&K_{\boeta_{n,t-1}}(\bw_n) \\
&= 
    \begin{bmatrix} 
    \exp{ \left (-h\beta\exp\{\bc_n^\top \bb_{S}\} \boeta_{n,t-1} - h\epsilon  \right ) }
    &
    1 - \exp{ \left (-h\beta\exp\{\bc_n^\top \bb_{S}\} \boeta_{n,t-1} - h \epsilon  \right ) }
    &
    0
    \\ 
    0
    &
    \exp{ \left ( -h\gamma_n  \right )}
    &
    1 - \exp{ \left ( -h\gamma_n \right )}\\
    0
    &
    0
    &
    1
    \end{bmatrix},
\end{split}
\end{equation}
with $\log \gamma_n = \log \gamma + \bc_n^\top \bb_R$ and where $\beta,\gamma >0$ and $\bb_S, \bb_R \in \mathbb{R}$.

Half of the population is forced to be unobserved and misreporting is not allowed. We call $U$ the set of individuals that are always unobserved, the observation model is then given by:
$$
    G(\bw_n) 
    =
    \begin{bmatrix}
    1 - q_S\mathbb{I}(n \notin U) 
    & 
    q_S\mathbb{I}(n \notin U) 
    & 
    0
    &
    0\\
    1 - q_I\mathbb{I}(n \notin U) 
    & 
    0 
    & 
    q_I\mathbb{I}(n \notin U)
    &
    0\\
    1 - q_R\mathbb{I}(n \notin U) 
    & 
    0 
    & 
    0
    &
    q_R\mathbb{I}(n \notin U)
    \end{bmatrix},
$$
where $q_S, q_I,q_R \in [0,1]$.

\paragraph{Misspecified model}

The initial infection probabilities, the transition matrix, and the observation model are the same as for the well-specified model. However, the interaction term is:
\begin{equation}
    \boeta_{n,t-1}
    =
    \frac{1}{N} \sum_{k \in [N]} \begin{bmatrix}
    0\\
    \exp\{\bc_{k}^\top \bb_I\} \frac{1}{\sqrt{2 \pi \phi^2}}\exp\left\{-\frac{(\|\bom_n-\bom_{k}\|\mathbb{I}(\|\bom_n-\bom_{k}\| \neq 0) + \bar{\bl}_n \mathbb{I}(\|\bom_n-\bom_{k}\| = 0) )^2}{2\phi^2}\right\}
    \end{bmatrix}^\top \bx_{k,t-1}.
\end{equation}
as explained in the main paper. We can notice that the big advantage of this formulation, as for Model 2, is the computational cost indeed:
\begin{equation}
\begin{split}
    \boeta_{n,t-1}
    &=
    \frac{1}{N} \sum_{k \in [N]} \begin{bmatrix}
    0\\
    \exp\{\bc_{k}^\top \bb_I\} \frac{1}{\sqrt{2 \pi \phi^2}}\exp\left\{-\frac{(\|\bom_n-\bom_{k}\|\mathbb{I}(\|\bom_n-\bom_{k}\| \neq 0) + \bar{\bl}_n \mathbb{I}(\|\bom_n-\bom_{k}\| = 0) )^2}{2\phi^2}\right\}
    \end{bmatrix}^\top \bx_{k,t-1}\\
    &=
    \frac{1}{N} \sum_i \sum_{k \in \ba_i} \begin{bmatrix}
    0\\
    \exp\{\bc_{k}^\top \bb_I\} \frac{1}{\sqrt{2 \pi \phi^2}}\exp\left\{-\frac{(\|\bom_n-\bom_{k}\|\mathbb{I}(\|\bom_n-\bom_{k}\| \neq 0) + \bar{\bl}_n \mathbb{I}(\|\bom_n-\bom_{k}\| = 0) )^2}{2\phi^2}\right\}
    \end{bmatrix}^\top \bx_{k,t-1}\\
    &=
    \frac{1}{N} \sum_i \sum_{k \in \ba_i} \begin{bmatrix}
    0\\
    \exp\{\bc_{k}^\top \bb_I\} \frac{1}{\sqrt{2 \pi \phi^2}}\exp\left\{-\frac{(\|\bom_j^c - \bom_{i}^c\|\mathbb{I}(\|\bom_j^c - \bom_{i}^c\| \neq 0) + \bar{\bl}_j^c \mathbb{I}(\|\bom_j^c - \bom_{i}^c\| = 0) )^2}{2\phi^2}\right\}
    \end{bmatrix}^\top \bx_{k,t-1}\\
\end{split}
\end{equation}
where $\bom_i^c$ is the centroid of community $i$, $\bar{\bl}_i^c$ is the mean distance within community $i$, $\ba_i$ is the set of individuals within community $i$, and we assume $n \in \ba_j$.

The considered metrics for comparison are:
\begin{itemize}
    \item cross-entropy loss: 
    $
    -\frac{1}{N T} \sum_{t=1}^T \sum_{n \in [N]} \bx_{n,t}^\top \log (\bpi_{n,t});
    $
    \item accuracy: 
    $
    \left [ \frac{1}{N T} \sum_{t=1}^T \sum_{n \in [N]} \mathbb{I} \left ( \argmax_i \bx_{n,t} = \argmax_i \bpi_{n,t}\right ) \right ] \cdot 100 \%.
    $
\end{itemize}
We now define the baseline classifiers. To build a classifier we need to create a vector of probabilities, which represents the probability of estimating the different states, e.g. for the CAL this is $\pi_{n,t}$. Let us start with the ``Random'' classifier, here our vector of probabilities for estimating is: $\bpi_{n,t}^g = \mathbb{I}\left (\by_{n,t}^{(M+1)}=1 \right ) (1_3 \oslash 3) + \mathbb{I}\left (\by_{n,t}^{(M+1)}=0 \right) \by_{n,t}$, meaning that we estimate at random if the individual is unreported, otherwise we estimate with what is reported. ``Prev. uncertain'' and ``Prev. certain'' are more complicated and we define them via Algorithm \ref{alg:guess}. Here, if the individual is reported, we estimate the individual's state with what is reported, otherwise, we have a confidence parameter $g$ which tells us how confident we are with predicting the $n$th individual at $t$ with their latest observed state. If $g=0.34$ we have ``Prev. uncertain'', while if $g=0.99$ we have ``Prev. certain''.

\begin{algorithm}[t!]
\caption{Previous guess} \label{alg:guess}
    \begin{algorithmic}
    \Require $\bW, \bY_{1:T}, g$
    \State Initialize $\bpi^g_{n,0}$ and $\bpi^{guess}_{n,0}$ with $1_M \oslash M$ for all $n \in [N]$
    \For{$t \in 1,\dots,T$}
        \For{$n \in [N]$}
            \If{$\by_{n,t}^{(M+1)}=0$}
            \State $\bpi^g_{n,t} = \by_{n,t}^{(1:M)}$
            \State $\bpi_{n,t}^{guess} = g \odot \by_{n,t}^{(1:M)} + \frac{1-g}{M-1} \odot \left (1-\by_{n,t}^{(1:M)} \right )$
            \Else
            \State $\bpi^g_{n,t} = \bpi_{n,t-1}^{guess}$
            \State $\bpi_{n,t}^{guess} = \bpi_{n,t-1}^{guess}$
            \EndIf
        \EndFor
    \EndFor 
    \end{algorithmic}
\end{algorithm}

\subsubsection{Changing the parameters values}\label{sec:app_tutorial_py}

Suppose now that we want to consider different parameters values, for instance we want $q_R=0$. This can be easily done by looking at the tutorial\_SIR.ipynb file in the GitHub repository \href{https://github.com/LorenzoRimella/CAL}{LorenzoRimella/CAL} and by changing the configuration:
\begin{verbatim}
    parameters = {
    "prior_infection":tf.convert_to_tensor([1-0.5, 0.5, 0.0], dtype = tf.float32),
    "log_beta":tf.math.log(tf.convert_to_tensor([3.0], dtype = tf.float32)),
    "b_S":tf.convert_to_tensor([+0.5], dtype = tf.float32),
    "b_I":tf.convert_to_tensor([+1.0], dtype = tf.float32),
    "log_gamma":tf.math.log(tf.convert_to_tensor([0.05], dtype = tf.float32)),
    "b_R":tf.convert_to_tensor([-0.1], dtype = tf.float32),
    "log_phi":tf.math.log(tf.convert_to_tensor([1.5], dtype = tf.float32)),
    "log_epsilon":tf.math.log(tf.convert_to_tensor([0.0001], dtype = tf.float32)),
    "logit_prob_testing":logit(tf.convert_to_tensor([0.1, 0.2, 0.5], 
    dtype = tf.float32))
    }
\end{verbatim}
For the case $q_R=0$ the above would change to: 
\begin{verbatim}
    parameters = {
    "prior_infection":tf.convert_to_tensor([1-0.5, 0.5, 0.0], dtype = tf.float32),
    "log_beta":tf.math.log(tf.convert_to_tensor([3.0], dtype = tf.float32)),
    "b_S":tf.convert_to_tensor([+0.5], dtype = tf.float32),
    "b_I":tf.convert_to_tensor([+1.0], dtype = tf.float32),
    "log_gamma":tf.math.log(tf.convert_to_tensor([0.05], dtype = tf.float32)),
    "b_R":tf.convert_to_tensor([-0.1], dtype = tf.float32),
    "log_phi":tf.math.log(tf.convert_to_tensor([1.5], dtype = tf.float32)),
    "log_epsilon":tf.math.log(tf.convert_to_tensor([0.0001], dtype = tf.float32)),
    "logit_prob_testing":logit(tf.convert_to_tensor([0.1, 0.2, 0.0], 
    dtype = tf.float32))
    }
\end{verbatim}

Rerunning the simulation leads to similar conclusion to Section \ref{sec:exp_mispec}, with parameter estimates close to the truth and similar cross-entropy loss and accuracy, see Table \ref{tab:app_well_mis_accuracy}.

\begin{table}[httb!]
    \centering
    \caption{Cross-entropy loss (the lower the better) and accuracy (the higher the better) for the CAL well-specified and misspecified, along with some baselines. The predicted state for accuracy is the argmax of the probability vector.}
    \label{tab:app_well_mis_accuracy}
    \begin{tabular}{l|ccccc}
        Metric & Random & Prev. uncertain & Prev. certain & CAL \\
        \hline
        Cross-entropy & 1.1 & 1.1 & 1.86 & 0.34 \\
        Accuracy      & $34.85\%$  & $41.65\%$  & $41.65\%$ & $86.09\%$ \\
        \bottomrule
    \end{tabular}
\end{table}

\subsection{Comparing CAL with some baselines}\label{sec:app_exp_SMC}

In this section, we compare the run time and marginal likelihood values obtained by CAL against those from: SMC for individual-based models with approximate optimal proposals by \cite{rimella2022approximating}, where $\alpha$ controls the number of future observations included in the lookahead scheme, and Simulation Based Composite Likelihood (SimBa-CL)  \citep{rimella2025simulation}. The SMC proposed by \cite{rimella2022approximating} provides a suitable proxy for comparison as it uses proposal distribution that are informed by future observations avoiding particles degeneracy. It has been recently discovered that \cite{rimella2022approximating} likelihood estimator is biased due to the exclusion of the initial weights when performing calculations. However we expect this bias to disappear as $P$ increases.

We consider a homogeneous-mixing SIS inspired by \cite{ju2021sequential}. Here the initial distribution is given by a logistic regression on $\bc_n$ with parameters $\bb_0\in \mathbb{R}^2$:
\begin{equation}
    p_0(\bw_n) 
    = 
    \begin{bmatrix} 
    1-\frac{1}{1 + \exp(-\bc_n^\top \bb_0)}\\ 
    \frac{1}{1 + \exp(-\bc_n^\top \bb_0)}
    \end{bmatrix}.
\end{equation}
We consider an interaction term $\boeta_{n,t-1} = \frac{1}{N} \sum_{k \in [N]} \bx_{k,t-1}^{(2)}$ or equivalently:
\begin{equation}
    \boeta_{n,t-1} = \frac{1}{N} \sum_{k \in [N]} 
    \begin{bmatrix}
        0\\
        1
    \end{bmatrix}^\top \bx_{k,t-1},
\end{equation}
as we are considering a homogeneous-mixing case, which is then used in the transition matrix:
\begin{equation}
    K_{\eta_{n,t-1}}(\bw_n) 
    = 
    \begin{bmatrix} 
    \exp \left ( - h \frac{\boeta_{n,t-1} +\epsilon }{1 + \exp{(-\bb_{S}^\top \bc_n)}} \right )
    & 
    1 - \exp \left (  -h \frac{\boeta_{n,t-1} +\epsilon }{1 + \exp{(-\bb_{S}^\top \bc_n)}} \right )\\
    1- \exp \left ( - \frac{h }{1 + \exp{(-\bb_{R}^\top \bc_n)}}\right )
    & 
    \exp \left ( - \frac{h }{1 + \exp{(-\bb_{R}^\top \bc_n)}}\right )
    \end{bmatrix},
\end{equation}
where $\bb_S,\bb_I \in \mathbb{R}^2$ and $\epsilon > 0$ and . The observation model is then given by:
\begin{equation}
    G(\bw_n) 
    =
    \begin{bmatrix}
    1 - q_S 
    & 
    q_S  
    & 
    0 \\
    1 - q_I 
    & 
    0 
    & 
    q_I
    \end{bmatrix},
\end{equation}
with $q_S,q_I \in [0,1]$. We consider a population of $1000$ individuals and a time horizon of $100$, with parameters set to $\bb_0=[-\log(100-1), 0]^\top, \epsilon=0.001, \bb_S=[-1, 2]^\top, \bb_R=[-1, -1]^\top, q_S = 0.6, q_I = 0.4, q_{Se} = 1.0, q_{Sp} = 1$. 

\begin{table*}[httb!]
    \centering
    \caption{Log-likelihood means and standard deviations for the SIS model. We denote \cite{rimella2022approximating}  with $\dag$, with $\alpha$ being the number of future observations included in the lookahead scheme. Log-likelihood results are averages and standard deviations over $100$ runs. Running times are reported for a single run and as averages across particles. }
    \label{tab:N1000_SIS}
    \resizebox{\textwidth}{!}{
    \begin{tabular}{l|ccc|c}
     Number of particles $P$                      & 512               & 1024              & 2048              & Time (sec)\\
    \midrule
    $\dag$ with $\alpha=5$              & -79551.92 (1.79)  & -79552.24 (1.6)   & -79552.81 (1.57)  & 3.78s \\
    $\dag$ with $\alpha=10$             & -79551.9 (1.81)   & -79552.22 (1.47)  & -79553.01 (1.56)  & 5.61s \\
    SimBa-CL           & -79612.74 (3.4)   & -79612.31 (2.37)  & -79612.34 (1.55)  & 1.03s \\
    \midrule
    CAL           &    &   &   -79550.69 & 1.93s\\
    CAL jit\_compiled           &    &   &   -79550.69 & 0.003s\\
    \bottomrule
    \end{tabular}
    }
\end{table*}

We consider $N=1000, T=100$, simulate from the model and, for that one realization of the data, estimate mean and standard deviation of the log-likelihood at the DGP for each method over $100$ runs, the CAL requires a single run as it is a deterministic algorithm. The results are reported in Table \ref{tab:N1000_SIS}. As expected, the method by \cite{rimella2022approximating} with the highest $\alpha$ has the lowest variance, but it is computationally intensive as a single run requires more than 5s. SimBa-CL performs well computationally but it is more biased. The CAL is the fastest and the log-likelihood estimate is close to the one from \cite{rimella2022approximating}, with a running time even dropping to 0.003s when considering just-in-time compilation, which is, as explained in Section \ref{sec:exper}, straightforward for the CAL.

The SEIR scenario is significantly more challenging as some transitions are not allowed, making the SMC more prone to particle impoverishment and degeneracy. In our SEIR we consider a homogeneous-mixing individual-based model with an initial distribution as in the aforementioned SIS, but an additional zero probability of being assigned to $E,R$ at the beginning of the epidemic. The model is again inspired by \cite{ju2021sequential}, now the initial distribution is:
\begin{equation}
    p_0(\bw_n) = 
    \begin{bmatrix}
         1- \frac{1}{1+\exp(-\bb_0^\top \bc_n)} \\
         0 \\
         \frac{1}{1+\exp(-\bb_0^\top \bc_n)} \\
         0
    \end{bmatrix},
\end{equation}
where $\bb_0 \in \mathbb{R}^2$. As we are again considering a homogeneous-mixing scenario the interaction term is:
\begin{equation}
    \boeta_{n,t-1} 
    = 
    \frac{1}{N} \sum_{k \in [N]} 
    \begin{bmatrix}
        0\\
        0\\
        1\\
        0
    \end{bmatrix}^\top \bx_{k,t-1},
\end{equation}
and used in the stochastic transition matrix:
\begin{equation}
\begin{split}
    &K_{\eta_{n,t-1}}(\bw_n)= \\
    &
    \begin{bmatrix} 
    \exp \left ( \frac{-h \boeta_{n,t-1} - h\epsilon }{1 + \exp{(-\bb_{S}^\top \bc_n)}} \right )
    & 
    1- \exp \left ( \frac{-h \boeta_{n,t-1} -h \epsilon }{1 + \exp{(-\bb_{S}^\top \bc_n)}} \right )& 0 & 0\\ 
    0 & \exp(-h \rho) & 1 - \exp(-h \rho) & 0 \\
    0 & 0 & \exp \left ( \frac{-h}{1 + \exp{(-\bb_{R}^\top \bc_n)}} \right )
    & 
    1- \exp \left ( \frac{-h}{1 + \exp{(-\bb_{R}^\top \bc_n)}} \right )\\
     0 & 0 & 0 & 1
    \end{bmatrix},
\end{split}
\end{equation}
with $\bb_S,\bb_I \in \mathbb{R}^2, \epsilon > 0$ and $\rho \in \mathbb{R}_+$. The observation matrix is:
\begin{equation}
    G(\bw_n) 
    =
    \begin{bmatrix}
    1 - q_S 
    & 
    q_S  
    & 
    0 & 0 & 0\\
    1 - q_E 
    & 
    0 
    & 
    q_E 
    & 0 & 0\\
    1 - q_I 
    & 
    0 
    & 
    0
    & 
    q_I 
    & 0 \\
    1 - q_R
    & 
    0 
    & 0 & 0
    & 
    q_R
    \end{bmatrix},
\end{equation}
with $q_S, q_E ,q_I, q_R \in [0,1]$.


\begin{table*}[httb!]
    \centering
    \caption{Log-likelihood means and log-likelihood standard deviations for the individual-based SEIR model with $N=1000$. We denote \cite{rimella2022approximating} with $\dag$, with $\alpha$ being the number of future observations included in the lookahead scheme. Log-likelihood results are averages and standard deviation over $100$ runs. Running times are reported for a single run. }
    \begin{tabular}{lcccc}
     Number of particles $P$                  & 512              & 1024             & 2048             & Time (sec)\\
    \hline
    $\dag$ with $\alpha=5$               & -43447.56 (52.04) & -43419.52 (51.08) & -43391.0 (52.41)  & 4.44s \\
     $\dag$ with $\alpha=20$              & -43004.55 (5.38)  & -43001.9 (4.65)   & -42999.76 (3.7)   & 11.08s \\
     $\dag$ with $\alpha=50$             & -42999.93 (3.44)  & -42998.13 (2.72)  & -42996.74 (2.39)  & 20.88s \\
     SimBa             & -43683.85 (9.54)  & -43683.67 (7.35)  & -43683.76 (5.16)  & 1.25s \\
    \hline
    CAL           &    &   & -43454.97 & 1.89s\\
    CAL jit\_compiled          &    &   & -43454.97 & 0.003s\\
    \hline
    \end{tabular}
    \label{tab:N1000_SEIR}
\end{table*}

As for the SIS scenario, we consider a population of $1000$ individuals and a time horizon of $100$. We set $\bb_0=[-\log(100-1), 0]^\top,\epsilon=0.001,\bb_S=[-1, 2]^\top, \rho=0.2, \bb_R=[-1, -1]^\top, q_S =0, q_E =0, q_I = 0.4, q_R = 0.6$ and we simulate from the model. The log-likelihood mean and standard deviation are then estimated over multiple runs. As expected, the method by \cite{rimella2022approximating} requires several future observations to achieve low variance, which implies a significant increase in the running time. SimBa-CL is computationally more efficient, but it shows higher variance and a significant bias. The CAL is closer to \cite{rimella2022approximating} with $\alpha=50$ compared to the other baselines. Again, if just-in-time compilation is considered, the CAL runs in about 0.003s.

\subsection{Comparing CAL with SMC} \label{sec:app_exp_SMC_only}

For the experiment of Section \ref{sec:exp_SMC_only} in the main manuscript, we use the parameter settings $p_0=0.01, \beta = 0.2, \bb_I = 0.3, \bb_S = -0.3, \gamma=0.1, \bb_R=0.2, q_S=0.2, q_I=0.5, q_{Se}=0.9, q_{Sp}=0.95$. 

We consider a batched Block APF to reduce memory consumption. Our implementation loads everything into memory to exploit GPU parallel computing, resulting in a memory blow-up when resampling is performed on all the individuals. Switching from an algorithm that is fully parallelized over both individuals and particles to a sequential version where parallelization is applied only within batches offers a trade-off between memory usage and running time. 

In the batched Block APF, the particles are always batched into four subgroups: instead of resampling $P$ particles in parallel, we sequentially sample $\frac{P}{4}$ particles and then concatenate the resulting samples. Individuals are batched according to the population size $N$. Specifically, there is no batching for $N = 10$; a batching size of $10$ for $N = 100$; a batching size of $100$ for $N = 1000$; and again a batching size of $100$ for $N = 10000$.

From Table \ref{tab:SMC_comparison_table} it appears that when $N,P\to \infty$ the (batched) Block APF has a similar asymptotic behavior to the CAL as $N\to \infty$. We therefore provide a graphical illustration of this behavior in Figure \ref{fig:smc_boxplot}. It can be observed that as $N$ increases the SMC methods require more and more particles to control the variance. Moreover, for $N = 1000$ and $N = 10000$, we notice that when $P$ increases the (batched) Block APF gets closer and closer to CAL. We expect this to be a consequence of the large-population limit of the model and the resulting decoupling of the individuals. Indeed, the CAL becomes exact as $N$ goes to infinity, and we expect the Block APF to have similar properties as it considers a similar approximation. This perspective differs from that taken by \cite{rebeschini2015can}, where the quality of the block particle filter was measured on coarser and coarser partitions rather than increasing dimension, and may be of independent interest.

\begin{figure}[httb!]
    \centering
    \includegraphics[width=0.85\linewidth]{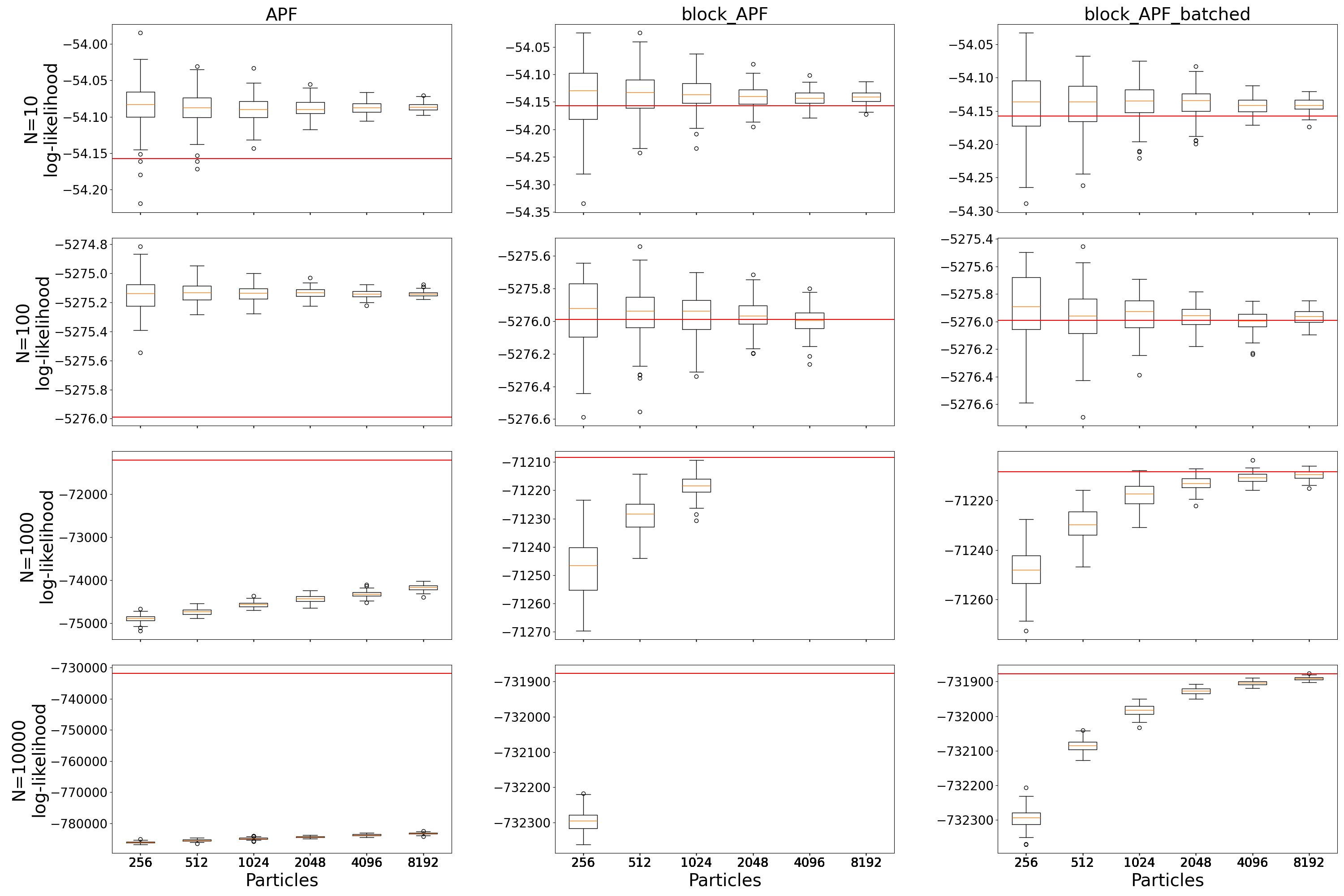}
    \caption{Log-likelihood boxplots for APF, Block APF, and batched Block APF when $N$ and $P$ increases for the SIS model from Section \ref{sec:motivating_example_first}. The CAL log-likelihood is reported as a horizontal red line.}
    \label{fig:smc_boxplot}
\end{figure}

\begin{table*}[httb!]
    \centering
    \caption{Log-likelihood means and standard deviations, over $100$ runs, for the SIS model from \cite{ju2021sequential}. Running times are reported in seconds for a single run and as averages across the $100$ runs.}
    \label{tab:SMC_comparison_table_ju}
    \resizebox{0.85\textwidth}{!}{
    \begin{tabular}{ll|llll|llllll}
    \hline
     N     & P    & CAL        & time(s) & CAL compiled & time(s) & APF                & time(s) & Block APF        & time(s) & batched Block APF & time(s) \\
     \hline
    10    & 256  & -67.13     & 1.354 & -67.13       & 0.001 & -67.01(0.03)       & 0.11  & -67.1(0.05)       & 0.13 & -67.11(0.06)      & 2.43    \\
     10    & 512  &            &       &              &       & -67.01(0.02)       & 0.12  & -67.1(0.03)       & 0.13 & -67.1(0.04)       & 2.42    \\
     10    & 1024 &            &       &              &       & -67.01(0.01)       & 0.11  & -67.1(0.02)       & 0.13 & -67.1(0.03)       & 2.41    \\
     10    & 2048 &            &       &              &       & -67.01(0.01)       & 0.12  & -67.1(0.02)       & 0.13 & -67.1(0.02)       & 2.41    \\
     10    & 4096 &            &       &              &       & -67.01(0.01)       & 0.12  & -67.1(0.01)       & 0.16 & -67.1(0.01)       & 2.41    \\
     10    & 8192 &            &       &              &       & -67.01(0.0)        & 0.12  & -67.1(0.01)       & 0.31 & -67.1(0.01)       & 2.31    \\
     \hline
     100   & 256  & -7101.4    & 2.128 & -7101.4      & 0.001 & -7107.8(4.94)      & 1.12  & -7102.37(1.82)    & 1.26 & -7102.82(1.72)    & 23.53   \\
     100   & 512  &            &       &              &       & -7105.34(3.21)     & 1.08  & -7102.38(1.3)     & 1.31 & -7102.25(1.24)    & 23.0    \\
     100   & 1024 &            &       &              &       & -7102.46(2.92)     & 1.11  & -7101.92(0.9)     & 1.46 & -7102.02(0.94)    & 23.02   \\
     100   & 2048 &            &       &              &       & -7100.6(2.37)      & 1.11  & -7101.78(0.61)    & 2.35 & -7101.7(0.67)     & 23.2    \\
     100   & 4096 &            &       &              &       & -7099.73(1.98)     & 1.14  & -7101.6(0.5)      & 5.93 & -7101.61(0.53)    & 23.26   \\
     100   & 8192 &            &       &              &       & -7099.1(1.47)      & 1.26  & Out of memory     &      & -7101.57(0.33)    & 37.83   \\
     \hline
     1000  & 256  & -82562.3   & 2.207 & -82562.3     & 0.001 & -84153.91(38.25)   & 1.17  & -82584.62(7.16)   & 1.47 & -82582.9(7.13)    & 25.75   \\
     1000  & 512  &            &       &              &       & -84043.09(35.49)   & 1.22  & -82573.03(4.79)   & 2.09 & -82573.19(4.48)   & 25.43   \\
     1000  & 1024 &            &       &              &       & -83963.92(37.44)   & 1.31  & -82567.92(3.23)   & 4.51 & -82567.3(2.98)    & 25.75   \\
     1000  & 2048 &            &       &              &       & -83869.86(33.91)   & 1.47  & Out of memory     &      & -82565.28(2.24)   & 30.94   \\
     1000  & 4096 &            &       &              &       & -83792.78(28.68)   & 1.86  & Out of memory     &      & -82563.78(1.88)   & 66.6    \\
     1000  & 8192 &            &       &              &       & -83721.4(27.97)    & 2.71  & Out of memory     &      & -82562.98(1.16)   & 208.23  \\
     \hline
     10000 & 256  & -822379.25 & 2.435 & -822379.25   & 0.001 & -847562.6(144.4)   & 1.44  & -822586.25(21.59) & 3.84 & -822587.7(19.98)  & 217.59  \\
     10000 & 512  &            &       &              &       & -847122.5(152.59)  & 1.94  & Out of memory     &      & -822482.56(13.75) & 218.71  \\
     10000 & 1024 &            &       &              &       & -846751.25(141.99) & 3.01  & Out of memory     &      & -822429.75(10.21) & 220.13  \\
     10000 & 2048 &            &       &              &       & -846343.2(130.99)  & 5.08  & Out of memory     &      & -822405.5(7.72)   & 283.91  \\
     10000 & 4096 &            &       &              &       & -845975.6(137.44)  & 9.7   & Out of memory     &      & -822393.2(4.76)   & 676.48  \\
     10000 & 8192 &            &       &              &       & -845647.8(132.39)  & 19.12 & Out of memory     &      & -822385.6(3.37)   & 2105.17 \\
    \hline
    \end{tabular}
    }
\end{table*}

\begin{figure}[httb!]
    \centering
    \includegraphics[width=0.85\linewidth]{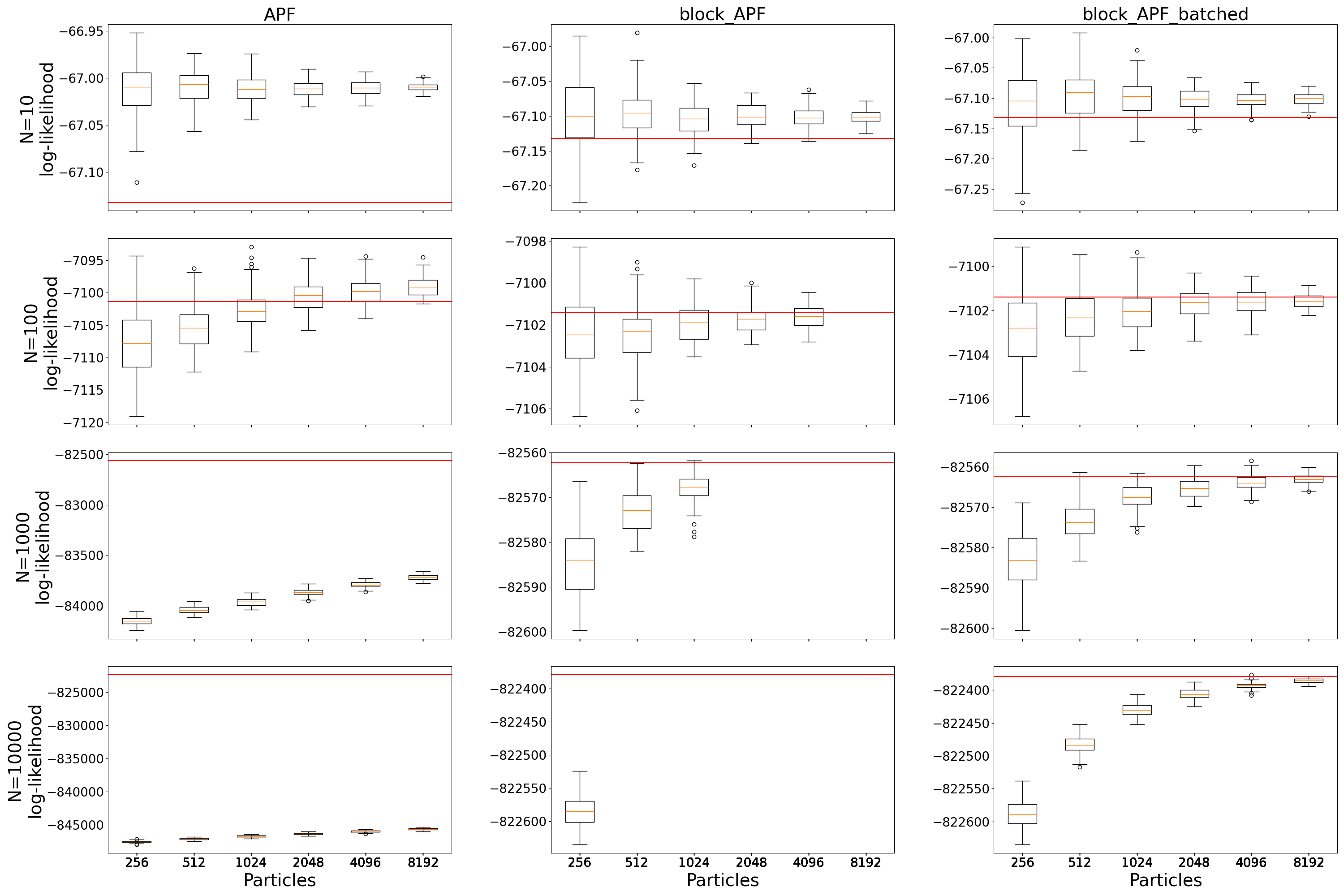}
    \caption{Log-likelihood boxplots for APF, Block APF, and batched Block APF when $N$ and $P$ increases for the SIS model by \cite{ju2021sequential}. The CAL log-likelihood is reported as a horizontal red line.}
    \label{fig:smc_boxplot_ju}
\end{figure}

We also replicated the same experiment under the \cite{ju2021sequential} model from the previous subsection with parameters set to $\bb_0=[-\log(100-1), 0]^\top, \epsilon=0.001, \bb_S=[-1, 2]^\top, \bb_R=[-1, -1]^\top, q_S = 0.6, q_I = 0.4, q_{Se} = 1.0, q_{Sp} = 1$. We report numerical results in Table \ref{tab:SMC_comparison_table_ju} and a graphical representation in Figure \ref{fig:smc_boxplot_ju}.

\subsection{2001 UK foot-and-mouth disease outbreak}\label{sec:app_exp_FM}

 \paragraph{Local authorities meta-population model.} For the local authorities we consider \cite{locauto}, reporting digital vector boundaries of the UK's local authority districts in December 2023. We did not find any open-source digital vector boundaries from 2001, which would have been ideal. Local authorities with less than five farms were excluded from the study, e.g. London and Birmingham.   The farm-specific covariates are then $\bw_n= [\bom_n, \bar{\bl}_n, \bc_n]$, where $\bom_n$ is the centroid (in EPSG:27700, the projected coordinate system for the UK) of the local authority individual $n$ is assigned to, $\bar{\bl}_n$ is the mean-distance (in km) across farms within the local authority, and $\bc_n$ is a bi-dimensional vector containing the log-number of cattle and the log-number of sheep. The components of $\bom_n$ are further divided by $1000$ so when computing Euclidean distances across local authorities the resulting distances are in Km. Observe that $\bc_n$ is still at an individual-level, while we have aggregated the spatial component.

\paragraph{Model.} We consider a heterogeneous-mixing individual-based SIR model, where transitions from $S$ to $R$ are also allowed, representing the culling\slash quarantine of healthy farms to create containment zones around infected farms. We consider two interaction terms:
\begin{align} \label{eq:int1}
    &\boeta_{n,t-1}^I
    =
    \frac{1}{N} \sum_{k \in [N]} \begin{bmatrix}
    0\\
    \frac{\exp\{\bc_{k}^\top \bb_I\}}{\sqrt{2 \pi \phi^2}}\exp\left\{-\frac{\|\bom_n-\bom_{k}\|^2 \mathbb{I}(\|\bom_n-\bom_{k}\| \neq 0) + \bar{\bl}_n^2 \mathbb{I}(\|\bom_n-\bom_{k}\| = 0)}{2\phi^2}\right\}\\
    0
    \end{bmatrix}^\top \bx_{k,t-1},\\
    &\boeta_{n,t-1}^C \label{eq:int2}
    =
    \frac{1}{N} \sum_{k \in [N]} \begin{bmatrix}
    0\\
    \frac{1}{\sqrt{2 \pi \psi^2}}\exp\left\{-\frac{\|\bom_n-\bom_{k}\|^2 \mathbb{I}(\|\bom_n-\bom_{k}\| \neq 0) + \bar{\bl}_n^2 \mathbb{I}(\|\bom_n-\bom_{k}\| = 0)}{2\psi^2}\right\}\\
    0
    \end{bmatrix}^\top \bx_{k,t-1},
\end{align}
where $\bb_I \in \mathbb{R}^2, \phi>0, \psi >0$, and we also define:
\begin{equation} \label{eq:int3}
\boeta_{n}^0 \coloneqq \frac{1}{N} \sum_{k \in [N]}\exp\{\bc_{k}^\top \bb_I\} \frac{\exp\left\{-\frac{ \|\bom_n-\bom_{k}\|^2 \mathbb{I}(\|\bom_n-\bom_{k}\| \neq 0) + \bar{\bl}_n^2 \mathbb{I}(\|\bom_n-\bom_{k}\| = 0) }{2\phi^2}\right\}}{\sqrt{2 \pi \phi^2}} \tau,
\end{equation}
where $\tau >0$ can be interpreted as the probability of being infected before $t=0$. We consider an initial distribution:
\begin{equation}
    p_0(\bw_n)
    =
    \begin{bmatrix}
        \exp\left ( -\beta \exp \{ \bc_n^\top \bb_S \} \boeta_{n}^0 -\epsilon \right )\\
        1-\exp\left ( -\beta \exp \{ \bc_n^\top \bb_S \} \boeta_{n}^0 -\epsilon \right )\\
        0
    \end{bmatrix}.
\end{equation}
We then define the culling\slash quarantine probability of farm $n$ by $P^C_n \coloneqq 1-\exp \left ( - h \rho \boeta_{n,t-1}^C \right )$ where $\rho>0$, the infection probability of a non-culled\slash quarantine farm by $P^I_n \coloneqq 1-\exp\left ( - h \beta \exp \{ \bc_n^\top \bb_S \} \boeta_{n,t-1}^I -h \epsilon\right )$ where $\beta>0$ and $\bb_S \in \mathbb{R}^2$, and the recovery probability $P^R_n \coloneqq 1- \exp \left ( -h \gamma \right )$ where $\gamma >0$. The stochastic transition matrix is then given by:
\begin{equation}
\begin{split}
    &K_{\eta_{n,t-1}^I,\eta_{n,t-1}^R}(\bw_n)
    =
    \begin{bmatrix} 
    (1- P^C_n)(1 - P^I_n)
    &
    (1- P^C_n)P^I_n
    &
    P^C_n\\
    0
    &
    (1- P^C_n)(1 - P^R_n)
    &
    P^C_n + (1 - P^C_n) P^R_n\\
     0 & 0 & 1
    \end{bmatrix}.
\end{split}
\end{equation}
For the observation model, we do not allow for misreporting and we assume only infected are observable:
\begin{equation}
    G(\bw_n) 
    =
    \begin{bmatrix}
    1 & 0 & 0 & 0\\ 
    1 - q_I & 0 & q_I & 0 \\
    1 & 0 & 0 & 0
    \end{bmatrix}.
\end{equation}
The aforementioned individual-based model has two interaction terms and an interaction term inside the initial distribution and does not strictly belong to the class of models that satisfy our assumptions. Nevertheless, the terms \eqref{eq:int1},\eqref{eq:int2},\eqref{eq:int3} follow a law of large numbers, as they are averages, hence our saturation theory can be developed for such individual-based models with multidimensional interaction terms. Consistency then follows using the same techniques of Section \ref{sec:theory_main}.

\paragraph{Maximum CAL estimation.}  We consider $100$ different random initializations of $\tau,\beta,\bb_S,$\\  $\bb_I,\phi,\gamma,$$\rho,\psi,\epsilon,q_I$ and run Adam optimizer for each of them for $10000$ gradients steps using auto-differentiation in TensorFlow. We then select the one with the highest log-CAL. We then perform a sequential optimization with Adam where every $50000$ iteration we change the value of the step size and reset the Adam optimizer. Specifically, we consider the step sizes $0.1, 0.01, 0.01, 0.001,0.001, $\\$ 0.001, 0.00001, 0.00001$, resulting in a total of $400000$ gradient steps.

\paragraph{Credible intervals.} We then run an HMC using the maximum CAL estimation as a warm start. We considered a step size of $0.00274625$, which was obtained via pre-tuning on the acceptance rate, 10 leapfrog steps, a vague prior for $\beta,\bb_S,\bb_I,\phi,\epsilon,q_I$ given by a Gaussian with mean 0 and standard deviation 100, and an informative prior on $\tau,\gamma,\rho,\psi$ given by Gaussians with means $12.5,3,12.5,1.2$ and standard deviation $0.25$. The informative prior choice was needed to improve mixing of the HMC. The mean of the prior was set to the maximum CAL estimation we used as a warm start, while the standard deviation was tuned to ensure proper mixing. Posterior means and $95\%$ credible interval are reported in Table \ref{tab:CI_mean}

\begin{table}
    \centering
    \begin{tabular}{l|llllll}
    \hline
    Param. & $\log(\tau)$ & $\log(\beta)$ & $\bb_S^{(1)}$ & $\bb_S^{(2)}$ & $\bb_I^{(1)}$ & $\bb_I^{(2)}$ \\
    \hline
     Mean & -12.52 & 4.83 & 0.37 & 0.11 & 0.23 & 0.09\\
     95\% CI & [-13.00,-11.97] & [4.49,5.14] & [0.35,0.40 ] & [0.10,0.13] & [0.18,0.30] & [0.05,0.12]\\
     \hline
     \end{tabular}

    \vspace{0.5cm}
     
     \begin{tabular}{l|llllll}
     \hline
     Param. & $\log(\phi)$           & $\log(\gamma)$ & $\log(\rho)$ & $\log(\psi)$ & $\log(\epsilon)$ & $logit(q_I)$ \\
    \hline
     Mean & 3.36 & 3.0 & 12.52 & 1.25 & -12.83 & 0.67 \\
     95\% CI & [3.32,3.40] & [2.47,3.49] & [12.08,12.99] & [0.92,1.44] & [-13.14,-12.54] & [0.51,0.85] \\
    \hline
    \end{tabular}
    \caption{HMC posterior means and $95\%$ credible intervals on the foot-and-mouth model.}
    \label{tab:CI_mean}
\end{table}

\begin{figure}[httb!]
    \centering
    \includegraphics[width=0.9\linewidth]{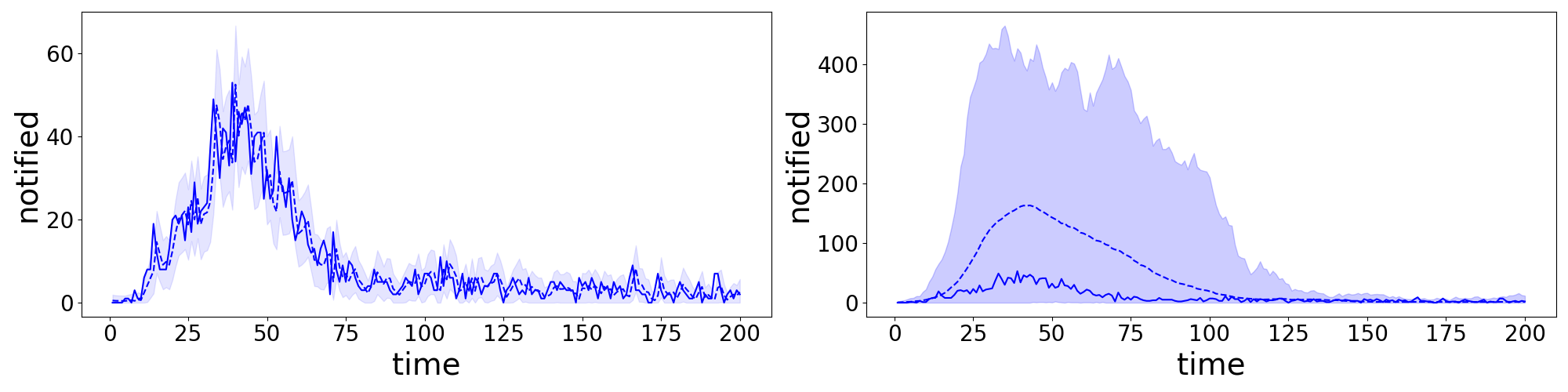}
    \caption{On the left, the CAL filter mean (dashed blue line) with $95\%$ bands based on the CAL filter variance. On the right, the Monte Carlo mean (dashed blue line) and $2.5\%$, $97.5\%$ quantile of the predictive distribution. Solid lines are used for the observations.}
    \label{fig:FM_predictive}
\end{figure}

 In Figure \ref{fig:FM_predictive} we report the CAL filter and the predictive distribution under the optimized parameters for the total number of notified. The CAL filter is simply obtained by running the CAL on the optimized parameters and by using the Categorical approximation to estimate mean and variance of $\sum_{n \in [N]} \by_{n,t}$. The predictive distribution is obtained via multiple simulations from the model with the DGP set to the optimized parameters, which are then used to get Monte Carlo estimates of mean and quantiles. In both cases we can observe that we get good coverage of the observations, showing that the optimized parameters are valid estimates.

The spatial location of a single farm cannot be disclosed for privacy, hence the left-hand side of Figure \ref{fig:FM_spread} and the whole Figure \ref{fig:FM_spat_track} cannot be reproduced. However, we are allowed to disclose information about local authorities and, in particular, the spread of the disease within the local authorities, which can be found in our GitHub repository. The equivalent of Figure \ref{fig:FM_spat_track} for the local authorities is Figure \ref{fig:FM_spat_track_locaut}. Even though this aggregated version is less informative compared to the fully spatial one, we can still recognize the same patterns in terms of the disease's spread of the infection and distribution of the removed.

\begin{figure}[httb!]
    \centering
    \includegraphics[width=\linewidth]{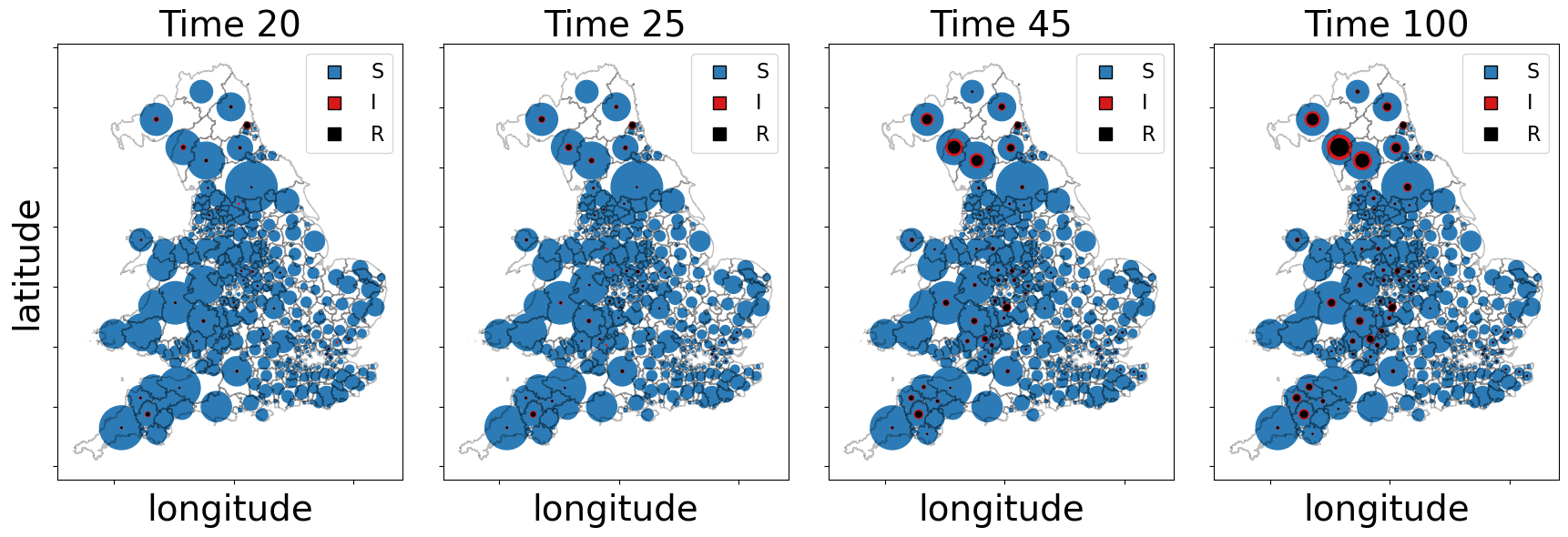}
    \caption{The CAL prediction over time of susceptible farms (blue), infected farms (red), and removed farms (black) for the local authority model. Black dots are double the radius and red dots are four times the radius for visual purposes.}
    \label{fig:FM_spat_track_locaut}
\end{figure}

\subsubsection{Benchmarking}

Consider $\widetilde{\by}_{n,t}$ which takes values $1,0$ depending on whether the farm is reported as infected or not. We model $p(\widetilde{\by}_{n,1}=1),p(\widetilde{\by}_{n,t}|\widetilde{\by}_{n,t-1})$ with an AR logistic regression:
$$
p(\widetilde{\by}_{n,1}=1) = \frac{\gamma}{1+e^{ -\bb^\top \bw_n}}, \quad p(\widetilde{\by}_{n,t}=1|\widetilde{\by}_{n,t-1}) = \frac{\gamma}{1+e^{-\beta \widetilde{\by}_{n,t-1} -\bb^\top \bw_n}},
$$
where $\gamma\in[0,1]$, $\beta \in \mathbb{R}$ and $\bb \in \mathbb{R}^{C+1}$, with $C$ number of covariates for each individual (+1 for the intercept). This results in the likelihood:
\begin{equation}
    \begin{split}
        p(\widetilde{\by}_{n,1:T}=1) &= \prod_{n\in[N]} \left( \frac{\gamma}{1+e^{ -\bb^\top \bw_n}} \right )^{\by_{n,1}} \left( 1- \frac{\gamma}{1+e^{ -\bb^\top \bw_n}} \right )^{1-\by_{n,1}}\\
        &\qquad \cdot \prod_{t=2}^T  \left ( \frac{\gamma}{1+e^{-\beta \widetilde{\by}_{n,t-1} -\bb^\top \bw_n}}\right )^{\by_{n,t}} \left ( 1- \frac{\gamma}{1+e^{-\beta \widetilde{\by}_{n,t-1} -\bb^\top \bw_n}}\right )^{1-\by_{n,t}}.
    \end{split}
\end{equation}

We optimize $\gamma,\beta,\bb$ using gradient ascent on the log-likelihood and Adam optimizer with $10000$ gradient steps and a learning rate of $0.1$, resulting in a log-likelihood of $-20858.994$. 

We then consider the same AR logistic regression but with local authority specific parameters. Specifically, we have
$$
p(\widetilde{\by}_{n,1}=1) = \frac{\gamma}{1+e^{ -\bb_k^\top \bw_n}}, \quad p(\widetilde{\by}_{n,t}=1|\widetilde{\by}_{n,t-1}) = \frac{\gamma_k}{1+e^{-\beta_k \widetilde{\by}_{n,t-1} -\bb_k^\top \bw_n}},
$$
where $\gamma_k\in[0,1]$, $\beta_k \in \mathbb{R}$ and $\bb_k \in \mathbb{R}^{C+1}$ and with $k$ being the local authority of individual $n$.

We similarly optimize $(\gamma_k,\beta_k,\bb_k)_k$ using gradient ascent on the log-likelihood and Adam optimizer with $10000$ gradient steps and a learning rate of $0.1$, resulting in a log-likelihood of $-18866.531$. 

\begin{algorithm}[t!]
\caption{CAL within bootstrap particle filter for ``shared'' overdispersion} \label{alg:overdispersion}
    \begin{algorithmic}
    \Require $\bW, \bY_{1:T}, p_0(\cdot), K_{\cdot}(\cdot), G(\cdot), P$
    \State Initialize $\bpi_{n,0}^P$ with $p_0(\bw_n)$ for all $n \in [N]$ and $p \in [P]$
    \For{$t \in 1,\dots,T$}
        \State $\bPi^p_{t-1} = (\bpi^p_{1,t-1}, \dots, \bpi^p_{N,t-1})$ for all $p \in [P]$
        \State Sample $\xi_{t}^p$ from the prior for all $p \in [P]$
        \For{$n \in [N]$}
            \State $\widetilde{\boeta}^p_{n,t-1} = \eta(\bw_n,\bW,\bPi^p_{t-1})$
            \State $\bpi^p_{n,t|t-1} =  \left [ (\bpi^p_{n,t-1})^{\top} K_{\widetilde{\boeta}_{n,t-1}}(\bw_n,\xi_t^p) \right ]^{\top}$ 
            \State $\bmu_{n,t}^p = \left [ (\bpi_{n,t|t-1}^p)^\top  G(\bw_n) \right ]^\top$
            \State $\bpi_{n,t}^p = \bpi_{n,t|t-1}^p \odot \left \{  \left [  G(\bw_n) \oslash  \left ( 1_M (\bmu_{n,t}^p)^\top \right ) \right ] \by_{n,t} \right \}$  
        \EndFor
        \State Set $w_t^p = \prod_{n \in [N]} \by_{n,t}^\top \bmu_{n,t}^p$ for all $p \in [P]$
        \State Resample the particles $\xi_{t}^p,(\bpi_{n,t}^p)_{n \in [N]}$ according to $\bar{w}_t^p\propto {w}_t^p$
    \EndFor
    
    \State Return the approximate likelihood $\prod_{t=1}^T \frac{1}{P}
    \sum_{p \in [P]}  w_t^p$
    \end{algorithmic}
\end{algorithm}

\begin{algorithm}[t!]
\caption{CAL within block bootstrap particle filter for ``local authority'' overdispersion} \label{alg:overdispersion_local}
    \begin{algorithmic}
    \Require $\bW, \bY_{1:T}, p_0(\cdot), K_{\cdot}(\cdot), G(\cdot), P$
    \State Initialize $\bpi_{n,0}^P$ with $p_0(\bw_n)$ for all $n \in [N]$ and $p \in [P]$
    \For{$t \in 1,\dots,T$}
        \State $\bPi^p_{t-1} = (\bpi^p_{1,t-1}, \dots, \bpi^p_{N,t-1})$ for all $p \in [P]$
        \For{$B \in$ local authorities}
            \State Sample $\xi_{B,t}^p$ from the prior for all $B \in$ local authorities and $p \in [P]$
            \For{$n \in B$}
                \State $\widetilde{\boeta}^p_{n,t-1} = \eta(\bw_n,\bW,\bPi^p_{t-1})$
                \State $\bpi^p_{n,t|t-1} =  \left [ (\bpi^p_{n,t-1})^{\top} K_{\widetilde{\boeta}_{n,t-1}}(\bw_n,\xi_{B,t}^p) \right ]^{\top}$ 
                \State $\bmu_{n,t}^p = \left [ (\bpi_{n,t|t-1}^p)^\top  G(\bw_n) \right ]^\top$
                \State $\bpi_{n,t}^p = \bpi_{n,t|t-1}^p \odot \left \{  \left [  G(\bw_n) \oslash  \left ( 1_M (\bmu_{n,t}^p)^\top \right ) \right ] \by_{n,t} \right \}$  
            \EndFor
            \State Set $w_{B,t}^p = \prod_{n \in B} \by_{n,t}^\top \bmu_{n,t}^p$
            \State Resample the particles $\xi_{B,t}^p,(\bpi_{n,t}^p)_{n \in B}$ according to $\bar{w}_{B,t}^p\propto {w}_{B,t}^p$
        \EndFor
    \EndFor
    
    \State Return the approximate likelihood $\prod_{t=1}^T \prod_{B} \frac{1}{P}
    \sum_{p \in [P]} w_{B,t}^p$
    \end{algorithmic}
\end{algorithm}

\subsubsection{Dealing with overdispersion}

Suppose that we want to include overdispersion in our foot-and-mouth model. This can be done by considering some stochastic parameters when the interaction term enters $P_n^I$.

\paragraph{Model overdispersed.} We consider the model as above but we set:
$$
P^I_n \coloneqq 1-\exp\left ( - h \beta \xi_{n,t} \exp \{ \bc_n^\top \bb_S \} \boeta_{n,t-1}^I -h \epsilon\right )
$$ 
where $\log \xi_{n,t}$ is Gaussian with mean $\mu_o$ and standard deviation $\sigma_o$. Here $\xi_{n,t}$ can be common across all the individuals, hence $\xi_t$, or across local authorities, hence $\xi_{B,t}$, resulting in what we call ``shared'' overdispersion and ``local authority'' overdispersion. $K_{\cdot}(\cdot)$ now depends on the stochastic parameters, meaning that we denote it with $K_{\eta_{n,t-1}^I,\eta_{n,t-1}^R}(\bw_n,\xi_{t})$ for ``shared'' overdispersion and with $K_{\eta_{n,t-1}^I,\eta_{n,t-1}^R}(\bw_n,\xi_{B,t})$ for ``local authority'' overdispersion. To make the notation lighter we refer to $K_{\eta_{n,t-1}^I,\eta_{n,t-1}^R}(\bw_n,\cdot)$ with just $K_{\eta_{n,t-1}}(\bw_n,\cdot)$.

\paragraph{Likelihood computation.} To provide likelihood estimate we nest the CAL within an SMC \citep{whitehouse2023consistent}. Precisely, we consider a bootstrap particle filter for ``shared'' overdispersion and a block bootstrap particle filter for ``local authority'' overdispersion. Here a particle approximation of $\xi_{n,t}|\bY_{1:t}$ is provided recursively via the SMC, where $\xi_{n,t}$ might be shared across different individuals. Algorithm \ref{alg:overdispersion} provides the pseudo-code for the bootstrap particle filter for the ``shared'' overdispersion, while Algorithm \ref{alg:overdispersion_local} provides the pseudo-code for the block bootstrap particle filter for the ``local authority'' overdispersion. 

\paragraph{Maximum CAL estimation.} Ideally we would like to optimize the parameters based on the likelihood approximations from both Algorithm \ref{alg:overdispersion} and Algorithm \ref{alg:overdispersion_local}. However, the current version of the algorithms is not suitable to automatic differentiation and it be JIT compiled only within a time step and not across time steps because of the resampling procedure. For computational reasons we hence optimize the parameters $\mu_o,\sigma_o$ on a grid while keeping the other parameters fixed to the maximum CAL estimator obtained without overdispersion. We decided to optimize also $\mu_o$ to ensure that shifts in the transmission rate are also possible if required. We consider $\mu_o \in \{ -2, -1, -0.5, -0.25, 0, 0.25, 0.5, 1, 2\}$ and $\sigma_o \in \{ 0.05, 0.125, 0.25, 0.5, 1, 2, 5\}$, and we get $\mu_o=0,\sigma_o=0.25$ as maximum on the grid, see Figure \ref{fig:ovd_grid} for a graphical illustration. We then consider $\mu_o=0,\sigma_o=0.25$ and rerun Algorithm \ref{alg:overdispersion} 100 times to get an estimate of the Monte Carlo variability. As we reported in the main manuscript we get a log-likelihood estimate of $-17926.988 \pm 0.4619721$. Considering Algorithm \ref{alg:overdispersion_local}, we set again $\mu_o=0,\sigma_o=0.25$ and run the algorithm 100 times. Here we obtain $-17907.555 \pm 0.835302$ as a log-likelihood estimate.

\begin{figure}[httb!]
    \centering
    \includegraphics[width=0.75\linewidth]{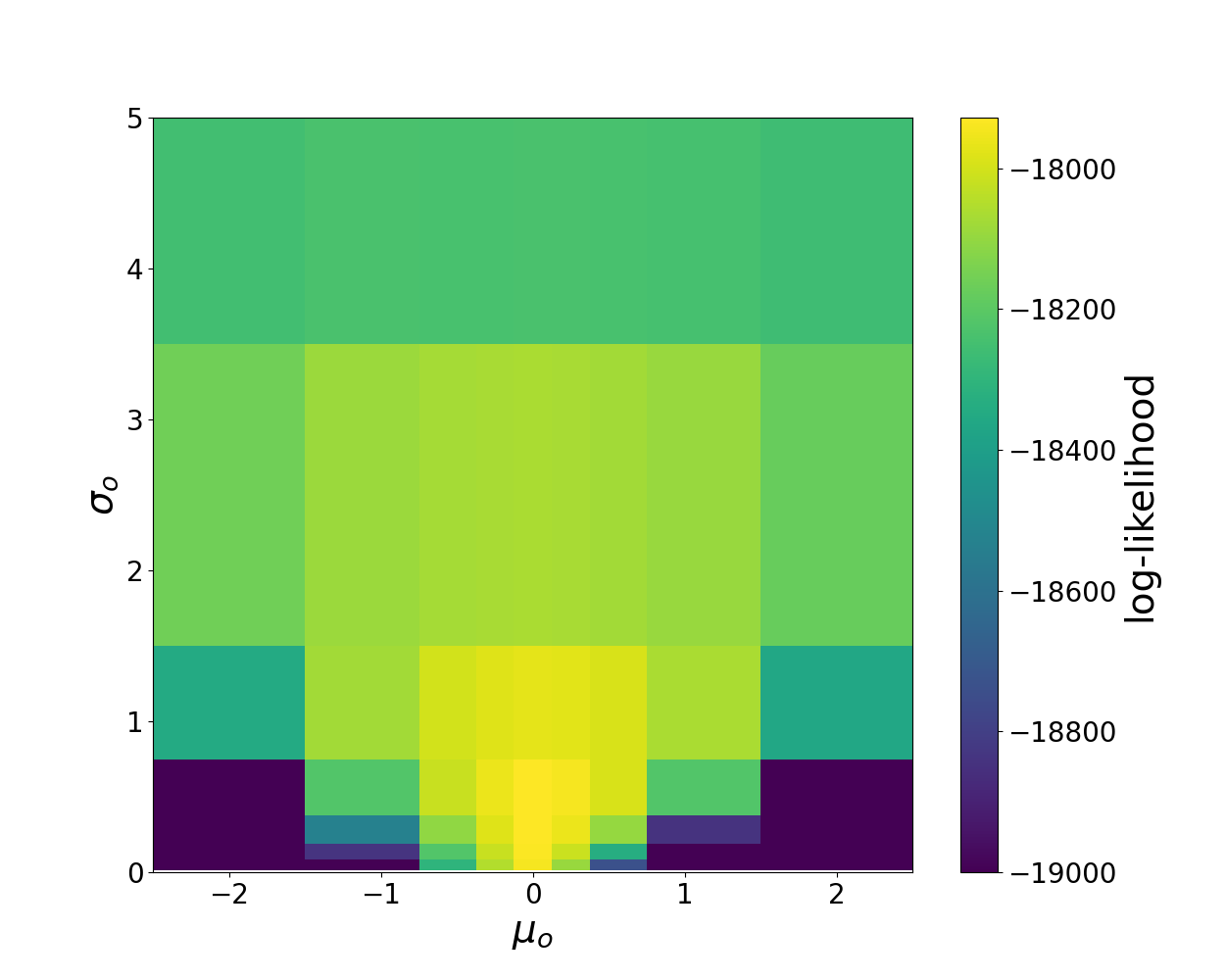}
    \caption{Log-likelihood estimates from Algorithm \ref{alg:overdispersion} on a grid of $\mu_o,\sigma_o$.}
    \label{fig:ovd_grid}
\end{figure}


\end{document}